\documentclass[utf8, 12pt]{article}

\usepackage{graphicx} 
\usepackage[left=2.5cm, right=2.5cm, top=2.5cm, bottom=2.5cm]{geometry}
\usepackage{setspace} 
\usepackage{amssymb} 
\usepackage{amsmath}
\usepackage{amsthm}
\usepackage{hyperref} 
\usepackage{booktabs} 
\usepackage{subfig} 
\usepackage[table]{xcolor} 
\usepackage{placeins} 
\usepackage{xspace} 
\usepackage{algorithm} 
\usepackage{algpseudocode} 

\usepackage[T1]{fontenc}
\usepackage[utf8]{inputenc}
\usepackage{lmodern}

\usepackage{tikz} 
\usetikzlibrary{decorations.pathreplacing,positioning, arrows.meta}

\usepackage[
    backend=biber,
    maxcitenames=2,
    mincitenames=1,
    maxbibnames=99,
    style=authoryear,
    natbib=true,
    url=true, 
    uniquename=false,
    uniquelist=false,
    doi=true,
    eprint=false
]{biblatex}
\newtheorem{theorem}{Theorem}[section]
\newtheorem{definition}[theorem]{Definition}
\newtheorem{corollary}[theorem]{Corollary}

\newcommand{\ind}{\perp\!\!\!\!\perp}

\newcommand{\minNumberChildrenByNurse}{100\xspace}
\newcommand{\numberUniqueNurses}{112\xspace} 
\newcommand{\numberUniqueNursesWithDistrictInfo}{111\xspace}
\newcommand{\shareRecordToDistrictMatch}{43.1-98.4\%\xspace}

\title{Optimal Treatment Allocation under Constraints\thanks{Acknowledgements: We thank Nadja van 't Hoff, Miriam Wüst, Lise Geisler Bjerregaard, Christian Møller Dahl, Andreas Bjerre-Nielsen, and seminar participants at University of Southern Denmark for helpful comments; Katrine Lindholm Nicolaisen for discussions on network flows; Christian Møller Dahl for his work on image segmentation; and Toru Kitagawa and Michael Lechner for discussions on optimal policy. We gratefully acknowledge financial support from the Independent Research Fund Denmark (8106-00003B). 
In this project, we transcribe and use data from the “Copenhagen Infant Health Nurse Records” that was established by the Center for Clinical Research and Prevention, the Capital Region of Denmark, the University of Copenhagen, and the University of Southern Denmark. 
We thank Thorkild I.A. Sørensen for his contributions to the establishment of the CIHNR.}}
\author{Torben S. D. Johansen\thanks{Department of Economics, University of Southern Denmark, Denmark, \texttt{tsdj@sam.sdu.dk}.}}
\date{April 2024}

\begin{document}
\maketitle


\begin{abstract}
    \noindent
    In optimal policy problems where treatment effects vary at the individual level, optimally allocating treatments to recipients is complex even when potential outcomes are known.
    We present an algorithm for multi-arm treatment allocation problems that is guaranteed to find the optimal allocation in strongly polynomial time, and which is able to handle arbitrary potential outcomes as well as constraints on treatment requirement and capacity. 
    Further, starting from an arbitrary allocation, we show how to optimally re-allocate treatments in a Pareto-improving manner.  
    To showcase our results, we use data from Danish nurse home visiting for infants.
    We estimate nurse specific treatment effects for children born 1959-1967 in Copenhagen, comparing nurses against each other.
    We exploit random assignment of newborn children to nurses within a district to obtain causal estimates of nurse-specific treatment effects using causal machine learning.
    Using these estimates, and treating the Danish nurse home visiting program as a case of an optimal treatment allocation problem (where a treatment is a nurse), we document room for significant productivity improvements by optimally re-allocating nurses to children.
    Our estimates suggest that optimal allocation of nurses to children could have improved average yearly earnings by USD 1,815 and length of education by around two months. 
\end{abstract}


\newpage
\section{Introduction}

Across a number of settings decision makers wish to make informed choices regarding provision of treatments to recipients in a way that maximizes some score.\footnote{Our use of treatment is broad and may include medical treatments, adverting, provision of services such as education, and more. Likewise, our use of recipient is broad and may include patients, customers, students, and more.}
Examples include \citet{kim2011battle, chan2012optimizing, ozanne2014development, athey2017beyond, swaminathan2017off, fukuoka2018objectively, zhou2018evaluating, bastani2018interpreting, ban2019big, farias2019learning, bastani2020sequential, bertsimas2020predictive, bastani2020online}, ranging from healthcare to digital advertising and public policies.
A common theme across these settings is heterogeneous treatment effects,\footnote{For example: Given two patients, one might respond better to drug A while the other responds better to drug B.} in turn leading to benefits from ``personalized treatment''.\footnote{This is also known as ``personalized medicine'' \citep{jain2002personalized, hamburg2010path, schork2015personalized}, and can be seen as a subset of the broader setting of personalized treatment, by which is meant a broader definition of treatments than exclusively medical treatments.}
However, to exploit this information requires insight into (estimates of) potential outcomes of recipients under different treatments and a policy mapping that, based on potential outcomes, is able to optimally allocate treatments to recipients, potentially while respecting constraints related to recipient treatment rights, treatment capacities, and equity concerns.

In this paper, we tackle the issue of optimally allocating treatments to recipients under constraints.
Related to our work is the literature on policy learning \citep{dudik2011doubly, zhang2012estimating, zhao2012estimating, zhao2015doubly, swaminathan2015batch, zhou2017residual, kallus2018confounding, kitagawa2018should, athey2021policy, zhou2023offline}.
While most of the existing literature focuses on binary treatments, exceptions include \citet{swaminathan2015batch, zhou2017residual, kallus2018balances, zhou2023offline}.
What sets this study apart, however, is our focus on optimal treatment allocation, and we attempt to solve the computational challenges related to optimal policy design given potential outcomes that arise due to large numbers of treatments and recipients and various constraints related to treatment allocation.
As such, our contribution is in developing a general method for finding policies that exactly maximize some objective function within a class of policies.\footnote{Or minimize some objective function, depending on whether a larger value signifies a better outcome.}
Thus, our work is related to the task of finding the policy that exactly minimizes the approximate value function within the class of tree-based policies discussed in \citet{zhou2023offline}, but where \citet{zhou2023offline} use a general mixed integer program (MIP) which quickly becomes computationally infeasible, our approach, based on a graph representation of the problem, is more suitable for larger problems.\footnote{Recall we abstract from any complications related to forming an approximate value function using estimators, and thus our objective is exclusively related to finding the optimal policy given a value function. \citet{zhou2023offline} consider the more general task of multi-action policy learning with observational data. Our contribution could help with certain aspect of this problem but is not an alternative to the full task.}

Our contribution is in proving that the general case of optimally allocating treatments to recipients taken potential outcomes as given is strongly polynomial, even when (1) treatment effects are allowed to vary arbitrarily between recipients, (2) there are arbitrarily many distinct types of treatments available, and (3) each treatment must respect capacity constraints.
In an extension, we show that our results apply also when only Pareto-improvements are allowed, i.e., in settings where a re-allocation of treatments is not allowed to make any individual worse off.\footnote{Our approach lends itself well to extensions such as differential capacities by treatment and respecting constraints related to certain types of treatments not being eligible for some recipients.}
For our proofs, we represent the \textit{optimal treatment allocation} problem as a \textit{flow} problem in an appropriately constructed \textit{network}.\footnote{See the textbooks by \citet{ford1962flows, ahuja1993network, bang2008digraphs} for introductions to flows in networks.}
In addition to proving our complexity results, this representation also lends itself well to the construction of an algorithm for solving the problem.

To showcase our method, we consider nurse home visiting (NHV) in Denmark, where we consider each individual nurse a treatment and each child a recipient.
NHV for infants started in 1937 in Denmark, and previous work has documented large, positive impacts of NHV and center care in Denmark \citep{wust2012early, hjort2017universal} as well as in Norway \citep{buetikofer2019infant}, Sweden \citep{bhalotra2017infant}, and the US \citep{hoehn2021long}.
We combine a rich dataset consisting of nurse journals following infants during their first year of life with Danish administrative data, providing information on education and labor market outcomes.\footnote{In joint, concurrent work, \citet{baker2023universal, bjerregaard2023cohort} have transcribed the contents of the nurse records using machine learning, based mostly on manual transcriptions from \citet{andersen2012weight, bjerregaard2014bmi}.} 
The nurse journals originate from the Danish NHV program and cover all children born between 1959-1967 and living in Copenhagen during their first year of life.
They contain information on birth characteristics (including birth weight, birth length, and preterm status), detailed data from visits taking place at child age two weeks and one, two, three, four, six, nine, and 12 months (including weight development, nutrition, and home and parent characteristics), and, crucially for this study, the name of the nurse performing the visits for each child.
Linking the journals to Danish administrative data, we can follow the development of these children until around age 50-60 (depending on cohort and register), from which we obtain information on education and labor market outcomes.

The NHV setting lends itself well to our proposed method:
First, using causal machine learning combined with detailed pre-treatment data for each child and her family, following children throughout their lives using Danish register data, and exploiting quasi-random allocation of nurses to children within nurse districts (and years), we can obtain individual-level treatment effects for all children of the Copenhagen NHV program, in turn allowing us to obtain estimates of potential outcomes under different treatments at the individual level.
Second, the setting is complex in that there are large numbers of distinct treatments and recipients, each treatment is capacity-constrained (for example, allocating all children to the same nurse is not a valid solution), each recipient must be treated by exactly one treatment (exactly one nurse allocated each child), and we may want to impose Pareto-improvement constraints to make sure a re-allocation does not leave any child worse off.\footnote{Or even that \textit{some} children must not be left worse off, e.g., those from families at the bottom of the socioeconomic distribution.}

We start by documenting that while the initial conditions of children (as measured by at-birth characteristics, such as birth weight, and parent characteristics) vary significantly between nurses, this is largely driven by differences between \textit{nurse districts} as well as birth years.
Once we remove these effects, differences in pre-treatment variables are much less pronounced, supporting our identification strategy relying on quasi-random allocation of nurses to children within district-by-year groups.
However, clear differences in outcomes remain between children to whom different nurses were allocated.
This suggests that nurses vary in terms of their treatment effects.\footnote{We also directly estimate differences in nurse treatment effects, results which are of independent interest. As we show later, the average (absolute) difference (beyond that which can occur due to noise) in length of education between groups of children for two separate nurses is around one month and in average income during ages 25-50 about 1.4\%.} 

We next turn to measuring child-specific treatment effects for each nurse.
Here, we once again exploit the quasi-random allocation of nurses to children within district-by-year groups, and now estimate heterogeneous treatment effects using causal machine learning while exploiting our information on pre-treatment child and family characteristics.\footnote{Specifically, we use multi-arm causal forest \citep{wager2018estimation, athey2019generalized, nie2021quasi} combined with child and family characteristics from the nurse records transcribed using machine learning, as well as information obtained from Danish registers by linking the transcribed nurse records to the registers.}
This allows us to obtain estimates of potential outcomes under different treatments, and using these we are able to estimate counterfactual distributions of child outcomes under different treatment allocation policy rules.
Using our method for finding optimal policy rules, we are able to estimate the maximum benefit from re-allocating nurses in a more efficient manner, as well as able to estimate the share of total theoretical benefits from re-allocation attained through other policy rules.
Our estimates suggest that optimal allocation of nurses to children could have improved average yearly earnings by USD 1,815 (4\%) and total length of education by around two months. 

We contribute to two broad strands of literature.
First, we contribute to the literature on optimal policy learning \citep{manski2004statistical, hirano2009asymptotics, dudik2011doubly, zhang2012estimating, zhao2012estimating, zhao2015doubly, hirano2020asymptotic, athey2021policy, zhou2023offline}, specifically on the problem of finding optimal policies taken potential outcomes as given. 
Finding optimal policies is important in the sense of optimally allocating resources and thus achieving full efficiency, but the complex allocation mechanism to achieve this might not always be desirable; however, even in cases where simpler allocation mechanisms are desired, comparing such mechanisms with optimal rules has the benefit of allowing policy makers to correctly weigh the benefits of simplicity against its costs in terms of efficiency losses.
In doing so, we also showcase how results from graph theory, and specifically network flows \citep{ford1956network, ford1962flows, murty1992network, ahuja1993network, dolan1993networks}, are helpful in solving complex economic problems in settings where these can be represented by appropriately constructed graphs.

Second, we contribute to the literature on the impact of early-life policies, and in particular the role of treatment providers.
The role of early-life circumstances on short- and long-run outcomes of individuals has long been an active area of research across disciplines, including economics \citep{Forsdahl1979, AlmondCurrieDuque2018}.
Within the realm of research investigating the enduring significance of early-life health policies, two primary currents emerge:\footnote{In addition to studies on the positive role of early-life investments, other studies have investigated the negative effects of early-life shocks, such as extreme weather events, and their interactions with early-life investments policies \citep{adhvaryu2018helping, duque2018early, garg2020temperature, aguilar2022nino}.}
The first strand consists of randomized trials that implement high-intensity targeted model programs such as the U.S. Nurse Family Partnership and the Perry Preschool Program, interventions which have underscored the substantial and lasting dividends reaped from precisely targeted investments in the well-being and development of socioeconomically disadvantaged children \citep{Oldsetal1986, Oldsetal1998, olds2019prenatal, belfield2006high, heckman2010rate}. 
The second strand exploits naturally occurring variations in access to early-life health policies, and noteworthy examples include diverse interventions such as nutritional and income support programs \citep{hoynes2016long, barr2022investing, bailey2023safetynet, barr2023fighting}, healthcare insurance and services \citep{wherry2018childhood, goodman2018public, miller2019long, noghanibehambari2022intergenerational, east2023multi}, early educational initiatives \citep{rossin2020preschool, bailey2021prep, anders2023effect}, as well as infant home visiting and center-based care \citep{hjort2017universal, bhalotra2017infant, buetikofer2019infant, hoehn2021long}.\footnote{In concurrent work, \citet{baker2023universal} combine some of the strengths of these two strands by evaluating a large-scale government trial of NHV in Copenhagen during the 1960s that quasi-randomly allocated some children to three vs. one year of NHV. This combines individual-level data on early-life circumstances with long-run administrative data on outcomes.}
Where we depart from much of this existing literature is on our focus on the role of treatment providers, rather than the effect of a program.
Understanding the role of providers is likely a key element in improving our understanding of optimal design and targeting of such programs, and we contribute with novel evidence on the magnitude of between-provider treatment effects and heterogeneity with respect to child characteristics.
Key to facilitating our application is the development and use of machine learning (ML) for handwritten text recognition, which we use to transcribe the handwritten nurse records and link those with administrative data.\footnote{In concurrent work, \citet{bjerregaard2023cohort, baker2023universal} transcribe handwritten nurse records from the 1960s Copenhagen NHV program and link those to Danish administrative data to construct a large cohort of children (92,902) to track from birth throughout their life. \citet{baker2023universal} use this data to study the impact of extended NHV, exploiting a 1960s trial that quasi-randomly allocated children to three vs. one year of NHV. This cohort also forms the population of children we use for our application.}
In doing so, we also contribute to the literature on layout detection \citep{clinchant2018tables, shen2021layoutparser, dahl2023bdad, dahl2023tableparser} and scene, optical character, and handwritten text recognition \citep{goodfellow2013multi, bluche2014htr, lee2016recursive, bluche2017scan, bartz2021htr, geetha2021htr, kang2022htr, dahl2022dare, dahl2023hana}.

The paper proceeds as follows:
In Section~\ref{sec: allocation}, we derive our allocation algorithm for solving the optimal allocation problem.
Section \ref{sec: background} then presents the institutional background of the 1960s Copenhagen infant NHV program and Section \ref{sec: data} presents our data.
Section \ref{sec: empirical-methods} describes our approach for estimation and Section \ref{sec: results} presents our results.
Section \ref{sec: conclusion} concludes.


\section{Optimal Allocation}
\label{sec: allocation}


We define the \textit{optimal allocation problem} as the problem of optimally allocating treatments to recipients (where a number of treatments are available), under constraints regarding recipients (receive exactly one treatment) and treatments (at most $m$ recipients can receive given treatment).\footnote{This encompasses the nurse allocation problem discussed in the paper, but it is not specific to it and may encompass a host of other such types of problems. It is straightforward to extend the problem to settings where capacities vary by treatments.}
Formally, we define the problem as:

\begin{definition}[Optimal allocation problem]
    \label{def: optimal-allocation}
    Let $Y_i^j \in \mathbb{R}$ denote the outcome of recipient $i$ under treatment $j$ and let $D_i^j \in \{0, 1\}$ be one if recipient $i$ receives treatment $j$.
    The \textit{optimal allocation problem} is then the problem of choosing $D_i^j$ such that the maximum average realized outcome is achieved, while respecting that each recipient must be allocated exactly one treatment and no treatment is allocated to more individuals than its capacity allows.
\end{definition}

Noting that $D_i^j\in\{0, 1\}$ is $1$ if individual $i$ is under treatment $j$ and $0$ otherwise, we see that assignment of everyone to exactly one treatment ($\sum_j D_i^j = 1$) and respecting capacities ($\sum_i D_i^j \leq m$) leads to a binary program:\footnote{If capacities of different treatments vary, we may instead write $\sum_i D_i^j \leq m_j$ as the second set of conditions. Otherwise the problem remains unchanged.}

\begin{align*}
    \underset{D^j_i}{\arg\max} & \, \sum_i \sum_j Y^j_iD^j_i
        \\
        \text{such that} \, \, \, & \sum_j D_i^j = 1 \,\,\, \forall \, i \in \{1, \dots, n\}
        \\
        & \sum_i D_i^j \leq m \,\,\, \forall \, j \in \{1, \dots, k\},
\end{align*}

\noindent
where $n$ is the number of recipients, $k$ the number of treatments, and $m$ the ``capacity'' of each treatment, i.e., the maximum number of recipients a given treatment may be allocated to.
The solution to this binary program is exactly the optimal solution to the allocation problem of Definition~\ref{def: optimal-allocation}.

We claim there exists a strongly polynomial algorithm for solving the allocation problem (despite its potential $\mathcal{NP}$-completeness given its structure as a binary program).
For our proof, we shall make use of \textit{flows in networks} by appropriately constructing a network and representing the optimal allocation problem as an \textit{integer minimum cost feasible flow}-problem in the constructed network.\footnote{See the textbooks by \citet{ford1962flows, ahuja1993network, bang2008digraphs} for introductions to flows in networks.}

Let $V_1$ be the set of treatments and $V_2$ the set of recipients, and let $s$ and $t$ be special vertices.
Define now $V := V_1 \cup V_2 \cup \{s, t\}$ as the set of vertices of a digraph and let $A := \{sv_1 \, : \, v_1 \in V_1\} \cup \{v_1v_2 \, : \, v_1 \in V_1, v_2 \in V_2\} \cup \{v_2t \, : \, v_2 \in V_2\}$ be the arcs of the digraph, i.e., $D = (V, A)$.
Define now the network $\mathcal{N} := (V, A, l, u, b, c)$, where $b: V \rightarrow \mathbb{R}$ is a function on the vertices and $l: A \rightarrow \mathbb{R}$, $u: A \rightarrow \mathbb{R}$, and $c: A \rightarrow \mathbb{R}$ are functions on the arcs.
By appropriately choosing $b$ and $l$, $u$, and $c$, it turns out that an \textit{integer minimum cost feasible flow} in the network exactly corresponds to the optimal solution of the allocation problem.
Formally, we define the network as:

\begin{definition}[Network representation]
    \label{def: network-representation}
    Let $\mathcal{N} := (V, A, l, u, b, c)$ be defined as above and let, for all $v_1 \in V_1$ and $v_2 \in V_2$, $l_{v_1v_2} = 0$, $u_{v_1v_2} = 1$, and $c_{v_1v_2} = -Y_{v_2}^{v_1}$.\footnote{In some algorithmic settings non-negative costs are desired. This is easily guaranteed by adding $\max \{Y_i^j\}$ to all costs (leaving the problem unchanged). We refrain from this for notational simplicity.}
    Further, let, for all $v_1 \in V_1$, $l_{sv_1} = 0$, $u_{sv_1} = m$, and $c_{sv_1} = 0$ and let, for all $v_2 \in V_2$, $l_{v_2t} = 1$, $u_{v_2t} = 1$, and $c_{v_2t} = 0$.
    Finally, let $b(v) = 0$ for all $v \in V \setminus \{s, t\}$ and $b(s) = -b(t)$.\footnote{This, in turn, implies that $\sum_{v \in V}b(v) = 0$.}
    We shall call the network constructed this way the \textit{network representation} of an instance of the \textit{optimal allocation problem} of Definition~\ref{def: optimal-allocation}.
\end{definition}

A \textit{flow} in a network is a function on the arcs $x: A \rightarrow \mathbb{R}_+$, and we call it an integer flow if $x: A \rightarrow \mathbb{N}_0$.
For a given flow, the balance vector of $x$ is 
\begin{align*}
    b_x(v) = \sum_{vw \in A}x_{vw} - \sum_{uv \in A}x_{uv} \quad \forall \, v \in V
\end{align*}

We say a flow $x$ is \textit{feasible} if, for all $uw \in A$, $l_{uw} \leq x_{uw} \leq u_{uw}$ and, for all $v \in V$, $b(v) = b_x(v)$.
Further, we define the \textit{cost} of a flow as $\sum_{uw \in A}c_{uw}x_{uw}$.
We are now in a position to prove the correspondence between the optimal allocation problem (Definition~\ref{def: optimal-allocation}) and its network representation (Definition~\ref{def: network-representation}).

\begin{theorem}
    \label{thm: correspondance}
    Any integer feasible flow $x$ in $\mathcal{N}$ defined as in Definition~\ref{def: network-representation} corresponds to a valid solution to the treatment allocation problem defined as in Definition~\ref{def: optimal-allocation}.
    Further, the cost of any feasible flow $x$ is exactly the negative value of its corresponding value in the treatment allocation problem.
\end{theorem}

\begin{proof}
    We sketch the main points of the proof here.
    Appendix \ref{sec: extended-proof} provides further details.

    Any integer feasible flow $x$ in $\mathcal{N}$ makes use of at most $m$ arcs from each $v_1 \in V_1$ (i.e., respects capacities) and exactly one arc from each $v_2 \in V_2$ (i.e., allocates exactly one treatment to each recipient).
    Let $D_{v_2}^{v_1} = x_{v_1v_2}$ for all $v_1 \in V_1$ and $v_2 \in V_2$.
    Then the set of $D_{v_2}^{v_1}$s satisfies the constraints of the binary program formulation of Definition~\ref{def: optimal-allocation} and is thus a valid solution to the allocation problem, with cost $\sum_{uw \in A}c_{uw}x_{uw} = -\sum_{v_2 \in V_2}\sum_{v_1 \in V_1}Y_{v_2}^{v_1}D_{v_2}^{v_1}$.
\end{proof}

It follows straightforwardly that the allocation problem has a solution if and only if an integer feasible flow exists in its network representation.
Further, any integer minimum cost feasible flow in an allocation problem's corresponding network representation corresponds to the optimal value of the allocation problem.

\begin{corollary}
    \label{cor: existence}
    The allocation problem as defined in Definition~\ref{def: optimal-allocation} has a solution if and only if $m|V_1| \geq |V_2|$, and any integer minimum cost feasible flow in its network representation (Definition~\ref{def: network-representation}) corresponds to an optimal solution of the associated allocation problem.
\end{corollary}

\begin{proof}
    Clearly, $\sum_j D_i^j = 1 \implies \sum_i\sum_j D_i^j = |V_2|$ and $\sum_iD_i^j \leq m \implies \sum_i\sum_j D_i^j \leq m|V_1|$.
    Hence, $m|V_1| \geq |V_2|$.
    
    Since, by Theorem~\ref{thm: correspondance}, the cost of any integer feasible flow in the network representation of the allocation problem is equal to the negative value of the optimal solution to the allocation problem, finding an integer minimum cost feasible flow in the allocation problem's network representation corresponds to finding an optimal solution to the allocation problem.
\end{proof}

To prove that a strongly polynomial algorithm for solving the optimal allocation problem exists, it suffices to show that an integer minimum cost feasible flow in its corresponding network can be found in strongly polynomial time.\footnote{It also requires that the network representation of an allocation problem can be found in strongly polynomial time, but this may be easily verified.}

\begin{theorem}
    \label{thm: complexity}
    Let $\mathcal{N} := (V, A, l, u, b, c)$ denote the network representation (Definition~\ref{def: network-representation}) of the allocation problem (Definition~\ref{def: optimal-allocation}) and let $n_1$ denote the number of treatments and $n_2$ the number of recipients.
    Given the existence of an integer feasible flow in $\mathcal{N}$, an integer minimum cost feasible flow $x$ in $\mathcal{N}$ (and, in turn, an optimal solution to the allocation problem) can be found in time $\mathcal{O}\left[(n_1^3n_2^2 + n_1^2n_2^3) \log^2 (n_1 + n_2)\right]$. 
\end{theorem}

\begin{proof}
    We sketch the main points of the proof here.
    Appendix \ref{sec: extended-proof} provides further details.

    Let $n = |V|$ denote the number of vertices and $m = |A|$ the number of arcs in the underlying digraph $D$ of $\mathcal{N}$.
    Noting that capacities (and lower bounds and balance vectors) are integers, using the cancel-and-tighten algorithm of \citet{goldberg1989finding} allows us to find a minimum cost feasible flow in $\mathcal{N}$ in time $\mathcal{O}(nm^2 (\log n)^2)$ (even when costs are arbitrary real-valued).

    Noting that the cancel-and-cut algorithm is a variant of the cycle cancelling algorithm \citep{klein1967primal}, we may use a convenient \textit{integrality} property of minimum cost flows, namely that given all integer lower bounds, capacities, and balance vectors, there exists an integer minimum cost flow:
    At each iteration, we augment our flow along the cycle $C$ by $\delta(C)$, by which we mean the minimum residual capacity of any arc on $C$ in $\mathcal{N}$, and thus our new flow after any iteration is $x^{t+1} = x^t \oplus \delta(C)$.
    Since the residual capacity on any arc is always an integer, $x^{t+1}$ is an integer flow so long as $x^t$ is, and thus by induction we have an integer minimum cost flow.

    Noting that $n = 2 + n_1 + n_2$ and $m = n_1 + n_2 + n_1n_2$, we have $\mathcal{O}(n m^2 (\log n)^2) = \mathcal{O}\left[(n_1^3n_2^2 + n_1^2n_2^3) \log^2 (n_1 + n_2)\right]$.
\end{proof}

For our application of allocating nurses to infants, our ultimate goal is to provide guidance for a re-allocation of nurses in such a way as to maximize their ``treatment effects''.
In cases where nurses differentially affect children, i.e., there exists $i_1, i_2, j_1, j_2, i_1 \neq i_2, j_1 \neq j_2$ such that $Y_{i_2}^{j_1} - Y_{i_1}^{j_1} > Y_{i_1}^{j_2} - Y_{i_2}^{j_2}$, re-allocating nurse $j_1$ from child $i_2$ to $i_1$ (and nurse $j_2$ from child $i_1$ to $i_2$) will improve average outcomes.\footnote{However, there is no guarantee that neither child is made worse off, an issue we turn to in Section~\ref{sec: pareto-allocation}.}
While we do not observe all $Y_i^j$, we may use sample estimates of these, as may be obtained by estimating heterogeneous nurse treatment effects by child pre-treatment characteristics.
Thus, our approach is designed to solve allocation problems once (estimates of) potential outcomes are obtained, and is then guaranteed to always be able to find the (non-parametric) optimal allocation of treatments to recipients in strongly polynomial time.

In a setting such as the one above, the number of treatments is required to grow as the number of recipients does, but this might not be the case if treatment is, e.g., a drug or an advertisement campaign.
In such cases we obtain complexity $\mathcal{O}\left[n_2^3 (\log n_2)^2\right]$.



\subsection{Pareto-Improvement Guarantees}
\label{sec: pareto-allocation}

While our approach taken thus far does not guarantee that no child is made worse off, it is easily extendable to cover scenarios where only Pareto improvements are allowed.
Given some ``initial solution'' (allocation of treatments to recipients), we call this new problem the \textit{Pareto-guaranteed optimal allocation problem} and formally define it as:

\begin{definition}[Pareto-guaranteed optimal allocation problem]
    \label{def: pareto-optimal-allocation}
    Let $Y_i^j \in \mathbb{R}$ denote the outcome of recipient $i$ under treatment $j$ and let $D_i^j \in \{0, 1\}$ be one if recipient $i$ receives treatment $j$.
    The \textit{Pareto-guaranteed optimal allocation problem} is then the problem of choosing $D_i^j$ such that the maximum average realized outcome is achieved, while respecting that each recipient must be allocated exactly one treatment, that no treatment is allocated to more individuals than its capacity allows, and that no re-allocation leads to any recipient be made worse off compared to the baseline choices of $D_i^j$.
\end{definition}

We shall make use of a slight modification to the network representation of Definition~\ref{def: network-representation} to prove our extension:

\begin{definition}[Pareto-guaranteed network representation]
    \label{def: pareto-network-representation}
    Let $\mathcal{N} := (V, A, l, u, b, c)$ be defined as in Definition~\ref{def: network-representation}.
    Remove arcs from $\mathcal{N}$ as follows:
    For each recipient $v_2^* \in V_2$, let $\bar{Y}_{v_2^*}^{v_1^*}$ denote the realized outcome under no re-allocation, i.e., nurse $v_1^*$ is allocated $v_2^*$ under no re-allocation.
    Now, for each $v_1 \in V_1$ and $v_2 \in V_2$, remove the arc $v_1v_2$ if $Y_{v_2}^{v_1} < \bar{Y}_{v_2^*}^{v_1^*}$ to construct the network $\mathcal{N}_P$.\footnote{Alternatively, one could set $u_{v_1v_2} = 0$ in these cases.}
    We shall call $\mathcal{N}_P$ constructed this way the \textit{Pareto-guaranteed network representation} of an instance of the \textit{Pareto-guaranteed optimal allocation problem} of Definition~\ref{def: pareto-optimal-allocation}.
\end{definition}

With our modified network definition, we are in a position to prove that our results extend to settings where only Pareto improvements are allowed.

\begin{theorem}
    \label{thm: pareto-solution}
    Let $\mathcal{N_P}$ as defined in Definition~\ref{def: pareto-network-representation} be the Pareto-guaranteed network representation of the Pareto-guaranteed optimal allocation problem of Definition~\ref{def: pareto-optimal-allocation}, and let $n_1$ denote the number of treatments and $n_2$ the number of recipients.    
    Given the existence of an integer feasible flow in $\mathcal{N_P}$, an integer minimum cost feasible flow $x$ in $\mathcal{N_P}$ is guaranteed to make no recipient worse off and to correspond to the (non-parametric) optimal allocation given the constraints.
    The solution can always be found in at most time $\mathcal{O}\left[(n_1^3n_2^2 + n_1^2n_2^3) \log^2 (n_1 + n_2)\right]$. 
\end{theorem}

\begin{proof}
    Under the constraints of Definition~\ref{def: pareto-optimal-allocation}, no flow $x$ with value more than $0$ may pass through one of the arcs deleted when modifying $\mathcal{N}$ to $\mathcal{N}_P$.
    Thus, we may without loss of generality remove these arcs.
    The correspondence of an integer minimum cost feasible flow in $\mathcal{N}_P$ and the optimal solution to the Pareto-guaranteed optimal allocation problem then follows from Theorem~\ref{thm: correspondance}.

    Using Theorem~\ref{thm: complexity} (which puts no restriction on $\mathcal{N}$), we can always find an integer minimum cost feasible flow (if one exists) $x$ in $\mathcal{N}_P$ in at most time $\mathcal{O}\left[(n_1^3n_2^2 + n_1^2n_2^3) \log^2 (n_1 + n_2)\right]$.
\end{proof}

Generally, our approach for solving the optimal allocation problem is readily generalizable to a broad range of alternative settings.
In fact, any alternative formulation of the problem which can be represented by the type of network we introduce is immediately covered (with one such example being the Pareto-guaranteed optimal allocation problem).
Other settings include the possibility of multiple (additive in effect) treatments available for some recipients (by appropriately changing capacities of the type $u_{v_2t}$) and treatments being differentially available to recipients (by appropriately removing arcs of the type $v_1v_2$).



\section{Institutional Background of the 1960s Copenhagen Nurse Home Visiting Program}
\label{sec: background}

In this section, we give an overview of the 1960s Copenhagen NHV program that serves as the background of our empirical application of the results derived in Section~\ref{sec: allocation}.
The 1960s Copenhagen NHV program provides the scene for the study by \citet{baker2023universal} on the role of extended NHV, and we refer the interested reader to that paper for additional details on the program, and in particular the Copenhagen trial on extended NHV.\footnote{The cohort profile we have created as part of our efforts of transcribing and linking the NHV records is described in more detail in \citet{bjerregaard2023cohort}.}

Universal home visiting for families with infants in Denmark has a rich history, dating back to 1937 when the Danish National Board of Health (DNBH) launched a program to address high infant mortality rates.
This initiative aimed to combat infant mortality rates of around six percent in the early 1930s. 
Utilizing staggered introductions across municipalities from 1937 to 1949, previous research has highlighted both short- and long-term health benefits of program participation \citep{wust2012early,  hjort2017universal}. 

By the 1960s, with significant improvements in living conditions and a decline in infant mortality rates to around two percent, the DNBH revised the program to emphasize broader health monitoring and encourage relevant parental health investments \citep{dodelighed19311960}.
This shift in focus aligns with the evolving landscape of early childhood programs, such as the US Head Start program, which also underscored parental investments during toddler years \citep{BarrGibbs2017}. 
With midwife-assisted home births being the norm and limited formal childcare options, interventions in the family home, especially during toddler years, gained prominence.

Amid these developments, the DNBH conducted experiments with extended home visiting, including the ``Copenhagen trial'' studied by \citet{baker2023universal}, which encompassed children born between 1959 and 1967.\footnote{Exploiting quasi-random allocation to a three (vs. baseline one) year NHV program in Copenhagen in the 1960s, \citet{baker2023universal} document positive long-run health (and to a lesser extend employment) effects, and notably significant heterogeneity with respect to characteristics such as birth weight (with low birth weight children much more positively affected than the average child).}
In this trial, nurses offered additional follow-up during existing first-year visits to families residing in Copenhagen. 
While the ``Copenhagen trial'' aimed at evaluating the efficacy of a longer visiting schedule, it also meticulously tracked child development during the baseline period (that is, the first year of a child's life) for everyone living in Copenhagen and born between 1959-1967.
Table~\ref{tab: visit-content-1yr} shows an overview of topics covered at different visits during a child's first year visits.\footnote{The position as an infant health nurse required education beyond that of an ordinary nurse, within maternity care, specialized paediatric care, epidemic or tuberculosis care, or care for patients with mental illness. All infant health nurses had to pass a course at Aarhus University \citep{kuhn1939vejledning, dha1954vejledning, dha1961vejledning}.}
Topics include objective development measurements (e.g., child weight), self-reported parenting decisions (e.g., type of nutrition), and nurse-assessed family characteristics (e.g., socioeconomic status and mother well-being).\footnote{While more than the eight visits indicated in Table~\ref{tab: visit-content-1yr} could take place (with an average of 13 visits per child taking place), structured information collection took place at these eight visits, with the nurse records containing pre-printed fields to be filled in at these specific visits.}

\begin{table}
    \caption{Content of the First Year Home Visiting Program in Copenhagen.}
    \label{tab: visit-content-1yr}
    \rowcolors{2}{gray!25}{white}
    \resizebox{1.0\linewidth}{!}{
        \begin{tabular}{p{0.3\linewidth} p{0.4\linewidth} cccc}
            \toprule
            \rowcolor{white}
            & & \multicolumn{3}{c}{Age of Child} \\ \cline{3-5}
            \rowcolor{white}
            Topic & Example Items & 2 weeks & 1, 2, 3, 4, 6, 9 months & 12 months \\
            \midrule
            Family  &                                              Socioeconomic status, mother mental and physical health &             &                  \checkmark &  \checkmark \\
            Mother labor force &                                           Employment status, childcare status &             &                  \checkmark &  \checkmark \\
            Nutrition &                                                                                                            &             &                             &             \\
            \hspace{3mm}a. &                                                             Infant feeding, number of meals &  \checkmark &                  \checkmark &  \checkmark \\
            \hspace{3mm}b. &                                                                                   Duration of breastfeeding &             &                             &  \checkmark \\        Child development &                                                                                                            &             &                             &             \\
            \hspace{3mm}a. &                                                                    Smiles, lifts head, babbles, sits alone &             &                  \checkmark &  \checkmark \\
            \hspace{3mm}b. &                                                                                                     Height &             &                             &  \checkmark \\
            \hspace{3mm}c. &                                                                                                     Weight &  \checkmark &                  \checkmark &  \checkmark \\
            \hspace{3mm}d. &  Walks, \#teeth, general health assessment, vaccination status, ever hospitalized &             &                             &  \checkmark \\
            
            Child sleeping conditions &                                                                                            Own bed &             &                  \checkmark &  \checkmark \\
            \bottomrule
        \end{tabular}
    }
    \begin{minipage}{1\linewidth}
        \vspace{1ex}
        \footnotesize{
            \textit{Notes:}
            The table shows topics covered in the first-year visits and example items for nurse registrations in the children's records. 
            At each age, more than one nurse visit could be performed (depending on family needs), with an average of around 13 first year visits during the trial (\citeauthor{stadsarkiv}, various years). 
            For each age-specific topic, nurse registrations were made at one of those visits.
            \textit{Source}: Table due to \citet{baker2023universal}.
        }
    \end{minipage}
 \end{table}

While the ``Copenhagen trial'' altered the landscape of child home visiting, it did so in a way unlikely to contaminate the study of differences in the impact of NHV by nurse, namely by quasi-random allocation to the prolonged program by day of the month of birth, assigning everyone born the first three days of any of the months of the nine year trial to the extended schedule.
While the number of visits each child received was reduced as a consequence, in order to compensate nurses for the additional visits to the children enrolled into the extended schedule,\footnote{The average number of first year visits was reduced to around 13, compared to the pre-trial average of around 14 (\citeauthor{stadsarkiv}, various years).} this was done in a ``homogeneous'' way across children and should not interfere with our design. 



\section{Data}
\label{sec: data}

Our study combines two primary sources of data, those being handwritten nurse records from the 1960s and Danish administrative data, to allow us to combine detailed information on childhood development with long-run outcomes from administrative register data, in total following individuals from their birth to when they are around 50-60 years old (depending on their year of birth).
Additionally, we use archive material detailing which nurse district each nurse of the Copenhagen infant nurse program worked in.
Combining these sources, we are able to identify which nurse visited each child and follow that child throughout her first year of life and from her adulthood until current time.



\subsection{Copenhagen Nurse Records}
\label{sec: data-nurse-records}

Nurse records with information on childhood development is available for all children born in Copenhagen between 1959-1967.\footnote{Years before or after are not available, and we speculate that these years were archived due to those being the years of the ``Copenhagen trial'' studied in \citet{baker2023universal}.}
Figure~\ref{fig: nurse-journal} shows the scan of a nurse journal of a child, with parts blackened for confidentially reasons (in the source material we have available, these black patches are not present).

\begin{figure}
    \centering
    \includegraphics[width=1\textwidth]{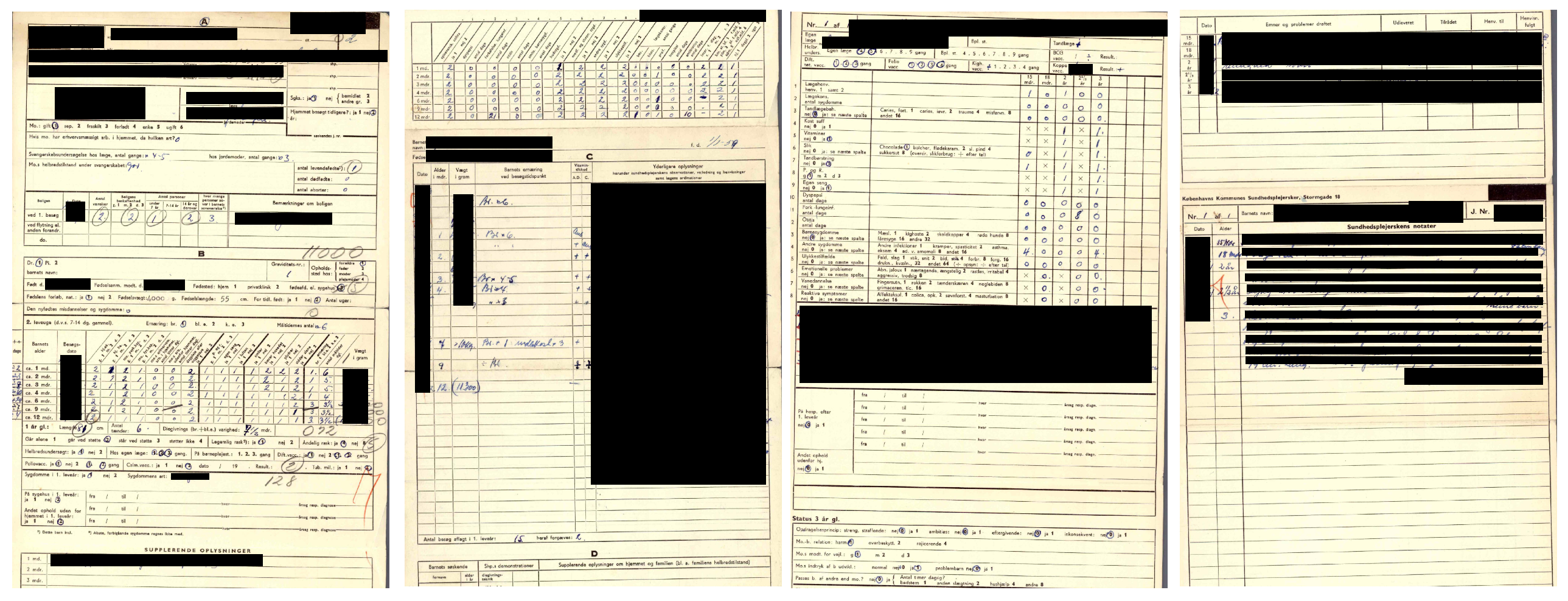}  
    \caption{Sample Nurse Record for a Copenhagen Child.}    
    \label{fig: nurse-journal} 
    \begin{minipage}{1\linewidth}
        \vspace{1ex}
 	  \footnotesize{
            \textit{Notes:}
            The pages depict a scanned nurse record of a child. For confidentiality reasons, parts of the pages are blackened. 
            The first page contains the table for first-year nurse registrations. The second page (flip side) is a page primarily for nurse comments in free text. 
            The third page contains the table for second- and third-year registrations (the ``treatment table''). 
            The final page (flip side) allows for further nurse comments during the second and third year.
            \textit{Source}: Figure due to \citet{baker2023universal}.
 	}
    \end{minipage}
\end{figure}

In joint, concurrent work, \citet{baker2023universal, bjerregaard2023cohort} use ML to transcribe the contents of these records and link them to Danish administrative register data.\footnote{See \citet{dahl2023bdad} for details on an unsupervised ML approach we use to identify types of journal pages and, as a consequence, identify which children were part of the treatment arm of the ``Copenhagen trial''.}
While the Danish unique personal identifier was introduced in 1968 -- i.e., after the birth of all the children of the records -- we nevertheless are able to obtain their personal identification number by using the name and date of birth of the child and her parents, linking them to the Danish Central Person Registry (CPR).

We transcribe the collection of nurse records by using \texttt{timmsn}, a Python library for image-to-text translation \citep{tsdj2023timmsn}.
The library is based on the PyTorch Image Models library by \citet{rw2019timm}, an image recognition Python library.
The full transcription code is available upon request and will be made available open-source at \url{https://github.com/TorbenSDJohansen/cihvr-transcription}.
For the interested reader we refer to Appendix~\ref{sec: transcription-details} for additional details.
We transcribe most of the contents of the nurse records with between 95\%-99\% accuracy, with slightly lower transcription accuracy of around 93\% for nurse names (see Appendix Table~\ref{tab: transc-accs} for details).

\paragraph{Data Linkage and Coverage}
We link the collection of nurse records to administrative data using the unique CPR number (similar to the US Social Security Number) of the children of the records.
Figure~\ref{fig: flowchart} shows the process from the raw, scanned nurse records to the link with outcomes from Danish registers.
While the full collection of scanned nurse records consist of 95,323 documents, some of these we identify as duplicates and others as non-records (e.g., notice of movement), leading to a total of 92,902 records (of which 92,279 contain date of birth, which we need to obtain the children's CPR numbers).
Some records we are unable to link to the Danish CPR, which may be due to poor scan quality or death prior to the establishment of the CPR in 1968, leaving us with 88,808 records we are able to link to the CPR.
Note, however, that 808 of these children do not result in a match with our administrative data (which starts in 1977), primarily due to emigration or death, leaving us with a sample of 88,000 records (children) we can link to administrative outcome data.
For the majority of our analyses, we shall make use of information on the nurse visiting the child and the district of the nurse; constraining the sample to only those records (children) with this information brings the sample size down to 75,318.\footnote{Here, we also restrain the sample to records on which the name of the nurse occur at least \minNumberChildrenByNurse times across the entire collection of nurse records. This is done for two reasons: First, for confidentially reasons we are not allowed to report estimates of small group sizes, and once we combine nurse name information with year and district information, subgroups otherwise become too small. Second, nurse names that occur rarely are more likely to be artefacts from transcription, such as an otherwise valid name having been transcribed slightly incorrectly.}

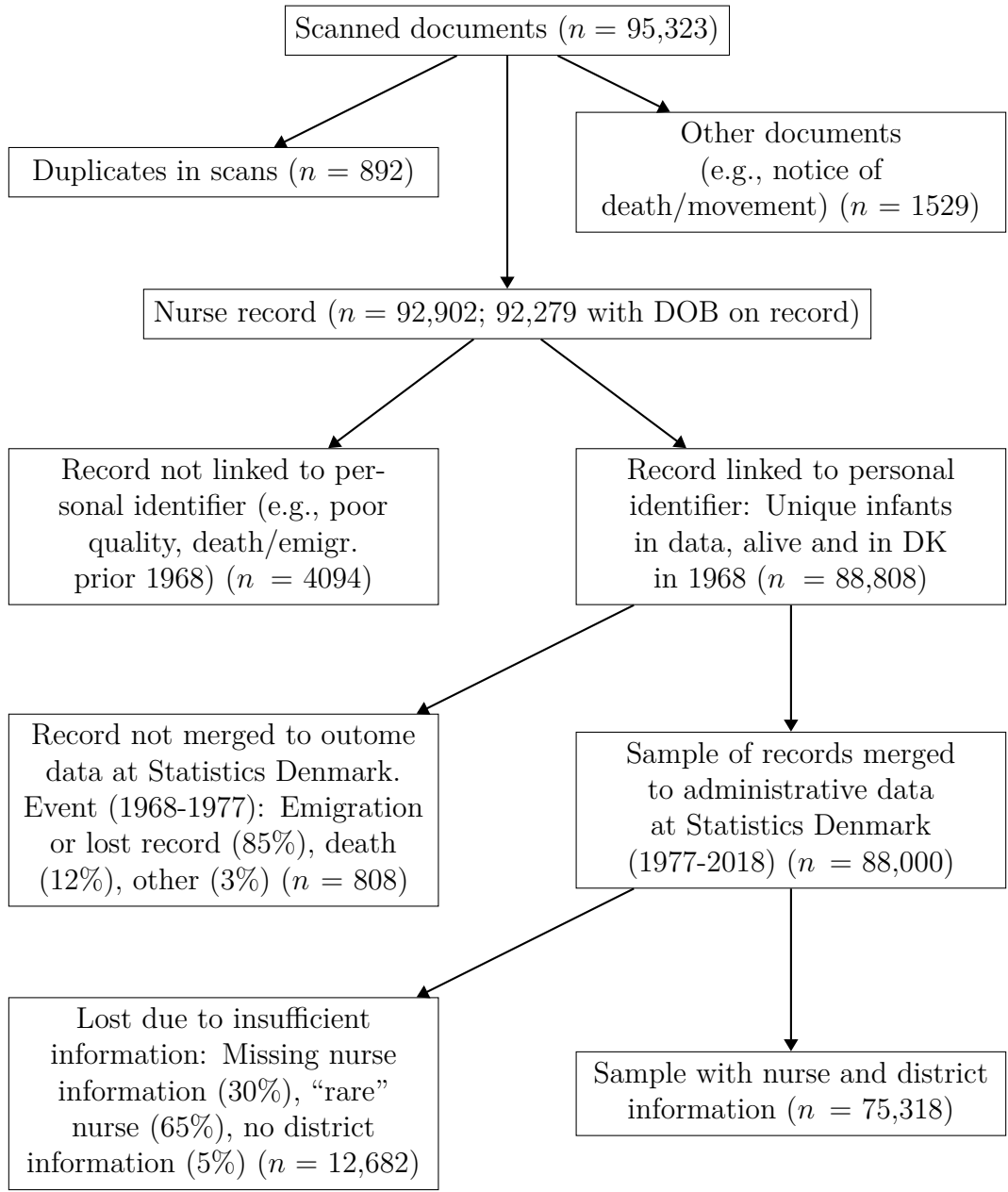
\begin{figure}
    \centering
    \singlespacing
    \begin{tikzpicture}
        \node[rectangle,draw,align=center] at (6, 0) (root) {Scanned documents ($n=$ 95,323)};
    
        \node[rectangle,draw,text width=14em,align=center] at (2, -2) (dupes) {Duplicates in scans ($n= 892$)};
    
        \node[rectangle,draw,text width=14em,align=center] at (10, -2) (deathmove) {Other documents (e.g., notice of death/movement) ($n=$ 1529)};
    
        \node[rectangle,draw,align=center] at (6, -4) (nurserecords) {Nurse record ($n=$ 92,902; 
        92,279 with DOB on record)};
    
        \node[rectangle,draw,text width=14em,align=center] at (2, -7) (notassigned) {Record not linked to personal identifier (e.g., poor quality, death/emigr. prior 1968) ($n=$ 4094)};
    
        \node[rectangle,draw,text width=14em,align=center] at (10, -7) (assigned) {Record linked to personal identifier: Unique infants in data, alive and in DK in 1968 ($n=$ 88,808)};
    
        \node[rectangle,draw,text width=14em,align=center] at (2, -11) (infantnotfound) {Record not merged to outome data at Statistics Denmark. Event (1968-1977): Emigration or lost record (85\%), death (12\%), other (3\%) ($n=$ 808)};
    
        \node[rectangle,draw,text width=14em,align=center] at (10, -11) (sample) {Sample of records merged to administrative data at Statistics Denmark (1977-2018) ($n = $ 88,000)};

        \node[rectangle,draw,text width=14em,align=center] at (2, -15) (nonurse) {Lost due to insufficient information: Missing nurse information (30\%), ``rare'' nurse (65\%), no district information (5\%) ($n = $ 12,682)};
    
        \node[rectangle,draw,text width=14em,align=center] at (10, -15) (withnurse) {Sample with nurse and district information ($n = $ 75,318)};
    
        \draw [thick,-Triangle] (root) to node {} (nurserecords);
        \draw [thick,-Triangle] (root) to node {} (dupes);
        \draw [thick,-Triangle] (root) to node {} (deathmove);
        \draw [thick,-Triangle] (nurserecords) to node {} (notassigned);
        \draw [thick,-Triangle] (nurserecords) to node {} (assigned);
        \draw [thick,-Triangle] (assigned) to node {} (infantnotfound);
        \draw [thick,-Triangle] (assigned) to node {} (sample);
        \draw [thick,-Triangle] (sample) to node {} (nonurse);
        \draw [thick,-Triangle] (sample) to node {} (withnurse);
    \end{tikzpicture}
    \caption{Flowchart: Linkage and Merge of Scanned Records to Administrative Data.}
    \label{fig: flowchart}
    \begin{minipage}{1\linewidth}
        \vspace{1ex}
        \footnotesize{
        \textit{Notes:}
        The figure shows the sample size in terms of the scanned nurse records to the individuals we are able to merge with administrative data.
        Numbers in parentheses indicate sample sizes.
        ``Rare'' nurse refers to nurses whose names occur on fewer than \minNumberChildrenByNurse    records.
        DOB is shorthand for date of birth.
        The sample sizes used for individual analyses might deviate slightly due to less stringent sample requirements due to (1) nurse name information not being important or (2) nurse district information not being important.
 	}
    \end{minipage}
\end{figure}


\paragraph{Child-Nurse Matches}
Our study relies on identifying the specific nurse responsible for visiting each child, and further on being able to identify which nurse district each nurse belonged to, in order for us to be able to compare children within the same nurse district with each other.
To do this, we use our transcriptions of nurse name for each journal, and merge this with archive material on which district each nurse belonged to.
We have been able to find information listing for each nurse the district they served in for the years 1957, 1963, 1965 (two separate accounts), and 1968, and it is this information we use to associate to each child the nurse district to which they belonged.
This is a potential limitation, as nurses may change district over time, and while the data on nurse districts cover the entire time span we consider, we lack information for certain years (if, e.g., a certain nurse in 1964 served in a different district than in 1963 and 1965) and may have incomplete information on all nurses (if, e.g., a nurse was employed only in 1964, thus not showing up in any of the years we observe).

We match nurse records to nurse districts by identifying ``most likely'' matches between the nurse name transcriptions of the collection of nurse records with the archive material listing nurses and their districts.
As the number of different nurses is limited (\numberUniqueNurses), often encountered issues in linking based on names pertaining to non-unique names is limited. 
However, we note the following potential issues:
First, while our transcription accuracy is high (see Appendix Table~\ref{tab: transc-accs}), it is not perfect, and this will result in a number of children with an incorrect nurse name transcription, leading to missed matches.\footnote{However, this is very unlikely to lead to a wrong match, as it would require a name to be transcribed incorrectly \textit{as another name}. Further, we observe that incorrect transcriptions are primarily related to poor scan quality or source material degradation, both of which are unlikely to lead to any systematic bias. As such, we mainly view this weakness as one decreasing our precision, not as one introducing potential bias.}
Second, nurses would not always write their name fully out.
Specifically, we often observe that nurses would write their initial (of their first name) as well as their last name, rather than fully writing out their name.
For this reason, we iteratively match nurses, looking first at exact matches where both first and last name of nurses are fully written, and then gradually relax this requirement over six steps.\footnote{In the second step, we -- for those we did not match in the previous step -- relax the requirements to allow use of initial in place of full first name. We then (3) match on first name initial and full last name, (4) match on full first name, (5) match on full last name, and (6) match on just first name initial. In cases where any of these steps leads to more than one match, we select the match closest in time, i.e., if a child is born in 1963 and we match the nurse on her record to our list of nurse districts for both 1963 and 1965, we select the match from the 1963 nurse district data.} 
Third, as we only observe certain ``snapshots'' with respect to the nurses' districts, we potentially assign a wrong district to some children, particularly for the years furthest away from when we have information on nurse districts.
For this reason, we experiment with samples where we limit the maximum time span between a child's birth and one of our snapshots.

In total, we are able to match \shareRecordToDistrictMatch of the children with a nurse district, depending on the strictness of our matching approach.\footnote{This refers to the share of nurse records with a non-rare nurse name transcription, of which there are 76,579 in total. This number is slightly higher than the final sample size shown in Figure~\ref{fig: flowchart} due to missing district information for a few of the 76,579 records.}
We focus on the larger sample that potentially includes lower quality matches, but also show robustness of our results to stricter versions of our matching approach.
In total, \numberUniqueNurses different nurses appear across the collection of nurse records, and using our most lenient method of matching information on nurse district allows us to obtain this information for \numberUniqueNursesWithDistrictInfo of the nurses.



\subsection{Danish Administrative Data}

\label{sec: data-adm}
We combine data from Danish registers to obtain information on long-run outcomes in the form of education and labor market outcomes. 
We obtain data from education and labor market registers for the years 1980-2018/2019, respectively. 
From education registers, we obtain information on the years of completed education of our focal individuals as well as their relatives, including information on highest completed educational level (e.g., mandatory, university, etc.).
From labor market registers, we obtain information on employment and earnings of our focal individuals and their relatives.
For each individual, we obtain the share of time in employment during each year for which we have data, as well as their earnings each year.\footnote{We inflation-adjust earnings to reflect 2015 values and winsorize one percent of each tail, the latter which we do separately for each age. 
}
Using our employment and income data, we obtain average share of time in employment and income during ages 25-50 for each individual.\footnote{If information is missing for one or multiple ages of an individual during these 25 years, we take the average of the ages with non-missing employment/income information.}



\section{Empirical Methods}
\label{sec: empirical-methods}

Directly comparing outcomes of children for whom different nurses were allocated is likely to lead to false conclusions given significant differences in resources across different areas of Copenhagen (such as a family's available resources).
To account for such differences, we make use of nurse districts and information on year of birth, exploiting that children born within a district-by-year group are likely to be ``as good as randomly'' assigned a nurse from that district.\footnote{In Section~\ref{sec: results}, we empirically verify this by comparing pre-treatment characteristics of children allocated different nurses but born within the same district-by-year group.}
Given that each nurse was supposed to be responsible for around 160 children, the birth of a new child in a district is likely to be assigned to the nurse from that district with the smallest current workload. 
Given that the Copenhagen nurse program followed all children during their first year of life, the allocation of a child to a nurse within a district is likely to be as good as random.

Cast in terms of the potential outcomes framework, we assume the potential outcomes of a child are independent of the nurse visiting the child, at least conditionally on the nurse district-by-year.\footnote{Conditioning on the district might be important if, e.g., the potential outcomes of a child changes if the child lived in another district where, e.g., the schools were better. Further, the parents of children in one district might vary systematically from the parents of children in another district, for instance reflecting socioeconomic differences.}
Letting $Y_{id}^j$ denote the potential outcome of child $i$ in district-by-year $d$ under the ``treatment'' (visits) of nurse $j$, we assume:

\begin{align}
    Y_{id}^j \ind D_i^j \, | \, G_{id},
\end{align}

\noindent
where $D_i^j \in \{0, 1\}$ is one if child $i$ was visited by nurse $j$ and $G_{id}$ symbols the nurse district-by-year groups. 
The realized outcome of child $i$ is then $Y_{id} = \sum_j Y_{id}^j D_i^j$, and our objective is to draw inference on ``treatment effects'' of the type $\tau_{ij_1j_2} = Y_{id}^{j_1} - Y_{id}^{j_2}$, where now both $j_1$ and $j_2$, $j_1 \neq j_2$, represent nurses.
This, then, represents the effect of assigning nurse $j_1$ rather than $j_2$ to child $i$, i.e., what is the change in the outcome of child $i$ by re-assigning nurses in such a way that child $i$ now receives visits form nurse $j_1$ rather than nurse $j_2$.

If we knew $Y_{id}^j$, we could directly calculate all $\tau_{idj_1j_2}$, but in the absence of this, where only the realized outcome is observed, we turn our attention first to the more coarse treatment effects of the type $\tau_{dj_1j_2} = \mathbb{E}\left[Y_d^{j_1}\right] - \mathbb{E}\left[Y_d^{j_2}\right]$, i.e., where we no longer subscript with $i$, instead averaging over children to obtain average treatment effects.
Given (conditional) random allocation of nurses to children, this can be estimated by plugging in sample equivalents of $\mathbb{E}\left[Y_d^{j}\right]$, namely $\frac{1}{\sum_i D_i^j} \sum_i D_i^jY_{id}$.

Due to potential differences between district-by-year groups we consider the simplest case, focusing on some specific district-by-year group (and then leaving out subscripts $d$ for notational simplicity).
Here, we may regress an outcome $Y_i$ of individual $i$ on dummy variables $D_i^j$, where $j$ enumerates the different nurses and where $D_i^j = 1$ if child $i$ was assigned nurse $j$:

\begin{align}
    Y_i = \beta_jD_i^j + \epsilon_i,
\end{align}

\noindent
where $\beta_j$ then denotes the expected outcome of children assigned nurse $j$; thus, our estimate of $\tau_{j_1j_2}$ is then:

\begin{align}
    \hat{\tau}_{j_1j_2} = \hat{\beta}_{j_1} - \hat{\beta}_{j_2} = \frac{1}{\sum_i D_i^{j_1}} \sum_i D_i^{j_1}Y_i - \frac{1}{\sum_i D_i^{j_2}} \sum_i D_i^{j_2}Y_i
\end{align}

The above approach allows us to compare the average outcomes of children by nurse within any specific district-by-year group.
A further challenge arises, however, if we try to compare treatment effect estimates between district-by-year groups.
The reason is as follows: 
Any treatment effect is a difference in average potential outcomes between two groups of children, and thus the size of any treatment effect depends not only on the visiting nurse, but also on the counterfactual alternative nurse you use as the base for your comparison.
For that reason, large treatment effects might arise from either or both of (1) the nurse $j_1$ being in the right tail of the skill distribution and/or (2) the nurse $j_2$ being in the left tail of the skill distribution.
For those reasons, directly comparing treatment effects between district-by-year groups is not possible

When we then turn to heterogeneity of nurse effects depending on child pre-treatment characteristics, we let $X_i \in \mathbb{R}^l$ denote the vector of characteristics for child $i$.\footnote{For example, $X_i$ might include low birth weight status or absence of father.}
We once again consider the simplest case of one district-by-year group, and now include $X_i$ to account for nurse-by-child-characteristic differences:

\begin{align}
    Y_i = f\left(D_i^j, X_i\right) + \epsilon_i,
\end{align}

\noindent
where we may use causal ML to estimate heterogeneous treatment effects, such as generalized random forests \citep{wager2018estimation, athey2019generalized}.\footnote{We use extensions by \citet{nie2021quasi} to allow estimation in our setting of multi-arm treatments.}
The heterogeneity we are able to capture this way reflects the cross-product of differences between nurses and children and, if present, allow us to consider welfare effects of nurse reallocation.

Our identification strategy relies on quasi-random allocation of children to nurses within nurse district-by-year groups.
If this is not so, and nurses are selectively allocated to children as might be the case if children from families the least well off are more likely to be allocated a specific nurse, conditional independence between potential outcomes and treatment no longer holds.
To mitigate such concerns, we perform several tests of differences in pre-treatment variables between the children of different nurses.

\section{Results}
\label{sec: results}


\subsection{Descriptive Statistics}
\label{sec: descr-stats}

We start by documenting the number of nurses and the number of children per year, district, and nurse.
We can estimate the number of nurses in two ways, based on either statistics from archives or directly from the collection of nurse records we transcribe.
From archive material, we know the number of nurses for the years from which we have data (1957, 1963, 1965-1969).
This includes their names, which allow us to track them over time, allowing us to calculate the total number of unique nurses.
From the collection of nurse records, we use our transcriptions of nurse names for each record to obtain statistics on the number of nurses for each year as well as the total number of unique nurses during the period 1959-1967.
Due to imperfect transcriptions, however, we err on the side of caution and define a nurse only when the name of the nurse appears on a sufficient number nurse records, to avoid a small error in a transcription resulting in a name with one letter off now occurring as a unique nurse with just one record; we require \minNumberChildrenByNurse occurrences of a name to include it.\footnote{Further, due to confidentiality reasons we are required to aggregate statistics up to include a certain minimum number of children.}
Table~\ref{tab: number-of-nurses} shows the number of nurses for each year as well as the total number of nurses, as estimated from either archive material or the collection of nurse records.
From the years for which we have data from both archive material and the nurse records, we generally see slightly fewer nurses from the archive material than from the nurse records (with 1965 being an exception).
This is expected, as the first reports a snapshot and the second includes nurses present just for parts of the year.\footnote{For example, a snapshot of nurses for May may not include a nurse that stopped working in that year before May or one that started later than May.}
In total, however, we identify more nurses from the archive material than from the nurse records; this is also not surprising as the archive material stretches over a longer period of time (1957-1969 vs. 1959-1967).
In our main analyses, we continue with the sample of \numberUniqueNurses unique nurses, the name of which we know occurred on at least \minNumberChildrenByNurse nurse records.\footnote{In our analyses exploiting information on district we continue with \numberUniqueNursesWithDistrictInfo nurses, due to incomplete district information.}

\begin{table}
    \centering
    \centering
    \caption{Number of Unique Nurses by Year.}
    \label{tab: number-of-nurses}
    \begin{tabular}{l rr}
         \toprule
          & Archive material & Nurse records \\
         \midrule
         1957           &   58  &     \\
         1959           &       &  88 \\
         1960           &       &  93 \\
         1961           &       &  93 \\
         1962           &       &  96 \\
         1963           &   74  &  99 \\
         1964           &       & 101 \\
         1965           &  113  &  96 \\
         1966           &   64  &  95 \\
         1967           &   62  &  93 \\
         1968           &   70  &     \\
         1969           &   69  &     \\
         \midrule
         Total (unique) & 166   & \numberUniqueNurses \\
         \bottomrule
    \end{tabular}
    \begin{minipage}{1\linewidth}
        \vspace{1ex}
        \footnotesize{
        \textit{Notes:}
        The table shows the number of unique nurses, identified through either archive materials or transcriptions of the collection of nurse records.
        We do not have an estimate of both sources for each year, which leads to some empty fields.
        The final row shows the total number of unique nurses, i.e., the same nurse present for multiple years counts only once.
		}
	\end{minipage}
\end{table}

Our identification strategy relies on the assumption of random allocation of children to nurses within each nurse district.
From archive material, we know that the aim of the Copenhagen nurse visiting program was to have each nurse be responsible for around 160 children, and we would thus expect each nurse at any point in time to be responsible for somewhere around 160 children (\citeauthor{stadsarkiv}, various years).
To assess the number of children by year, district, and nurse, Figure~\ref{fig: obs-by-year-district-nurse} shows the number of children in our primary sample by year of birth (Panel~\ref{subfig: obs-by-year}), district (Panel~\ref{subfig: obs-by-district}), and nurse (Panel~\ref{subfig: obs-by-nurse}).
Note that only nurse names occurring at least \minNumberChildrenByNurse times are included.

\begin{figure}
    \centering
    \subfloat[Children by Year]{\label{subfig: obs-by-year}\includegraphics[width=0.33\linewidth]{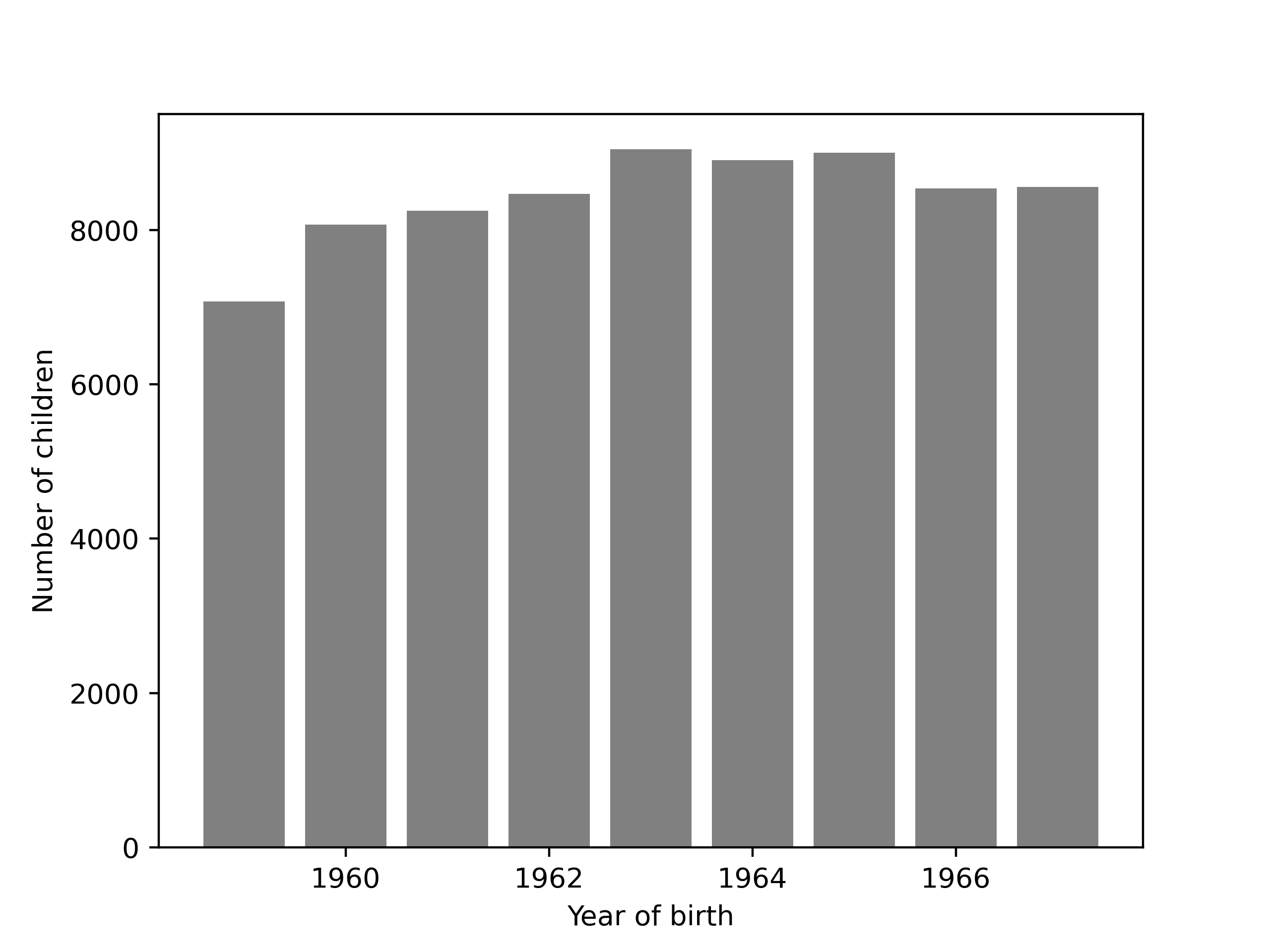}}
    \subfloat[Children by District]{\label{subfig: obs-by-district}\includegraphics[width=0.33\linewidth]{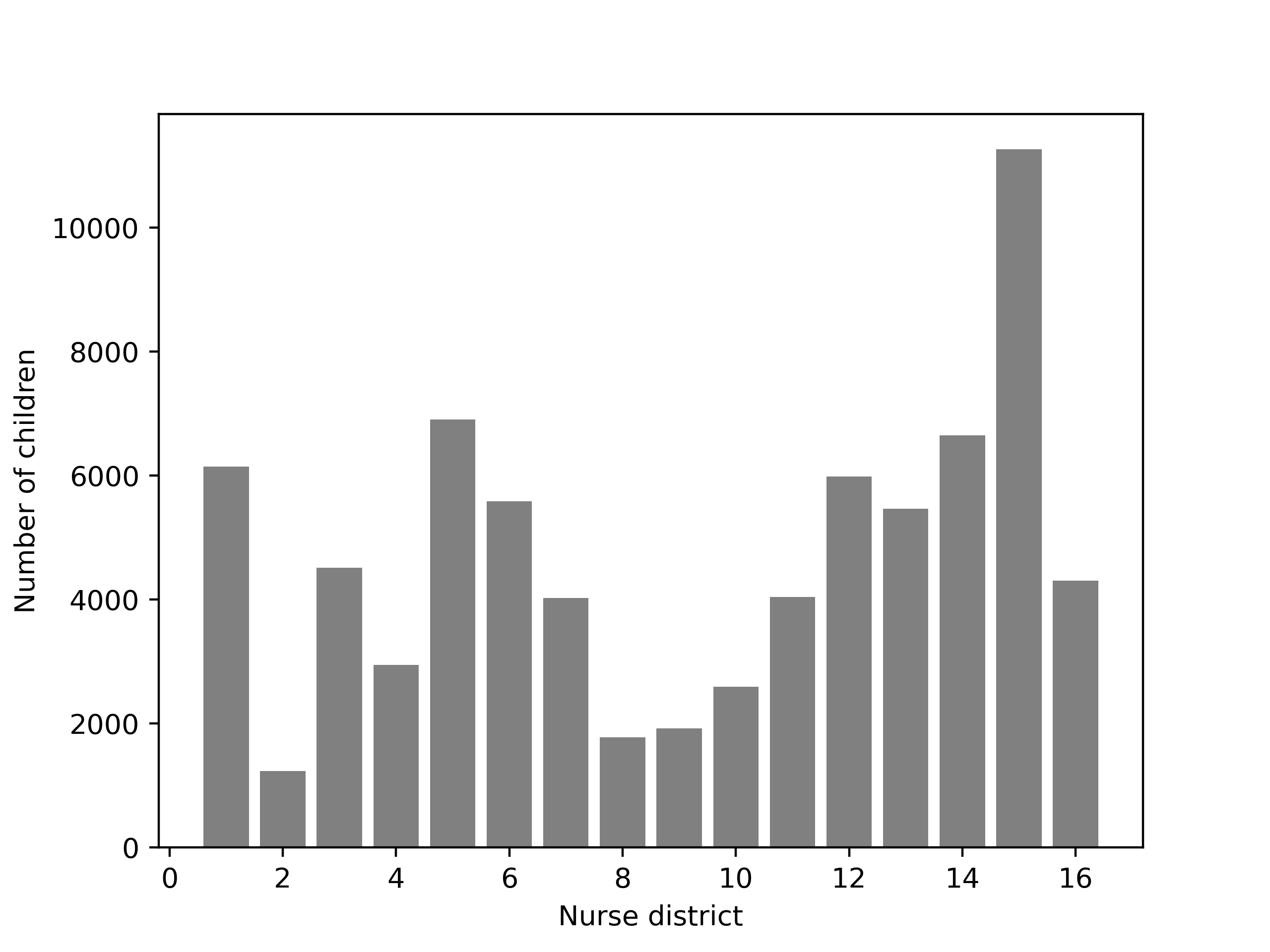}}
    \subfloat[Children by Nurse]{\label{subfig: obs-by-nurse}\includegraphics[width=0.33\linewidth]{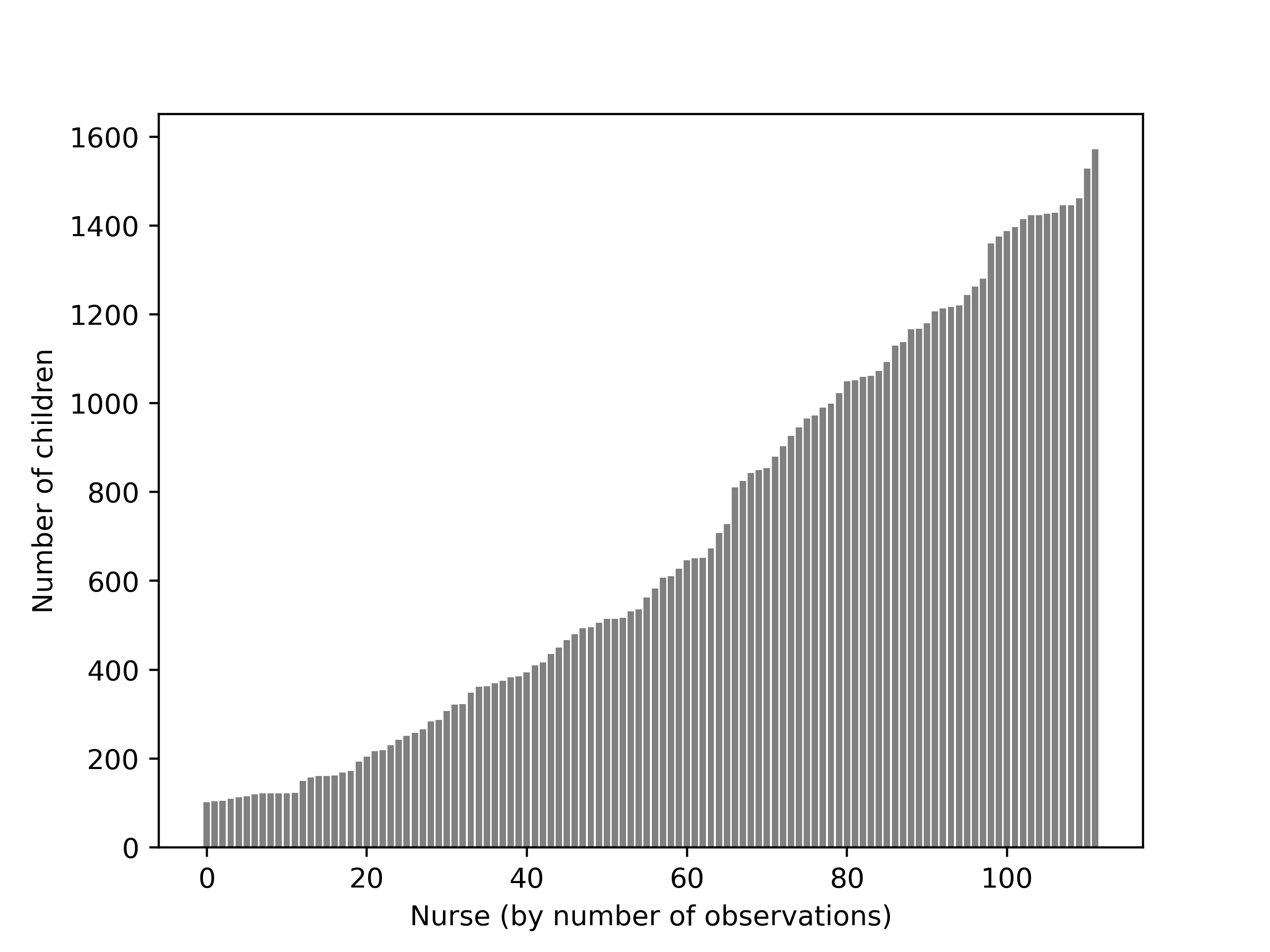}}
    \caption{Number of Children by Year, District, and Nurse.}
    \label{fig: obs-by-year-district-nurse}
    \begin{minipage}{1\linewidth}
        \vspace{1ex}
        \footnotesize{
        \textit{Notes:}
        The figure shows the number of children (i.e., nurse records) in our primary sample by year of birth (Panel~\ref{subfig: obs-by-year}), district (Panel~\ref{subfig: obs-by-district}), and nurse (Panel~\ref{subfig: obs-by-nurse}).
        Note that only nurse names occurring at least \minNumberChildrenByNurse times are included.
		}
	\end{minipage}
\end{figure}

As is evident from Figure~\ref{fig: obs-by-year-district-nurse}, the number of children by nurse varies substantially.
While this at first appears problematic for our design, this is explained by differences in how many years the individual nurses were present:
A nurse present in all of the years 1959-1967 will occur on more journals than one only present in 1959.
To gauge the real number of children present at the same time per nurse, we therefore calculate the number of children by \textit{nurse-year}.
We do this by calculating the number of children born in each calendar year for each nurse, and only include a nurse-year if at least one child born in each month of the given year was assigned to the specific nurse.\footnote{This is done to eliminate the issue of a nurse only being present for part of a year.}
We expect this to lead to a density with high mass centered at close to but below 160 (given that our sample includes close to but not all of the potential children of the nurse program and the program aimed at each nurse being allocated around 160 children).
Indeed, Figure~\ref{fig: obs-by-district-year-and-nurse-year}, which shows the density of nurse-years as well as district-years, exhibits this pattern:
Panel~\ref{subfig: obs-by-nurse-year} shows the density of the number of children by nurse-year, which exhibits a tight center of mass, with the average number of children by nurse-year being around 140. 
Turning to Panel~\ref{subfig: obs-by-district-year}, which shows the density of the number of children by district-year, there is more dispersion, indicating that some districts were larger than others in terms of number of children in the district.
Appendix Figure~\ref{fig: obs-by-year-by-district} shows the number of children within each district-by-year combination, showcasing that districts that were larger at the start of the period (1959) generally continue to be so over all the available years (1959-1967).
Further, as we would expect, the number of nurses is larger in those districts with more children. 

\begin{figure}
    \centering
    \subfloat[Children by Nurse-Year]{\label{subfig: obs-by-nurse-year}\includegraphics[width=0.5\linewidth]{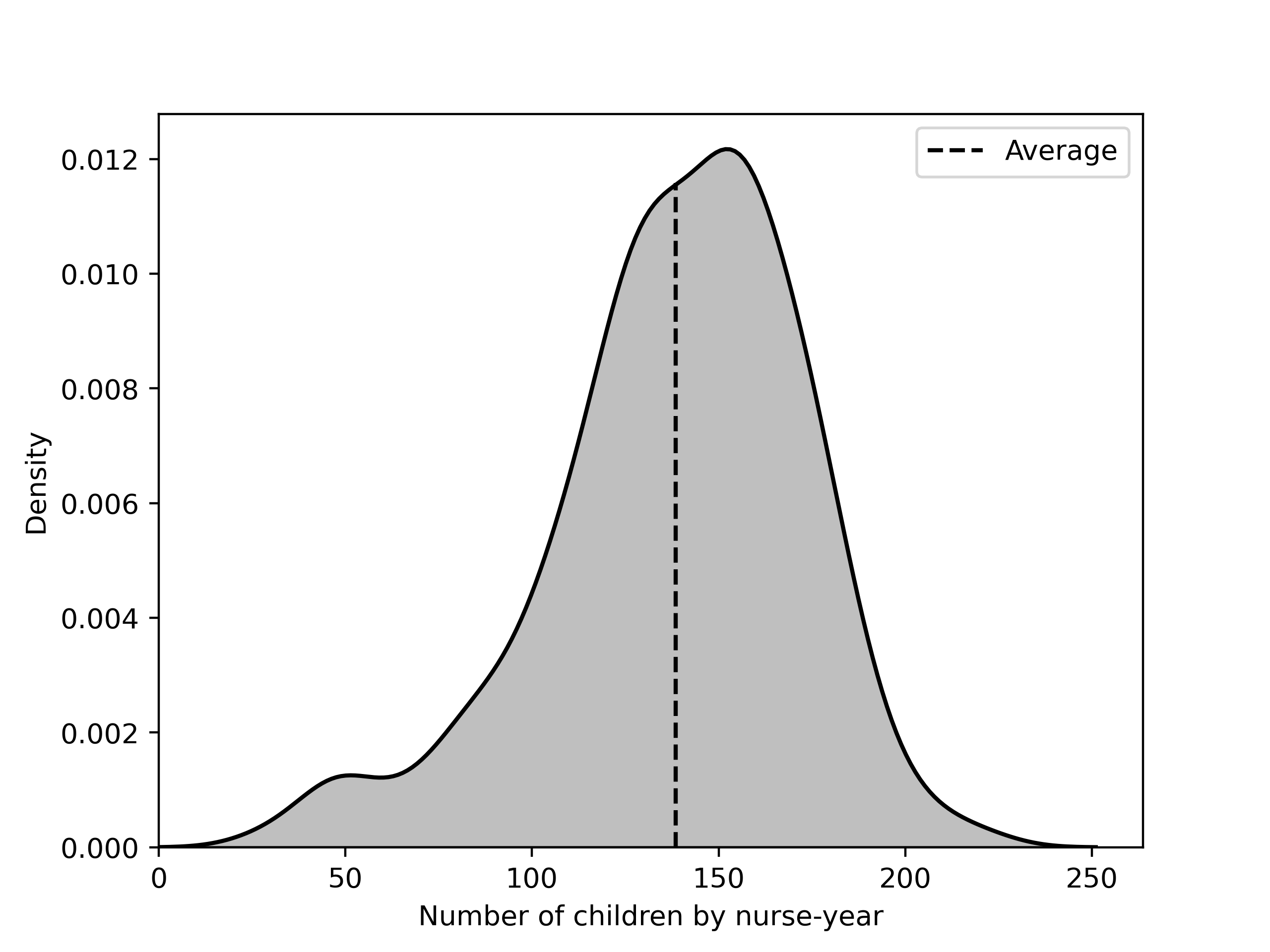}}
    \subfloat[Children by District-Year]{\label{subfig: obs-by-district-year}\includegraphics[width=0.5\linewidth]{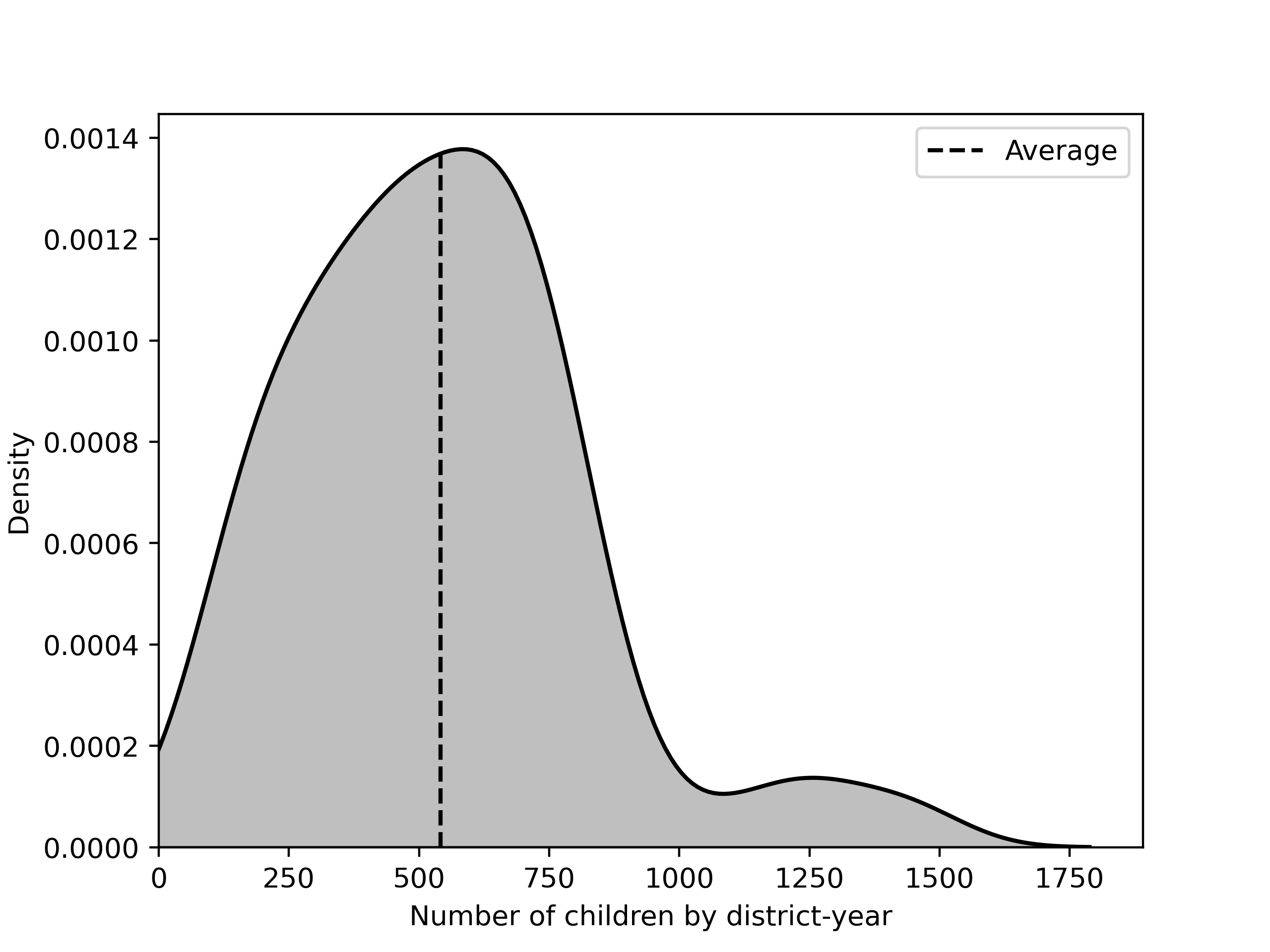}}
    \caption{Number of Children by Nurse-Year and District-Year.}
    \label{fig: obs-by-district-year-and-nurse-year}
    \begin{minipage}{1\linewidth}
        \vspace{1ex}
        \footnotesize{
        \textit{Notes:}
        The figure shows the density of the number of children by nurse-year (Panel~\ref{subfig: obs-by-nurse-year}) and district-year (Panel~\ref{subfig: obs-by-district-year}).
		}
	\end{minipage}
\end{figure}

Table~\ref{tab: descriptives} shows means, standard deviations, and number of observations for pre-treatment and outcome variables for our primary sample, with the last column showing p-values from a Kruskal-Wallis H-test for independent samples \citep{kruskal1952use}, comparing the population of nurses against each other with a null-hypothesis of no difference in the median outcome of children of different nurses.\footnote{Testing for differences in means (i.e., an ANOVA test) leads to the same picture, but we prefer the Kruskal-Wallis test due to its milder assumptions.}
If children were randomly assigned nurses, we would expect the p-values of the pre-treatment variables (\textit{Panel A}) to be large, while the p-values of the outcome variables (\textit{Panel B}) could still be small if nurses differed with respect to their treatment effect.
However, as is clear from the table there are statistically significant differences between nurses in all measures, with the largest p-value being for sex (0.069). 
This is not surprising, given that nurses served in different districts (and time periods) with different populations of children, due to, e.g., different levels of socioeconomic status of the parents of our focal individuals between nurse districts.

\begin{table}
    \centering
    \caption{Descriptive Statistics.}
    \label{tab: descriptives}
    \begin{tabular}{l rrrr}
        \toprule
        & Mean & Standard deviation & No. of obs. & P-value \\
        \midrule
        \multicolumn{5}{l}{\textit{Panel A: Background Characteristics, Nurse Records \& Administrative Data}} \\
        \midrule
        Birth weight (g) & 3,326.76 & 527.81 & 76,467 & 0.000 \\
        Low BW &     0.05 &   0.22 & 76,467 & 0.000 \\
        Birth length (cm) &    51.42 &   2.40 & 76,257 & 0.000 \\
        Born prior to due date &     0.12 &   0.33 & 75,118 & 0.000 \\
        Weeks prior to due date &     3.41 &   1.87 &  8,619 & 0.001 \\
        Year of birth & 1,963.12 &   2.53 & 75,899 & 0.000 \\
        Born 1-3 &     0.10 &   0.30 & 75,899 & 0.015 \\
        Child parity &     1.63 &   0.79 & 75,184 & 0.000 \\
        Firstborn &     0.53 &   0.50 & 75,184 & 0.000 \\
        Female &     0.49 &   0.50 & 75,360 & 0.069 \\
        Mother yrs. of educ. &    11.25 &   3.13 & 72,303 & 0.000 \\
        Mother age at birth &    25.17 &   5.38 & 75,184 & 0.000 \\
        Father yrs. of educ. &    12.37 &   3.40 & 64,494 & 0.000 \\
        Father age at birth &    28.41 &   6.67 & 72,087 & 0.000 \\
        Father missing &     0.06 &   0.23 & 76,579 & 0.000 \\
        High SES, 1 mo. &     0.12 &   0.32 & 49,942 & 0.000 \\
        \midrule
        \multicolumn{5}{l}{\textit{Panel B: Outcome Measures, Nurse Records \& Administrative Data}} \\
        \midrule
        Breastfed, 1 mo. & 0.58 & 0.49 & 70,452 &  0.0 \\
        Breastfed, 6 mo. & 0.03 & 0.17 & 63,670 &  0.0 \\
        Duration of breastfeeding (mo.) & 2.50 & 2.86 & 64,853 &  0.0 \\
        Avg. inc. 25-50 (DKK) & 269,452.55 & 162,909.39 & 74,922 &  0.0 \\
        Share empl. 25-50 &       0.79 &       0.28 & 74,168 &  0.0 \\
        Yrs. of educ. &      13.74 &       2.54 & 74,167 &  0.0 \\
        Above mand. edu. &       0.74 &       0.44 & 74,167 &  0.0 \\
        \bottomrule
    \end{tabular}
    \begin{minipage}{1\linewidth}
        \vspace{1ex}
        \footnotesize{
        \textit{Notes:}
        The table shows means, standard deviations, and number of observations for pre-treatment (\textit{Panel A}) and outcome (\textit{Panel B}) variables.
        Additionally, the final column shows the p-value from a Kruskal-Wallis H-test for independent samples, comparing the groups of children of different nurses against each other.
        Note that \textit{Weeks prior to due date} is reported only for those children born at least a week prior to their due date, which is what explains the low number of observations and the large value of the mean of this variable.
		}
	\end{minipage}
\end{table}

While Table~\ref{tab: descriptives} document statistically significant differences between the groups of children of different nurses, it does little to report on the magnitudes of the differences.
Table~\ref{tab: descriptives-quantiles} shows statistics for the variables of Table~\ref{tab: descriptives}, but where, for each variable, results are grouped by ``nurse rank'':\footnote{This is done separately for each variable, meaning that the order of nurses vary to some degree between different rows of the table. However, the rank-order correlation between variables is high, and thus those nurses for whom the average outcome of one variable of their allocated children is high also tend to score highly on other variables.}
Nurses for whom the average outcome of the given variable of the children they visit is below the first quantile forms the first group, the second group consists of those between the first and second quantile, the third group of those between the second and third quantile, and the final group of those above the third quantile.
If nurses were randomly allocated across all districts, we would expect to see relatively little difference between the different groups for the variables in \textit{Panel A}, and large differences in \textit{Panel B} only if significant heterogeneity in nurse treatment effects is present.\footnote{Some variation is still expected due to finite sample sizes.}
However, as is clear even for the pre-treatment variables, differences are also economically large between nurses, implying differences due to, e.g., variation in socioeconomic status between nurse districts (e.g., the difference between the lowest ranked quarter of nurses vs. the highest ranked quarter of nurses in mother years of education is over one year). 

\begin{table}
    \centering
    \centering
    \caption{Descriptive Statistics by Nurse Group (Rank).}
    \label{tab: descriptives-quantiles}
    \resizebox{\linewidth}{!}{
        \begin{tabular}{l rrrrr}
            \toprule
            & \multicolumn{4}{c}{Mean (std. dev) by quartile} & No. of obs. \\ \cline{2-5}
            & 0-25\% & 25-50\% & 50-75\% & 75-100\% & \\
            \midrule
            \multicolumn{6}{l}{\textit{Panel A: Background Characteristics, Nurse Records \& Administrative Data}} \\
            \midrule
                   Birth weight (g) & 3,284.49 (536.91) & 3,313.56 (528.39) & 3,334.62 (531.45) & 3,357.27 (515.09) & 76,467 \\
                 Low BW &       0.04 (0.19) &       0.05 (0.21) &       0.05 (0.23) &       0.07 (0.25) & 76,467 \\
      Birth length (cm) &      51.23 (2.47) &      51.37 (2.42) &      51.46 (2.41) &      51.59 (2.32) & 76,257 \\
 Born prior to due date &       0.09 (0.29) &       0.12 (0.32) &       0.13 (0.34) &       0.15 (0.36) & 75,118 \\
Weeks prior to due date &       3.08 (1.62) &       3.32 (1.82) &       3.48 (1.88) &       3.74 (2.06) &  8,619 \\
          Year of birth &   1,961.01 (1.75) &   1,962.88 (2.48) &   1,963.36 (2.46) &   1,964.81 (1.84) & 75,899 \\
               Born 1-3 &       0.08 (0.27) &       0.09 (0.29) &       0.10 (0.30) &       0.12 (0.32) & 75,899 \\
           Child parity &       1.54 (0.72) &       1.60 (0.76) &       1.64 (0.80) &       1.71 (0.86) & 75,184 \\
              Firstborn &       0.49 (0.50) &       0.52 (0.50) &       0.54 (0.50) &       0.58 (0.49) & 75,184 \\
                 Female &       0.46 (0.50) &       0.48 (0.50) &       0.50 (0.50) &       0.52 (0.50) & 75,360 \\
   Mother yrs. of educ. &      10.55 (2.98) &      10.99 (3.05) &      11.36 (3.12) &      11.79 (3.20) & 72,303 \\
    Mother age at birth &      23.97 (4.76) &      24.78 (5.21) &      25.32 (5.42) &      26.14 (5.67) & 75,184 \\
   Father yrs. of educ. &      11.61 (3.34) &      12.06 (3.37) &      12.50 (3.37) &      13.01 (3.37) & 64,494 \\
    Father age at birth &      27.23 (6.31) &      28.00 (6.48) &      28.60 (6.75) &      29.38 (6.86) & 72,087 \\
         Father missing &       0.04 (0.20) &       0.05 (0.23) &       0.07 (0.25) &       0.08 (0.27) & 76,579 \\
        High SES, 1 mo. &       0.02 (0.13) &       0.05 (0.22) &       0.10 (0.31) &       0.31 (0.46) & 49,942 \\
            \midrule
            \multicolumn{6}{l}{\textit{Panel B: Outcome Measures, Nurse Records \& Administrative Data}} \\
            \midrule
            Breastfed, 1 mo. & 0.49 (0.50) & 0.56 (0.50) & 0.59 (0.49) & 0.64 (0.48) & 70,452 \\
            Breastfed, 6 mo. & 0.01 (0.09) & 0.02 (0.14) & 0.03 (0.17) & 0.06 (0.23) & 63,670 \\
            Duration of breastfeeding (mo.) & 1.81 (2.42) & 2.40 (2.75) & 2.61 (2.94) & 3.01 (3.09) & 64,853 \\
            Avg. inc. 25-50 (DKK) & 248,926.10 (156,939.84) & 260,964.34 (162,325.75) & 271,131.54 (161,560.47) & 285,750.18 (166,404.97) & 74,922 \\
            Share empl. 25-50 &             0.76 (0.30) &             0.78 (0.29) &             0.80 (0.27) &             0.82 (0.26) & 74,168 \\
            Yrs. of educ. &            13.34 (2.49) &            13.57 (2.51) &            13.79 (2.55) &            14.08 (2.55) & 74,167 \\
            Above mand. edu. &             0.69 (0.46) &             0.72 (0.45) &             0.75 (0.44) &             0.79 (0.41) & 74,167 \\
            \bottomrule
        \end{tabular}
    }
    \begin{minipage}{1\linewidth}
        \vspace{1ex}
        \footnotesize{
        \textit{Notes:}
        The table shows means (standard deviations) for pre-treatment variables (\textit{Panel A}) and outcome variables (\textit{Panel B}) by different groups of nurses, as defined by their rank compared to other nurses.
        Nurses for whom the average outcome of the given variable of the children they visit is below the first quantile forms the leftmost group (0-25\%), while the next group consists of those between the first and second quantile, the third group of those between the second and third quantile, and the final group of those above the third quantile.
		}
	\end{minipage}
\end{table}

To get a more detailed look into the distribution of variables by nurse, Figure~\ref{fig: var-by-nurse} shows mean values (and associated 95\% confidence intervals) of selects pre-treatment (i.e., determined before visit by nurse) variables by nurse, ranked such that nurses are sorted in ascending order.
As expected following the results in \textit{Panel A} of Table~\ref{tab: descriptives-quantiles}, the figure shows significant differences in averages of pre-treatment variables by nurse, indicating that the population of children allocated different nurses vary substantially; as these variables are determined before the visits take place, the differences cannot be explained by differences in nurse treatment effects, but must instead occur due to non-random allocation.

\begin{figure}
    \centering
    \subfloat[Low BW ($<$ 2500 g)]{\label{subfig: var-by-nurse-lowbw}\includegraphics[width=0.5\linewidth]{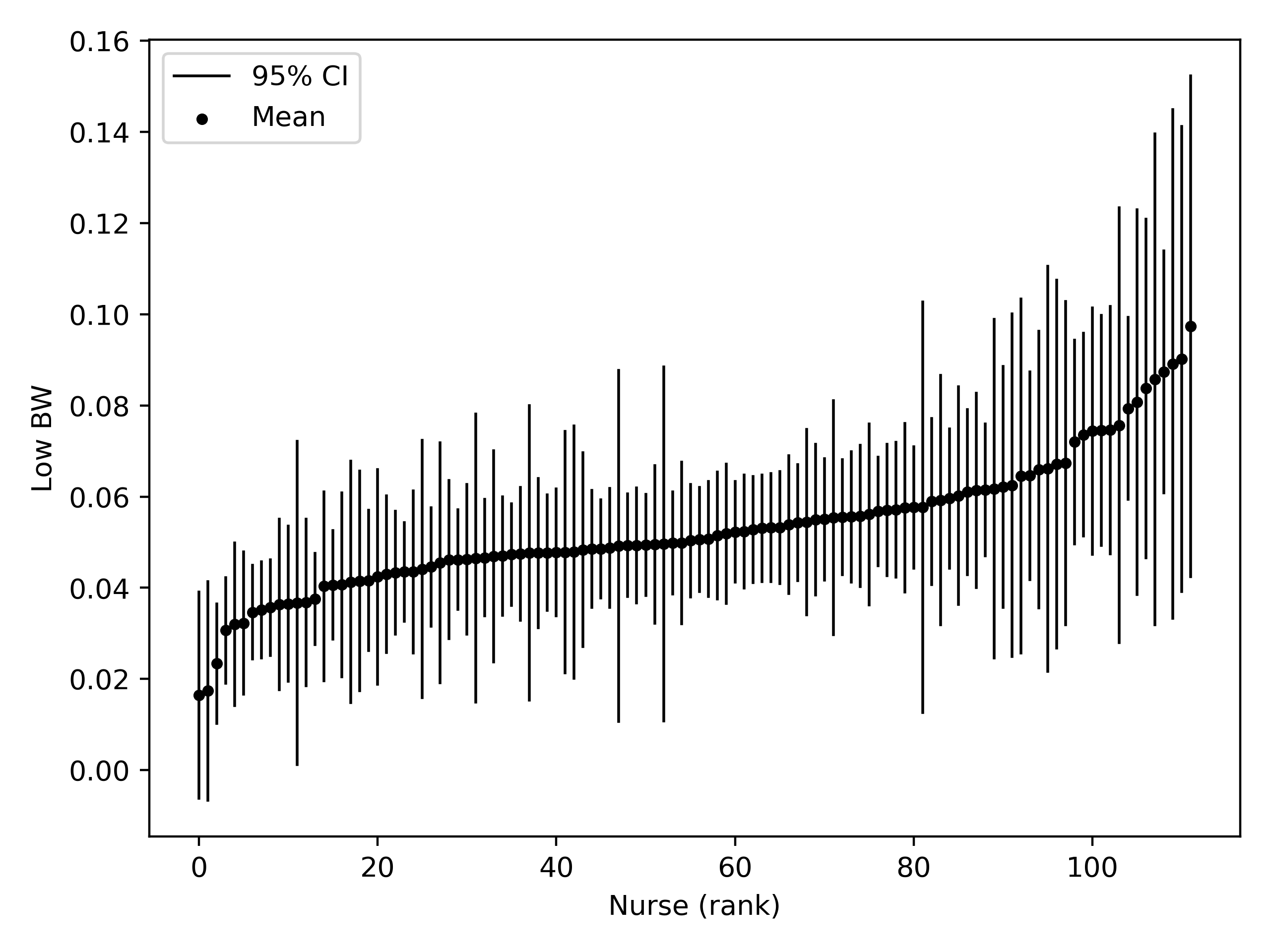}}
    \subfloat[Mother years of education]{\label{subfig: var-by-nurse-medulen}\includegraphics[width=0.5\linewidth]{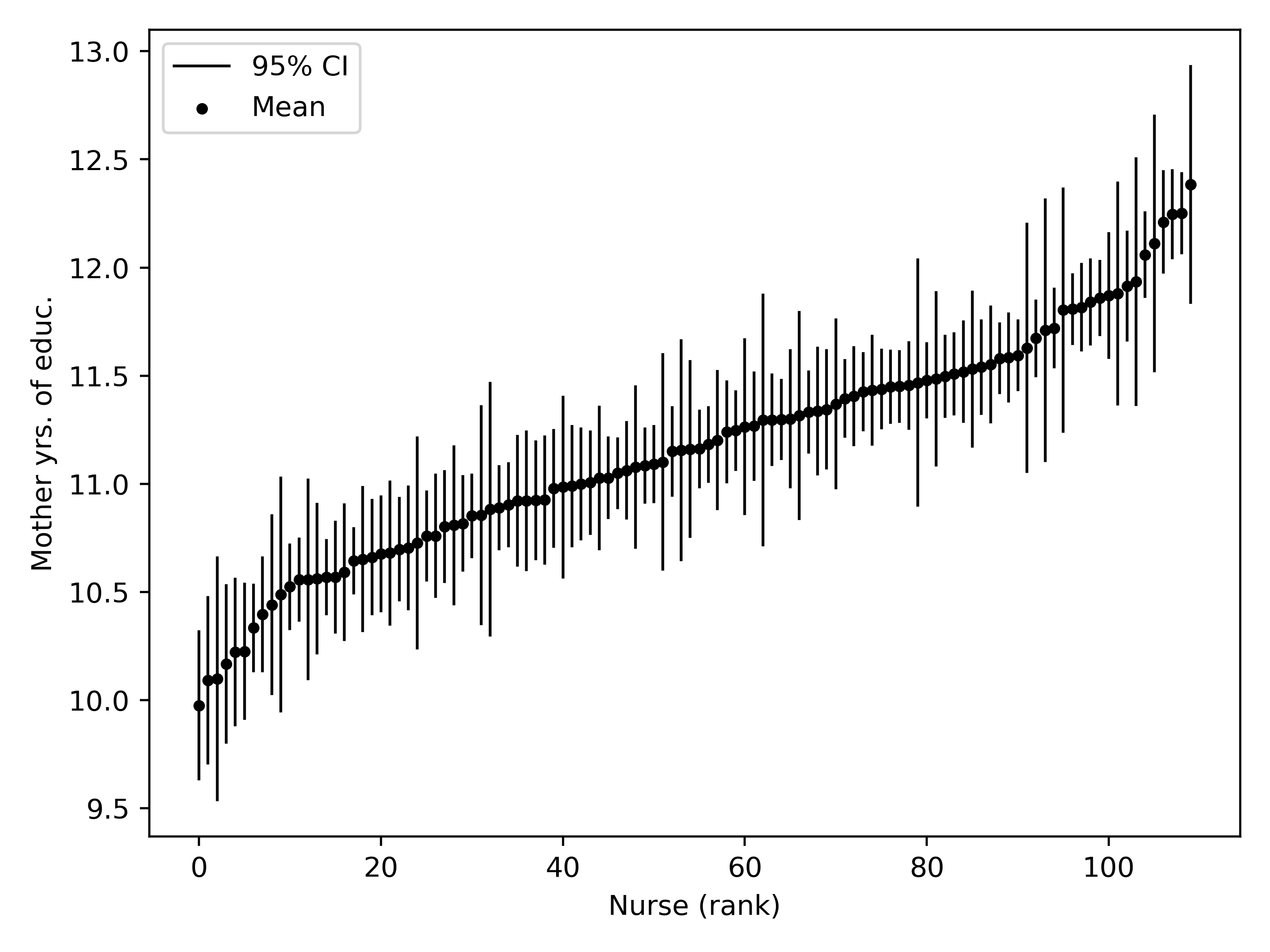}}
    \caption{Pre-Treatment Variables by Nurse.}
    \label{fig: var-by-nurse}
    \begin{minipage}{1\linewidth}
        \vspace{1ex}
        \footnotesize{
        \textit{Notes:}
        The figure shows means and 95\% confidence intervals for select pre-treatment (i.e., determined before visit) variables by nurse, where nurses are sorted in ascending order in terms of the mean of the variable for the children allocated to the specific nurse.
        Panel~\ref{subfig: var-by-nurse-lowbw} shows the results for an indicator of low birth weight (defined as a birth weight below 2500 g) and Panel~\ref{subfig: var-by-nurse-medulen} for the years of education of the mother of the child.
		}
	\end{minipage}
\end{figure}

Another way of assessing the degree of sorting that happens between nurses and children is to ask what the rank-order correlations between pre-treatment and outcome variables are.
In a setting of no sorting, this would be (asymptotically) zero:
While children that score ``poorly'' on pre-treatment characteristics (e.g., someone born low birth weight) are expected to score poorer on outcome variables (e.g., income), we would expect zero rank-order correlation between the average value of a pre-treatment and an outcome variable between nurses.
We therefore calculate these averages for each nurse for select variables and then use a Spearman rank-order correlation (a non-parametric measure of monotonicity) to compare nurses \citep{fieller1957tests}.\footnote{With 112 nurses, the asymptotic approximation of the p-value may yet be inaccurate, and the exact p-value is computationally impossible to derive. For these reasons, we use a permutation test that randomly draws 100,000 permutations, calculating the rank-order correlation coefficient of each and then compares it to the rank-order correlation coefficient of the non-permuted sample, letting the p-value be the share of permutations that results in a more extreme value of the rank-order correlation coefficient. 
}
Appendix Figure~\ref{fig: correlation-coefficient-nurses} shows the correlation coefficients and p-values for pairs of select pre-treatment and outcome variable means by nurses, verifying that substantial rank-order correlation is present, also when comparing pre-treatment and outcome variable pairs.

Taken at face value, the above results could be interpreted as detrimental to our identification strategy.
However, the differences we observe between nurses turn out to be explained largely by differences between nurse districts (and to a smaller extent differences between cohorts).
Figure~\ref{fig: var-by-district} shows the equivalent of Figure~\ref{fig: var-by-nurse} but now by \textit{nurse district} rather than nurse.
Notably, the variability of pre-treatment mean outcomes between districts is nearly as large as the variability of mean outcomes between nurses.\footnote{Note that the minimum and maximum mean outcomes of a district are bounded by the most extreme values of mean outcomes of any nurse.} 
The same pattern emerges when we consider the rank-order correlation in means of variables between districts (similarly to what we did above for nurses):
Appendix Figure~\ref{fig: correlation-coefficient-districts} shows the correlation coefficients and p-values for pairs of select pre-treatment and outcome variable means by nurse districts, verifying that substantial rank-order correlation is present, also when comparing pre-treatment and outcome variable pairs.\footnote{With 16 nurse districts, the asymptotic approximation of the p-value will be inaccurate, and the exact p-value is computationally impossible to derive. For these reasons, we use a permutation test that randomly draws 100,000 permutations, calculating the rank-order correlation coefficient of each and then compares it to the rank-order correlation coefficient of the non-permuted sample, letting the p-value be the share of permutations that results in a more extreme value of the rank-order correlation coefficient.}

\begin{figure}
    \centering
    \subfloat[Low BW ($<$ 2500 g)]{\label{subfig: var-by-district-lowbw}\includegraphics[width=0.5\linewidth]{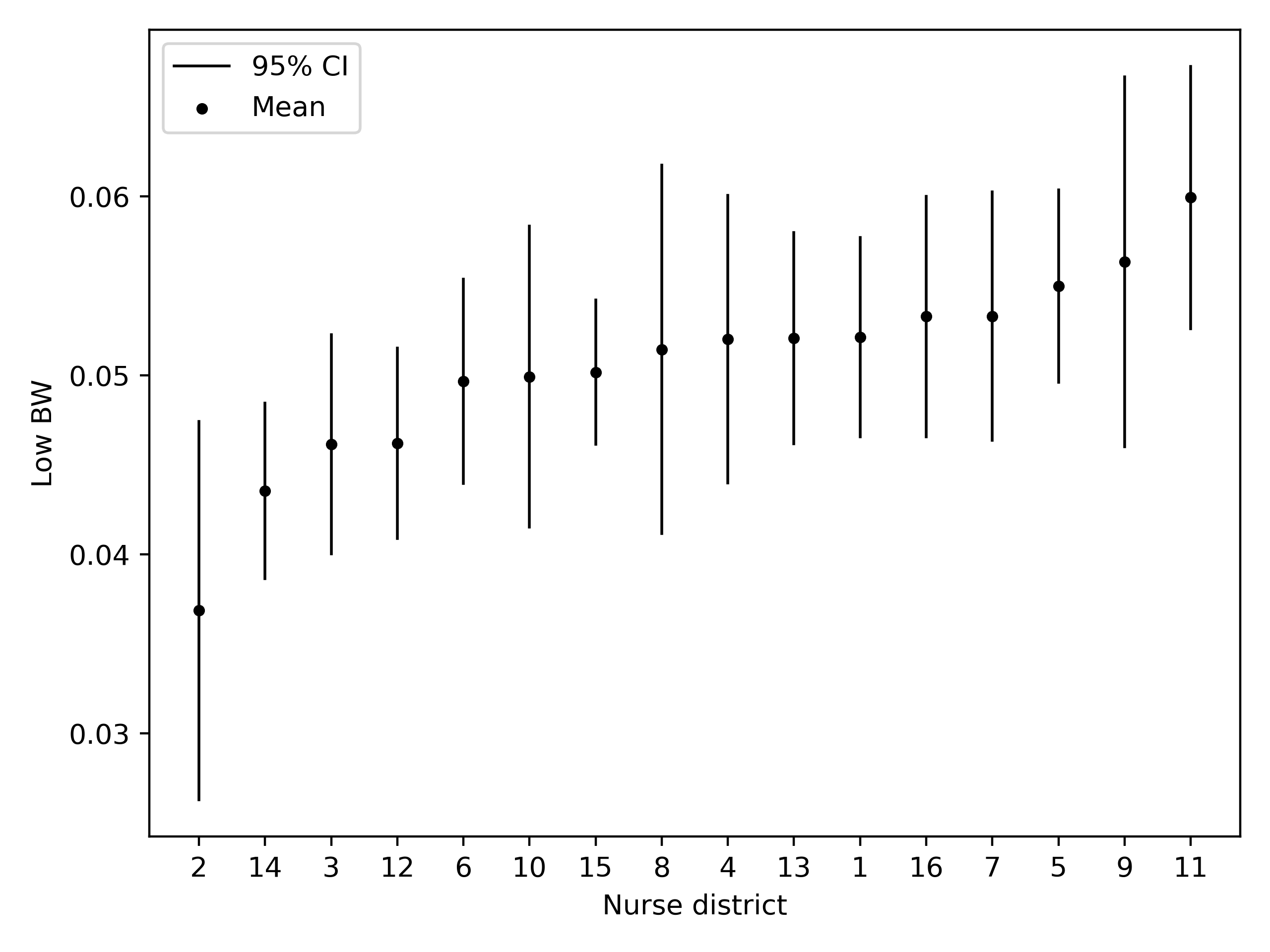}}
    \subfloat[Mother years of education]{\label{subfig: var-by-district-medulen}\includegraphics[width=0.5\linewidth]{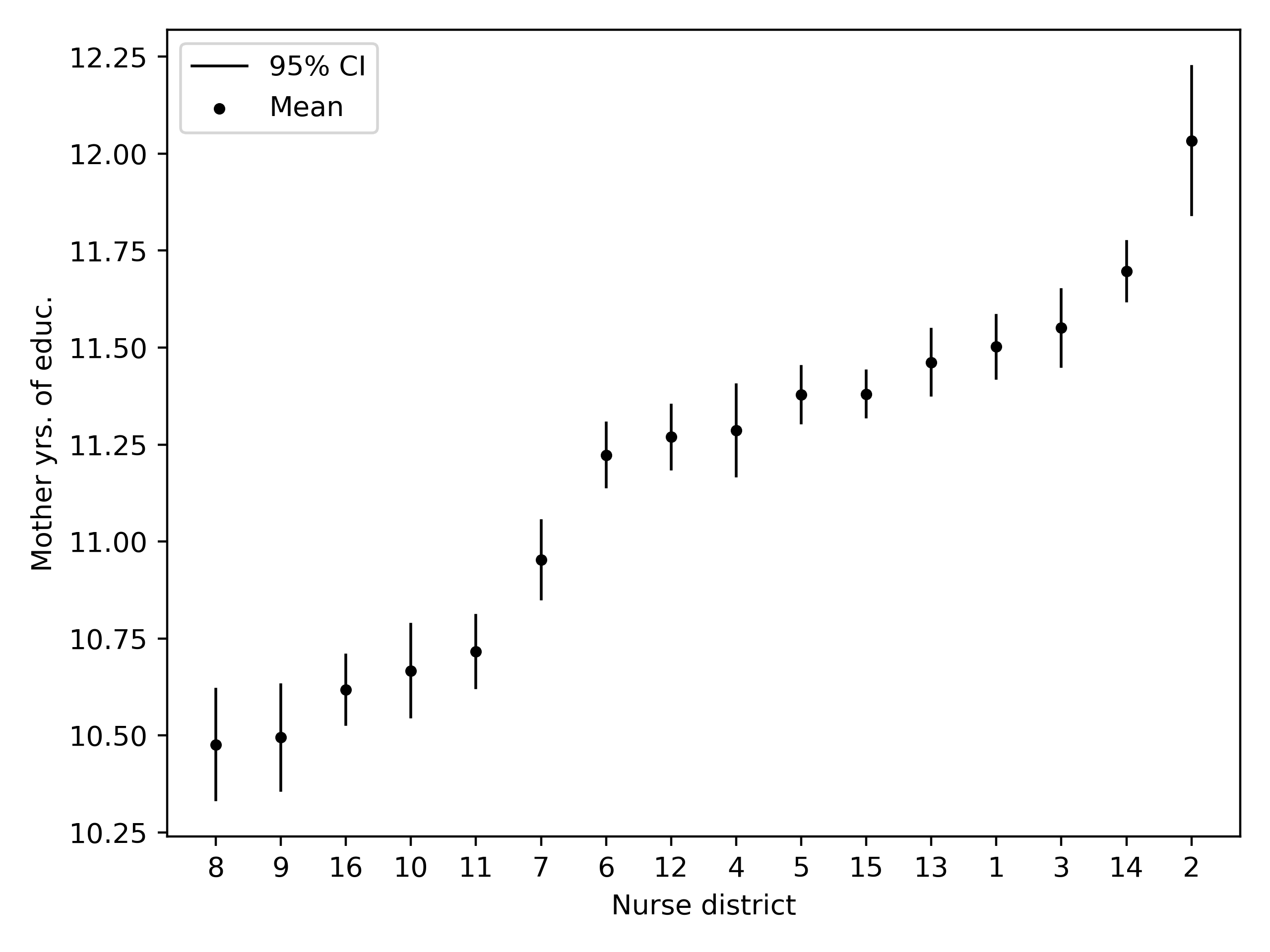}}
    \caption{Pre-Treatment Variables by Nurse District.}
    \label{fig: var-by-district}
    \begin{minipage}{1\linewidth}
        \vspace{1ex}
        \footnotesize{
        \textit{Notes:}
        The figure shows means and 95\% confidence intervals for select pre-treatment (i.e., determined before visit) variables by nurse district, where districts are sorted in ascending order in terms of the mean of the variable for the children allocated a nurse in the given district.
        Panel~\ref{subfig: var-by-district-lowbw} shows the results for an indicator of low birth weight (defined as a birth weight below 2500 g) and Panel~\ref{subfig: var-by-district-medulen} for the years of education of the mother of the child.
		}
	\end{minipage}
\end{figure}

Differences in outcomes between districts persist throughout life:
Figure~\ref{fig: age-plot-by-district} shows income and share of time in employment for our focal individuals during ages 25-50 by nurse district.
While there is some variability, those districts ``doing well'' in the early years of the focal individuals' lives tend to do so throughout their life-cycle.\footnote{We end at age 50 as that is about the highest age for which we observe the outcomes for all our focal individuals.}

\begin{figure}
    \centering
    \subfloat[Income (DKK)]{\label{subfig: age-plot-by-district-income}\includegraphics[width=0.5\linewidth]{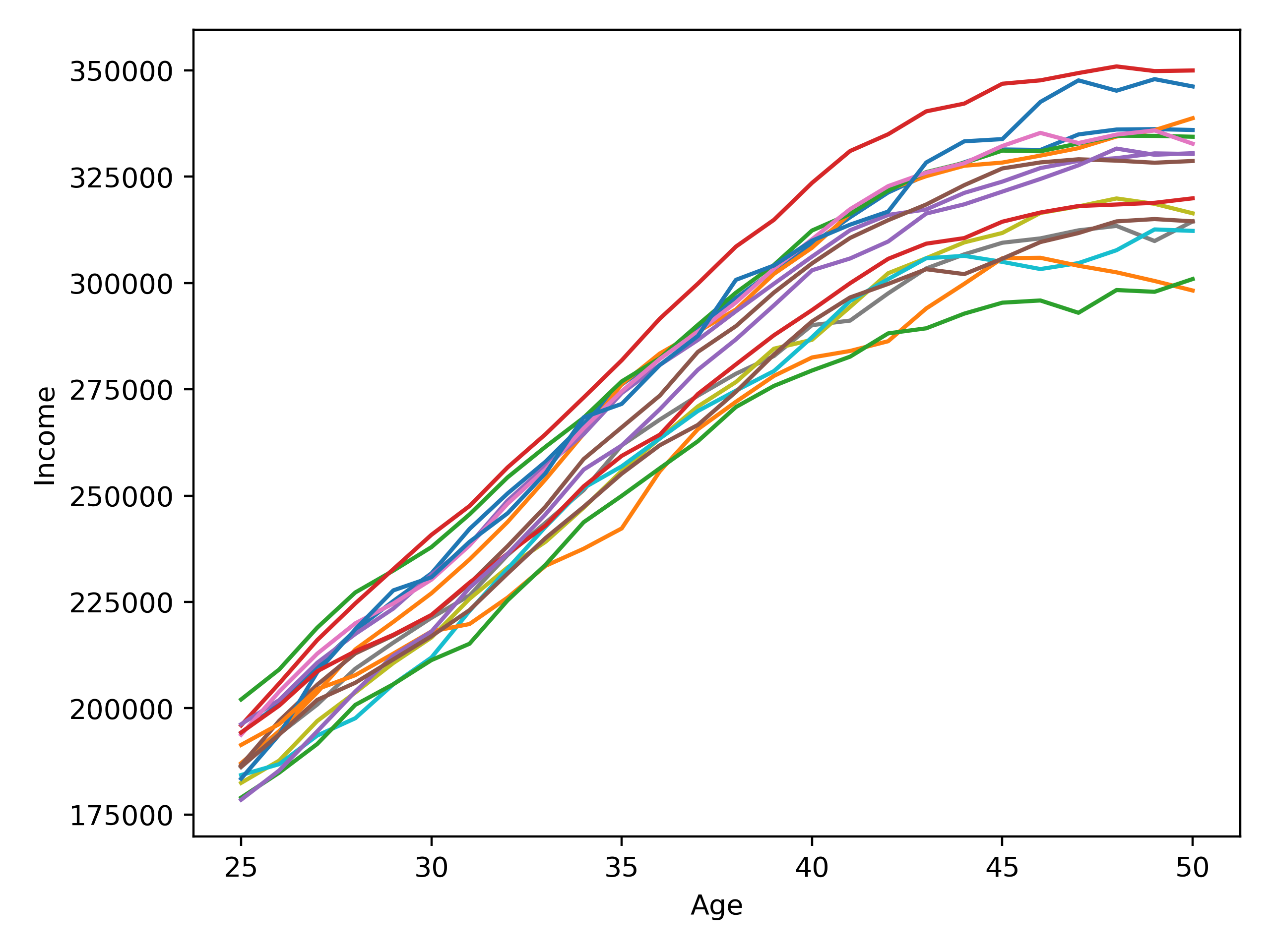}}
    \subfloat[Employed]{\label{subfig: age-plot-by-district-employed}\includegraphics[width=0.5\linewidth]{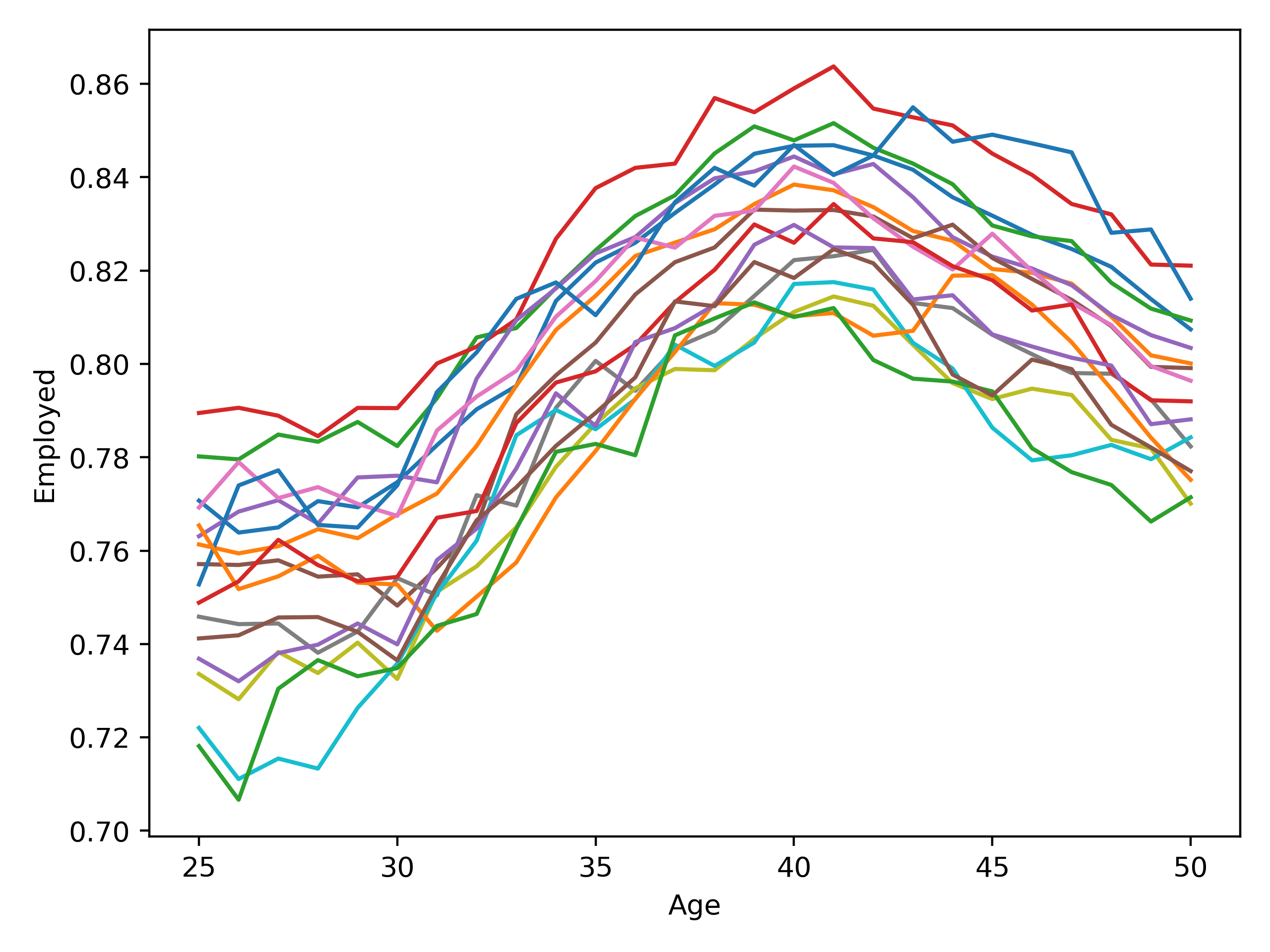}}
    \caption{Life-Cycle Outcomes by Nurse District: Ages 25-50.}
    \label{fig: age-plot-by-district}
    \begin{minipage}{1\linewidth}
        \vspace{1ex}
        \footnotesize{
        \textit{Notes:}
        The figure shows the average income (Panel~\ref{subfig: age-plot-by-district-income}) and share of time in employment (Panel~\ref{subfig: age-plot-by-district-employed}) for our focal individuals during ages 25-50.
        Each line represents a separate nurse district, and its value is the average value of the respective variable at the given age.
		}
	\end{minipage}
\end{figure}

We claim these differences between districts (and to a smaller extent cohorts) largely drive the apparent non-random allocation of children to nurses, and that properly handling this leads to close-to random allocation within district-by-year groups (see Appendix Figure~\ref{fig: obs-by-district-year-and-nurse-year} for district-by-year group sample sizes).
To show this, we once again turn to the rank-order correlation coefficient tests we performed above for nurses, but now do so within each district-by-year group of our sample.
This leads to 137 (instead of 144, due to a few district-by-year groups containing exactly one nurse) separate district-by-year groups, and for each we obtain the p-value for the Spearman rank-order correlation coefficient test.\footnote{The number of nurses within each district-by-year group vary, and we thus apply a flexible inference approach that calculates the exact p-value whenever this is possible in fewer than 10,000 permutations and otherwise sample 10,000 permutations at random.}
Figure~\ref{fig: correlation-coefficient-cdf-nurses} shows the empirical cumulative density function (CDF) of the p-values obtained this way (colored lines).
The black, dashed line indicates the asymptotic CDF given no rank-order correlation; given random allocation of children to nurses within these district-by-year groups, we expect the pairs with (at least) one pre-treatment variable to lie close to this line.
Indeed, we also see this pattern, while also observing that the rank-order correlation between average income during ages 25-50 and years of education is statistically significant, which implies that while nurses did appear to have been allocated children at random, the children they visited systematically differ in how well they fare during their life.

\begin{figure}
    \centering
    \includegraphics[width=0.75\linewidth]{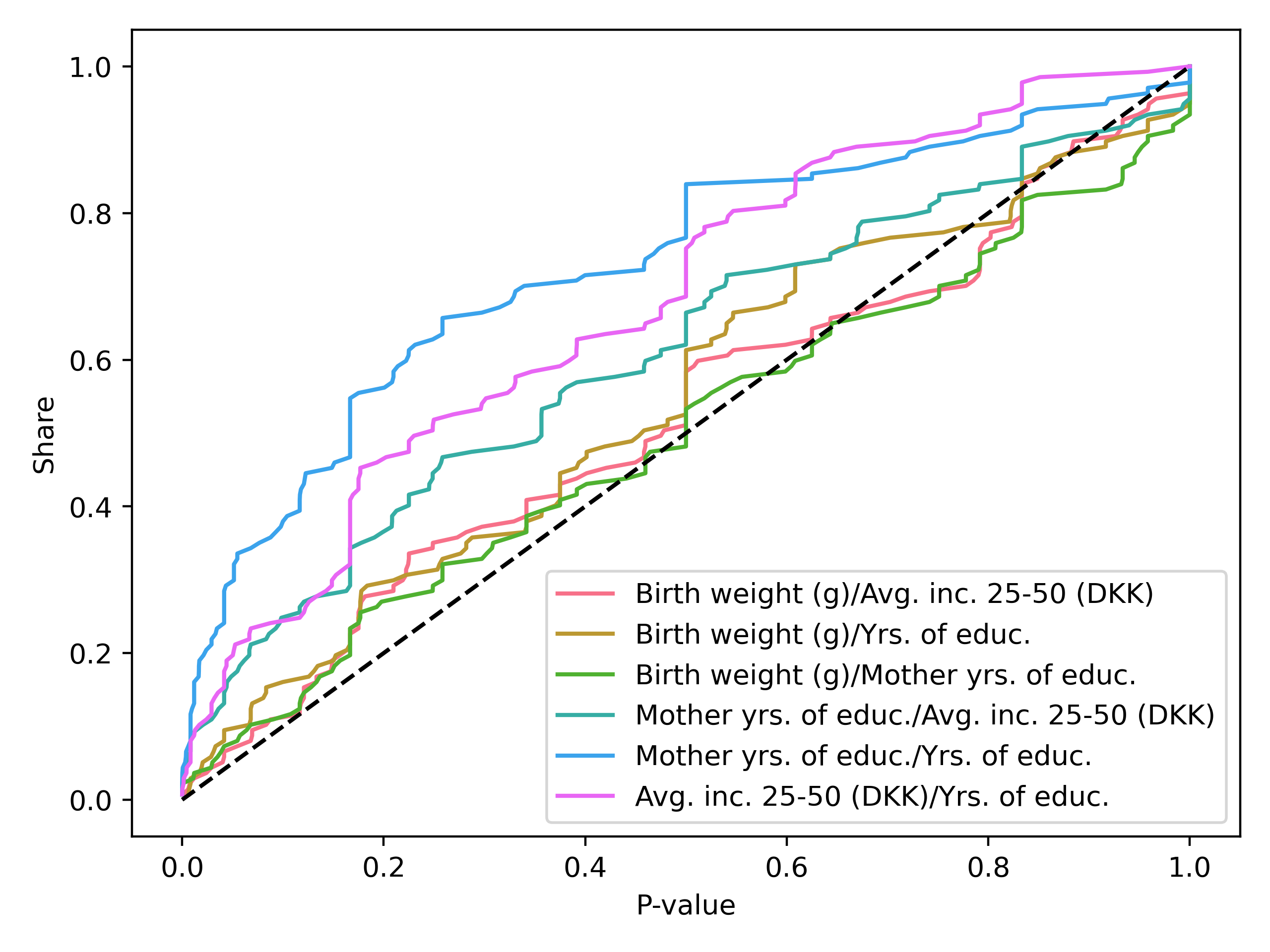}
    \caption{Rank-Order Correlation  CDF of Mean Outcomes by Nurses Within Same District-by-Year.}
    \label{fig: correlation-coefficient-cdf-nurses}
    \begin{minipage}{1\linewidth}
        \vspace{1ex}
        \footnotesize{
        \textit{Notes:}
        The figure shows p-values from Spearman rank-order correlation tests.
        For each district-by-year group, we obtain the mean value of each variable for each nurse, and we then calculate the rank-order correlation between nurses (i.e., comparing the ranks of nurses within the district-by-year group for one variable with the ranks of nurses for another variable, again within the specific district-by-year group).
        To obtain p-values, we apply a flexible inference approach that calculates the exact p-value whenever this is possible in fewer than 10,000 permutations and otherwise sample 10,000 permutations at random.
        For each pair of variables, we then plot the empirical CDF of the obtained p-values.
        The jump at $0.5$ arise as a consequence of the presence of district-by-year groups with exactly two nurses.\footnote{Exact p-values are 0.5 by design in such cases, adding mass to the midpoints.}
        }
	\end{minipage}
\end{figure}

To further investigate whether some selection between nurses could still take place within the district-by-year groups, we return to our earlier strategy to non-parametrically compare groups of children allocated different nurses by means of a Kruskal-Wallis H-test for independent samples \citep{kruskal1952use}.
Now, however, we do so separately for each district-by-year group and then plot the empirical CDF (as well as the theoretical CDF given no selection).
Figure~\ref{fig: pval-cdf-controls} shows the results for select pre-treatment variables, with each colored line representing a pre-treatment variable and the black, dashed line indicating the theoretical CDF given no selection.
Reassuringly, all empirical CDFs lie relatively close to the dashed line, indicating limited differences between groups of children allocated different nurses (within district-by-year groups); however, differences do not vanish entirely for all variables.

\begin{figure}
    \centering
    \includegraphics[width=0.75\linewidth]{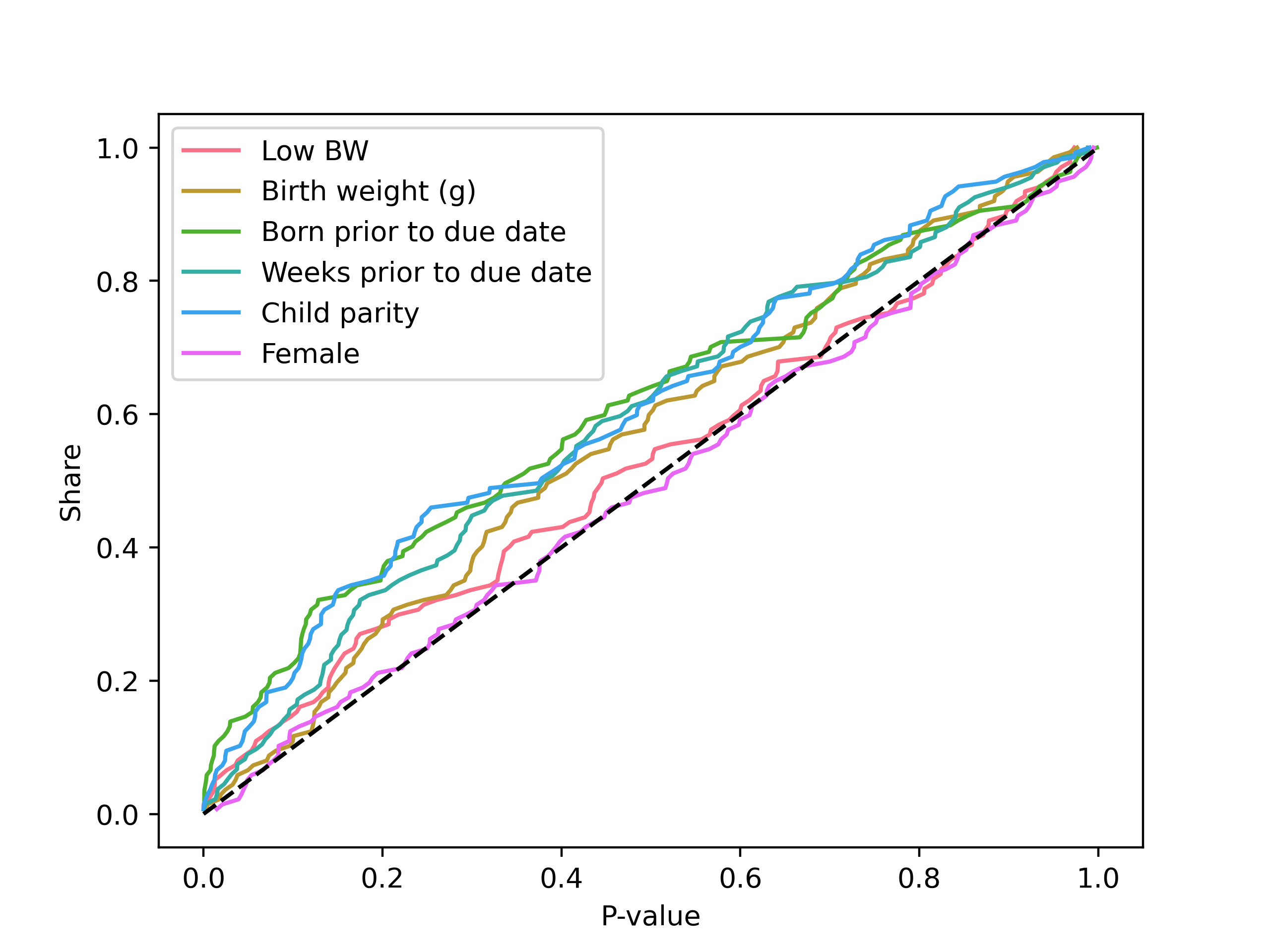}
    \caption{Non-Parametric P-Values for Differences Between Samples of Children by Nurse within District-by-Year: Pre-Treatment Variables.}
    \label{fig: pval-cdf-controls}
    \begin{minipage}{1\linewidth}
        \vspace{1ex}
        \footnotesize{
        \textit{Notes:}
        The figure shows empirical CDFs for p-values obtained from a Kruskal-Wallis H-test for independent samples for select pre-treatment variables.
        Within each district-by-year group, we test for differences in medians for groups of children allocated different nurses.
        }
	\end{minipage}
\end{figure}



\subsection{Estimation Results}
\label{sec: estimation-results}

Having established what appears to be close-to random allocation of children to nurses within nurse district-by-year groups, we next turn to estimating the ``treatment effects'' of nurses.
Since children appear to have been more or less randomly allocated a nurse from the population of nurses in the district-by-year group they belong to, we can estimate the average potential outcome of children assigned a specific nurse by simply calculating the sample average for that nurse.
First, however, we verify that such differences are indeed present.

Figure~\ref{fig: pval-cdf-outcomes} shows empirical CDFs for p-values from a Kruskal-Wallis H-test for independent samples between nurses in the same district-by-year group for select outcome variables (i.e., done similarly to Figure~\ref{fig: pval-cdf-controls} but now for outcomes).
Panel~\ref{subfig: pval-cdf-outcomes-childhood} shows results for childhood outcomes (obtained from the nurse records within the first year of a child's life) and Panel~\ref{subfig: pval-cdf-outcomes-adulthood} for adulthood outcomes (obtained from registers).
The colored lines indicate outcomes and the dashed, black lines the theoretical CDFs in a setting of no differences between samples.
As is evident from both panels, and particularly so for the childhood outcomes of Panel~\ref{subfig: pval-cdf-outcomes-childhood}, the empirical CDFs lie far above the black, dashed lines, indicating statistically significant differences between groups of children allocated different nurses, even when only comparing against nurses within the same district-by-year.
Compared to the empirical CDFs for pre-treatment variables shown in Figure~\ref{fig: pval-cdf-controls}, where the CDFs were close to the dashed line indicating no differences between groups, there now emerge differences far more statistically significant.

\begin{figure}
    \centering
    \subfloat[Childhood outcomes]{\label{subfig: pval-cdf-outcomes-childhood}\includegraphics[width=0.5\linewidth]{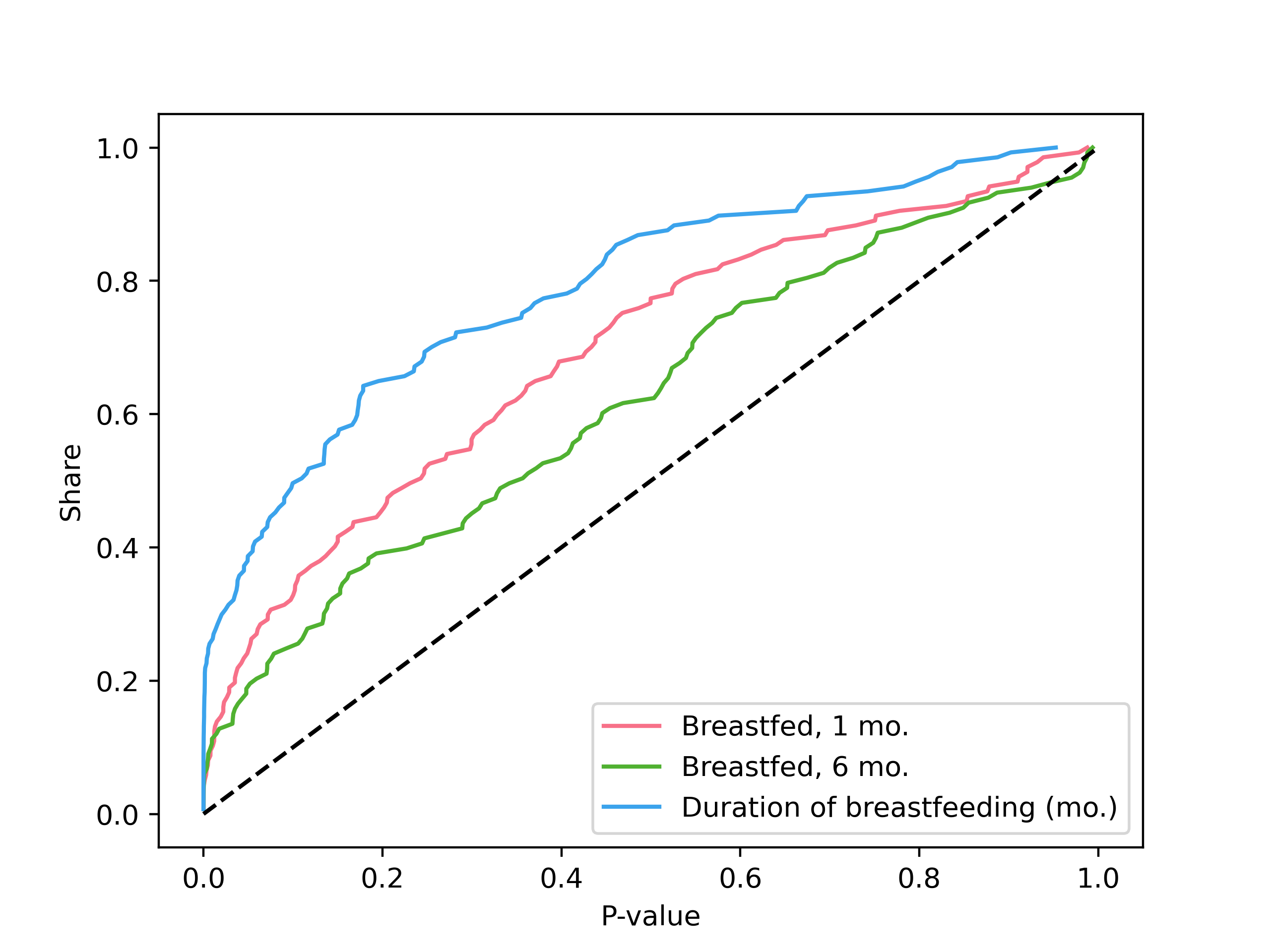}}
    \subfloat[Adulthood outcomes]{\label{subfig: pval-cdf-outcomes-adulthood}\includegraphics[width=0.5\linewidth]{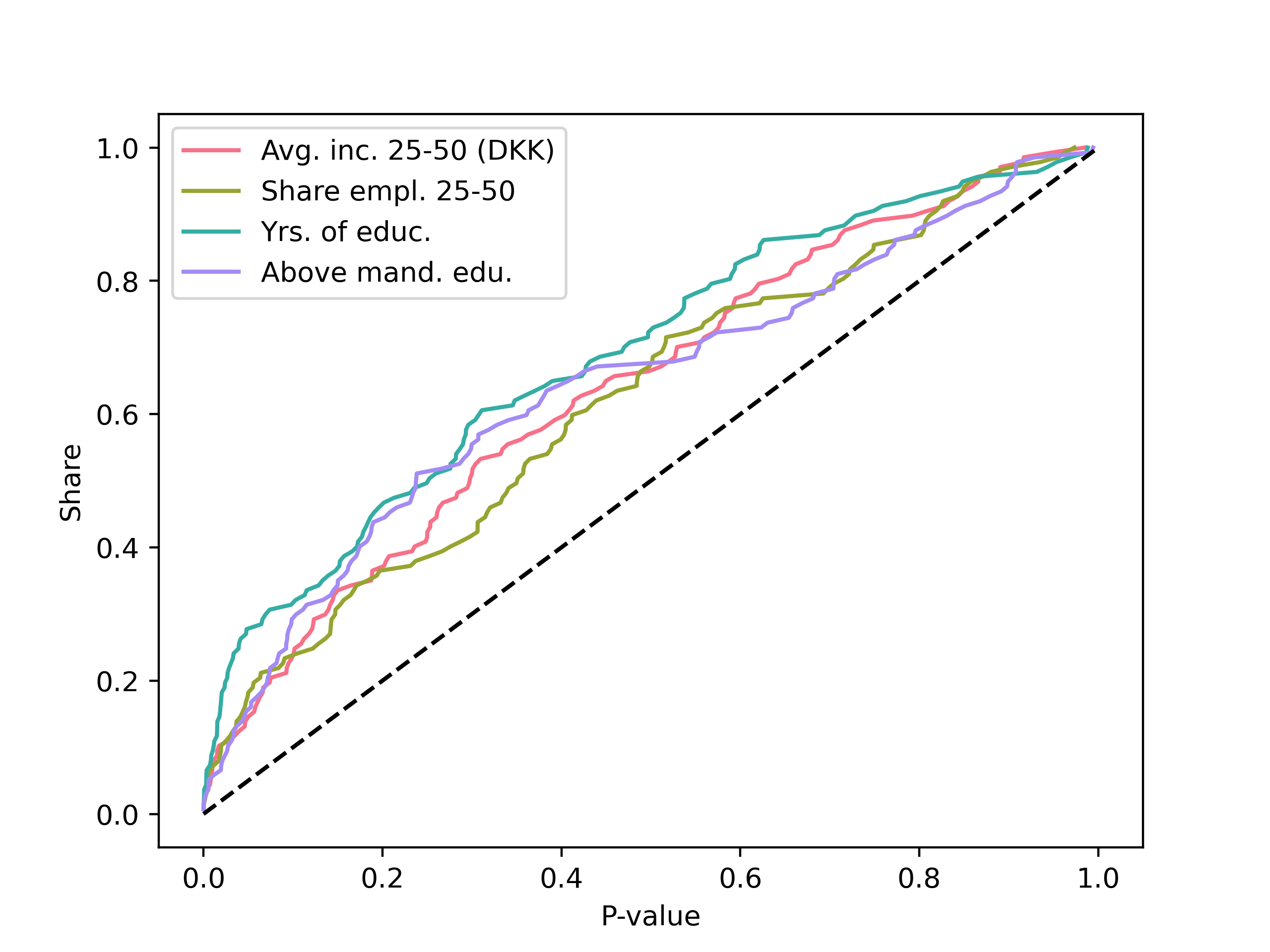}}
    \caption{Non-Parametric P-Values for Differences Between Samples of Children by Nurse within District-by-Year: Outcome Variables.}
    \label{fig: pval-cdf-outcomes}
    \begin{minipage}{1\linewidth}
        \vspace{1ex}
        \footnotesize{
        \textit{Notes:}
        The figure shows empirical CDFs for p-values obtained from a Kruskal-Wallis H-test for independent samples for select outcome variables; Panel~\ref{subfig: pval-cdf-outcomes-childhood} for childhood outcomes and Panel~\ref{subfig: pval-cdf-outcomes-adulthood} for adulthood outcomes.
        Within each district-by-year group, we test for differences in medians for groups of children allocated different nurses.
		}
	\end{minipage}
\end{figure}

While the evidence above indicates strong, statistically significant differences between nurses in terms of their treatment effects, it does little to gauge the magnitudes of these differences.
To illustrate the magnitudes of the differences, Figure~\ref{fig: box-plot-by-district-outcomes} shows box-plots of average outcomes by nurse-year for each nurse district.
Evidently, there is significant heterogeneity both within and between districts in terms of averages of children allocated different nurses.
Further, as supported by Appendix Figure~\ref{fig: correlation-coefficient-districts}, districts that do well on one measure tends to also do well on others.
However, while there are large differences between districts, differences within districts (i.e., between nurse-years within the same district) are in some cases larger still.

\begin{figure}
    \centering
    \subfloat[Years of education]{\label{subfig: box-plot-by-district-edulen}\includegraphics[width=0.5\linewidth]{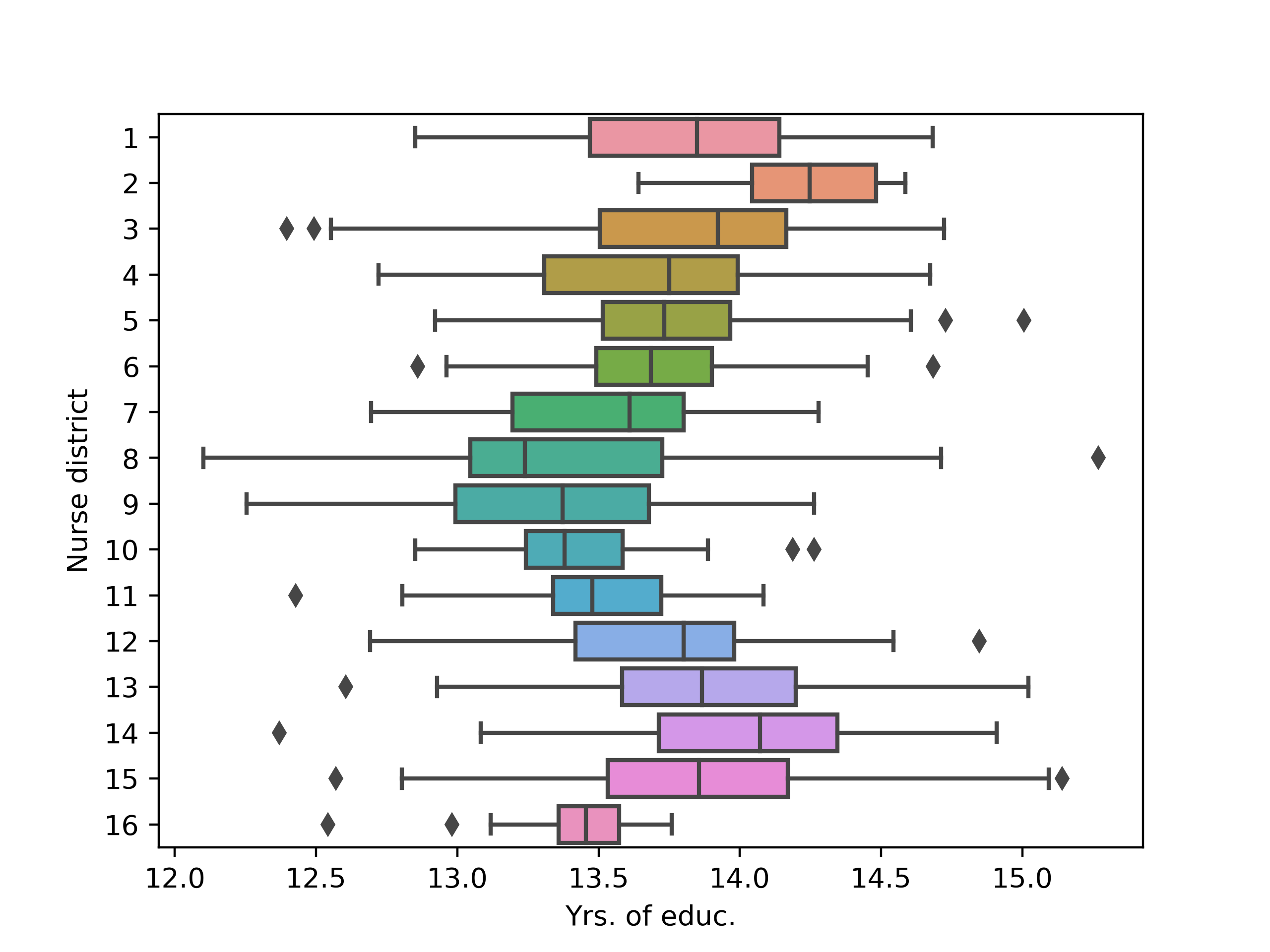}}
    \subfloat[Above mandatory education]{\label{subfig: box-plot-by-district-above_mand_edu}\includegraphics[width=0.5\linewidth]{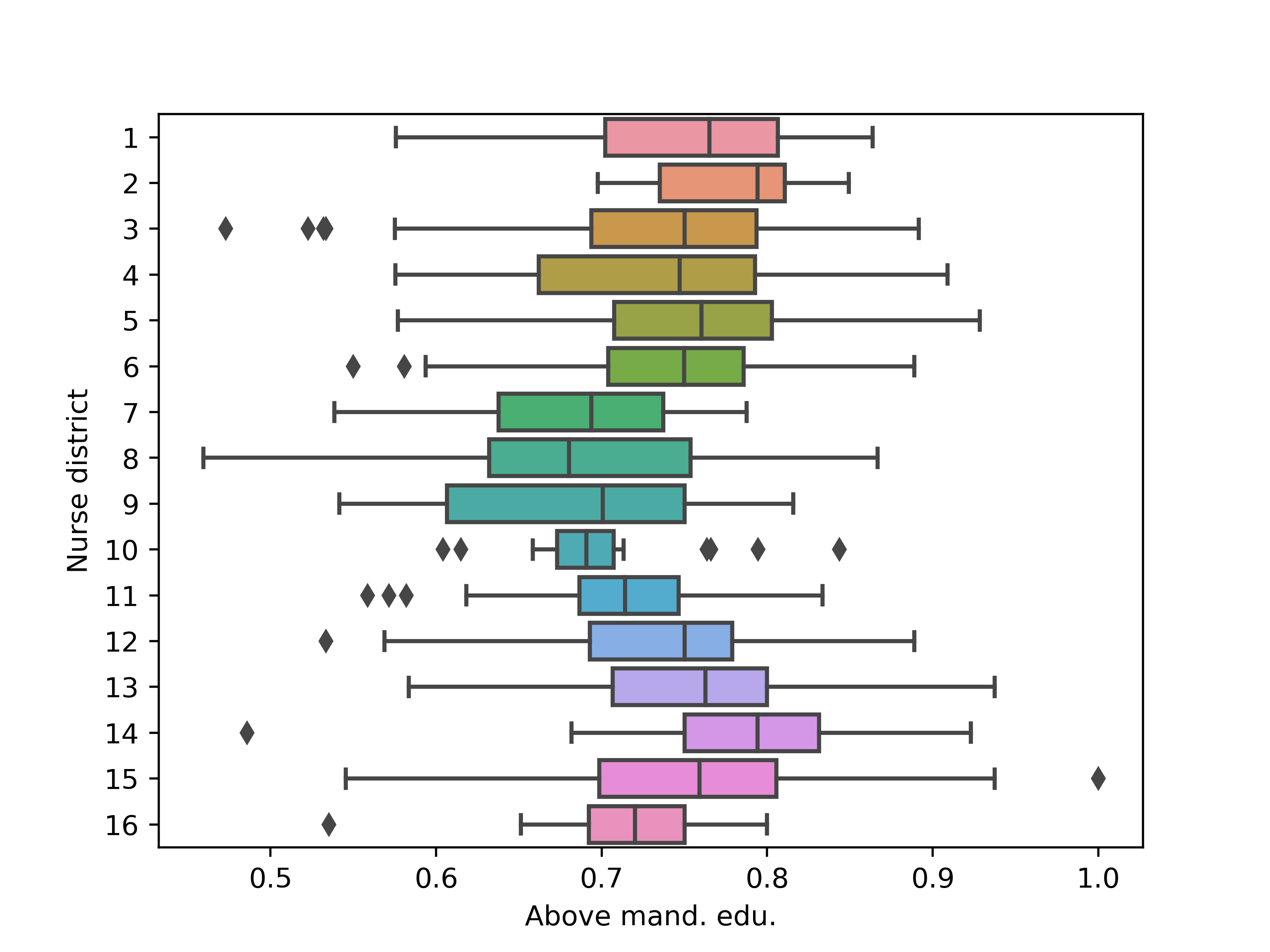}}
    
    \subfloat[Share of time employed 25-50]{\label{subfig: box-plot-by-district-avg_emp_25_50}\includegraphics[width=0.5\linewidth]{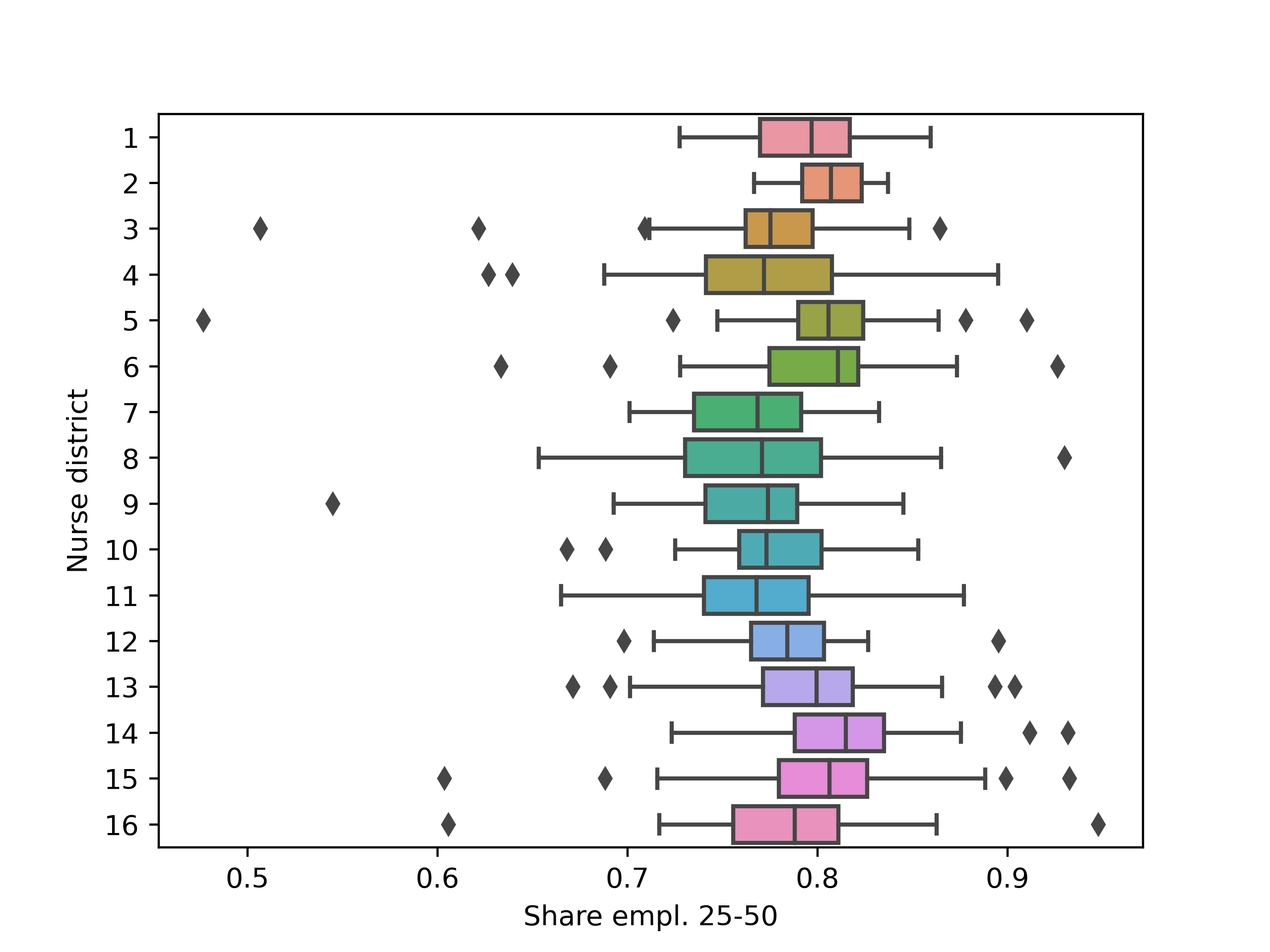}}
    \subfloat[Average income 25-50]{\label{subfig: box-plot-by-district-avg_inc_25_50}\includegraphics[width=0.5\linewidth]{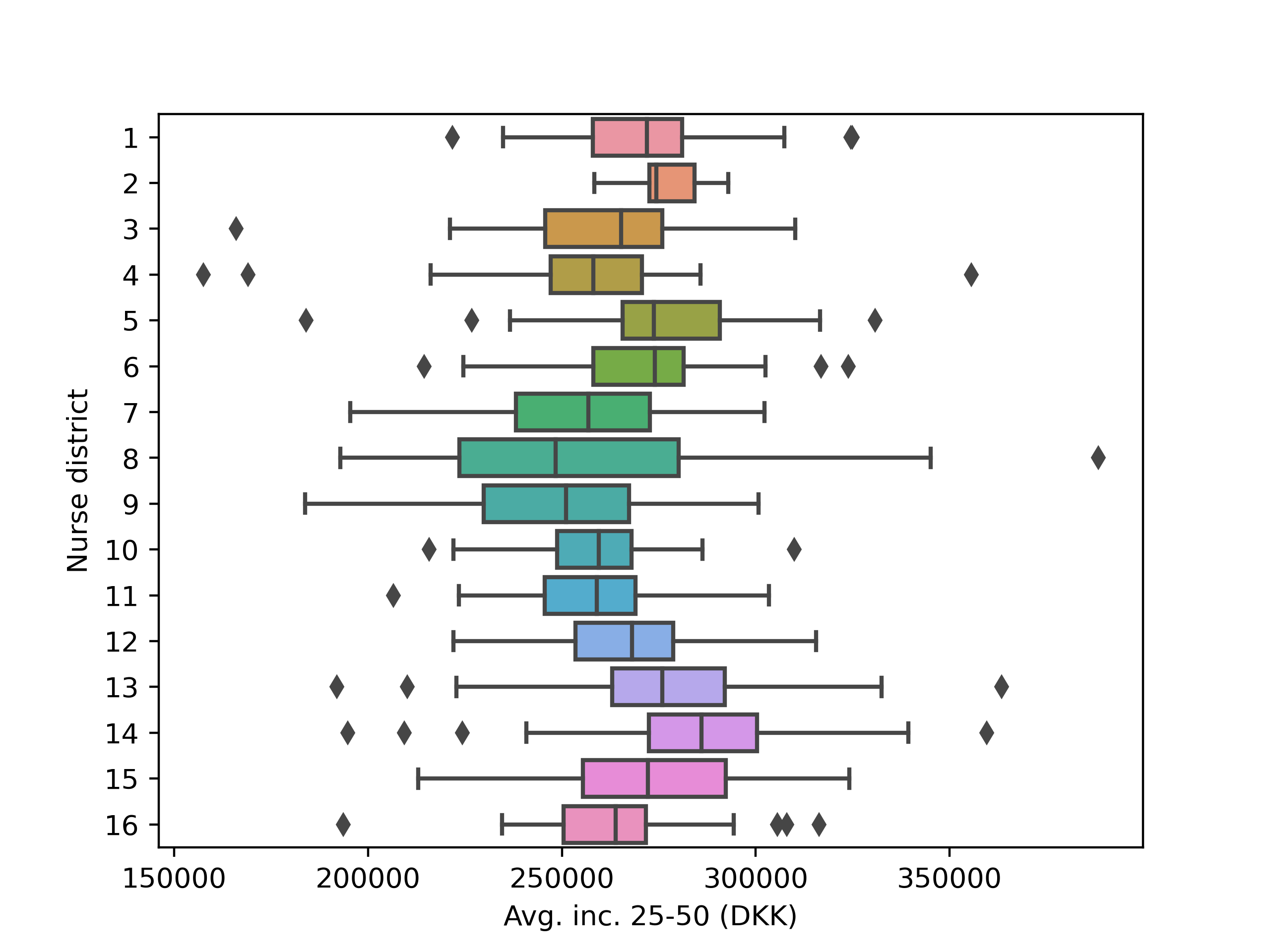}}
    \caption{Box Plots of Average Outcomes by Nurse-Year for Each Nurse District.}
    \label{fig: box-plot-by-district-outcomes}
    \begin{minipage}{1\linewidth}
        \vspace{1ex}
        \footnotesize{
        \textit{Notes:}
        The figure shows box plots of average outcomes by nurse-year for each district, meaning that each point represents an average of the children allocated a specific nurse born within a specific year.
        Panel~\ref{subfig: box-plot-by-district-edulen} shows box plots for average years of education by nurse-year for each of 16 nurse districts, Panel~\ref{subfig: box-plot-by-district-above_mand_edu} for the share of individuals attaining higher than mandatory education, Panel~\ref{subfig: box-plot-by-district-avg_emp_25_50} for the share of time in employment during ages 25-50, and Panel~\ref{subfig: box-plot-by-district-avg_inc_25_50} for average income during ages 25-50.
		}
	\end{minipage}
\end{figure}

While Figure~\ref{fig: box-plot-by-district-outcomes} provides an overview of differences within and between nurse districts, it ``aggregates'' across cohorts, meaning that each box plot contains averages for different years, meaning that some of the variation comes from differences between cohorts rather than nurses.
For this reason, Appendix Figures~\ref{fig: box-plot-by-district-and-year-edulen} and \ref{fig: box-plot-by-district-and-year-avg_inc_25_50} show equivalents of Panels \ref{subfig: box-plot-by-district-edulen} and \ref{subfig: box-plot-by-district-avg_inc_25_50} of Figure~\ref{fig: box-plot-by-district-outcomes}, respectively, but where each panel now refers to a different year of birth of the children.
While the figures are more noisy, significant heterogeneity is still present both between and within nurse districts, documenting that the differences of Figure~\ref{fig: box-plot-by-district-outcomes} are not solely driven by cohort effects.

A potential concern regarding the differences between nurses documented above is that they may arise through noise:
Since raw averages constitute the points of the box plots, the variance of these estimates are not accounted for.
For this reason, we plot both the point estimate and its associated 95\% confidence interval for each nurse, split by district and year, in Appendix Figure~\ref{fig: ci-plot-by-district-and-year-edulen} (for years of education) and Appendix Figure~\ref{fig: ci-plot-by-district-and-year-edulen} (for average earnings during ages 25-50).
While some estimates are relatively imprecise (with wide confidence intervals), there is nonetheless clear differences between the outcomes of different nurses, even within the same district-by-year group, and these differences are economically significant  (e.g., cases of more than one year of education on average between the children of two different nurses within the same district and year).

While the evidence above supports the hypothesis of differences between nurse skills leading to differences in outcomes between children allocated different nurses, its magnitude and interpretation is hampered by the complex setting of many nurses as well as district-by-year groups.
We therefore turn to a simpler method of assessing what differences can be expected between two groups of otherwise similar children which happen to be allocated different nurses.
Specifically, we calculate, for each district-by-year group, the average values of our outcomes by each nurse, and then take all distinct pairs of nurses and calculate the absolute value of the difference in the the average values of their children.
The average of these absolute differences is then an estimate of the expected difference between two otherwise similar groups of children, but which happen to have been allocated different nurses.

Figure~\ref{fig: average-treatment-effects} shows the empirical distributions of these absolute differences in the averages of children with different nurses but within the same district-by-year group, along with the average value of these absolute differences (black, solid line).
As shown, these differences are economically large:
For example, two otherwise similar groups of children that happened to be allocated different nurses differ, on average, by nearly half a year of education and around DKK 27,000 in annual earnings (note, however, that large parts of these differences arise due to considering absolute differences; nevertheless, as we show later, differences beyond those that might arise due to noise are present). 
Further, these differences are not driven by few, extreme differences, but rather by relatively large differences by a large share of nurse-pairs. 

\begin{figure}
    \centering
    \subfloat[Breastfed at 1 mo.]{\label{subfig: average-treatment-effects-breastfeeding_1_mo_pred}\includegraphics[width=0.33\linewidth]{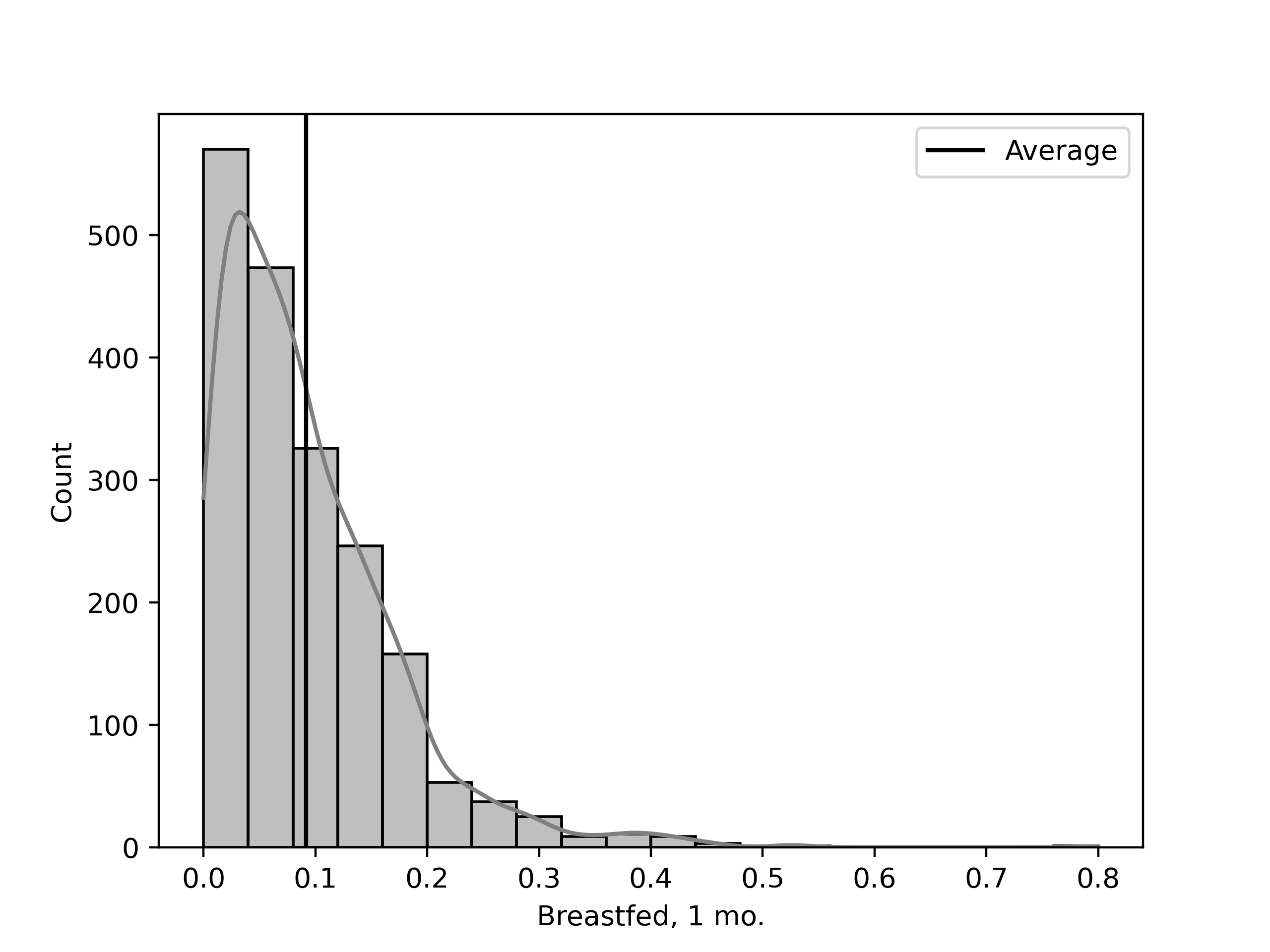}}
    \subfloat[Breastfed at 6 mo.]{\label{subfig: average-treatment-effects-breastfeeding_6_mo_pred}\includegraphics[width=0.33\linewidth]{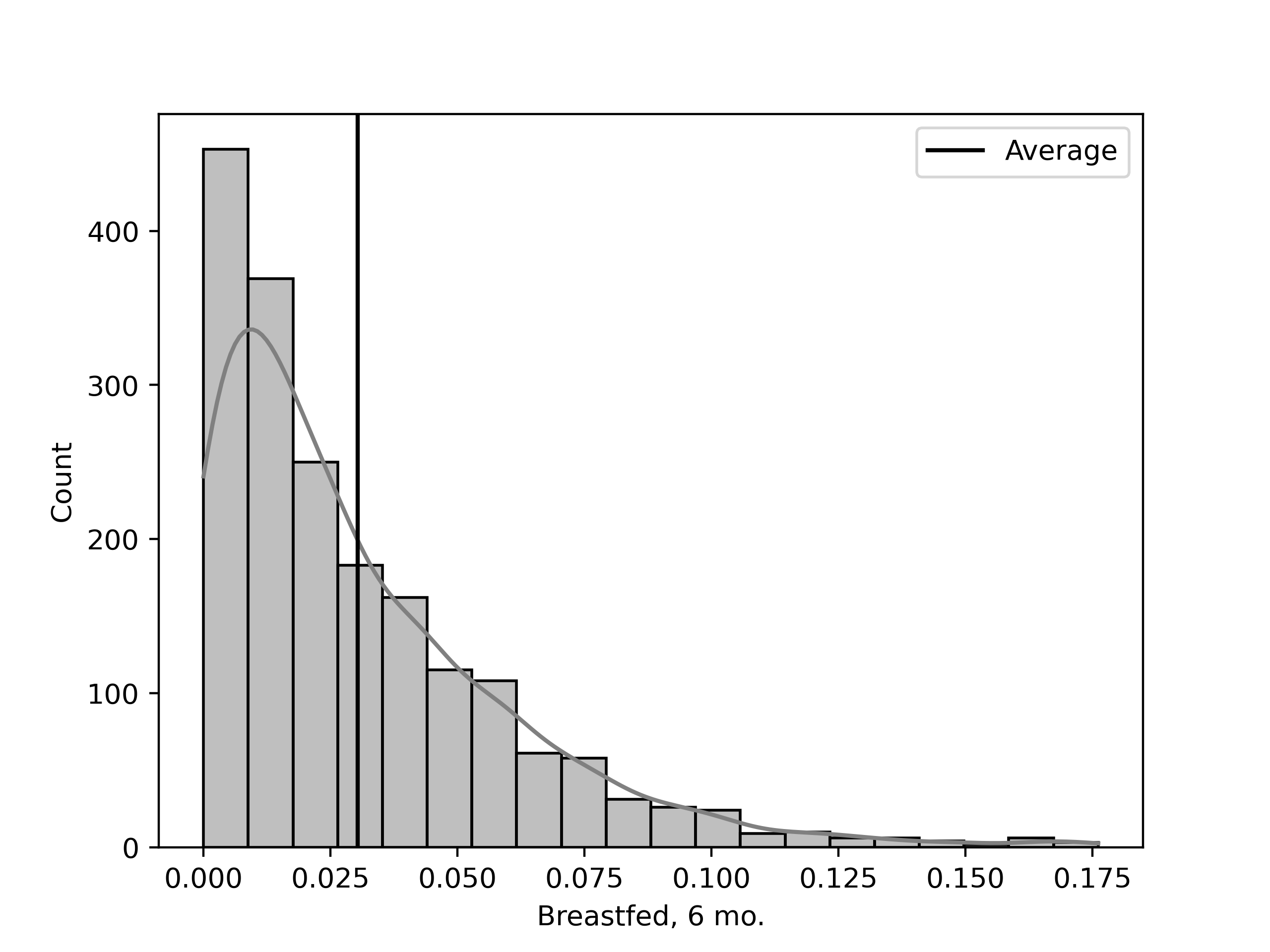}}
    \subfloat[Duration of breastfeeding (mo.)]{\label{subfig: average-treatment-effects-bfdurany_pred}\includegraphics[width=0.33\linewidth]{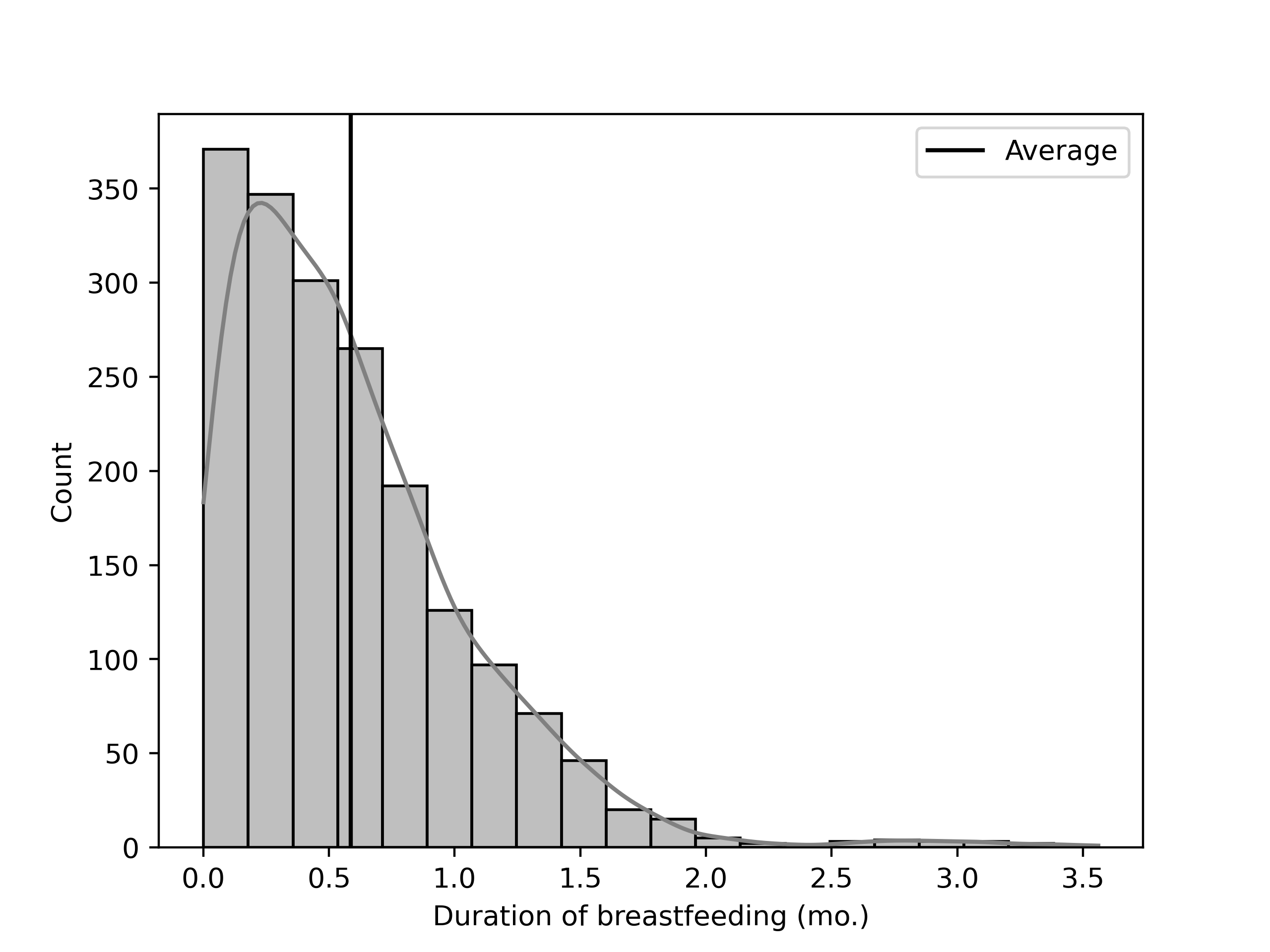}}

    \subfloat[Years of education]{\label{subfig: average-treatment-effects-edulen}\includegraphics[width=0.4\linewidth]{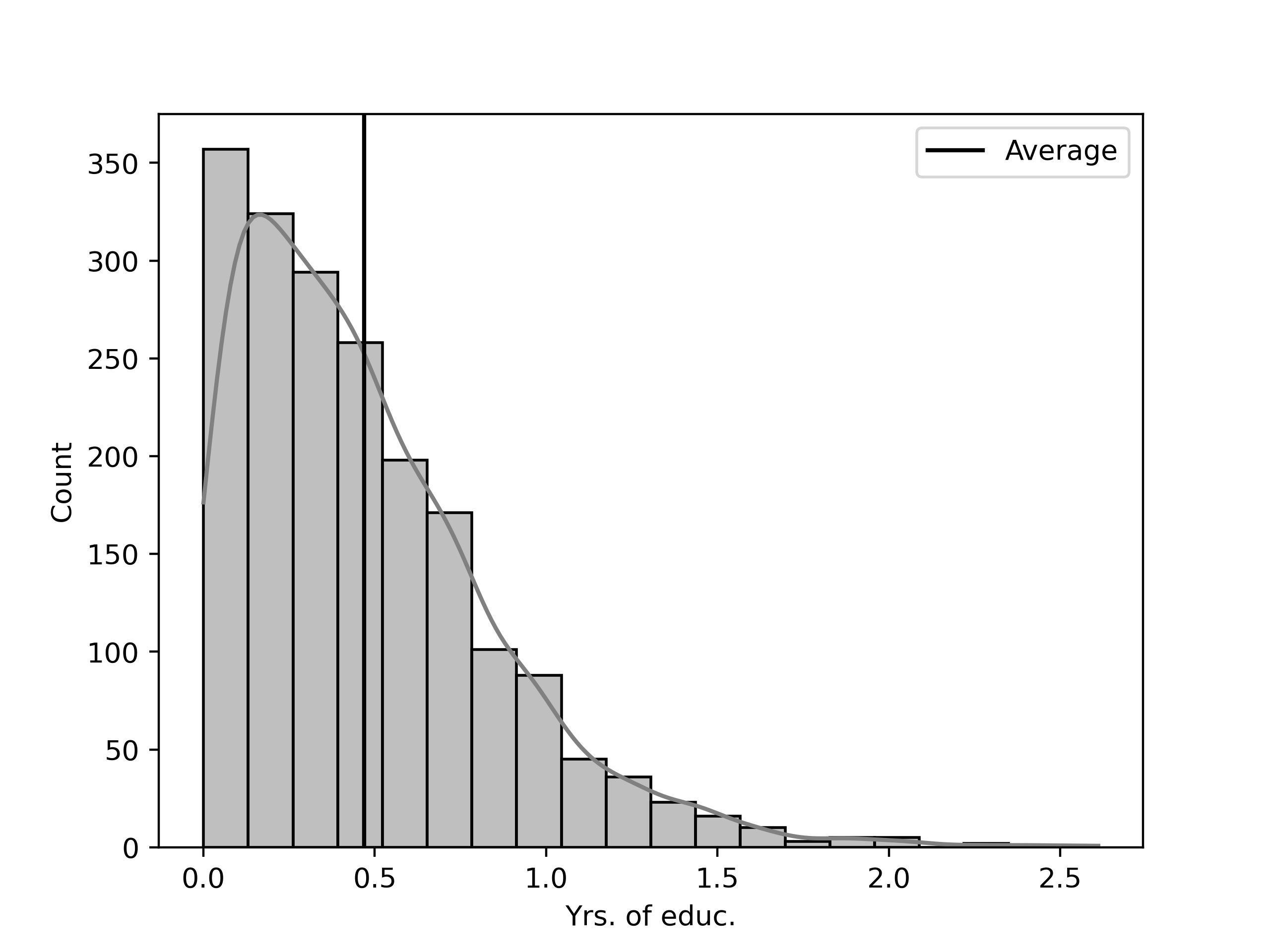}}
    \subfloat[Above mandatory education]{\label{subfig: average-treatment-effects-above_mand_edu}\includegraphics[width=0.4\linewidth]{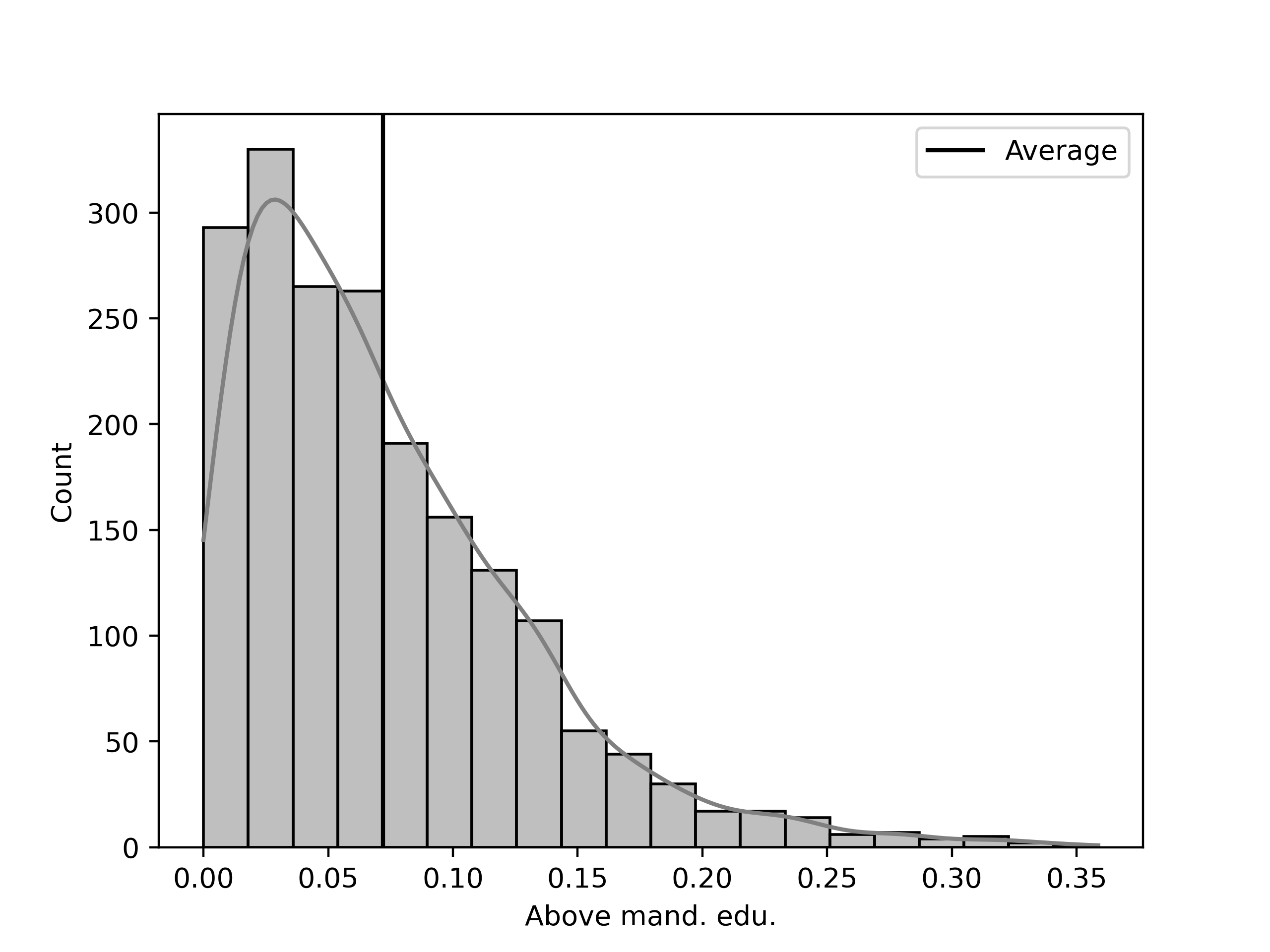}}

    \subfloat[Share of time employed 25-50]{\label{subfig: average-treatment-effects-avg_emp_25_50}\includegraphics[width=0.4\linewidth]{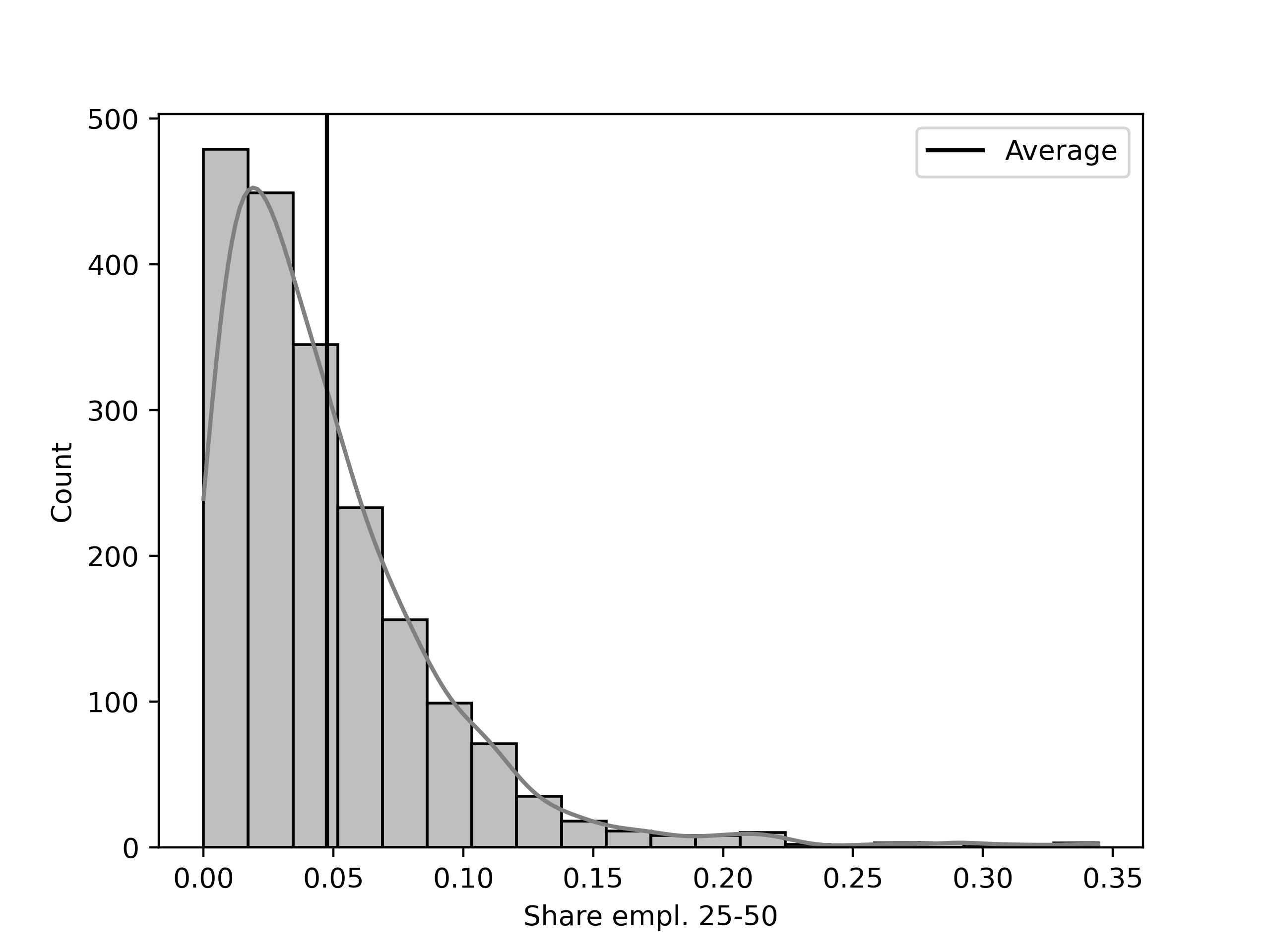}}
    \subfloat[Average income 25-50]{\label{subfig: average-treatment-effects-avg_inc_25_50}\includegraphics[width=0.4\linewidth]{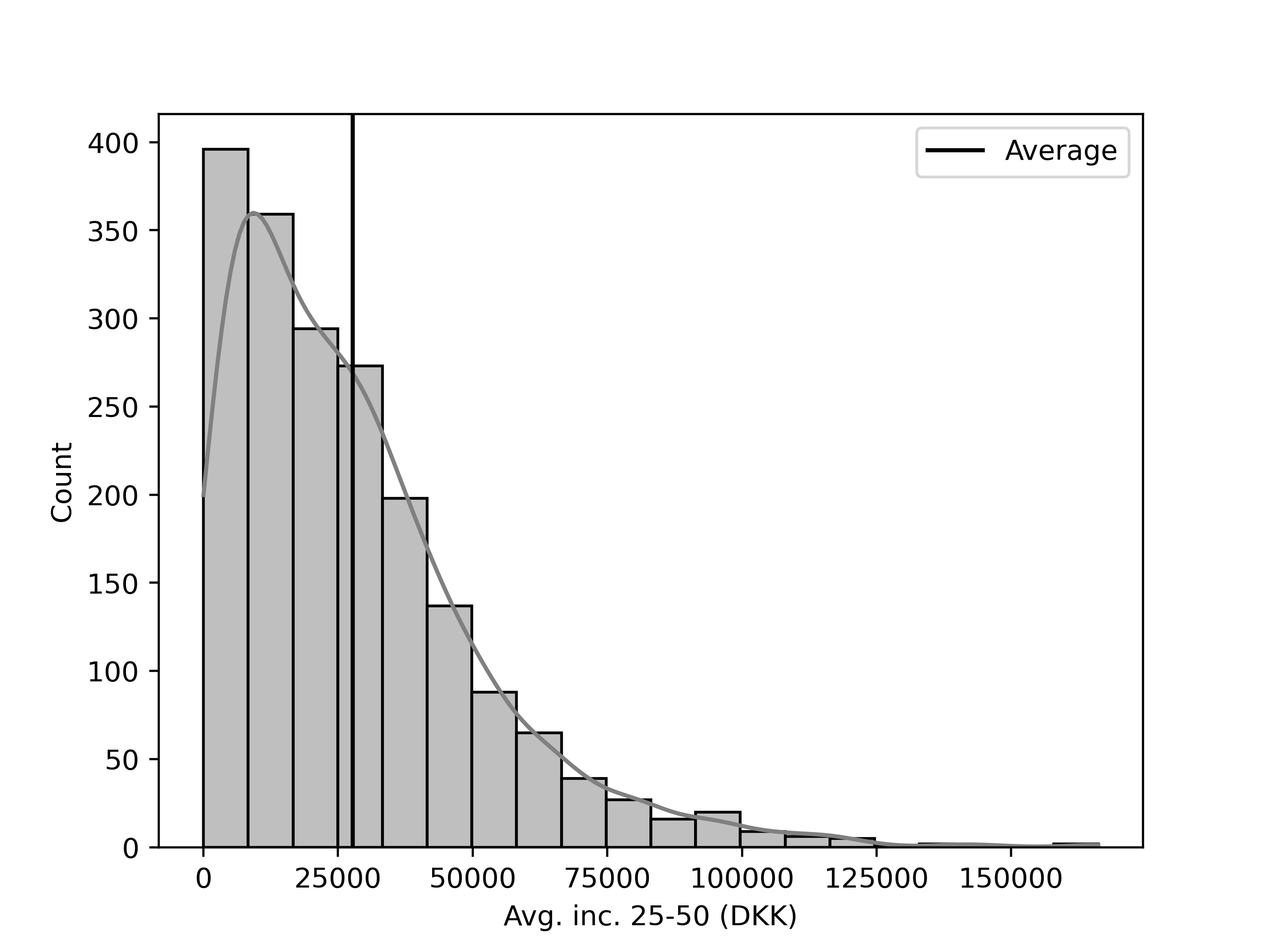}}
    
    \caption{Average Absolute Differences Between Children of Different Nurses Within Same District-by-Year.}
    \label{fig: average-treatment-effects}
    \begin{minipage}{1\linewidth}
        \vspace{1ex}
        \footnotesize{
        \textit{Notes:}
        The figure shows histograms (with associated density plots) of absolute values of differences in average values of children with different nurses but within same district-by-year group.
        The black, solid line in each panel shows the average of the underlying absolute differences.
        The different panels refer to different outcome variables.
		}
	\end{minipage}
\end{figure}

While the findings in Figure~\ref{fig: average-treatment-effects} document large differences in outcomes of children allocated different nurses, the measure is by design bounded below at zero, and so the average differences will be larger than zero, even asymptotically and with no actual differences between nurses.
Therefore, we want to ensure that these differences are not artefacts due to noise but represent real differences between nurses.
However, the complex nature of the setup does not lend itself well to normal inference, and we therefore turn to a permutation-based strategy for inference.

Suppose nurses do not differentially affect children, but that potential outcomes differ by districts and cohorts.
Then we could, separately for each district-by-year group, permute the children, leading us to -- in the counterfactual setting -- assign at random $k$ children within the same district-by-year group to some nurse with $k$ children in the real setting.\footnote{We need to keep the number of children for each nurse fixed to not bias our results to show as always significant: Since we measure un-weighted differences between pairs, removing noise in number of children by nurse would lead us to over-estimate the statistical significance of our findings, as this equalization in sample size would reduce statistical noise.} 
Doing so for each nurse within each district-by-year group, we obtain a permutation which, if nurses did not differentially affect children, would follow the same distribution as our actual sample.
Repeating this would then lead to a counterfactual distribution of average absolute differences, and we can then get an estimate of the exact p-value by calculating the share of permutations that led to a larger average absolute difference than the one in the actual data.\footnote{Exact p-values are computationally infeasible, which is why we only obtain an estimate. The test is designed such that a more extreme value always corresponds to a greater value, which leads to the right-tail focus.}
Specifically, we perform 1000 permutations and obtain for each a counterfactual estimate of the average absolute difference in the permuted sample.

Table~\ref{tab: average-treatment-effects} shows the average absolute differences in outcomes between children of different nurses and, to take potential differences arising due to noise into account, the ``excess'' difference, i.e., the absolute average difference beyond those obtained on average in the permutation tests.
Further, mean values of outcomes are reported (MDV), and the table also shows p-values obtained from the permutation tests, i.e., each p-value is the share of times a counterfactual estimate based on a permutation attained a more extreme value than the one from the non-permuted sample.
As is evident, all the differences are highly statistically significant, as are they economically significant (e.g., differences of more than one month in length of education and income differences of around 1.4\%).
Appendix Figure~\ref{fig: counterfactual-permutation-treatment-effects} shows, for each variable in Table~\ref{tab: average-treatment-effects}, histograms of the averages obtained from the permutation tests, clearly showcasing that the averages obtained in the actual sample are larger than those which could appear due to random noise.

\begin{table}
    \centering
    \centering
    \caption{Average Differences in Outcomes by Nurses Within District-by-Year Groups.}
    \label{tab: average-treatment-effects}
    \resizebox{1\linewidth}{!}{
        \begin{tabular}{l rrrr}
            \toprule
            Outcome & Average Difference & Excess Difference & MDV & P-value \\
            \midrule
            Breastfed, 1 mo. &      0.09 &     0.02 &       0.58 & 0.000 \\
            Breastfed, 6 mo. &      0.03 &     0.00 &       0.03 & 0.004 \\
            Duration of breastfeeding (mo.) &      0.59 &     0.14 &       2.50 & 0.000 \\
            Avg. inc. 25-50 (DKK) & 27,722.42 & 3,767.50 & 269,452.55 & 0.000 \\
            Share empl. 25-50 &      0.05 &     0.01 &       0.79 & 0.000 \\
            Yrs. of educ. &      0.47 &     0.09 &      13.74 & 0.000 \\
            Above mand. edu. &      0.07 &     0.01 &       0.74 & 0.006 \\
            \bottomrule
        \end{tabular}
    }
    \begin{minipage}{1\linewidth}
        \vspace{1ex}
        \footnotesize{
        \textit{Notes:}
        The table shows estimates of average (absolute) differences in outcomes between two otherwise similar groups of children allocated different nurses (but within the same district-by-year group).
        Since the test statistic we use is not centered around zero, we also report the ``excess'' difference, which is the difference between nurses above that which is a result of noise.
        MDV is the mean of the dependent variable.
        The p-values are from permutation-based tests which randomly match children and nurses within a district-by-year group and is calculated as the share which results in a more extreme value than the one observed in the non-permuted data, and we use a Holm–Šidák approach to account for multiple hypothesis testing.
		}
	\end{minipage}
\end{table}

Taken together, our results showcase large, economically as well as statistically significant, differences in how well children fare as a consequence of which nurse they got allocated for their visits during their first year of life.
Further, recall that these estimates are not estimates of the impact of the nurse visiting program; rather, they are specifically the differences \textit{within} the program as a consequence of the visiting nurse.
As such, our findings suggest that an integral part of understanding the impact of such policies lies in understanding the role of the treatment provider, in addition to the program itself.

When comparing nurses and nurse ``treatment'' effects in this setting it is important, however, to apply some caution:
Any treatment effect is a difference in averages between two nurses within the same nurse district, and so a large treatment effects may be any one or both of (1) a skilled nurse and (2) a poor comparison nurse.
Further, it is only possible to obtain these estimates within district-by-year group without imposing stronger assumptions on the underlying potential outcomes, and thus comparing differences across district-by-year groups must be done with caution.


\subsection{Child-Specific Treatment Effects \& Allocation Mechanisms}
\label{sec: heterogeneity}

While our results documenting important differences in how well children allocated different nurses fare are interesting by themselves, and may have policy implications in terms of added focus on improving the skill of the worst nurses, some differences are likely impossible to erase, due either to skills attained only through years of experience or some inherent differences between nurses which are non-trivial to mitigate, another avenue for potential gains is in optimally ``matching'' children and nurses.
Our hypothesis posits that certain children are particularly likely to benefit from a highly skilled nurse compared to others; for example, perhaps children with poor initial health or born to parents in the lower end of the socioeconomic distribution are more likely to benefit from a high rather than low skill nurse compared to a child born in perfect health to parents from the higher end of the socioeconomic distribution.
If this is the case, there is potentially room for improved impacts of the nurse visiting program by exploiting such information to allocate children to nurses in a better way.

For it to be possible to identify potential gains from improving the allocation mechanism of the nurse visiting program, we require estimates of treatment effects that take into account such heterogeneity.
Further, for any such allocation mechanism to be feasible, it needs to rely on readily available information of the child and/or her family.

In terms of relevant information readily available for potential (re)allocation, the nurse records provide information available immediately from birth in terms of birth weight and length, sex, and number of weeks born prior to due date (if relevant), and from administrative data we obtain other information visible directly at birth in the form of parent characteristics such as age and education.\footnote{While our use of register data here obviously was not available at the time of the sample we study, we only use information from the registers that would have been readily available from the parents and which, today, is immediately available. Nonetheless, we choose to drop parents' years of education, as we observe this only later in the registers and parents may potentially have obtained additional education between the birth of their child and the point at which we observe their education.}
Specifically, we use birth weight, birth length, born prior to due date (and number of weeks when applicable), an indicator for born one of the first three days of a month,\footnote{This is potentially an important factor to consider given the ``Copenhagen trial'', which selected children born during the first three days of a month to an extended, three year NHV program \citep{baker2023universal}.}, parity, an indicator for sex, mother and father age at birth, and an indicator for nurse assessment of socioeconomic status of the family (low, average, high).\footnote{Nurse assessment of the socioeconomic status of the home is not strictly speaking observed before the first visit, as the nurse records it during her visit at child age one month. However, it is unlikely to be affected by the nurse during such a short time span. Our results are robust to excluding this variable from the list of variables we use to estimate heterogeneous treatment effects.} 

With this information, we estimate individual-level treatment effects by means of a series of multi-arm causal forest \citep{wager2018estimation, athey2019generalized, nie2021quasi}:\footnote{These effects are ``individual'' in the sense of being specific with respect to the at-birth information we have available on the child and her family. Two children which are identical with respect to all our measures will not differ in terms of their estimated treatment effects.}
For each district-by-year group, we treat each nurse as a separate treatment arm and estimate heterogeneous treatment effects (for each pair of arms, i.e., comparing pairs of nurses) with respect to the information on the child and her family available to us at her birth.\footnote{Here, we add a new sample restriction in the form of requiring at least 100 children per nurse (within our district-by-year groups) for inclusion, as we otherwise risk imprecise estimates related to very low estimated propensity scores for some nurses.} 

Having estimated a multi-arm causal forest for each district-by-year group, we can use the out-of-bag predictions of the forests to estimate counterfactual outcome distributions under different treatment allocation mechanisms.
Limiting the counterfactual treatment allocation rules to those feasible within the nurse program, we can then estimate potential gains in the nurse visiting program from alternative allocation mechanisms.
Further, we can estimate the total room for gains from an alternative allocation rule by estimating the gains from the optimal allocation mechanism, something which we proved, in Section~\ref{sec: allocation}, is always possible in strongly polynomial time.

Table~\ref{tab: average-outcome-by-allocation} shows estimated average values of our outcome variables under the current as well as under two alternative allocation mechanisms:
One uses a greedy heuristic to identify gains from re-allocation (but which is not guaranteed to achieve the optimal allocation) and one finds the optimal allocation.
The results indicate economically significant room for improvements:
Our estimates for average earnings during ages 25-50 suggest that improved allocation of nurses to children could have resulted in additional yearly income of around (2023) USD 1,815.\footnote{We arrive at this number by adjusting for Danish inflation between 2015 and 2023 (17\%) and using an exchange rate of 0.14 DKK/USD.} 
Turning to education, we estimate the total gains in terms of length of education to be around two months on average. 

\begin{table}
    \centering
    \centering
    \caption{Outcomes Under Alternative Allocation Mechanisms.}
    \label{tab: average-outcome-by-allocation}
    \resizebox{1\linewidth}{!}{
        \begin{tabular}{l rrrr}
            \toprule
            & \multicolumn{3}{c}{Allocation mechanism} & \\ \cline{2-4}
            Outcome & Actual & Greedy & Optimal & No. of obs. \\
            \midrule
            Breastfed, 1 mo. &       0.58 &       0.60 &       0.62 & 43,259 \\
            Breastfed, 6 mo. &       0.03 &       0.04 &       0.04 & 31,449 \\
            Duration of breastfeeding (mo.) &       2.55 &       2.63 &       2.74 & 33,139 \\
            Avg. inc. 25-50 (DKK) & 275,018.01 & 280,699.60 & 286,100.90 & 50,108 \\
            Share empl. 25-50 &       0.80 &       0.81 &       0.82 & 49,186 \\
            Yrs. of educ. &      13.83 &      13.92 &      14.00 & 48,916 \\
            Above mand. edu. &       0.76 &       0.77 &       0.79 & 48,916 \\
            \bottomrule
        \end{tabular}
    }
    \begin{minipage}{1\linewidth}
        \vspace{1ex}
        \footnotesize{
        \textit{Notes:}
        The table shows estimates of average values of outcomes under alternative allocation mechanisms, estimated separately for each outcome.
        The leftmost rule (``Actual'') reports the actual averages as they are under no new allocation rule.
        The next rule (``Greedy'') shows the results of employing a heuristic, greedy re-allocation mechanism.
        The final rule (``Optimal'') shows the results obtained when optimally allocating nurses to children.
		}
	\end{minipage}
\end{table}

How much of the gains from improved allocation is attainable through a simpler re-allocation algorithm?
To answer this question, Table~\ref{tab: average-outcome-by-allocation} reports results from a re-allocation mechanism that attempts to find re-allocations using a greedy heuristic (the column ``Greedy'').
Specifically, it keeps track of the number of children allocated each nurse and then goes through each child in a district-by-year group and allocates the nurse to that child that (1) is not yet at capacity and (2) results in the highest potential outcome for that child (thus obtaining an $\mathcal{O}(n_2)$ re-allocation mechanism).
However, as is clear from our estimates, using this heuristic algorithm leaves significant room for improvements.
This highlights the importance of an exact algorithm for complex tasks like optimal treatment allocation, emphasizing the value of deriving a strongly polynomial algorithm for this task.

While Table~\ref{tab: average-outcome-by-allocation} documents significant room for average improvements through an optimized nurse allocation mechanism, it potentially hides important ways in which these average improvements are obtained:
If those children currently the worst off are negatively impacted (e.g., through complementarity of initial conditions and nurse skill leading to matching of the best nurses to the best off children), welfare may not be improved under the re-allocation (in terms of a social planner with preferences for equality).
Conversely, if those children the worst off are (particularly) positively impacted (e.g., through substitutability of initial conditions and nurse skill leading to matching of the best nurses to the worst off children), welfare may be more positively impacted under re-allocations compared to what the average improvement suggests (in terms of a social planner with preferences for equality).

To study the impact of nurse re-allocation on the distribution (rather than just its expectation) of outcomes, Figure~\ref{fig: density-realized-vs-optimal-outcomes} shows the distribution of our non-binary outcomes under the current as well as under the optimal allocation mechanisms.\footnote{Binary outcomes are left out as Table~\ref{tab: average-outcome-by-allocation} (i.e., their expectations) fully captures the effects on these.}
Across the outcomes, we observe a shift to the right across the entire distribution when moving from the actual to the optimal allocation, suggesting neither strong complementarity nor substitutability.

\begin{figure}
    \centering
    \subfloat[Duration of breastfeeding (mo.)]{\label{subfig: density-realized-vs-optimal-bfdurany_pred}\includegraphics[width=0.5\linewidth]{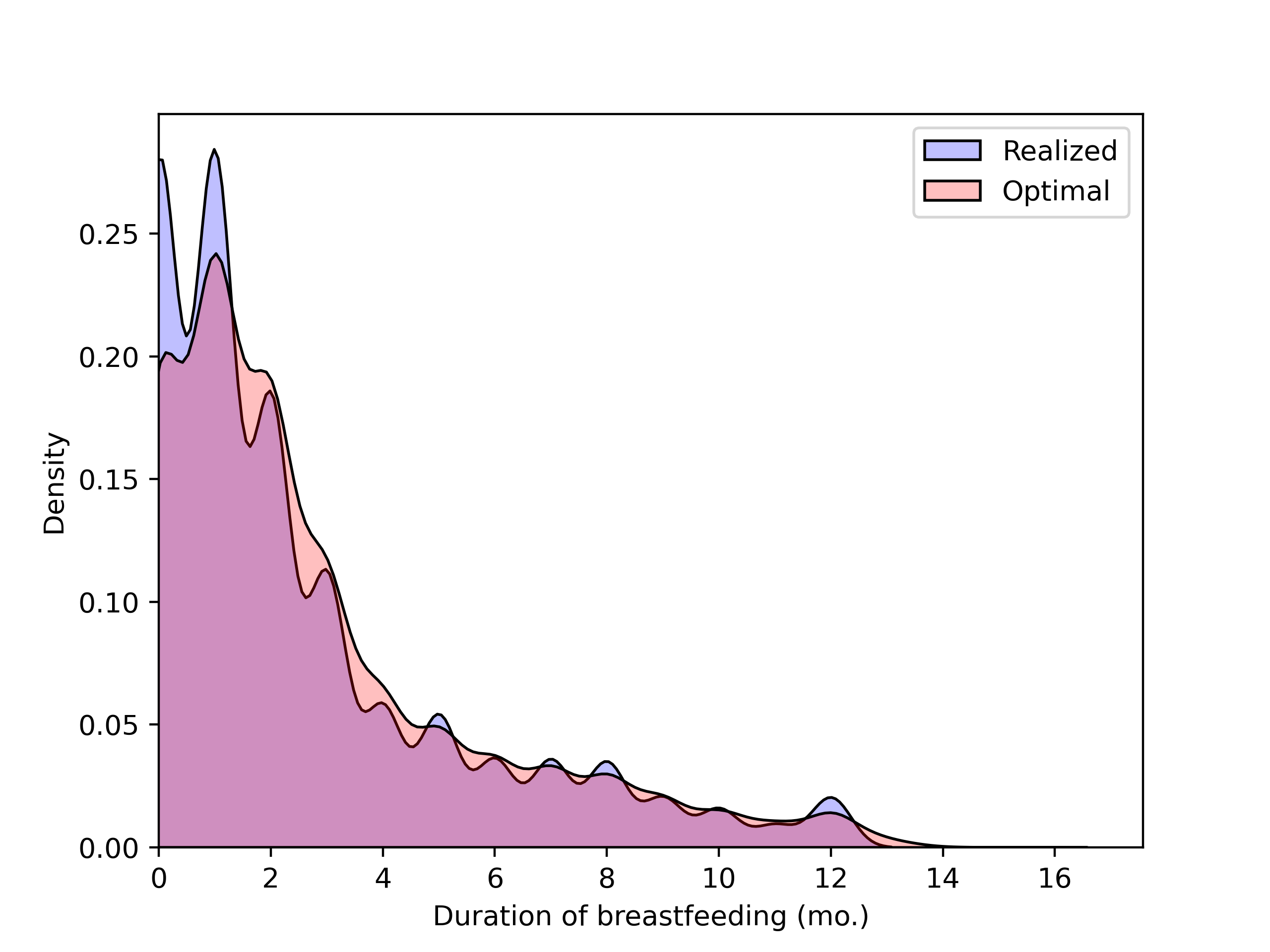}}
    \subfloat[Years of education]{\label{subfig: density-realized-vs-optimal-edulen}\includegraphics[width=0.5\linewidth]{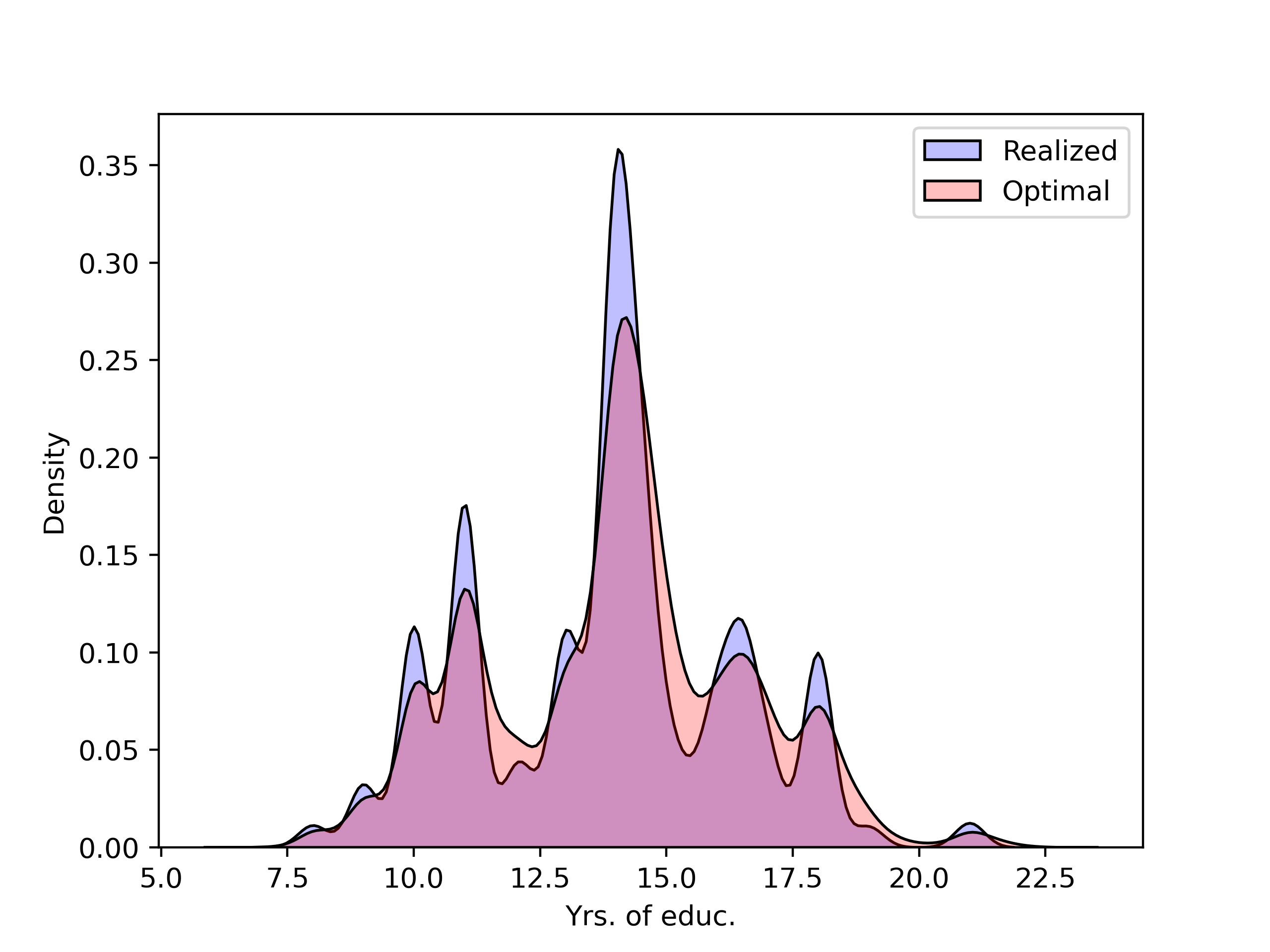}}

    \subfloat[Share of time employed 25-50]{\label{subfig: density-realized-vs-optimal-avg_emp_25_50}\includegraphics[width=0.5\linewidth]{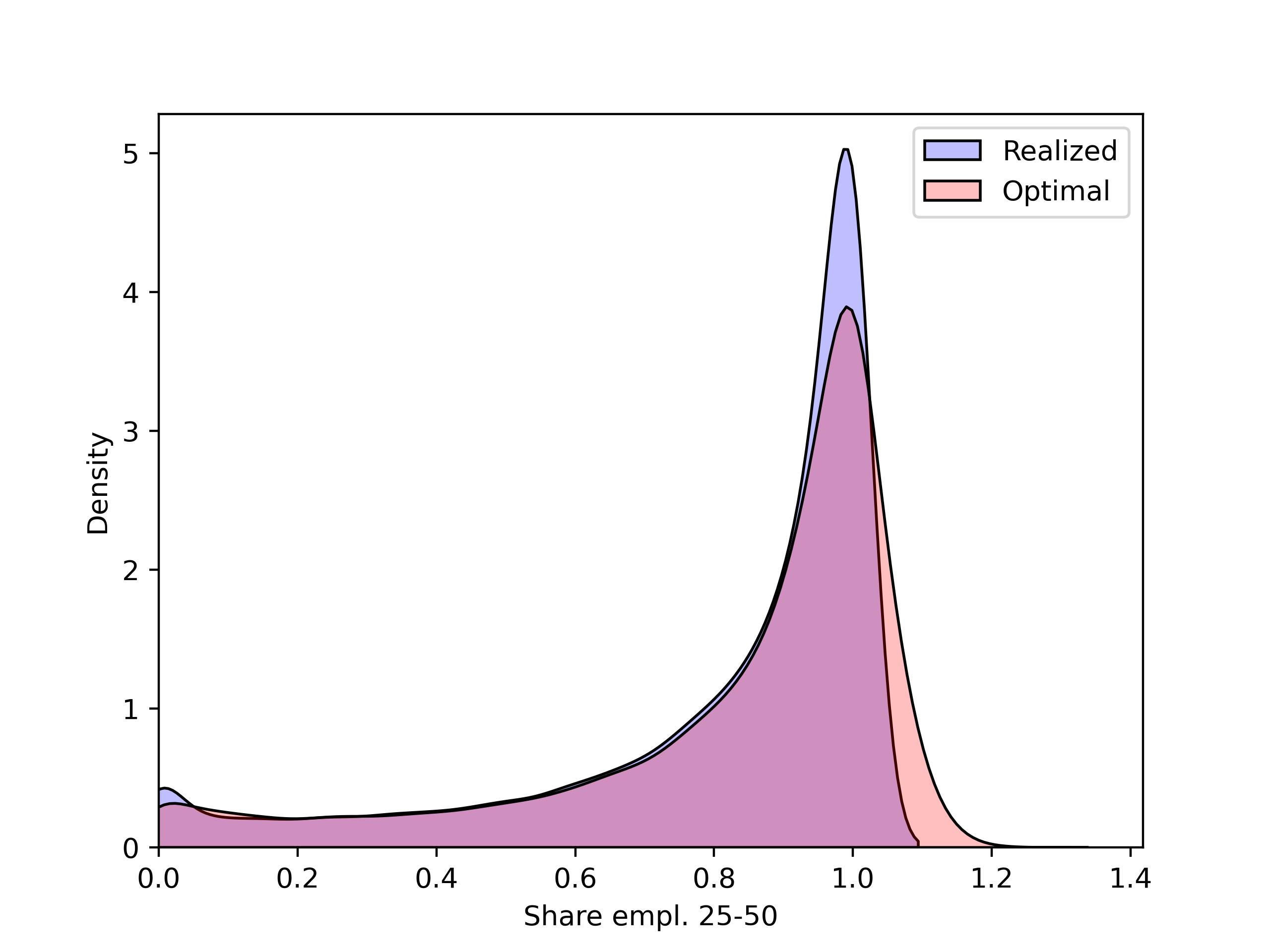}}
    \subfloat[Average income 25-50]{\label{subfig: density-realized-vs-optimal-avg_inc_25_50}\includegraphics[width=0.5\linewidth]{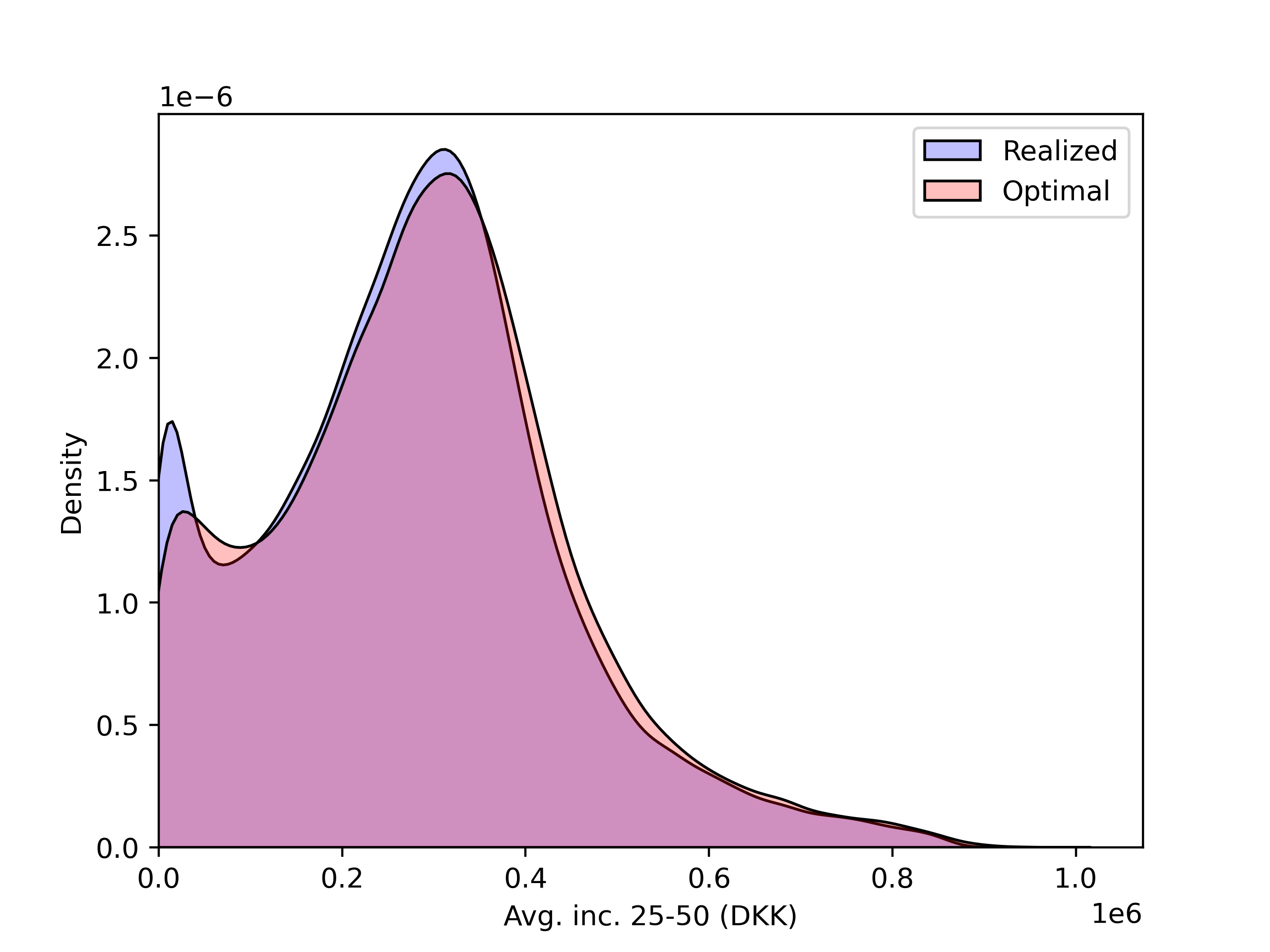}}
    
    \caption{Realized and Counterfactual Distributions of Outcomes by Nurse Allocation Mechanism.}
    \label{fig: density-realized-vs-optimal-outcomes}
    \begin{minipage}{1\linewidth}
        \vspace{1ex}
        \footnotesize{
        \textit{Notes:}
        The figure shows the actual and estimated counterfactual distributions of outcome variables by nurse allocation mechanism.
        Each panel shows the current distribution of an outcome variable (in blue) and the estimated counterfactual distribution under  optimal allocation of nurses (in red).
		}
	\end{minipage}
\end{figure}


\subsection{Robustness Checks}
\label{sec: robustness}

The biggest concern to our empirical design is selection in the allocation mechanism:
If children with relatively good potential outcomes are selectively allocated to one nurse, and children with relatively poor potential outcomes to another, then our estimates of the differences in treatment effects between nurses will be biased.

To mitigate such potential concerns, Table~\ref{tab: average-treatment-effects-pre-treatment-variables} shows average differences in children pre-treatment variables by nurse (with associated p-values from permutation tests), i.e., similarly to Table~\ref{tab: average-treatment-effects} but now for pre-treatment rather than outcome variables. 
A necessary condition for our assumption of no selective allocation within district-by-year groups is that differences in pre-treatment variables should not differ between groups of children allocated different nurses.
Indeed, with the exception of ``born prior to due date'', the estimated differences are at most marginally significant.\footnote{The ``excess'' difference in the share of children born prior to due date between nurses is around one percent. However, the number of weeks born prior to due date does not systematically differ between children. We interpret this as differences between nurses in terms of when they would note a child as being born prior to due date, with some nurses being stricter than others.} 
While these results do not rule out potential selection, they are nevertheless reassuring for our design.
Further, Appendix Figure~\ref{fig: counterfactual-permutation-treatment-effects-pre-treatment-variables} shows the full distribution of counterfactual differences from our permutation tests (similarly to Appendix Figure~\ref{fig: counterfactual-permutation-treatment-effects}, but for pre-treatment rather than outcome variables).

\begin{table}
    \centering
    \centering
    \caption{Average Differences in Pre-Treatment Variables by Nurses Within District-by-Year Groups.}
    \label{tab: average-treatment-effects-pre-treatment-variables}
    \begin{tabular}{l rrrr}
        \toprule
        Outcome & Average Difference & Excess Difference & MDV & P-value \\
        \midrule
        Low BW & 0.03 &  0.00 & 0.05 & 0.119 \\
        Born prior to due date & 0.06 &  0.01 & 0.12 & 0.000 \\
        Weeks prior to due date & 0.56 &  0.03 & 3.41 & 0.290 \\
        Born 1-3 & 0.05 &  0.00 & 0.10 & 0.252 \\
        Female & 0.07 & -0.00 & 0.49 & 0.924 \\
        Father missing & 0.03 & -0.00 & 0.06 & 0.924 \\
        \bottomrule
    \end{tabular}
    \begin{minipage}{1\linewidth}
        \vspace{1ex}
        \footnotesize{
        \textit{Notes:}
        The table shows estimates of average (absolute) differences in pre-treatment variables between two otherwise similar groups of children allocated different nurses (but within the same district-by-year group).
        Since the test statistic we use is not centered around zero, we also report the ``excess'' difference, which is the difference between nurses above that which is a result of noise.
        MDV is the mean of the dependent variable.
        The p-values are from permutation-based tests which randomly match children and nurses within a district-by-year group and is calculated as the share which results in a more extreme value than the one observed in the non-permuted data, and we use a Holm–Šidák approach to account for multiple hypothesis testing.
		}
	\end{minipage}
\end{table}

Could our relatively lenient strategy for matching nurse records and nurse district information introduce some potential bias?
To answer this question, we re-estimate the average differences in outcomes of children between nurses, similarly to Table~\ref{tab: average-treatment-effects}, but now for a subsample where records are only included if we can match it to nurse district information using nurse first name initial and full last name.
Table~\ref{tab: average-treatment-effects-mr3} shows our findings, indicating that our results are robust to this additional sample restriction.
Magnitudes remain relatively stable and our estimated differences are still highly statistically significant.

\begin{table}
    \centering
    \centering
    \caption{Average Differences in Outcomes by Nurses Within District-by-Year Groups: Stricter Matching Criteria.}
    \label{tab: average-treatment-effects-mr3}
    \resizebox{1\linewidth}{!}{
        \begin{tabular}{l rrrr}
            \toprule
            Outcome & Average Difference & Excess Difference & MDV & P-value \\
            \midrule
            Breastfed, 1 mo. &      0.09 &     0.02 &       0.58 & 0.000 \\
            Breastfed, 6 mo. &      0.03 &     0.00 &       0.03 & 0.003 \\
            Duration of breastfeeding (mo.) &      0.59 &     0.15 &       2.50 & 0.000 \\
            Avg. inc. 25-50 (DKK) & 27,965.44 & 4,221.04 & 269,780.38 & 0.000 \\
            Share empl. 25-50 &      0.05 &     0.01 &       0.79 & 0.004 \\
            Yrs. of educ. &      0.46 &     0.09 &      13.74 & 0.000 \\
            Above mand. edu. &      0.07 &     0.01 &       0.74 & 0.009 \\
            \bottomrule
        \end{tabular}
    }
    \begin{minipage}{1\linewidth}
        \vspace{1ex}
        \footnotesize{
        \textit{Notes:}
        The table shows estimates of average (absolute) differences in outcomes between two otherwise similar groups of children allocated different nurses (but within the same district-by-year group), and where we restrict the sample to those with a nurse match based on nurse first name initial and full last name.
        Since the test statistic we use is not centered around zero, we also report the ``excess'' difference, which is the difference between nurses above that which is a result of noise.
        MDV is the mean of the dependent variable.
        The p-values are from permutation-based tests which randomly match children and nurses within a district-by-year group and is calculated as the share which results in a more extreme value than the one observed in the non-permuted data, and we use a Holm–Šidák approach to account for multiple hypothesis testing.
		}
	\end{minipage}
\end{table}

Another potential concern is misclassification of children to districts as a result of nurses changing districts over time and us only observing nurse district for a subset of our sample years (1963 and 1965). 
To mitigate such concerns, we re-estimate the average differences in outcomes of children between nurses, but now for a subsample where records are only included if the child is born in one of the years for which we have data on nurse district.
Table~\ref{tab: average-treatment-effects-ty} shows our findings, indicating that our results are robust to this additional sample restriction.
The magnitudes of our estimates for this restricted sample are close to the results of Table~\ref{tab: average-treatment-effects}, though less precisely estimated.

\begin{table}
    \centering
    \centering
    \caption{Average Differences in Outcomes by Nurses Within District-by-Year Groups: Stricter Year of Birth Criteria.}
    \label{tab: average-treatment-effects-ty}
    \resizebox{1\linewidth}{!}{
        \begin{tabular}{l rrrr}
            \toprule
            Outcome & Average Difference & Excess Difference & MDV & P-value \\
            \midrule
            Breastfed, 1 mo. &      0.08 &     0.01 &       0.56 & 0.141 \\
            Breastfed, 6 mo. &      0.03 &     0.00 &       0.03 & 0.087 \\
            Duration of breastfeeding (mo.) &      0.52 &     0.09 &       2.39 & 0.012 \\
            Avg. inc. 25-50 (DKK) & 28,080.40 & 4,449.58 & 274,341.07 & 0.036 \\
            Share empl. 25-50 &      0.05 &     0.01 &       0.80 & 0.015 \\
            Yrs. of educ. &      0.45 &     0.09 &      13.80 & 0.007 \\
            Above mand. edu. &      0.07 &     0.01 &       0.76 & 0.141 \\
            \bottomrule
        \end{tabular}
    }
    \begin{minipage}{1\linewidth}
        \vspace{1ex}
        \footnotesize{
        \textit{Notes:}
        The table shows estimates of average (absolute) differences in outcomes between two otherwise similar groups of children allocated different nurses (but within the same district-by-year group), and where we restrict the sample to those born in the years 1963 or 1965 (the years exactly coinciding with our data on nurse districts).
        Since the test statistic we use is not centered around zero, we also report the ``excess'' difference, which is the difference between nurses above that which is a result of noise.
        MDV is the mean of the dependent variable.
        The p-values are from permutation-based tests which randomly match children and nurses within a district-by-year group and is calculated as the share which results in a more extreme value than the one observed in the non-permuted data, and we use a Holm–Šidák approach to account for multiple hypothesis testing.
		}
	\end{minipage}
\end{table}

What role might spillover effects between siblings introduce, and could nurses be selectively allocated children from families of which they already were assigned previous children?
Table~\ref{tab: average-treatment-effects-firstborn} shows our estimates of average differences in outcomes of children between nurses, but now for a sample of only firstborn children.
Here, some of our results are slightly attenuated (when compared to Table~\ref{tab: average-treatment-effects}), but the same overarching pattern remains, again statistically significantly estimated.

\begin{table}
    \centering
    \centering
    \caption{Average Differences in Outcomes by Nurses Within District-by-Year Groups: Firstborn Children.}
    \label{tab: average-treatment-effects-firstborn}
    \resizebox{1\linewidth}{!}{
        \begin{tabular}{l rrrr}
            \toprule
            Outcome & Average Difference & Excess Difference & MDV & P-value \\
            \midrule
            Breastfed, 1 mo. &      0.10 &     0.01 &       0.61 & 0.020 \\
            Breastfed, 6 mo. &      0.03 &     0.00 &       0.03 & 0.059 \\
            Duration of breastfeeding (mo.) &      0.65 &     0.13 &       2.59 & 0.000 \\
            Avg. inc. 25-50 (DKK) & 31,249.23 & 1,927.78 & 278,822.14 & 0.090 \\
            Share empl. 25-50 &      0.05 &     0.00 &       0.80 & 0.114 \\
            Yrs. of educ. &      0.51 &     0.06 &      13.98 & 0.012 \\
            Above mand. edu. &      0.08 &     0.01 &       0.77 & 0.039 \\
            \bottomrule
        \end{tabular}
    }
    \begin{minipage}{1\linewidth}
        \vspace{1ex}
        \footnotesize{
        \textit{Notes:}
        The table shows estimates of average (absolute) differences in outcomes between two otherwise similar groups of children allocated different nurses (but within the same district-by-year group), and where we restrict the sample to firstborn children.
        Since the test statistic we use is not centered around zero, we also report the ``excess'' difference, which is the difference between nurses above that which is a result of noise.
        MDV is the mean of the dependent variable.
        The p-values are from permutation-based tests which randomly match children and nurses within a district-by-year group and is calculated as the share which results in a more extreme value than the one observed in the non-permuted data, and we use a Holm–Šidák approach to account for multiple hypothesis testing.
		}
	\end{minipage}
\end{table}

Our above placebo tests and results for alternative samples, combined with our earlier tests for rank-order correlation between pre-treatment and outcome variables by nurse within district-by-year groups (see Figure~\ref{fig: correlation-coefficient-cdf-nurses}) and our non-parametric tests for differences in pre-treatment variables by nurse within district-by-year groups (see Figure~\ref{fig: pval-cdf-controls}), support our identification strategy.
As far as pre-treatment variables are concerned, there does not appear to have been significant selection, as long as we only compare nurses within the same district-by-year group.



\subsection{Discussion}
\label{sec: discussion}

How do our estimated treatment effects and effects of re-allocation compare to other early-life interventions?
In particular, how much can be gained by optimizing an existing program -- as is the case of the nurse re-allocation -- vs. adding or changing a program with associated increases in required resources?

The most similar early-life interventions to our setting are the introduction of NHV programs and center care in Scandinavia \citep{wust2012early, hjort2017universal, bhalotra2017infant, buetikofer2019infant}, as well as the study by \citet{baker2023universal} which examines the impacts of extending the Danish NHV program from one to three years.
What sets our setting apart from the studies above is that we measure the impact of an, in principal, ``free'' intervention, which instead of adding resources considers the impact of improved allocation of existing resources.

Compared to the study on the introduction of NHV in Denmark (1937) by \citet{hjort2017universal}, which finds health effects but none for education and labor market outcomes, we document positive impacts on education and labor market outcomes.
In Sweden, \citet{bhalotra2017infant} find that the introduction of a health intervention (pioneered in 1931-1933) that provided information to mothers and infant monitoring through home visits and clinics led to substantial health benefits and increased likelihood for secondary schooling enrolment for females by around 3-4\%-points (15\%).\footnote{They find no effects on likelihood of secondary schooling enrolment for males. They interpret this as being due to relatively fixed supply of schooling during that time.}
Our estimates suggest that optimal re-allocation could have improved the probability of completing at least secondary education by around 3\%-points. 
In Norway, \citet{buetikofer2019infant} study increased access to mother and child health care centers during a child's first year of life (introduced during the 1930s), finding that access to these centers improved length of schooling by around 1.8 months and earnings by around 2\%, as well as improved health.
Our estimates suggest that optimal re-allocation could have improved the length of education by around two months and income by around 4\%. 

Compared to the above studies on the \textit{introduction} of early-life policies in Scandinavia, our results show that, in terms of education and labor market outcomes, optimally allocating providers (nurses) to recipients (children) may have effects of similar magnitude.
In the US, \citet{hoehn2021long} studies the impact of county-level health departments (instituted from 1908-1933) on long-run outcomes, showing that affected men's later-life earnings were improved by 2-5\%, again similar in magnitude to our results.
However, our study is set at a later point in time and focuses not on the introduction but rather the provision of early-life investments.
The closest study to ours in terms of timing and target group is the study by \citet{baker2023universal} on the impact of extended NHV (three vs. one year) for the same group of children comprising our study.
Their findings suggest positive and persistent health effects of enrolment into the extended NHV program, with some effects on labor market participation of women (around 1.4\%-points for share of time in employment during ages 30-50) but none for education or income.
In comparison, our results suggest that optimal re-allocation could have improved share of time in employment during ages 25-50 by around 2\%-points.

How come our estimates are of similar magnitudes to impacts of introduction of similar services in Scandinavia and the US?
A likely answer lies in findings across these studies of significant effect heterogeneity, with disadvantaged children disproportionately positively impacted.
Unlike the studies above, which focus on universal programs, our counterfactual policy focuses precisely on those children most likely to be positively impacted.
We view this as yet another motivation for studying the role played by treatment providers when evaluating and designing policies. 



\section{Conclusions}
\label{sec: conclusion}

We tackle the issue of solving optimal treatment allocation problems of the type where each recipient may be differentially affected by each treatment, and where there are constraints on  treatment capacities.
We prove that the problem is solvable in strongly polynomial time and present an algorithm based on flows in networks for problems of this type.
We also extend our algorithm to cases where only Pareto-improving re-allocations of treatments to recipients are allowed by modifying the underlying network to ensure that only re-allocations that make no individual worse off are selected.

To showcase our method, we study NHV in the 1960s Copenhagen.
Earlier work in Denmark, as well as other countries, has documented positive impacts of the introduction of NHV and center care \citep{wust2012early, hjort2017universal, bhalotra2017infant, buetikofer2019infant, hoehn2021long} and extended NHV \citep{baker2023universal}, but we go beyond these studies by estimating differences in treatment effects between individual nurses.
We do so by transcribing and linking data from historical nurse records to Danish administrative data, identifying the nurse allocated each child and following the child throughout her life using Danish register data.\footnote{This transcription work has been done concurrently and in collaboration with the work of \citet{bjerregaard2023cohort, baker2023universal}.}

We show that outcomes of children vary significantly by the visiting nurse, even when comparing children in the same nurse district and born the same year.
We show that these differences are not artefacts due to selection (once operating within district-by-year groups) by comparing pre-treatment variables between children allocated different nurses.
Further, using causal machine learning, we show that children are heterogeneously affected by nurses.
This, in turn, allows us to obtain estimates of potential outcomes under different allocation mechanisms, and we document that a re-allocation of the nurses in the Copenhagen NHV program could have resulted in significant efficiency improvements.
Our estimates suggest that optimal allocation of nurses to children could have improved average yearly earnings by USD 1,815 (4\%) and length of education by around two months.

We contribute with new knowledge within optimal policy and the literature on the role of early-life conditions and investments for long-run outcomes.
Within optimal policy, we introduce a strongly polynomial algorithm for optimal treatment allocation problems under constraints.
Within the literature on early-life investments, we add novel evidence on the role of treatment providers for the effects of such policies.
Further, we show that optimal allocation of such investments to recipients may play a crucial role in fully exploiting potential benefits of early-life investments.

While our approach for solving optimal treatment allocation problems is applicable to settings with various constraints, we assume that potential outcomes can be sufficiently precisely estimated.
Optimal policy learning is a rapidly evolving field with a number of challenges related to identification.
Directly incorporating policy learning with our approach for solving allocation problems could prove a fruitful avenue for new research.

For our empirical application, two challenges remain for proper integration of our optimal allocation design.
First, our results for long-run outcomes are not observable before a significant time-gap, and thus these results are not directly applicable in practice for estimating potential outcomes.
However, given the strong rank-order correlation between short- and long-run outcomes, using estimates for short-run outcomes are likely sufficient to obtain re-allocation rules to improve also long-run outcomes. 
Second, if \textit{all} nurses are re-allocated based on our algorithm, it would no longer be possible to obtain causal estimates of potential outcomes if the underlying distribution changes (e.g., as a result of the population of children changing or the population of nurses changing) since the algorithm by design exploits selection.
Implementing an approach in practice would thus require only partial selection, leaving some children to be randomly allocated to update estimates of potential outcomes as its underlying distribution changes.
This leaves a complex ``exploration vs. exploiting'' problem \citep{sutton2018reinforcement}, where the optimal share of children to randomly allocate to nurses (in order to update estimates of potential outcomes, i.e., explore) must be decided in a way to optimize average outcomes over time (i.e., exploit our knowledge to improve allocation and, in turn, outcomes).


\newpage 
\singlespacing
\printbibliography


\newpage
\onehalfspacing
\FloatBarrier
\appendix
\section*{Appendices}

\setcounter{table}{0}
\setcounter{figure}{0}

\renewcommand{\thetable}{\Alph{subsection}.\arabic{table}}  
\renewcommand{\thefigure}{\Alph{subsection}.\arabic{figure}}
\renewcommand{\thesection}{\Alph{subsection}}
\renewcommand{\thesubsection}{\Alph{subsection}}


\subsection{Extended Proof}
\label{sec: extended-proof}

This section provides additional details on the proofs of Theorems~\ref{thm: correspondance} and \ref{thm: complexity}.

\begin{proof}[Extended Proof of Theorem~\ref{thm: correspondance}]
    For each $v_1 \in V_1$ (treatment) and $v_2 \in V_2$ (recipient), let $D_{v_2}^{v_1} = x_{v_1v_2}$ with $x$ a feasible flow in $\mathcal{N}$.
    To prove that $x$ corresponds to a valid solution of the optimal allocation problem, it is sufficient to prove that any feasible flow $x$ respects the constraints of the optimal allocation problem's binary program representation.

    Recall that a flow $x$ is feasible if, for all $uw \in A$, $l_{uw} \leq x_{uw} \leq u_{uw}$ and, for all $v \in V$, $b(v) = b_x(v)$.    
    Exactly one arc enters each $v_1 \in V_1$, namely $sv_1$, and its capacity is $u_{sv_1} = m$.
    Since $b(v_1) = b_x(v_1)$, this implies that at most $m$ units of flow leaves $v_1$, and since $x$ is an integer flow this implies that for each $v_1 \in V_1$, at most $m$ of the arcs $v_1v_2$, $v_2 \in V_2$, have non-zero flow.
    This proves that a feasible flow allocates any treatment to at most $m$ recipients.

    Exactly one arc leaves each $v_2 \in V_2$, namely $v_2t$, and its lower bound and capacity are $l_{v_1t} = 1$ and $u_{v_2t} = 1$, respectively.
    Since $b(v_2) = b_x(v_2)$ and $x$ is an integer flow, this implies that for each $v_2 \in V_2$, exactly one of the arcs $v_1v_2$, $v_1 \in V_1$, have non-zero flow.
    This proves that a feasible flow allocates exactly one treatment to any recipient.

    To prove the relationship between the cost of a feasible flow and the value of a solution to the optimal allocation problem's binary program representation, note that the cost of any feasible flow $x$ is:

    \begin{align*}
        \sum_{uw \in A}c_{uw}x_{uw} & = \sum_{v_1 \in V_1}\sum_{v_2 \in V_2}c_{v_1v_2}x_{v_1v_2} & \text{Since $c_{uw}$ is 0 for all other arcs}
        \\
        & = \sum_{v_2 \in V_2}\sum_{v_1 \in V_1}c_{v_1v_2}x_{v_1v_2} & \text{Switching order of summations}
        \\
        & = \sum_{v_2 \in V_2}\sum_{v_1 \in V_1}c_{v_1v_2}D_{v_2}^{v_1} & \text{Since $x_{v_1v_2} = D_{v_2}^{v_1}$ for all $v_1 \in V_1$, $v_2 \in V_2$}
        \\
        & = -\sum_{v_2 \in V_2}\sum_{v_1 \in V_1}Y_{v_2}^{v_1}D_{v_2}^{v_1} & \text{Using $c_{v_1v_2} = - Y_{v_2}^{v_1}$}
    \end{align*}
\end{proof}

To prove Theorem~\ref{thm: complexity}, we shall make use of the following algorithm \citep{klein1967primal}:

\begin{algorithm}
    \caption{Cycle Cancelling Algorithm}
    \label{alg: cycle-cancelling}
    \begin{algorithmic}
        \Require A network $\mathcal{N} = (V, A, l, u, b, c)$.
        \Ensure A minimum cost feasible flow $x$ in $\mathcal{N}$.
        
        \Procedure{CycleCanceling}{$\mathcal{N}$}
        \State Find a feasible flow $x$ in $\mathcal{N}$.
        \Repeat
            \State Search for a negative cycle $C$ in $\mathcal{N}(x)$.
            \If{such a cycle $C$ is found}
                \State Augment $x$ by $\delta(X)$ units along $C$.
            \EndIf
        \Until{no negative cycle is found}
        \State \Return $x$.
    \EndProcedure
    \end{algorithmic}
\end{algorithm}

\FloatBarrier

\begin{proof}[Extended Proof of Theorem~\ref{thm: complexity}]
    Let $n = |V|$ denote the number of vertices and $m = |A|$ the number of arcs in the underlying digraph $D$ of $\mathcal{N}$.
    Using that the Bellman-Ford-Moore algorithm allows us to check for the existence of a negative cycle in an arbitrary digraph in time $\mathcal{O}(nm)$ \citep{bellman1958routing, ford1956network, moore1959shortest}, we may use the \textit{cycle cancelling algorithm} \citep{klein1967primal} to find a minimum cost feasible flow in a network (Algorithm~\ref{alg: cycle-cancelling}).\footnote{This is due to the result that $x$ is a minimum cost feasible flow in $\mathcal{N}$ if and only if $\mathcal{N}(x)$ contains no directed cycle of negative cost \citep{bang2008digraphs}.}
    However, in the absence of guidance as to how to find the cycles to augment along, this algorithm is guaranteed to work only when lower bounds, capacities, costs, and balance vectors are integers, and may be exponential in the size of the input (i.e., not \textit{strongly} polynomial).

    The problem with the ``naïve'' cycle cancelling algorithm is in its searching step for a negative cycle.
    Without searching for an appropriate negative cycle, the algorithm is not useful for the type of network we consider.
    To confront this issue, we use the approach pioneered in \citet{goldberg1985new} and generalized in \citet{goldberg1988new} of using \textit{preflows} \citep{karzanov1974determining} for solving maximum flow problems.
    Using generalizations due to \citet{goldberg1987efficient, goldberg1989finding}, this allows us to modify the searching step of the algorithm in a way to obtain a strongly polynomial algorithm by always augmenting along a cycle of minimum mean cost.\footnote{The mean cost of a cycle $C$ is defined as $c(C) / |A(C)|$.}
    Specifically, we can find a minimum mean cost cycle in time $\mathcal{O}(\log n)$, resulting in total time $\mathcal{O}(n m^2 (\log n)^2)$.

    Recall that we require any solution to be an integer flow.\footnote{Otherwise some recipient could be allocated different ``shares'' of different treatments, which is incompatible with the binary program structure of the optimal allocation problem.}
    Here, we may use a convenient \textit{integrality} property of minimum cost flows, namely that given all integer lower bounds, capacities, and balance vectors, there exists an integer minimum cost flow.
    To see why, we may start by assuming, without loss of generality, that our initial feasible flow is an integer flow.\footnote{Either by using the initial allocation as the flow or by exploiting the layered structure of $\mathcal{N}$ to easily find a feasible integer flow in time $\mathcal{O}(m)$.}
    At each iteration, we augment our flow by $\delta(C)$, by which we mean the minimum residual capacity of any arc on $C$ in $\mathcal{N}$, and thus our new flow after any iteration is $x^{t+1} = x^t \oplus \delta(C)$.
    Since the residual capacity on any arc is always an integer, $x^{t+1}$ is an integer flow so long as $x^t$ is, and thus by induction we have an integer minimum cost flow.

    Using the structure of the network, we can obtain the complexity in terms of the inputs to the optimal allocation problem.
    Noting that $n = 2 + n_1 + n_2$ and $m = n_1 + n_2 + n_1n_2$, we have:
    \begin{align*}
        \mathcal{O}\left[nm^2(\log n)^2\right] & = \mathcal{O}\left[(2 + n_1 + n_2) (n_1 + n_2 + n_1n_2)^2 \log^2 (2 + n_1 + n_2)\right] &
        \\
        & = \mathcal{O}\left[(n_1 + n_2) (n_1n_2)^2\log^2 (n_1 + n_2)\right] &
        \\
        & = \mathcal{O}\left[(n_1^3n_2^2 + n_1^2n_2^3) \log^2 (n_1 + n_2)\right] &
    \end{align*}
    
\end{proof}


\FloatBarrier
\newpage
\subsection{Transcription Details}
\label{sec: transcription-details}


We split the handwritten text recognition task into three steps, those being (1) page classification, (2) image segmentation, and (3) text recognition (this is similar to the abstract approach described in \citet{dahl2023bdad}).
For page detection, we use an unsupervised approach based on clustering (specifically, we use the DBSCAN algorithm of \citet{ester1996dbscan} on features extracted using a pre-trained neural network).
For image segmentation, we use a strategy based on point set registration, the task of aligning points between an image and a template \citep{basl1992registration}.
We use the semantic segmentation model of \citet{dahl2023tableparser} to extract vertical and horizontal lines of the nurse records.
We align the set of intersection points of each page with a pre-specified template using the efficient probabilistic point-set registration (FilterReg) by \citet{gao2019filtereg}, thus obtaining a transformation matrix for each image which we use, in combination with a pre-specified overlay, to crop each field of interest into a separate image.

Having thus obtained segmented images of each field of interest, we use neural networks to transcribe their contents.
Here, we use a combination of convolutional neural networks (CNN), those being based on the EfficientNetV2 architecture \citep{tan2021efficientnetv2} and based on the approach in \citet{dahl2022dare}, and vision transformer (ViT) architecture by \citet{dosovitskiy2020vit}.
We use different models depending on the field we transcribe, as the content of the fields vary, with some consisting of numbers, others characters, and yet others circling a number rather than writing anything.
Appendix Table~\ref{tab: field-info} provides information on all the different groups of fields we consider.
Note that some of these groups consist of many fields with the same type of information (and as such the same ``alphabet'', consisting of the set of characters/digits that can occur in the field), while others refer to just a single field.
All fields are present on each nurse record.

\begin{table}
    \centering
    \caption{Grouping of Fields from the Nurse Records.}
    \label{tab: field-info}
    \begin{tabular}{l rrr}
         \toprule
         Group                  & \#Fields& Maximum Length & Alphabet                    \\
         \midrule
         breastfeed-7-do        & 1       & 1          & $\{1, 2, 3\}^\ast$              \\
         dura-any-breastfeed    & 1       & 2          & $\{0, 1, \dots, 9\}$            \\
         date                   & 7       & 4          & $\{0, 1, \dots, 9\}$            \\
         length                 & 2       & 3          & $\{0, 1, \dots, 9\}$            \\
         birth prior to due date          & 1       & 1          & $\{1, 2\}^\ast$       \\
         no. of weeks prior to due date    & 1       & 2          & $\{0, 1, \dots, 9\}$ \\
         tab-b                  & 112     & 2          & $\{0, 1, \dots, 9\}$            \\
         weight                 & 8       & 5          & $\{0, 1, \dots, 9\}$            \\
         nurse-name (first)     & 3       & $k^{\ast\ast}$ & $\{a, b, \dots, \text{å}\}$ \\
         nurse-name (last)      & 3       & $k^{\ast\ast}$ & $\{a, b, \dots, \text{å}\}$ \\
         \bottomrule
    \end{tabular}
    \begin{minipage}{1\linewidth}
        \linespread{1.0}\selectfont
        \vspace{1ex}
        \scriptsize{
            \textit{Notes:}
            The table shows fields of the nurse records grouped together in such a way that fields which are similar with respect to sequence length and alphabet of their contents are put into one group.
            The first column refers to the name given to the specific group of fields.
            The second column shows the number of fields of the given group.
            The third column shows the maximum length of the content of any field of the given group.
            The fourth column shows the alphabet of the fields of the given group.
            $^\ast$The alphabet of these fields are specifically a \textit{circle} being put around one of the digits shown, the digits being pre-printed on the records.
            $^{\ast\ast}$While there is no clear limit to the length of a name, the longest name in our training dataset contains 14 characters.
            \textit{Source}: Table due to \citet{baker2023universal}.
        }
    \end{minipage}
\end{table}

For each group of Appendix Table~\ref{tab: field-info}, we train two neural networks, one based on a CNN structure and one based on a ViT structure.
Further, we create two ``meta''-groups, consisting of different groups for which a single alphabet easily covers the different types of content.
First, we create a ``Circle''-group consisting of breastfeeding at seven days old (breastfeed-7-do) and born prior to due date (birth prior to due date), both of which consist of circling a number.
Second, we create an ``Integer''-group consisting of duration of breastfeeding (dura-any-breastfeed), length (length), number of weeks born prior to due date (weeks prior to due date), Table B visit information (tab-b), and weight (weight), as all groups consist of sequences of one to five integers.
For the ``Cirlce''-group, we train both a CNN and a ViT, and for the ``Integer''-group we train only a ViT.
In total, this results in training 23 neural networks, of which six models were chosen.
For each group of Appendix Table~\ref{tab: field-info}, we evaluate the appropriate models and select the one that performs best on a held-out test set not used for training.
Appendix Table~\ref{tab: nn-params-differences} shows information on the selected transcription models; for full details see \citet{baker2023universal} and/or the official GitHub (\url{https://github.com/TorbenSDJohansen/cihvr-transcription}, to be made public). 

\begin{table}
    \centering
    \caption{Model Differences -- Selected Models.}
    \label{tab: nn-params-differences}
    \resizebox{1.0\linewidth}{!}{
        \begin{tabular}{l p{0.4\linewidth} rrr}
            \toprule
            \textbf{Model} & \textbf{Fields (Appendix Table \ref{tab: field-info})} & \textbf{Resolution} & \textbf{Seq. len.} & \textbf{Batch size} \\
            \midrule
            & \multicolumn{4}{l}{\textit{Panel A: ViT-based models}} \\
            \midrule
            Circle       &                                   breastfeed-7-do, birth prior to due date &    100x350 &            4 &        256 \\
            Integer seq. &  dura-any-breastfeed, length, weeks prior to due date, tab-b, weight &     90x230 &            7 &       1024 \\
            Last name    & nurse-name &     91x530 &           20 &        512 \\
            \midrule
            & \multicolumn{4}{l}{\textit{Panel B: CNN-based models}} \\
            \midrule
            Date       &        date &     67x181 &            3 &       1024 \\
            First name &  nurse-name &     91x530 &           18 &        256 \\
            Weight     &      weight &     80x258 &            5 &       1024 \\
            \bottomrule
        \end{tabular}
    }
    \begin{minipage}{1\linewidth}
        \linespread{1.0}\selectfont
        \vspace{1ex}
        \scriptsize{
            \textit{Notes:}
            The table shows differences of hyperparameters of the final six neural networks selected.
            The differences are all related to the groups of fields used for training, as they differ in resolution (including aspect ratio) and sequence length, which, together with varying availability of 1 vs. 2 GPUs, led to different batch sizes.
            Note how the sequence lengths of this table differ from those of Appendix Table \ref{tab: field-info}, often being longer.
            This is due to the addition of certain special characters such as beginning of sequence and end of sequence tokens being pre- and appended to the sequences, respectively, for some models.
            \textit{Source}: Table due to \citet{baker2023universal}.
        }
    \end{minipage}
\end{table}


\paragraph{Transcription Performance}
Appendix Table~\ref{tab: transc-accs} shows the transcription performance for the various fields of the nurse records.
These groups offer a slightly more detailed perspective than those of Appendix Table~\ref{tab: field-info}. 
Notably, the Table B group has been further subdivided into 16 distinct groups, aligning with the 16 different columns within the group, where each column corresponds to seven unique fields (visits at months one, two, three, four, six, nine, and 12).

It is worth noting that we intentionally avoid character accuracy evaluation, as our objective is to ensure the correctness of entire sequences. 
For instance, in the context of character accuracy, transcribing a weight of 6,000 grams and transcribing 5,000 grams would yield a relatively high similarity score, despite their significant practical difference. 
Our decision to emphasize ``sequence accuracy'' helps mitigate such discrepancies.

Appendix Table~\ref{tab: transc-accs} also shows performance metrics that account for scenarios where empty fields are excluded. 
This analysis demonstrates that our high transcription accuracy is not solely based on accurately predicting a substantial number of empty fields, which might be relatively straightforward to achieve. 
Notably, the last column of the table indicates the proportion of non-empty fields. 
It is important to note that this proportion does not represent the share of non-empty fields across the entire collection of records within a particular group. 
Instead, it specifically pertains to the evaluation set used for testing purposes.

\begin{table}
    \centering
    \setstretch{1.0}
    \caption{Automated Transcription Performance.}
    \label{tab: transc-accs}
    \rowcolors{2}{gray!25}{white}
    \resizebox{1.0\linewidth}{!}{
        \begin{tabular}{l ccc}
            \toprule
            \rowcolor{white}
            & \multicolumn{2}{c}{\underline{Transcription accuracy (\%)}} & Share non-empty (\%) \\
            \rowcolor{white}
            &  All &  Non-empty &  \\
            \midrule
            Babbles &            92.9 &                  97.5 &             61.5 \\
            Breastfeeding 7 days &            99.4 &                  99.7 &             91.4 \\
            Care and cleanliness &            96.7 &                  99.0 &             60.8 \\
            Date &            97.2 &                  96.8 &             77.5 \\
            Duration breastfeeding &            97.6 &                  97.9 &             97.3 \\
            Home economic status &            96.0 &                  97.6 &             46.5 \\
            Home harmony &            97.7 &                  98.9 &             26.5 \\
            In air &            91.7 &                  97.5 &             61.2 \\
            Length &            99.0 &                  99.0 &             97.2 \\
            Lifts head &            92.6 &                  94.9 &             63.4 \\
            Mother daily hours working at home &            88.8 &                  99.5 &             60.7 \\
            Mother daily hours working outside home &            93.0 &                  99.5 &             70.8 \\
            Mother mental capacity &            96.8 &                  97.4 &             40.9 \\
            Mother physical capacity &            97.1 &                  98.1 &             41.6 \\
            Number of daily meals &            74.5 &                  74.0 &             72.9 \\
            Nursery or kindergarten &            93.9 &                  93.3 &             69.9 \\
            Nutrition &            96.2 &                  97.0 &             80.2 \\
            Own bed &            96.0 &                  99.7 &             59.7 \\
            Birth prior to due date &            99.0 &                  99.5 &             88.4 \\
            Weeks prior to due date  &            97.3 &                  80.1$^{\ast}$ &             12.5 \\
            Sits &            91.9 &                  91.8 &             70.3 \\
            Smiles &            93.0 &                  98.0 &             61.5 \\
            Weight &            97.8 &                  97.7 &             97.3 \\
            Nurse first name &            95.2 &                  93.6 &             57.0 \\
            Nurse last name &            95.0 &                  93.2 &             57.0 \\
            \bottomrule
        \end{tabular}
    }
    \begin{minipage}{1\linewidth}
        \vspace{1ex}
        \footnotesize{
        \textit{Notes:} 
            The table shows the accuracy (\%) of the ML transcriptions for separate groups of fields, measured on an independent test set not part of the data used to train our neural networks.
            The second column shows the accuracy on the full test sample.
            The third column shows the accuracy when excluding empty fields.
            The fourth column shows the share of observations of the test set that is non-empty for each group.
            $^\ast$The low sequence accuracy for non-empty number of weeks prior to due date is due to inconsistencies regarding manual labelling of ranges such as ``1-2''.
            In those cases, the label might either say 1 or 2, meaning that it is not possible to do better than guessing one of the two numbers for a number of these cases.
            Allowing the number of weeks born prior to the due date to differ by one increases the sequence accuracy to 95.3\% for the non-empty cases and to 99.2\% for the full sample.
            \textit{Source}: Table due to \citet{baker2023universal}.
 	}
    \end{minipage}
\end{table}


\FloatBarrier
\newpage
\subsection{Additional Figures and Tables}
\label{sec: additional-figures-and-tables}

\begin{figure}[h]
    \centering
    \includegraphics[width=0.75\linewidth]{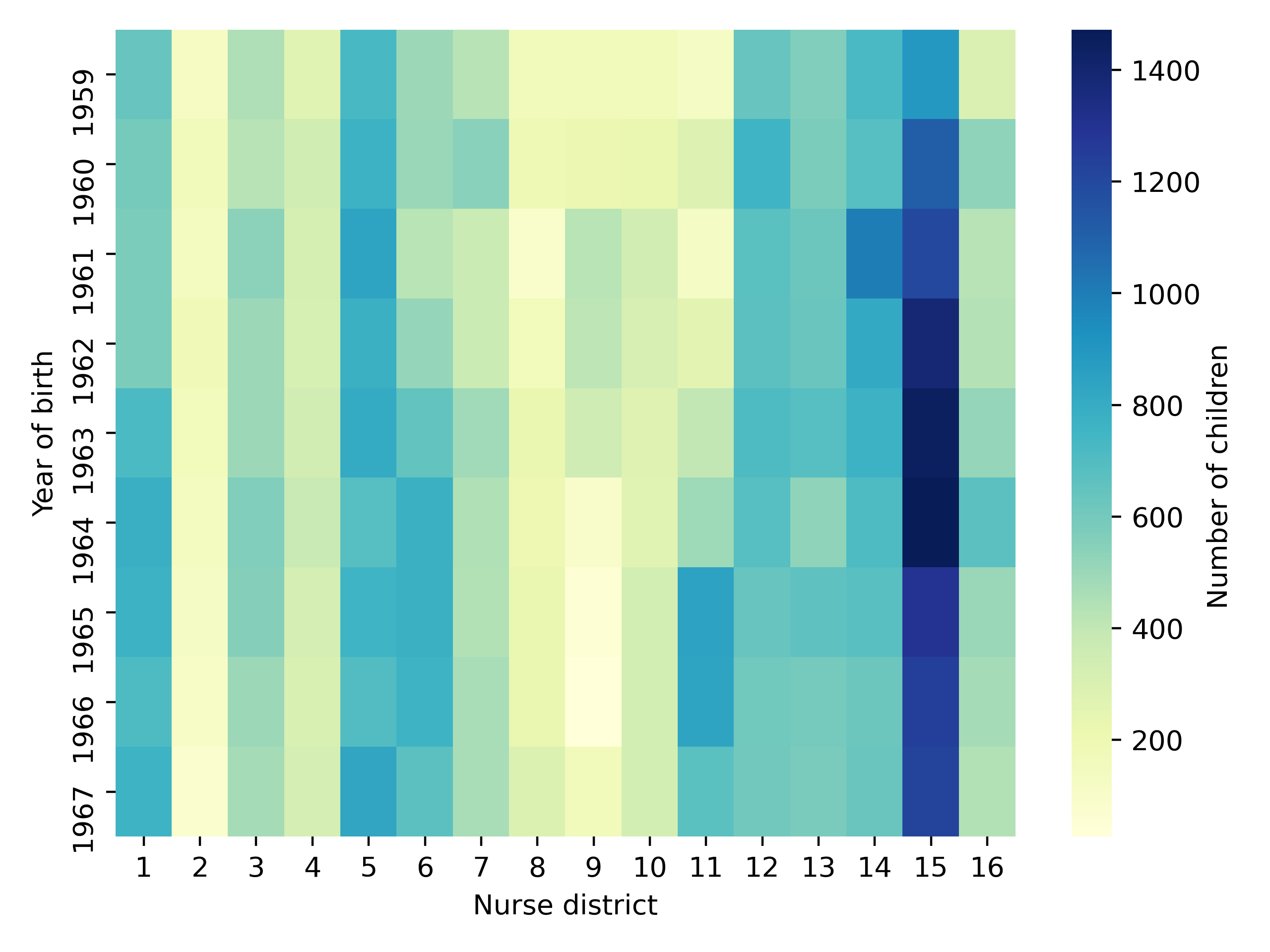}
    \caption{Number of Children by Nurse District and Year of Birth.}
    \label{fig: obs-by-year-by-district}
    \begin{minipage}{1\linewidth}
        \vspace{1ex}
        \footnotesize{
        \textit{Notes:}
        The figure shows the number of children in our primary sample by nurse district and year of birth.
		}
	\end{minipage}
\end{figure}

\begin{figure}
    \centering
    \includegraphics[width=0.75\linewidth]{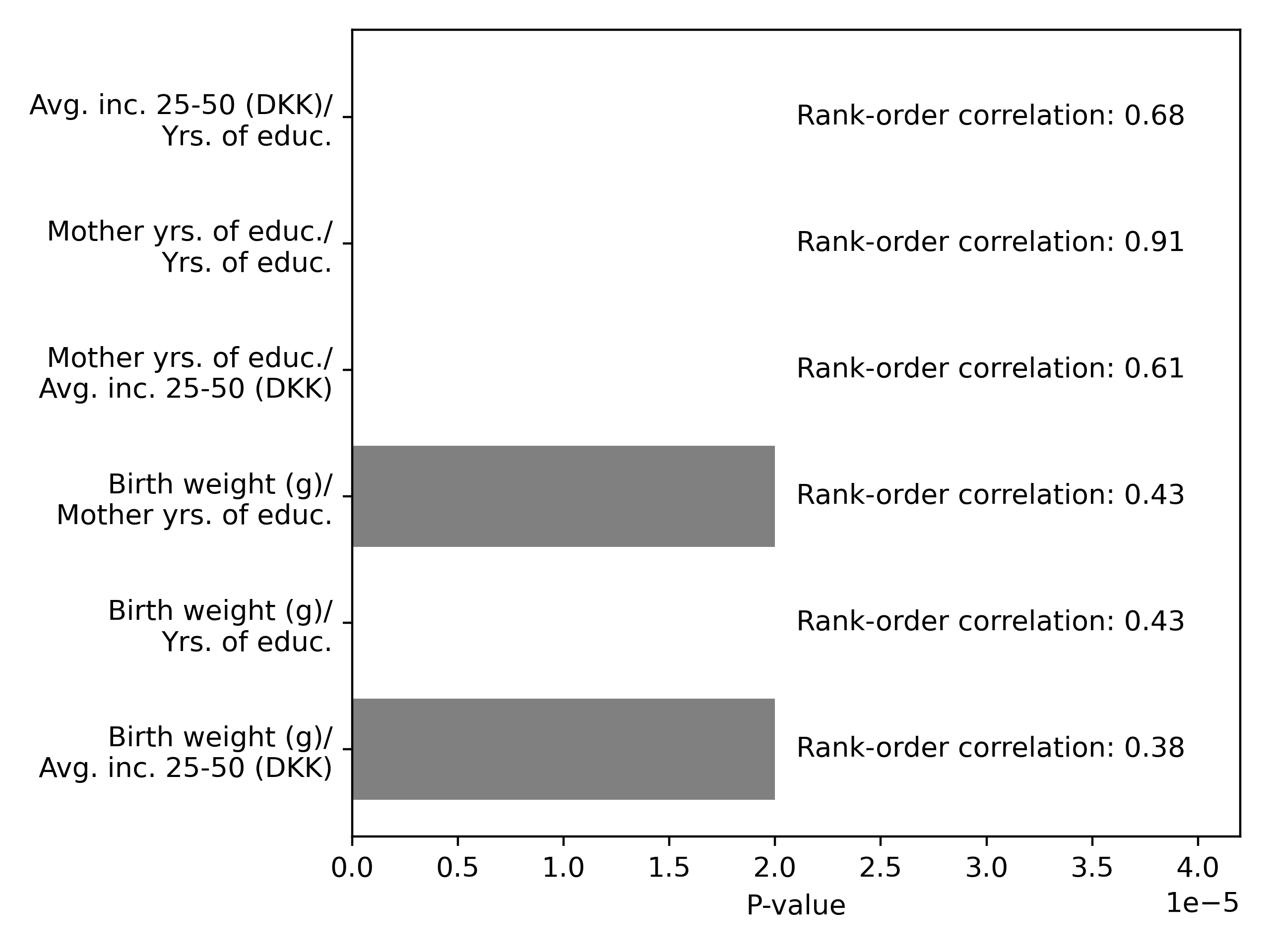}
    \caption{Rank-Order Correlation Between Mean Outcomes of Nurses.}
    \label{fig: correlation-coefficient-nurses}
    \begin{minipage}{1\linewidth}
        \vspace{1ex}
        \footnotesize{
        \textit{Notes:}
        The figure shows p-values and rank-order correlation coefficients from Spearman rank-order correlation tests.
        For each variable, the mean value for the children of each nurse is calculated, and the rank-order correlation between nurses is then calculated (i.e., comparing the ranks of nurses for one variable with the ranks of nurses for another variable).
        Given the number of nurses (\numberUniqueNurses), the p-values are calculated by means of a permutation test which randomly drew 100,000 permutations and then calculated the share of correlation coefficients below that of the non-permuted data, using this for the p-value.
        }
	\end{minipage}
\end{figure}

\begin{figure}
    \centering
    \includegraphics[width=0.75\linewidth]{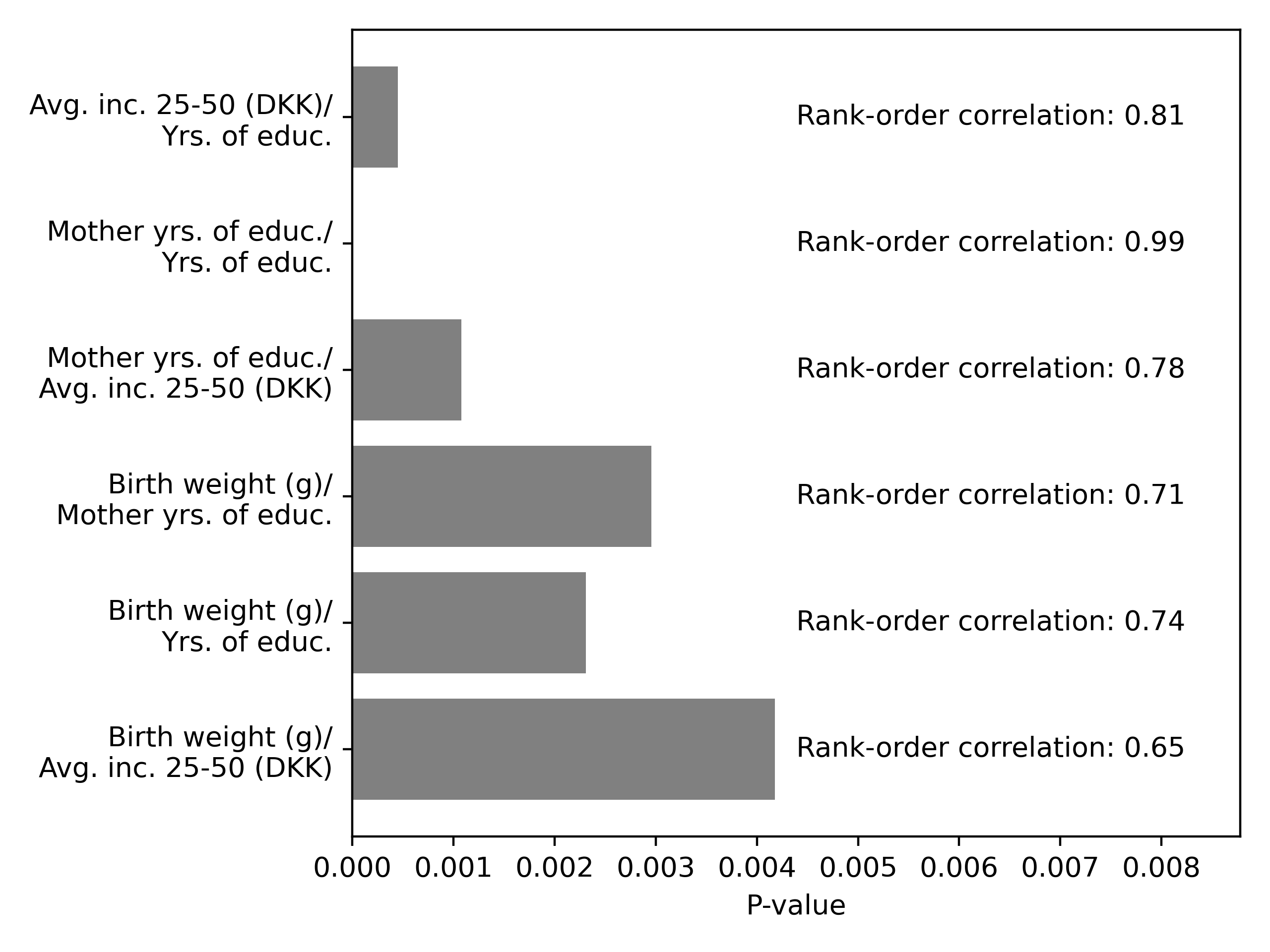}
    \caption{Rank-Order Correlation Between Mean Outcomes of Nurse Districts.}
    \label{fig: correlation-coefficient-districts}
    \begin{minipage}{1\linewidth}
        \vspace{1ex}
        \footnotesize{
        \textit{Notes:}
        The figure shows p-values and rank-order correlation coefficients from Spearman rank-order correlation tests.
        For each variable, the mean value for the children of each nurse district is calculated, and the rank-order correlation between nurses districts is then calculated (i.e., comparing the ranks of nurse districts for one variable with the ranks of nurse districts for another variable).
        Given the number of nurse districts (16), the p-values are calculated by means of a permutation test which randomly drew 100,000 permutations and then calculated the share of correlation coefficients below that of the non-permuted data, using this for the p-value.
        }
	\end{minipage}
\end{figure}


\begin{figure}
    \centering
    \subfloat[1959]{\label{subfig: box-plot-by-district-and-year-edulen-1959}\includegraphics[width=0.33\linewidth]{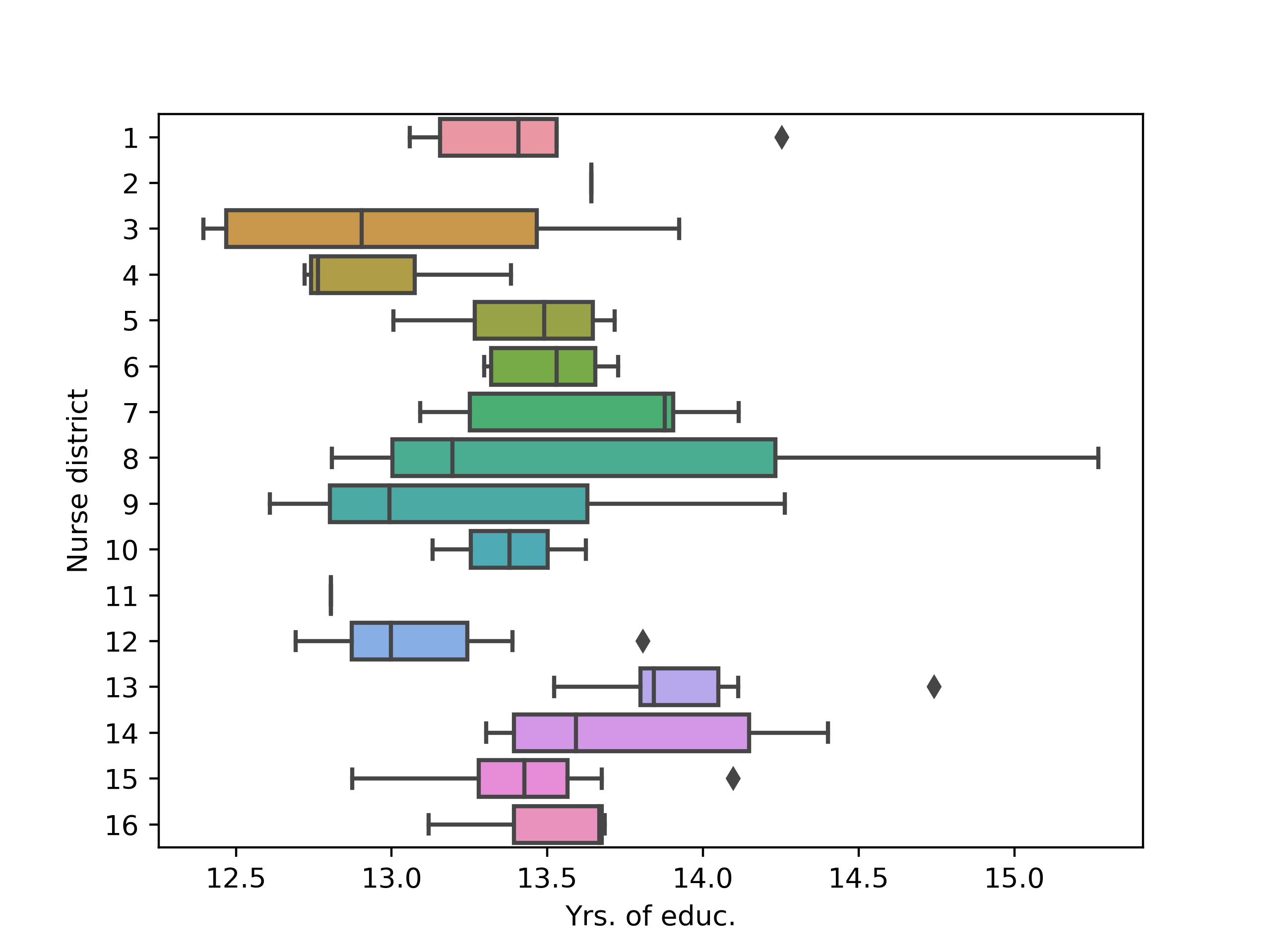}}
    \subfloat[1960]{\label{subfig: box-plot-by-district-and-year-edulen-1960}\includegraphics[width=0.33\linewidth]{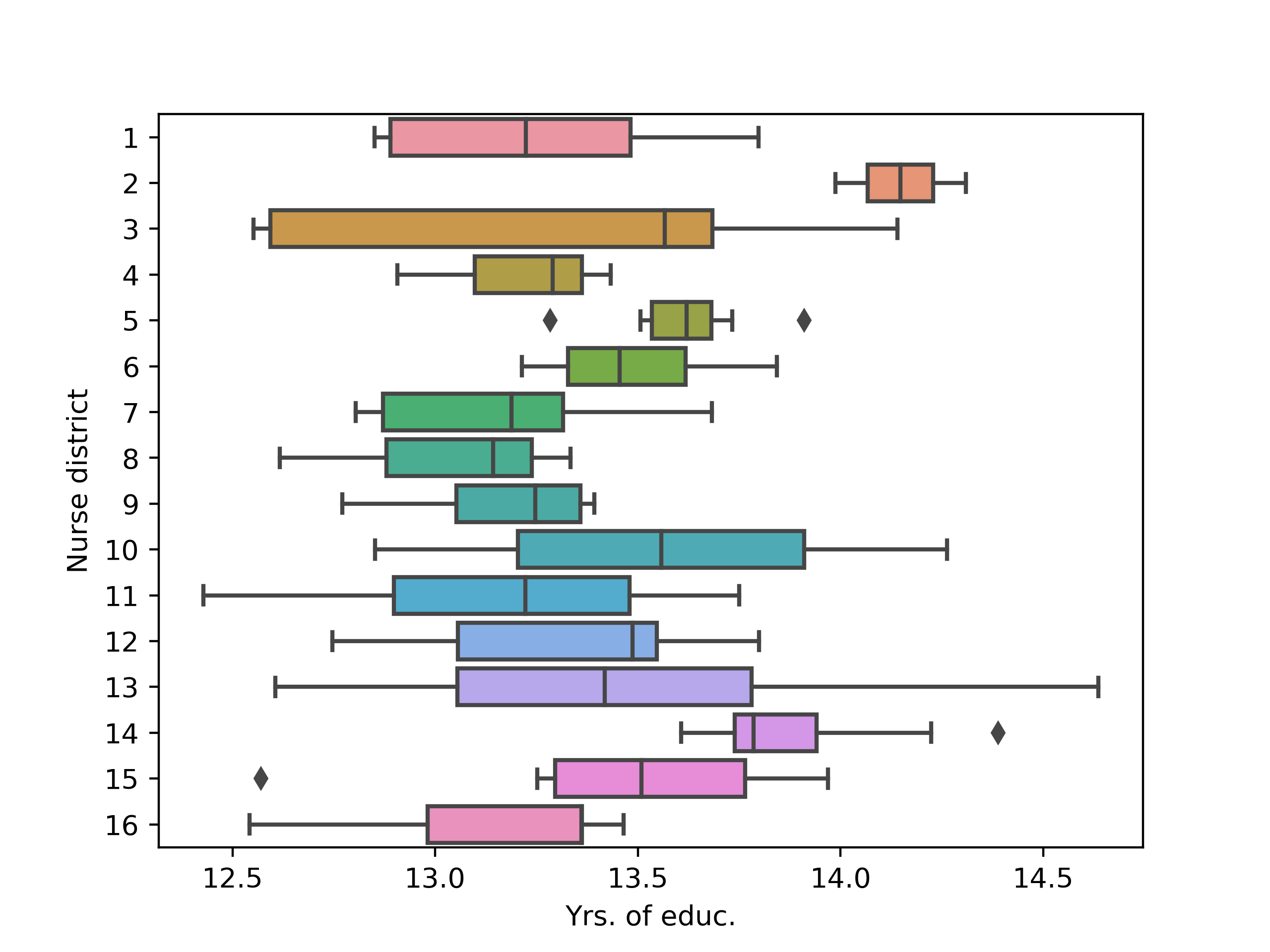}}
    \subfloat[1961]{\label{subfig: box-plot-by-district-and-year-edulen-1961}\includegraphics[width=0.33\linewidth]{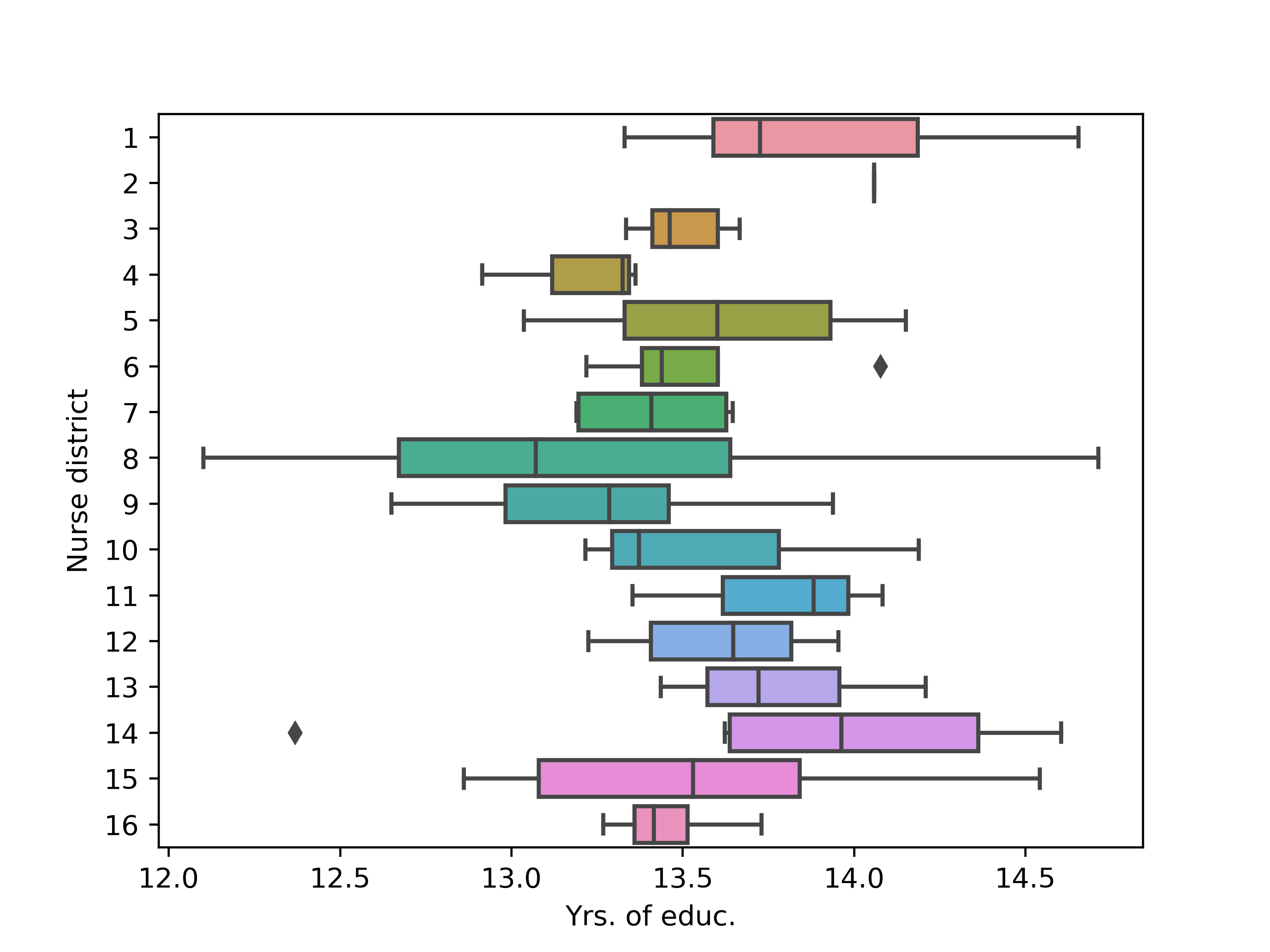}}

    \subfloat[1962]{\label{subfig: box-plot-by-district-and-year-edulen-1962}\includegraphics[width=0.33\linewidth]{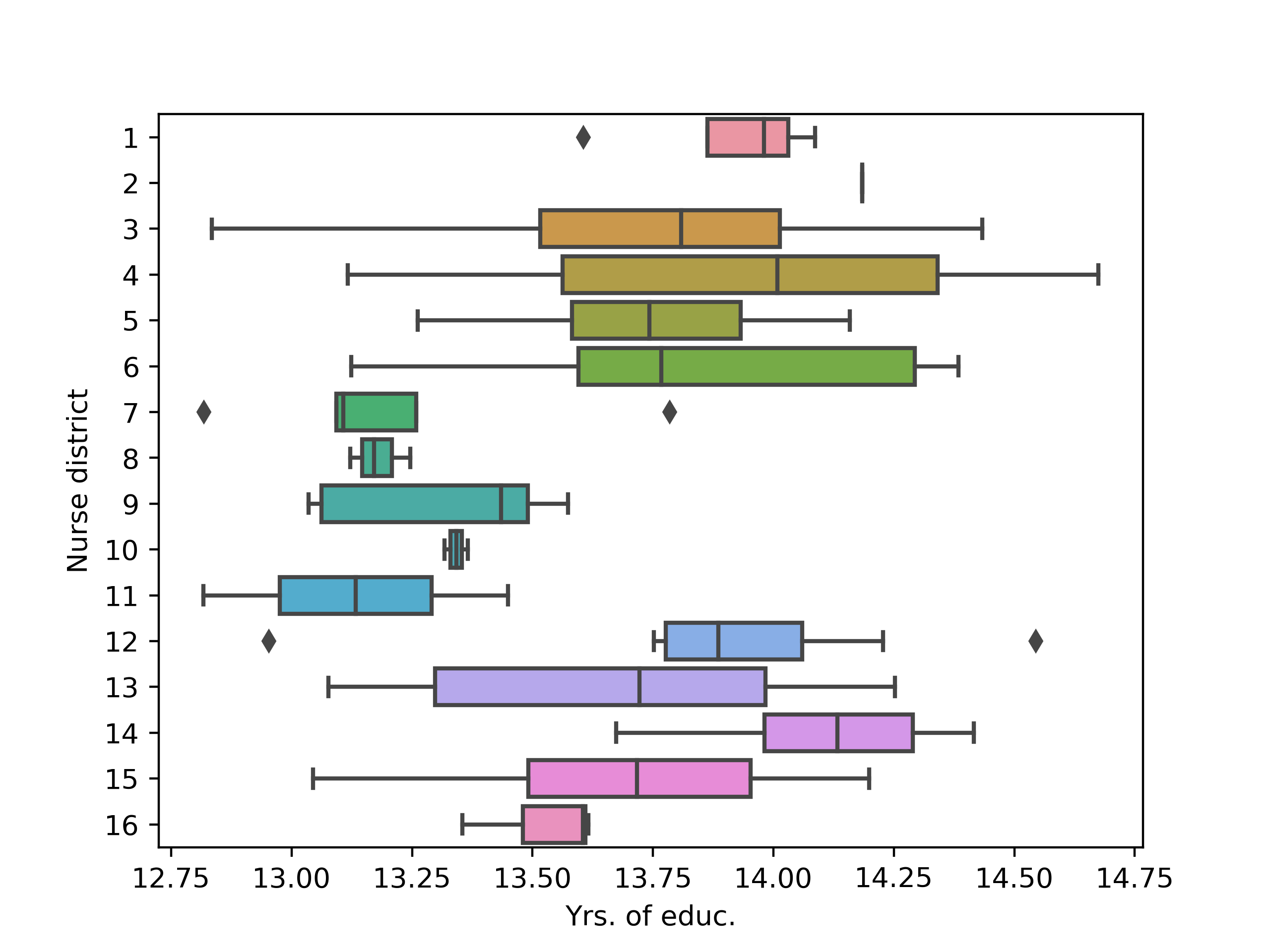}}
    \subfloat[1963]{\label{subfig: box-plot-by-district-and-year-edulen-1963}\includegraphics[width=0.33\linewidth]{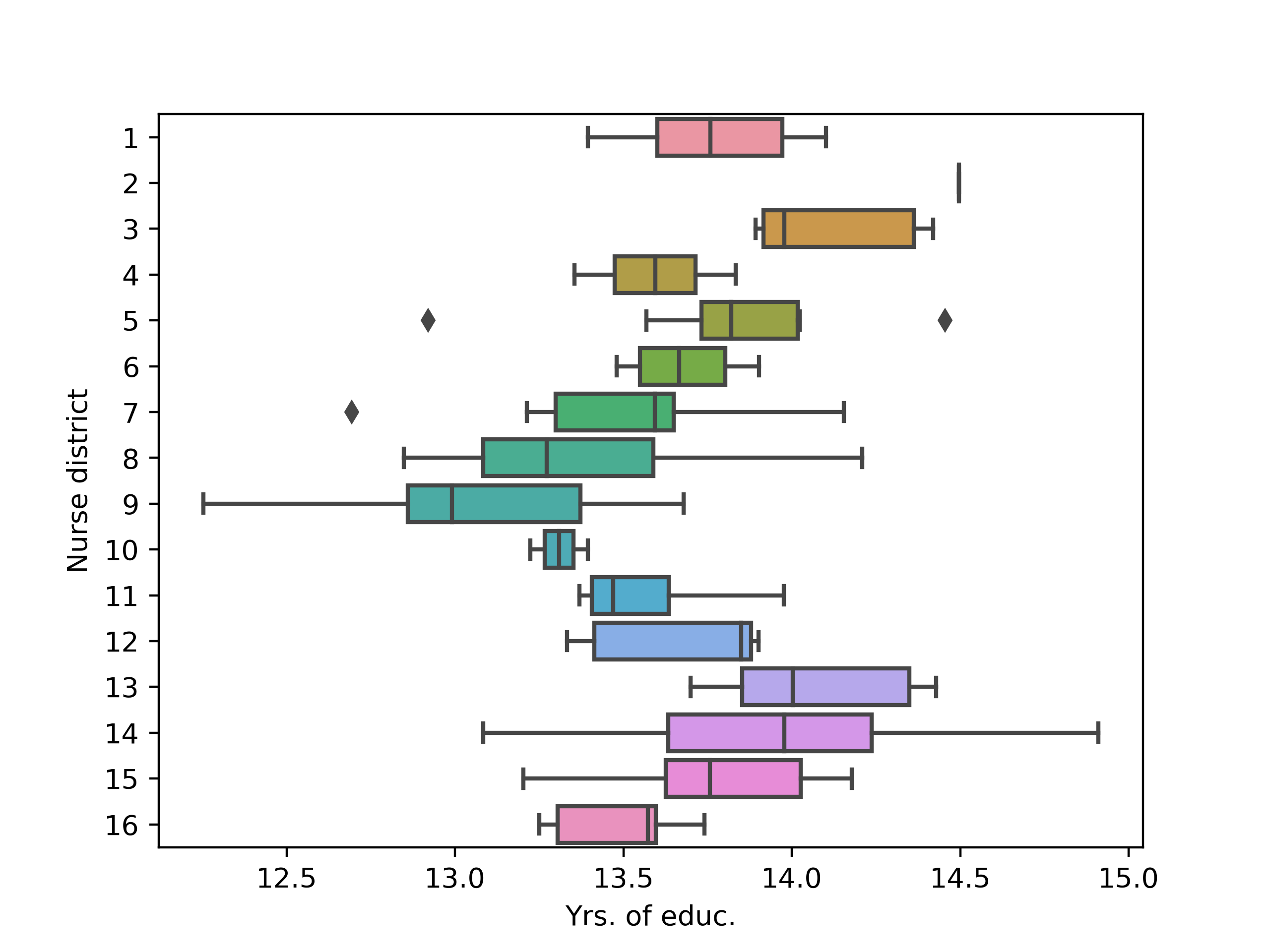}}
    \subfloat[1964]{\label{subfig: box-plot-by-district-and-year-edulen-1964}\includegraphics[width=0.33\linewidth]{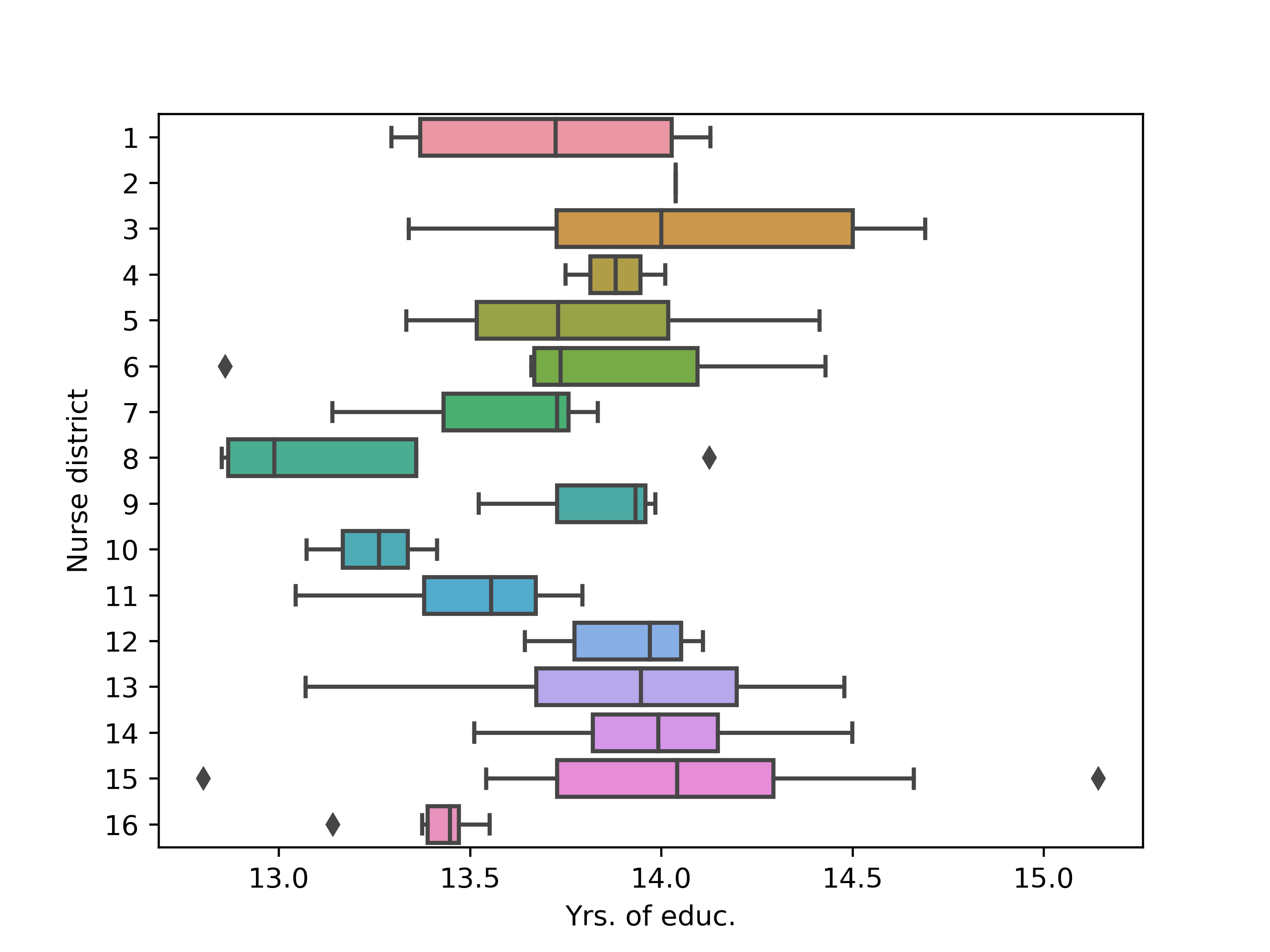}}

    \subfloat[1965]{\label{subfig: box-plot-by-district-and-year-edulen-1965}\includegraphics[width=0.33\linewidth]{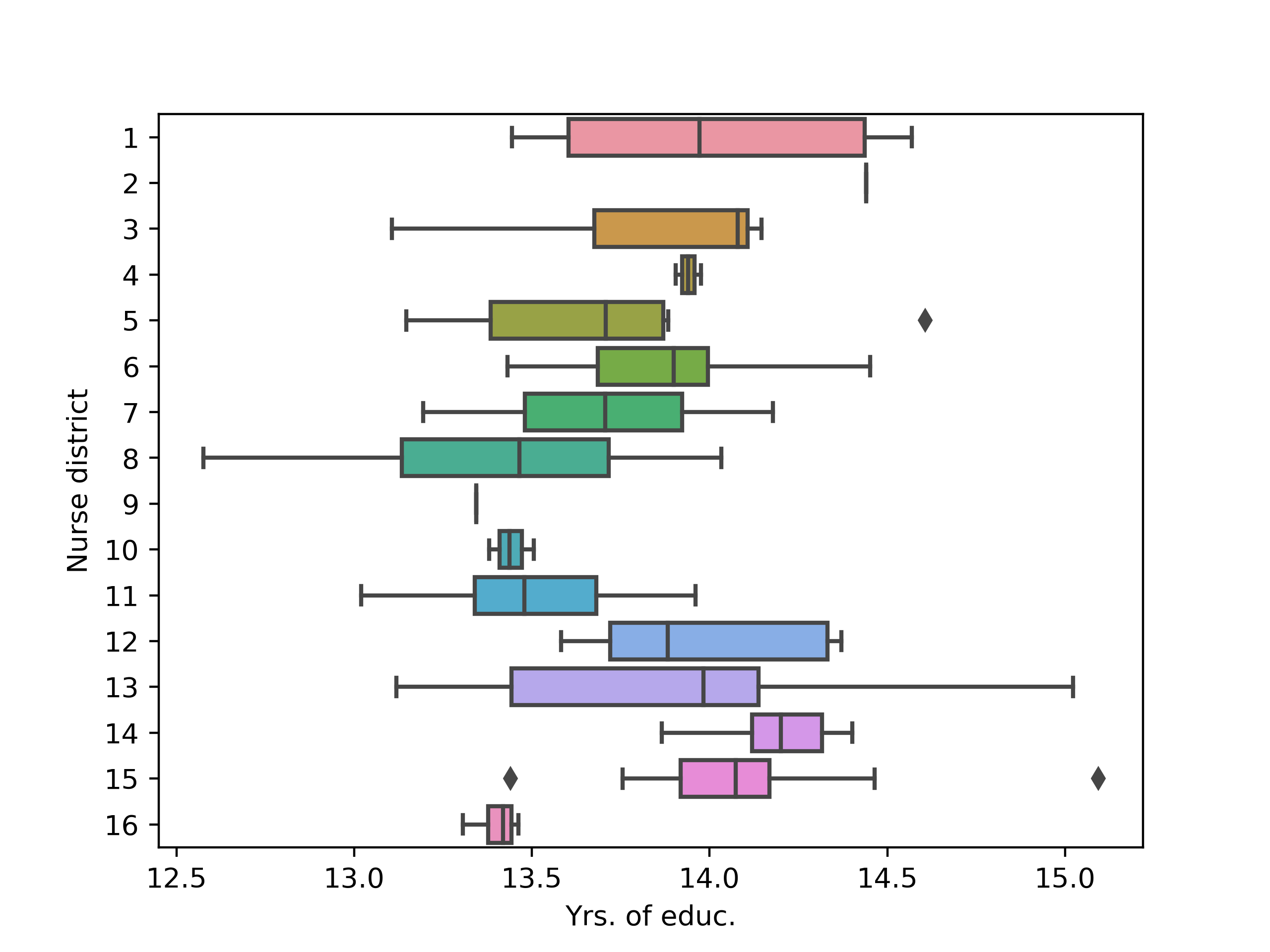}}
    \subfloat[1966]{\label{subfig: box-plot-by-district-and-year-edulen-1966}\includegraphics[width=0.33\linewidth]{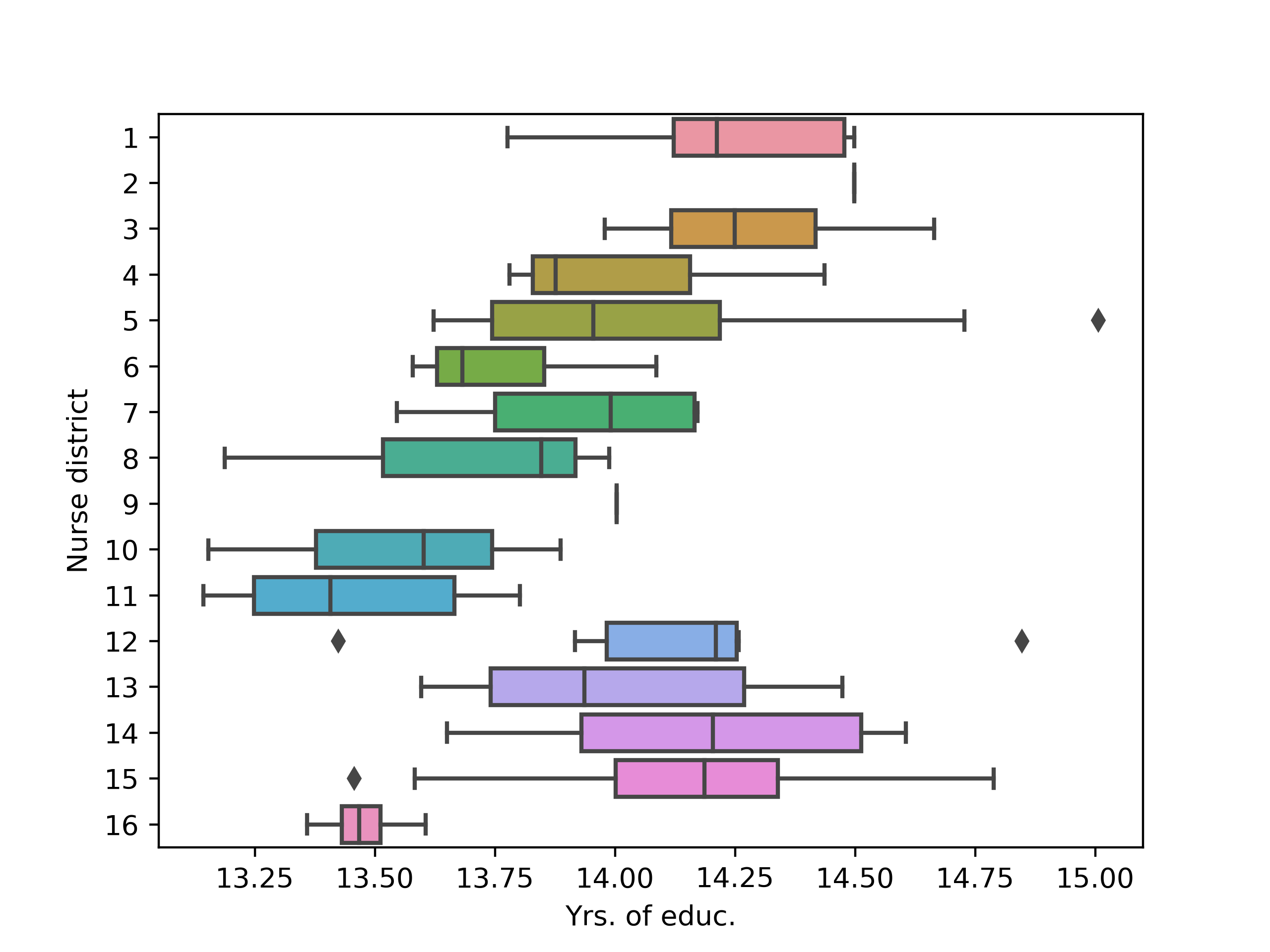}}
    \subfloat[1967]{\label{subfig: box-plot-by-district-and-year-edulen-1967}\includegraphics[width=0.33\linewidth]{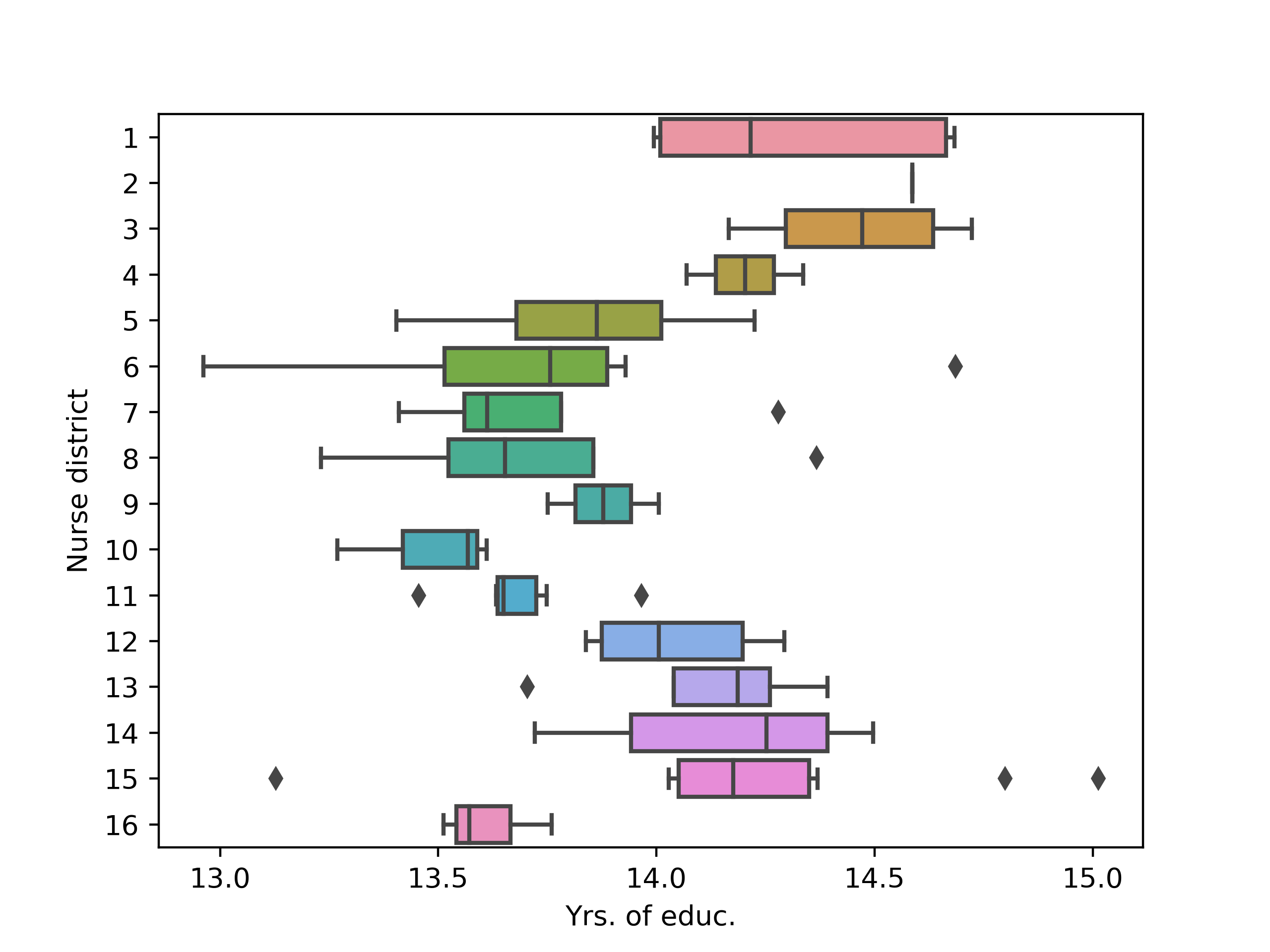}}
    
    \caption{Box Plots of Average Years of Education by Nurses for Each Nurse District and Year.}
    \label{fig: box-plot-by-district-and-year-edulen}
    \begin{minipage}{1\linewidth}
        \vspace{1ex}
        \footnotesize{
        \textit{Notes:}
        The figure shows box plots of average years of education of the children by nurse for each district and year, meaning that each point represents an average of the children allocated a specific nurse born in a specific year.
        The different panels refer to different years (1959-1967) and the different box plots within each panel refer to different nurse districts (1-16).
		}
	\end{minipage}
\end{figure}

\begin{figure}
    \centering
    \subfloat[1959]{\label{subfig: box-plot-by-district-and-year-avg_inc_25_50-1959}\includegraphics[width=0.33\linewidth]{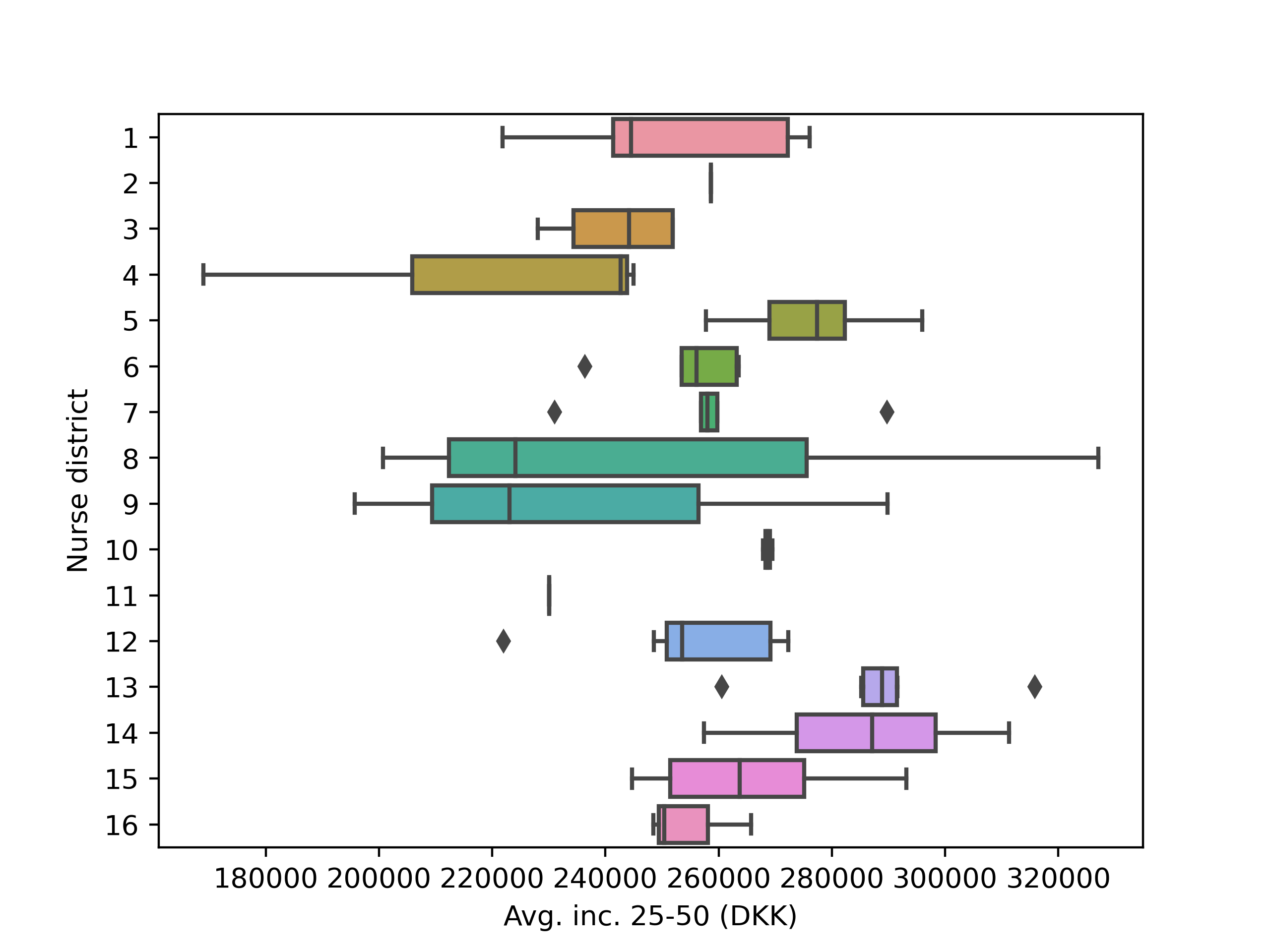}}
    \subfloat[1960]{\label{subfig: box-plot-by-district-and-year-avg_inc_25_50-1960}\includegraphics[width=0.33\linewidth]{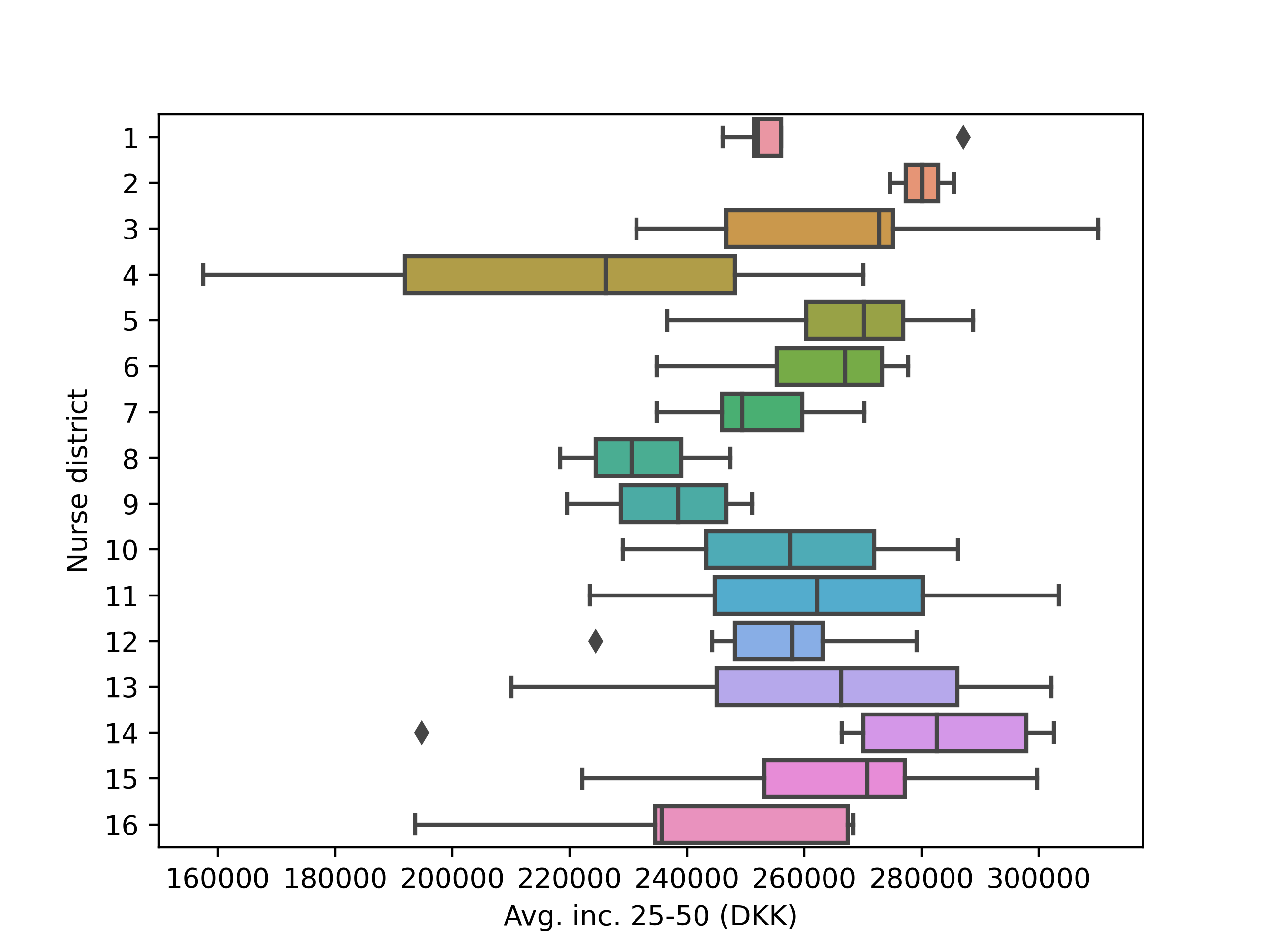}}
    \subfloat[1961]{\label{subfig: box-plot-by-district-and-year-avg_inc_25_50-1961}\includegraphics[width=0.33\linewidth]{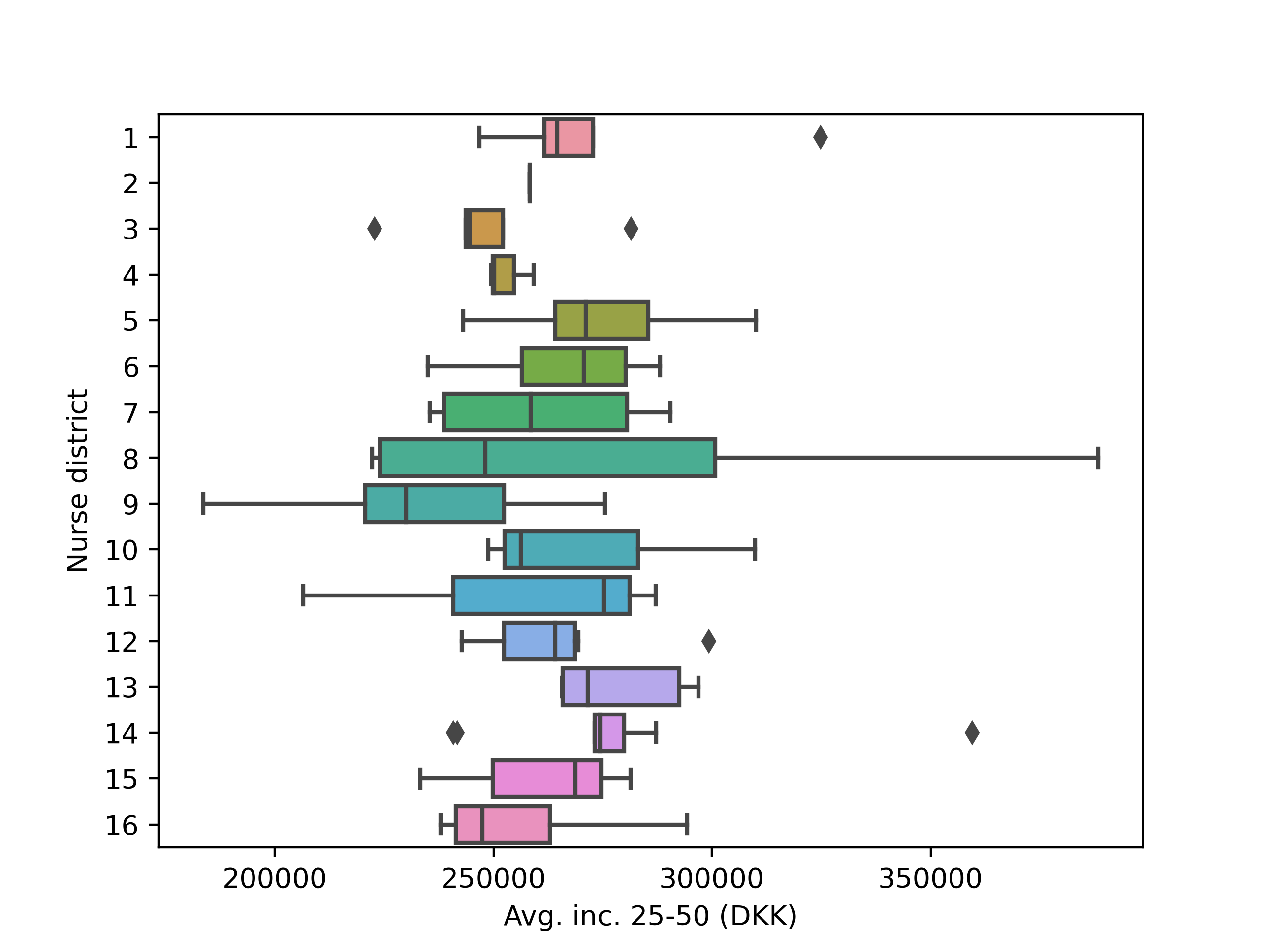}}

    \subfloat[1962]{\label{subfig: box-plot-by-district-and-year-avg_inc_25_50-1962}\includegraphics[width=0.33\linewidth]{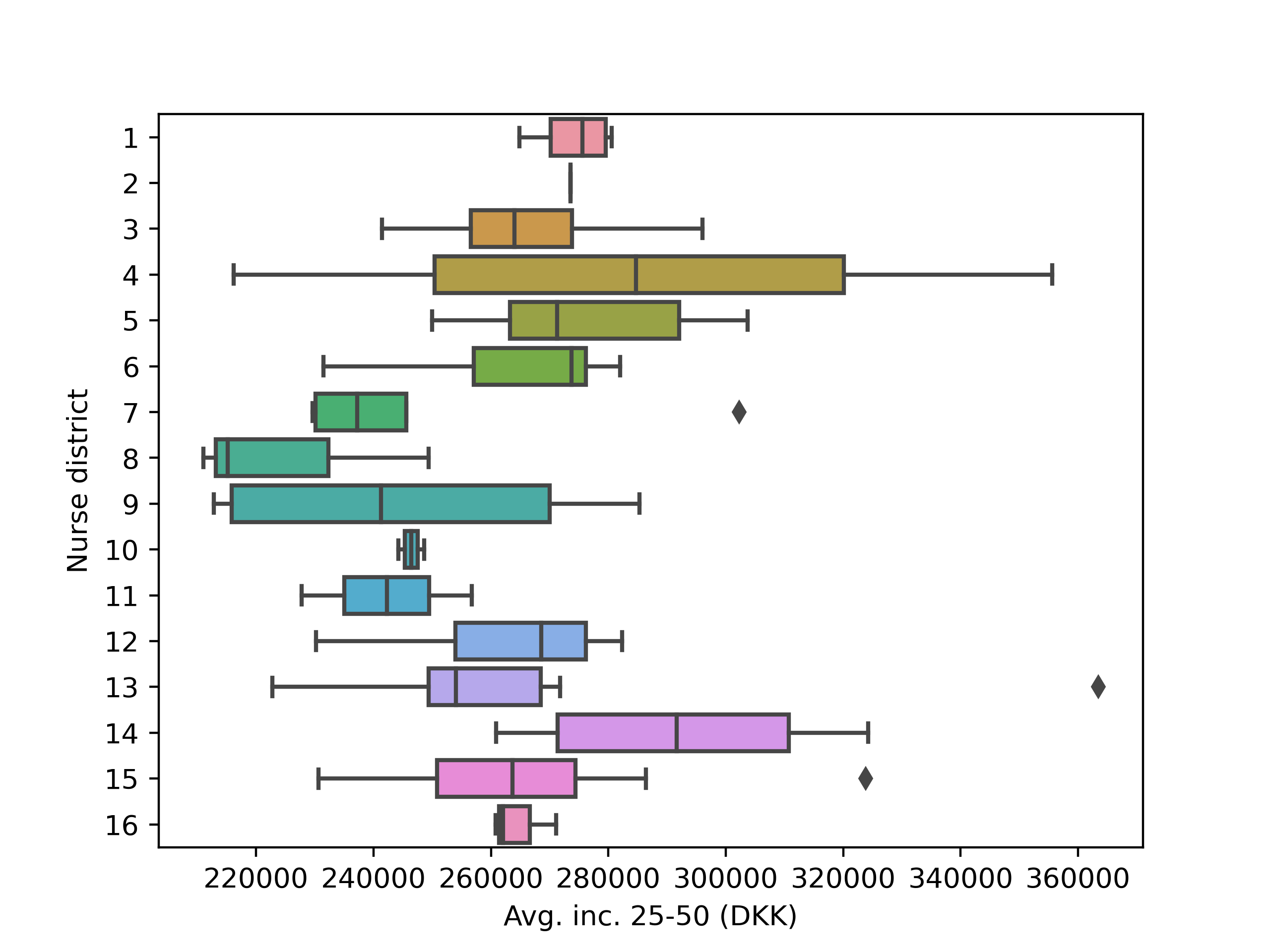}}
    \subfloat[1963]{\label{subfig: box-plot-by-district-and-year-avg_inc_25_50-1963}\includegraphics[width=0.33\linewidth]{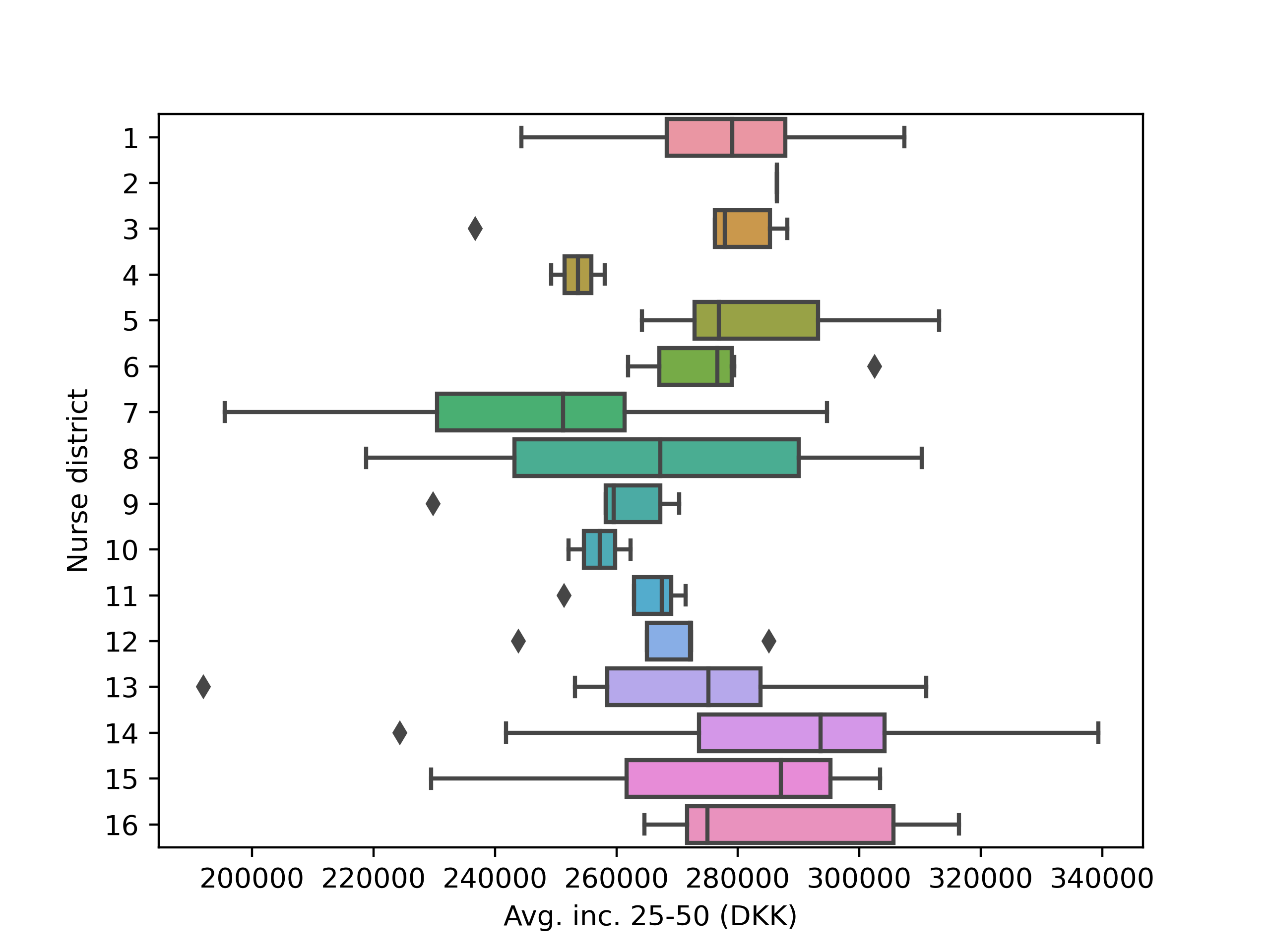}}
    \subfloat[1964]{\label{subfig: box-plot-by-district-and-year-avg_inc_25_50-1964}\includegraphics[width=0.33\linewidth]{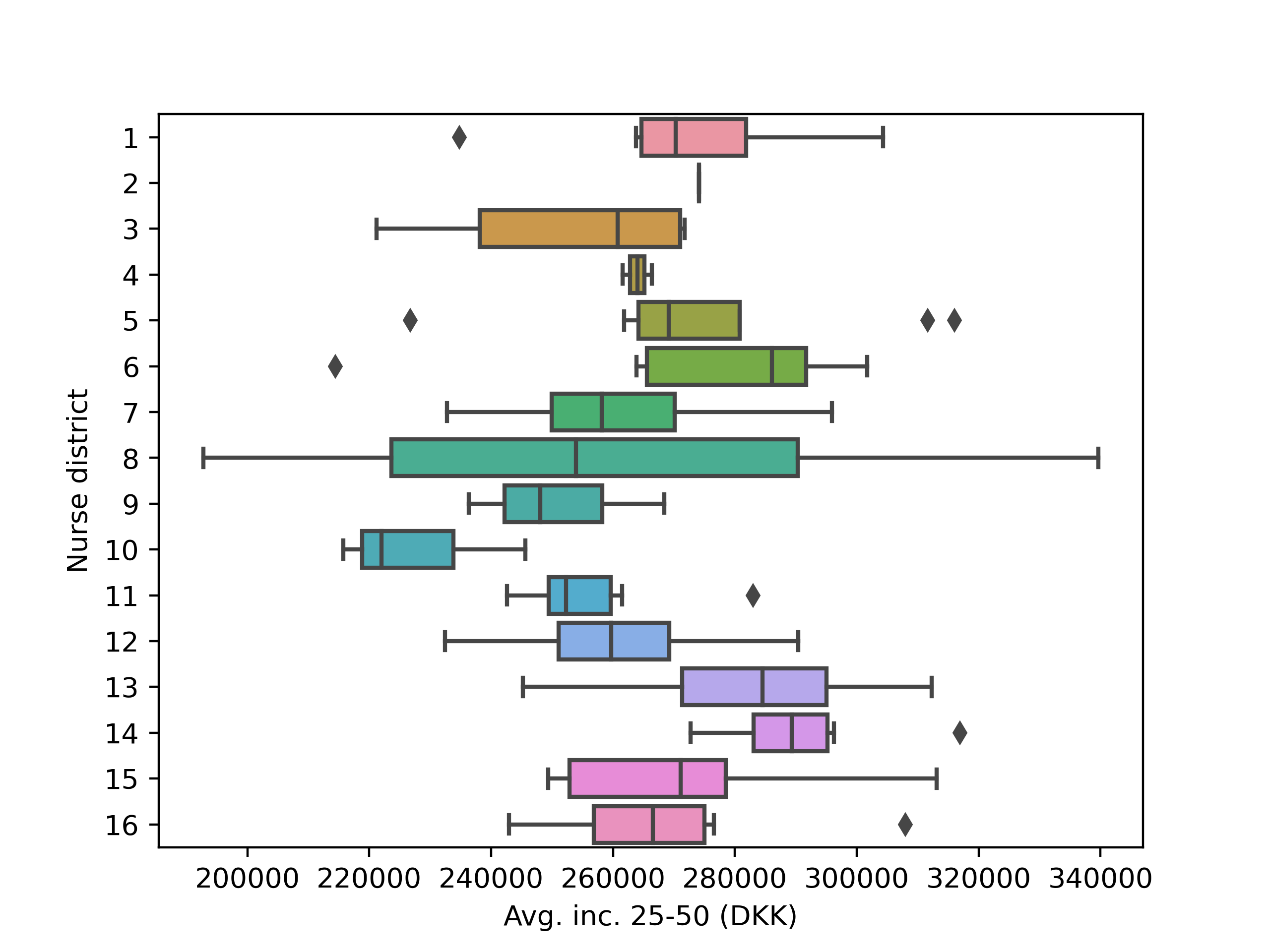}}

    \subfloat[1965]{\label{subfig: box-plot-by-district-and-year-avg_inc_25_50-1965}\includegraphics[width=0.33\linewidth]{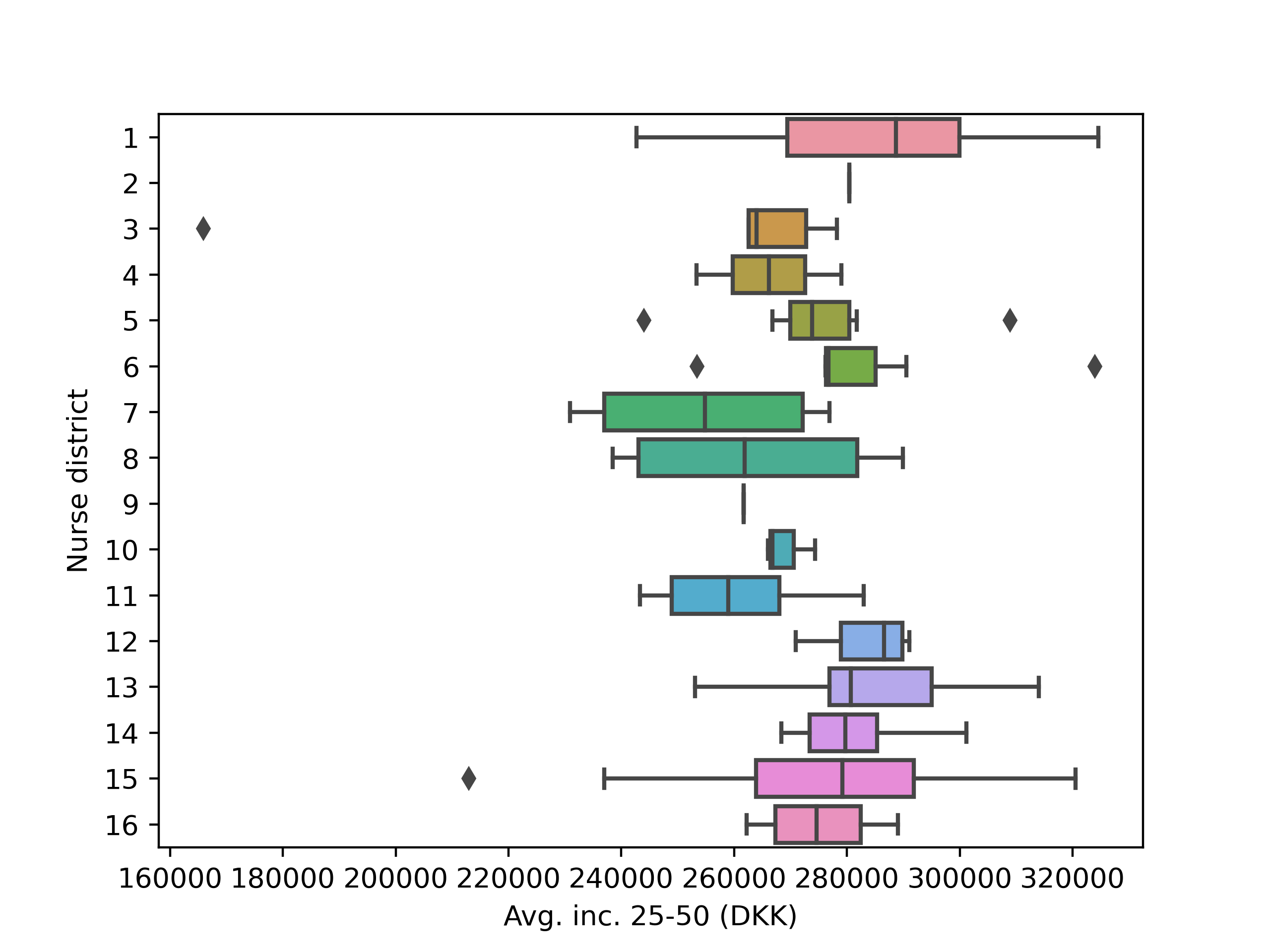}}
    \subfloat[1966]{\label{subfig: box-plot-by-district-and-year-avg_inc_25_50-1966}\includegraphics[width=0.33\linewidth]{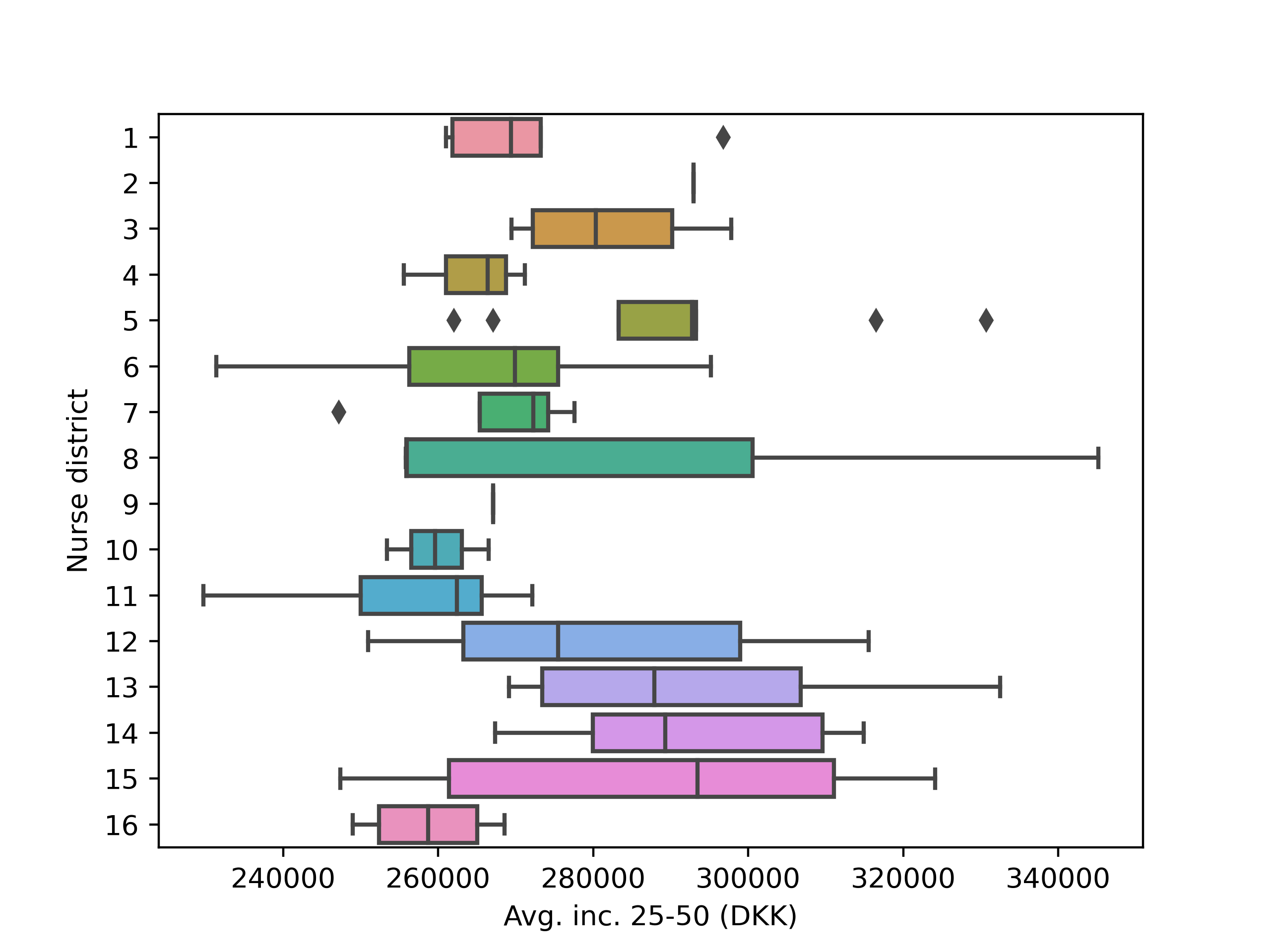}}
    \subfloat[1967]{\label{subfig: box-plot-by-district-and-year-avg_inc_25_50-1967}\includegraphics[width=0.33\linewidth]{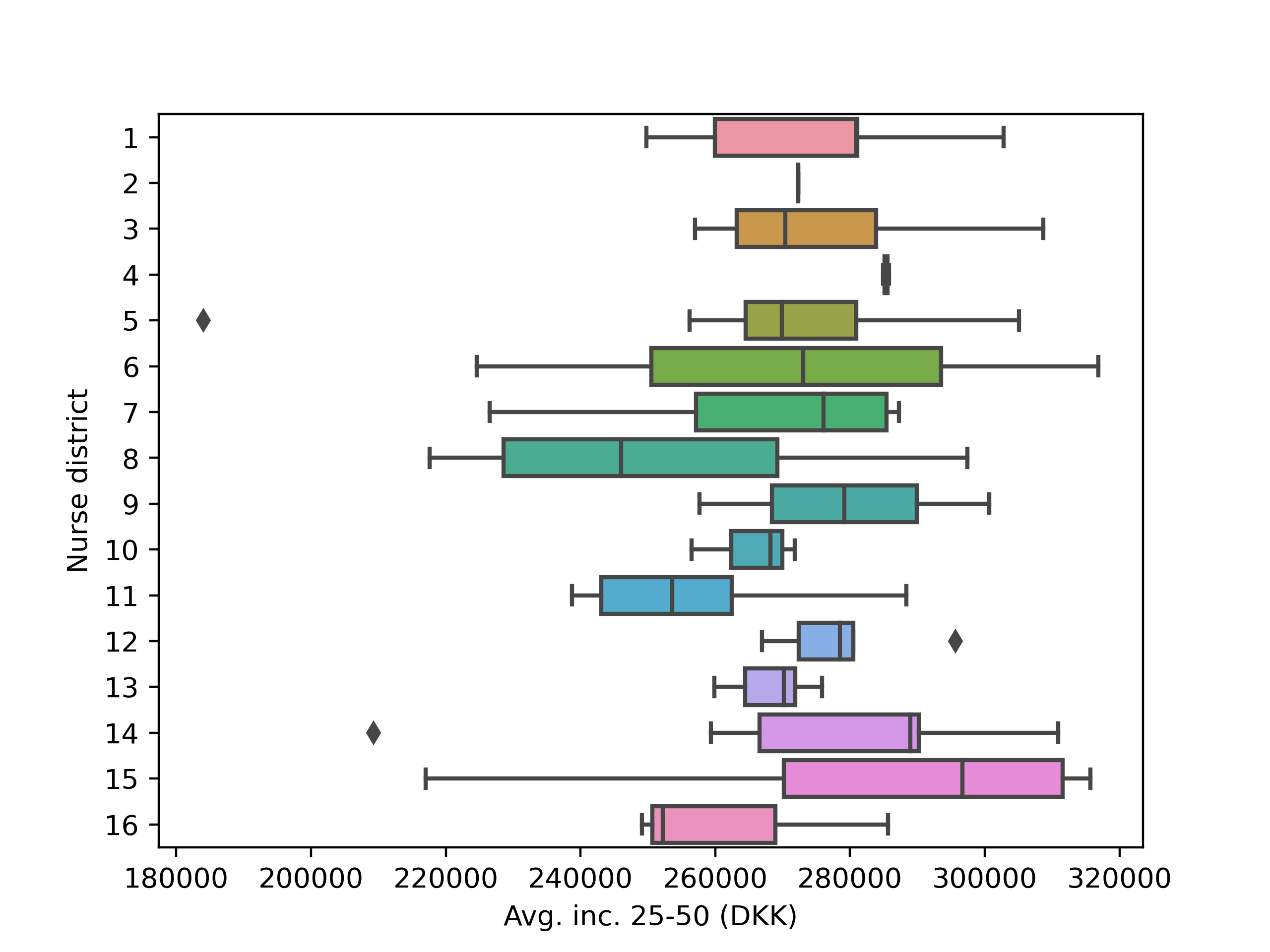}}
    
    \caption{Box Plots of Average Earnings During Ages 25-50 by Nurses for Each Nurse District and Year.}
    \label{fig: box-plot-by-district-and-year-avg_inc_25_50}
    \begin{minipage}{1\linewidth}
        \vspace{1ex}
        \footnotesize{
        \textit{Notes:}
        The figure shows box plots of average earnings during ages 25-50 of the children by nurse for each district and year, meaning that each point represents an average of the children allocated a specific nurse born in a specific year.
        The different panels refer to different years (1959-1967) and the different box plots within each panel refer to different nurse districts (1-16).
		}
	\end{minipage}
\end{figure}


\begin{figure}
    \centering
    \subfloat[1959]{\label{subfig: ci-plot-by-district-and-year-edulen-1959}\includegraphics[width=0.33\linewidth]{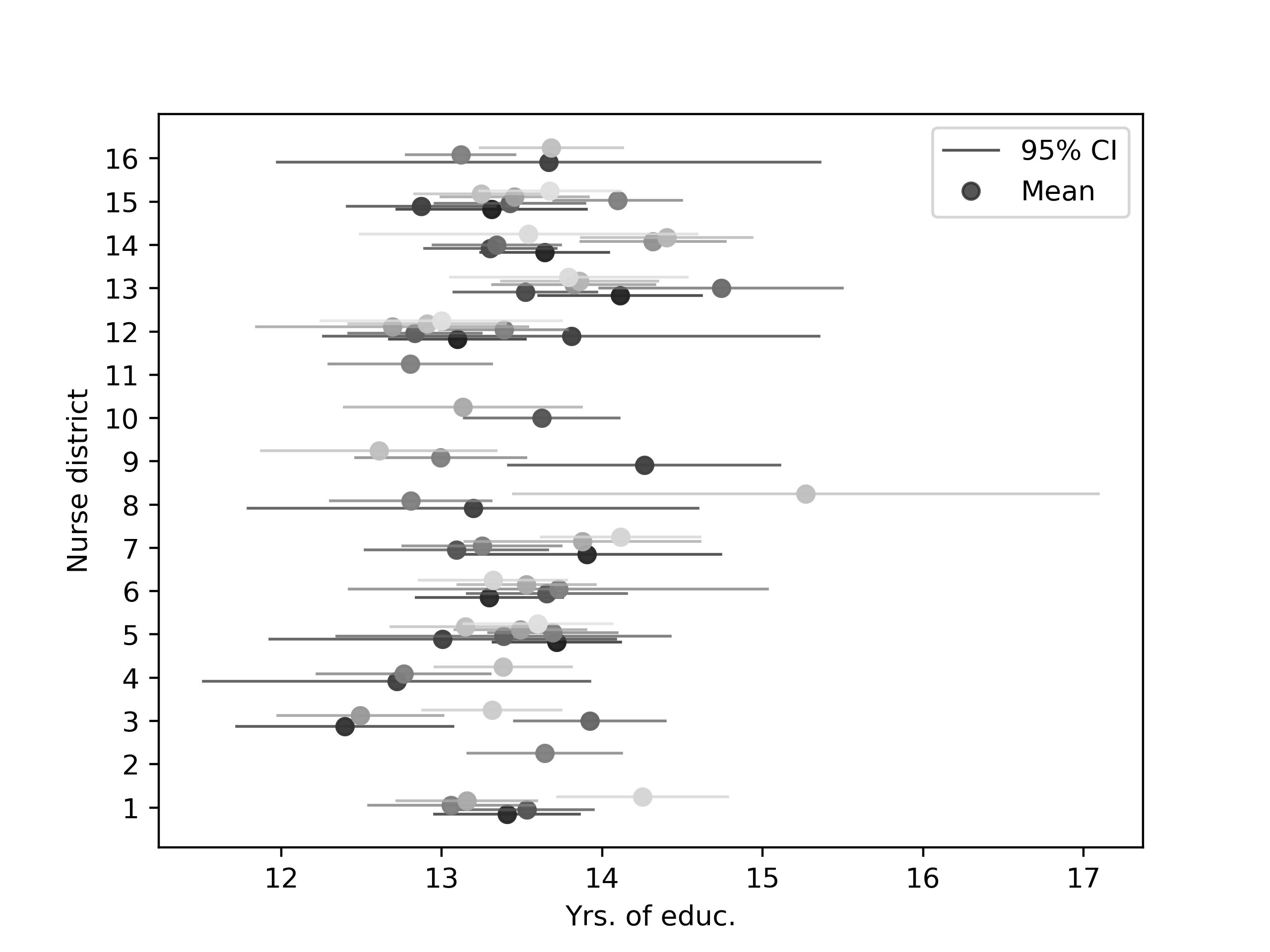}}
    \subfloat[1960]{\label{subfig: ci-plot-by-district-and-year-edulen-1960}\includegraphics[width=0.33\linewidth]{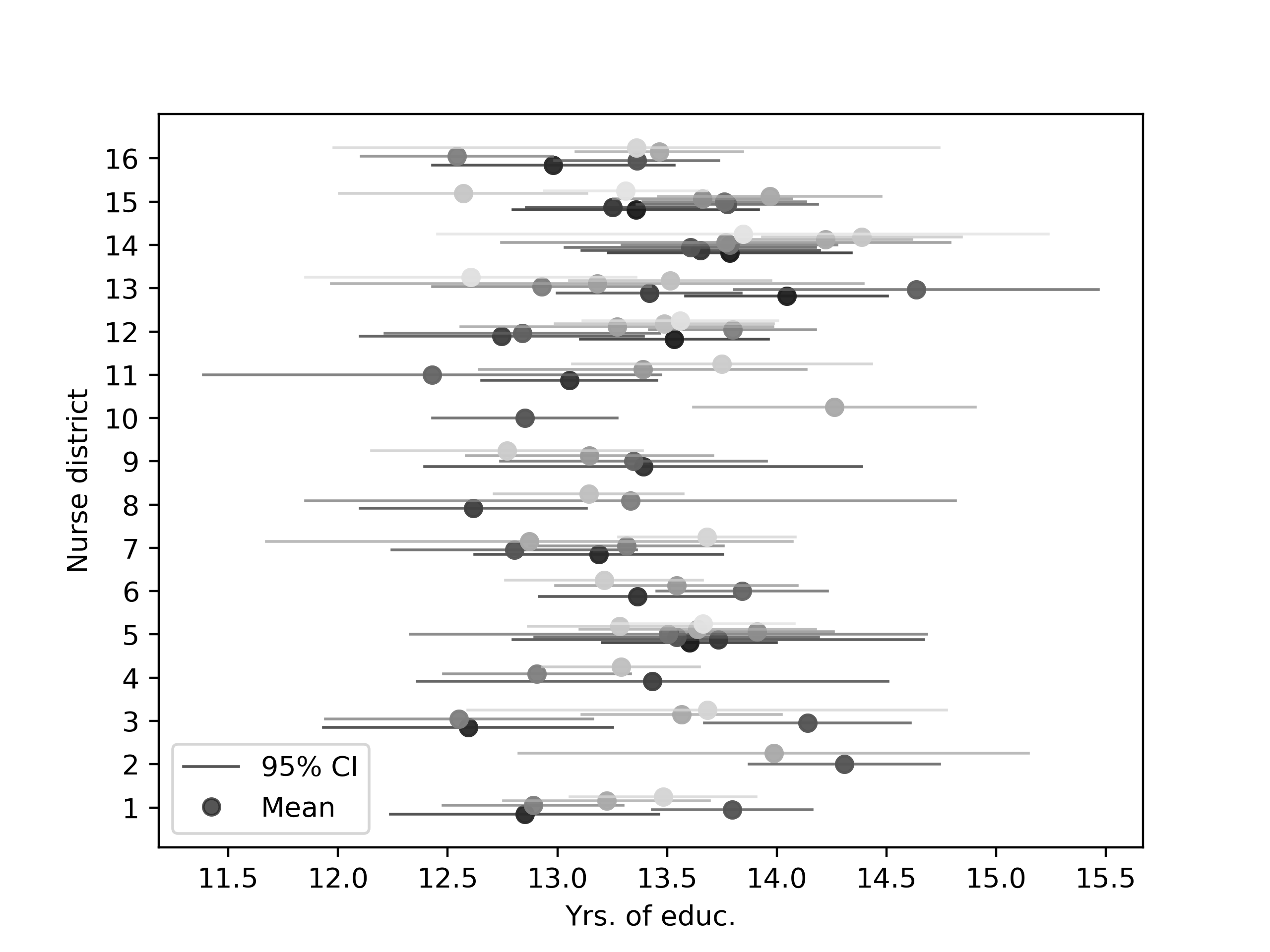}}
    \subfloat[1961]{\label{subfig: ci-plot-by-district-and-year-edulen-1961}\includegraphics[width=0.33\linewidth]{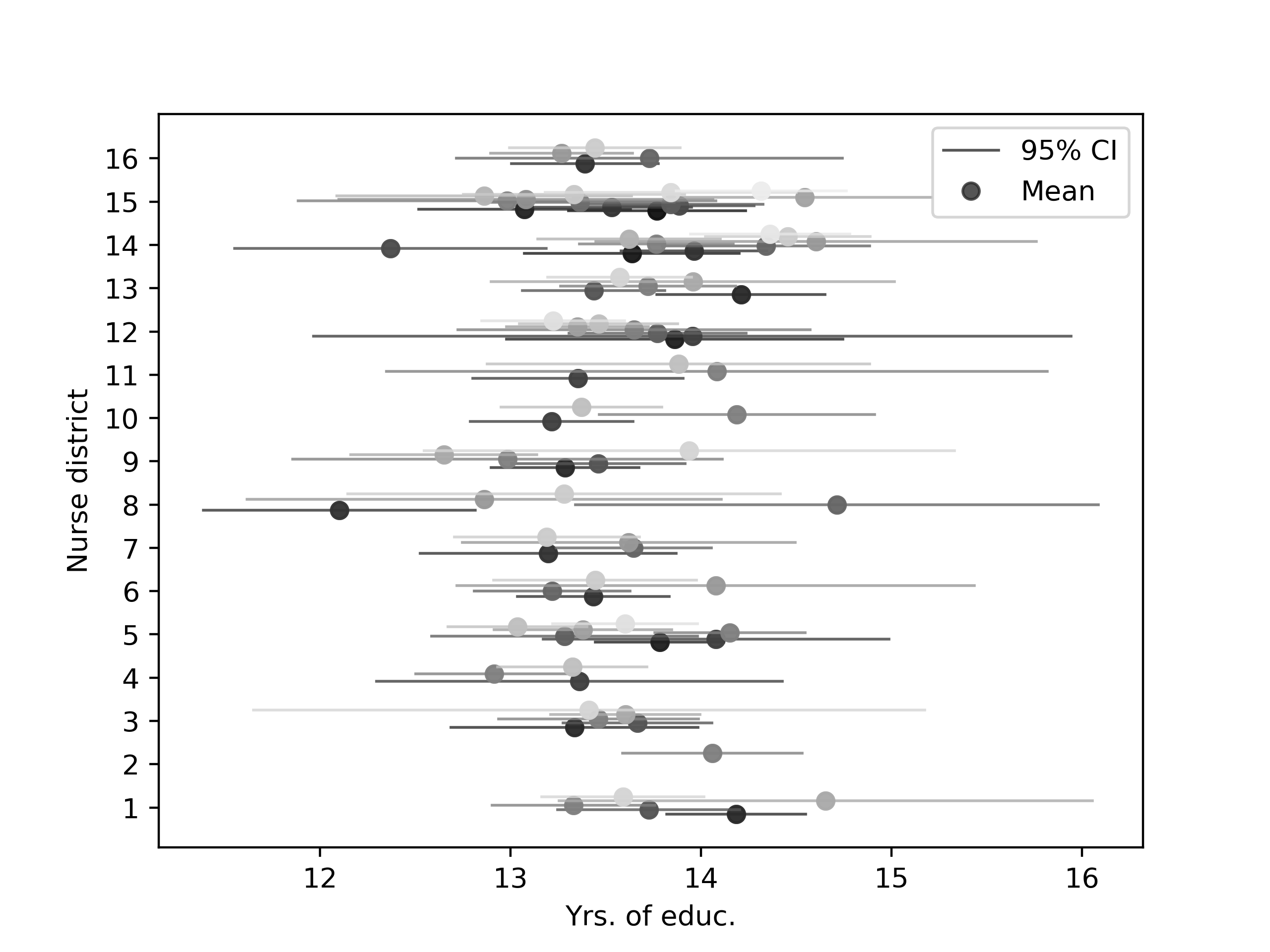}}

    \subfloat[1962]{\label{subfig: ci-plot-by-district-and-year-edulen-1962}\includegraphics[width=0.33\linewidth]{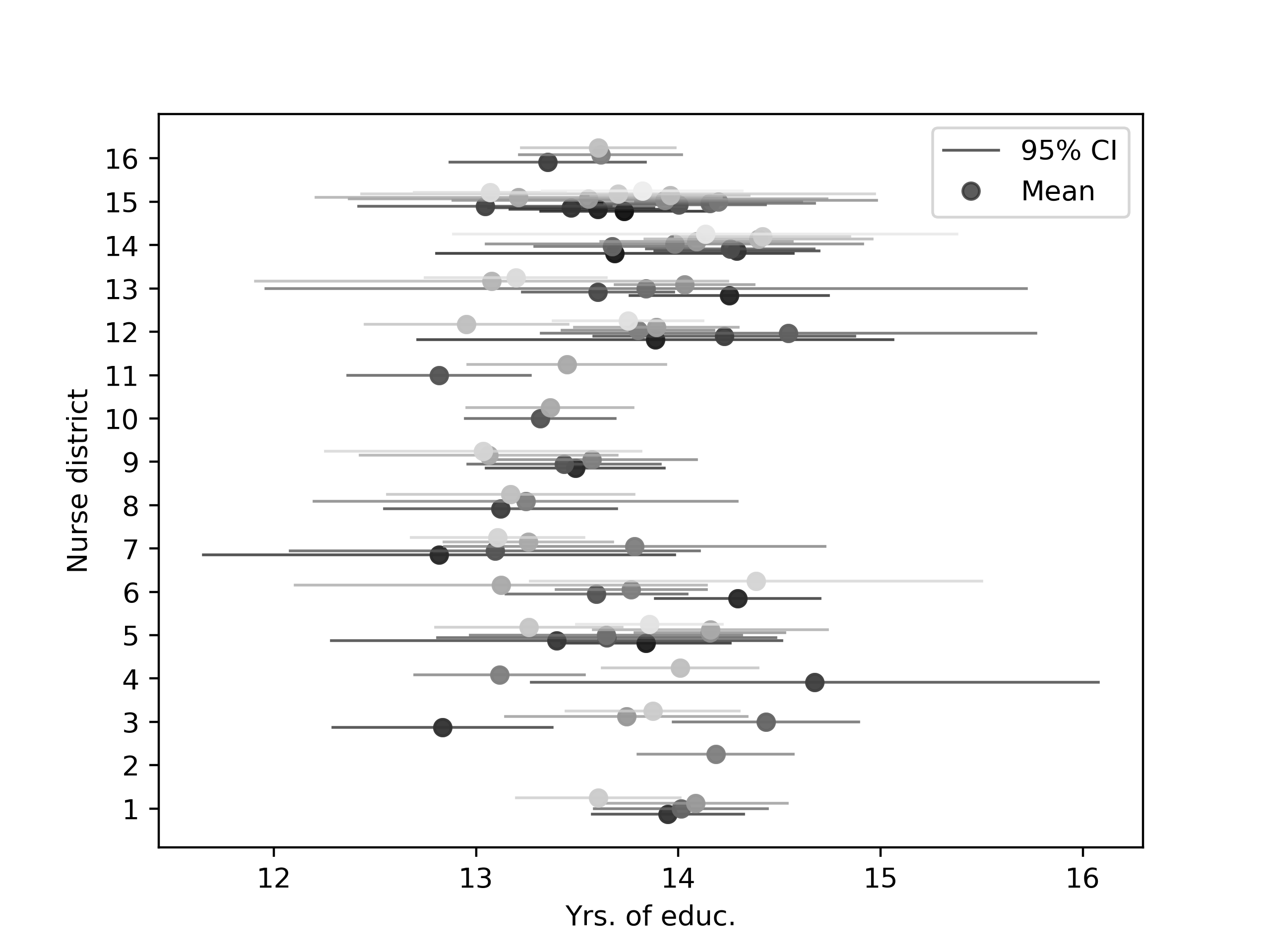}}
    \subfloat[1963]{\label{subfig: ci-plot-by-district-and-year-edulen-1963}\includegraphics[width=0.33\linewidth]{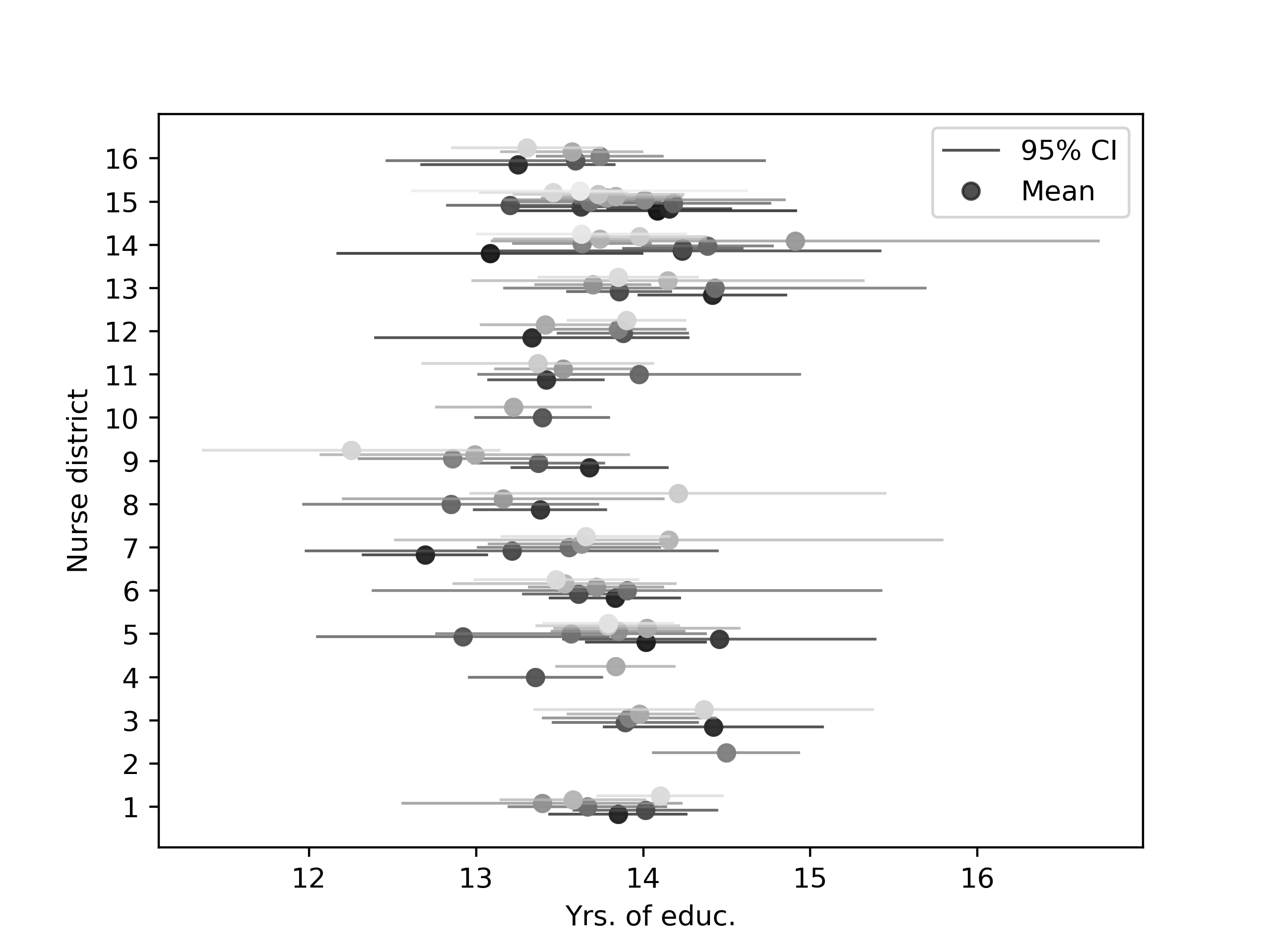}}
    \subfloat[1964]{\label{subfig: ci-plot-by-district-and-year-edulen-1964}\includegraphics[width=0.33\linewidth]{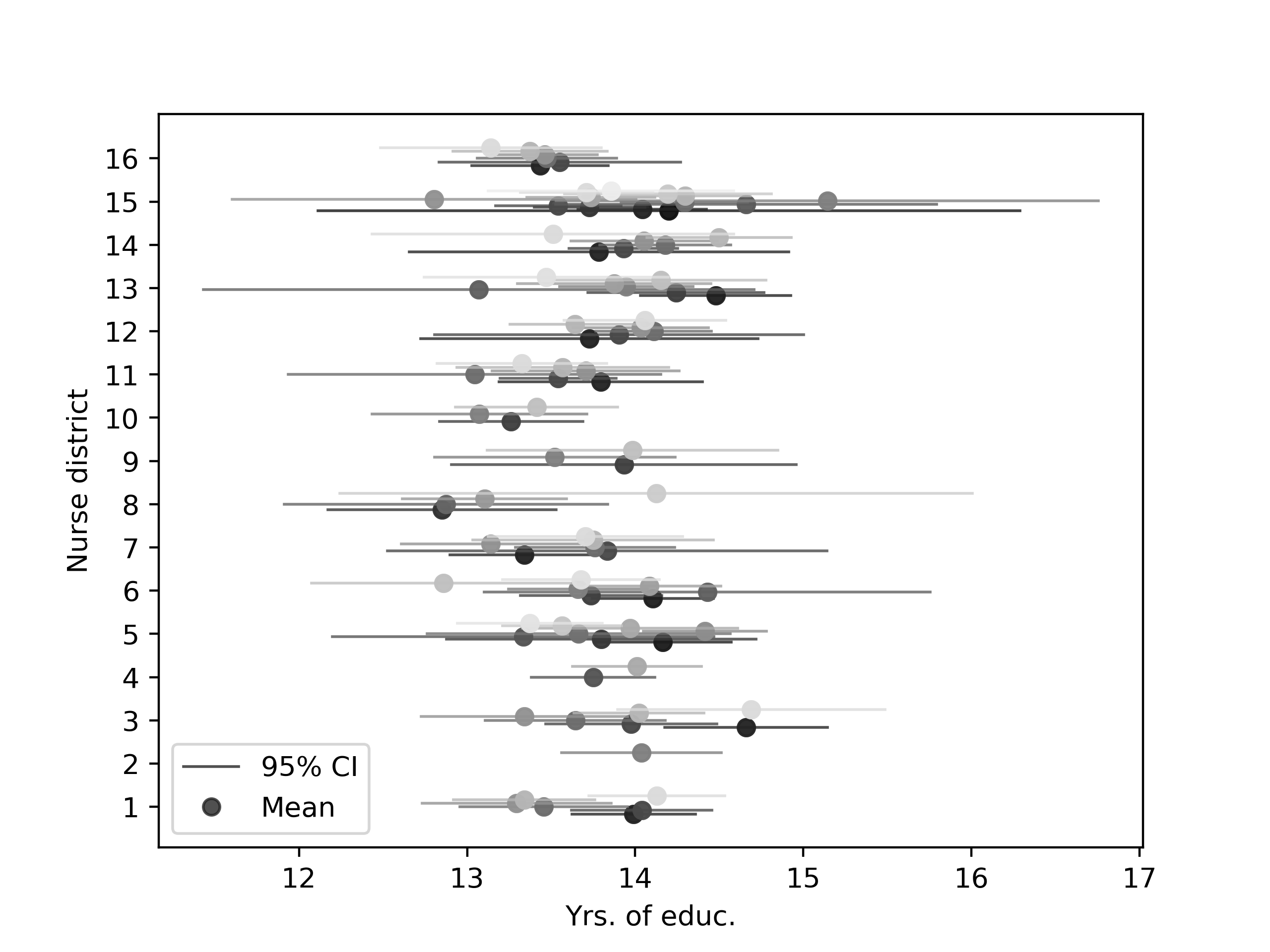}}

    \subfloat[1965]{\label{subfig: ci-plot-by-district-and-year-edulen-1965}\includegraphics[width=0.33\linewidth]{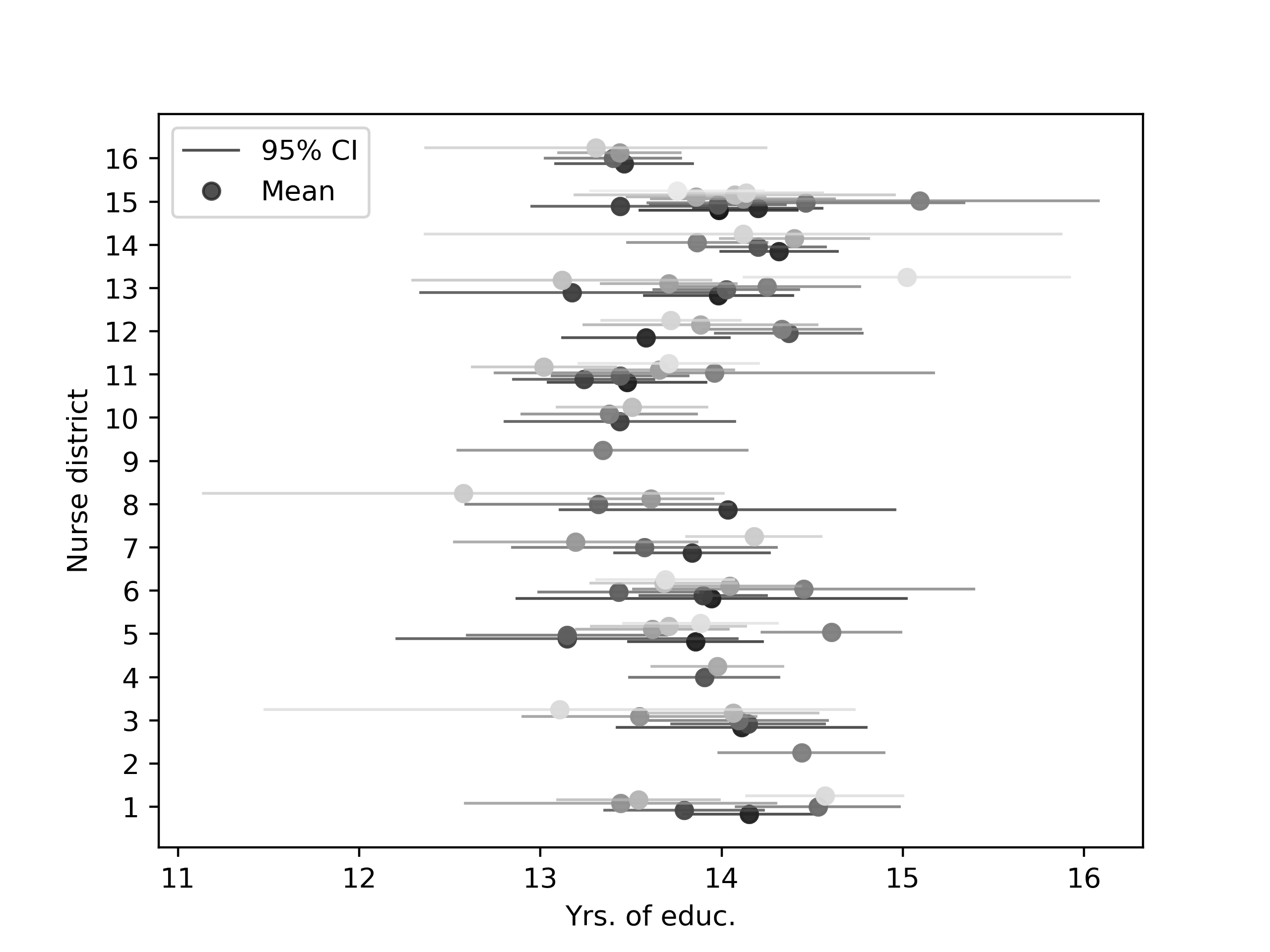}}
    \subfloat[1966]{\label{subfig: ci-plot-by-district-and-year-edulen-1966}\includegraphics[width=0.33\linewidth]{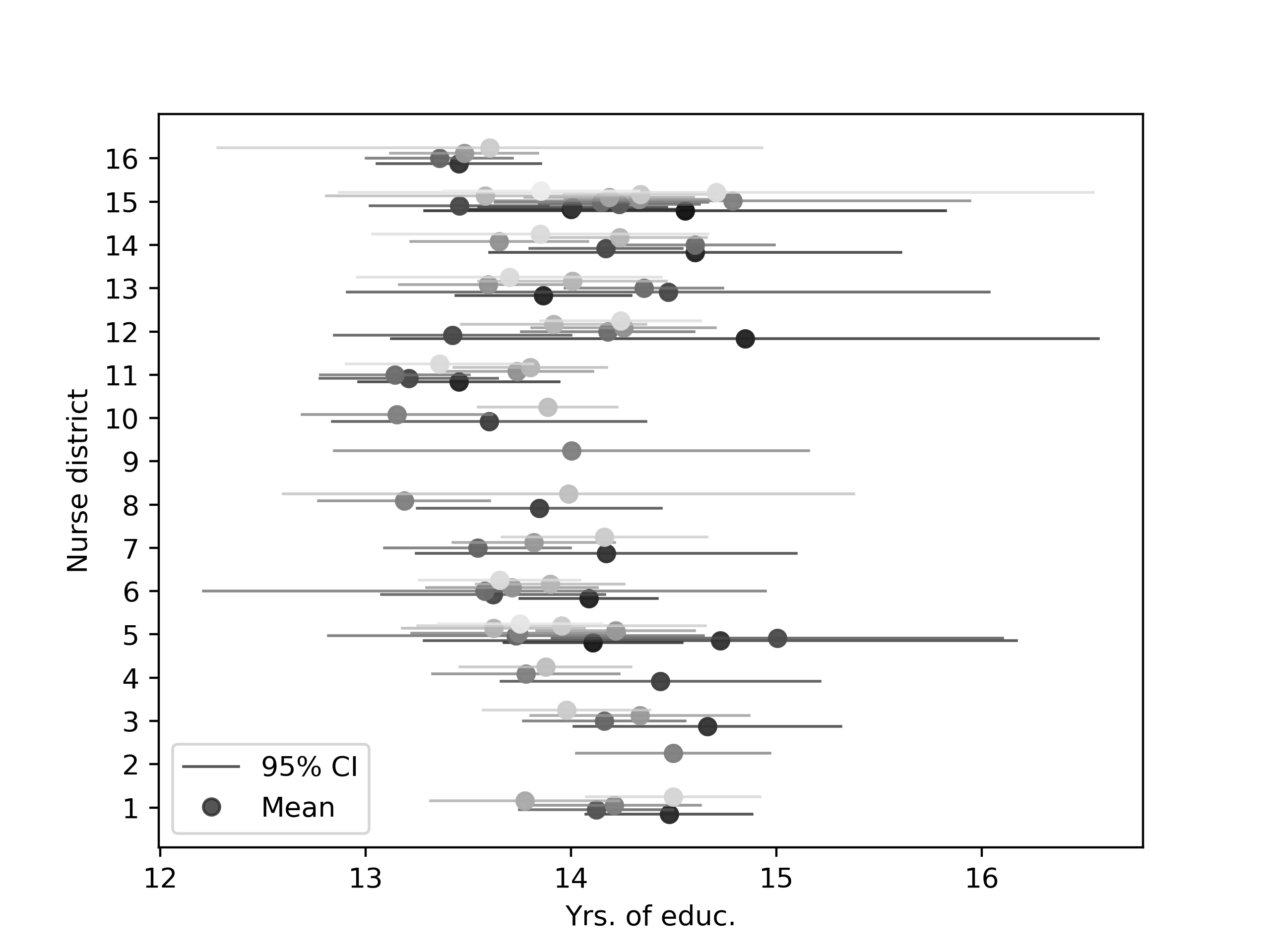}}
    \subfloat[1967]{\label{subfig: ci-plot-by-district-and-year-edulen-1967}\includegraphics[width=0.33\linewidth]{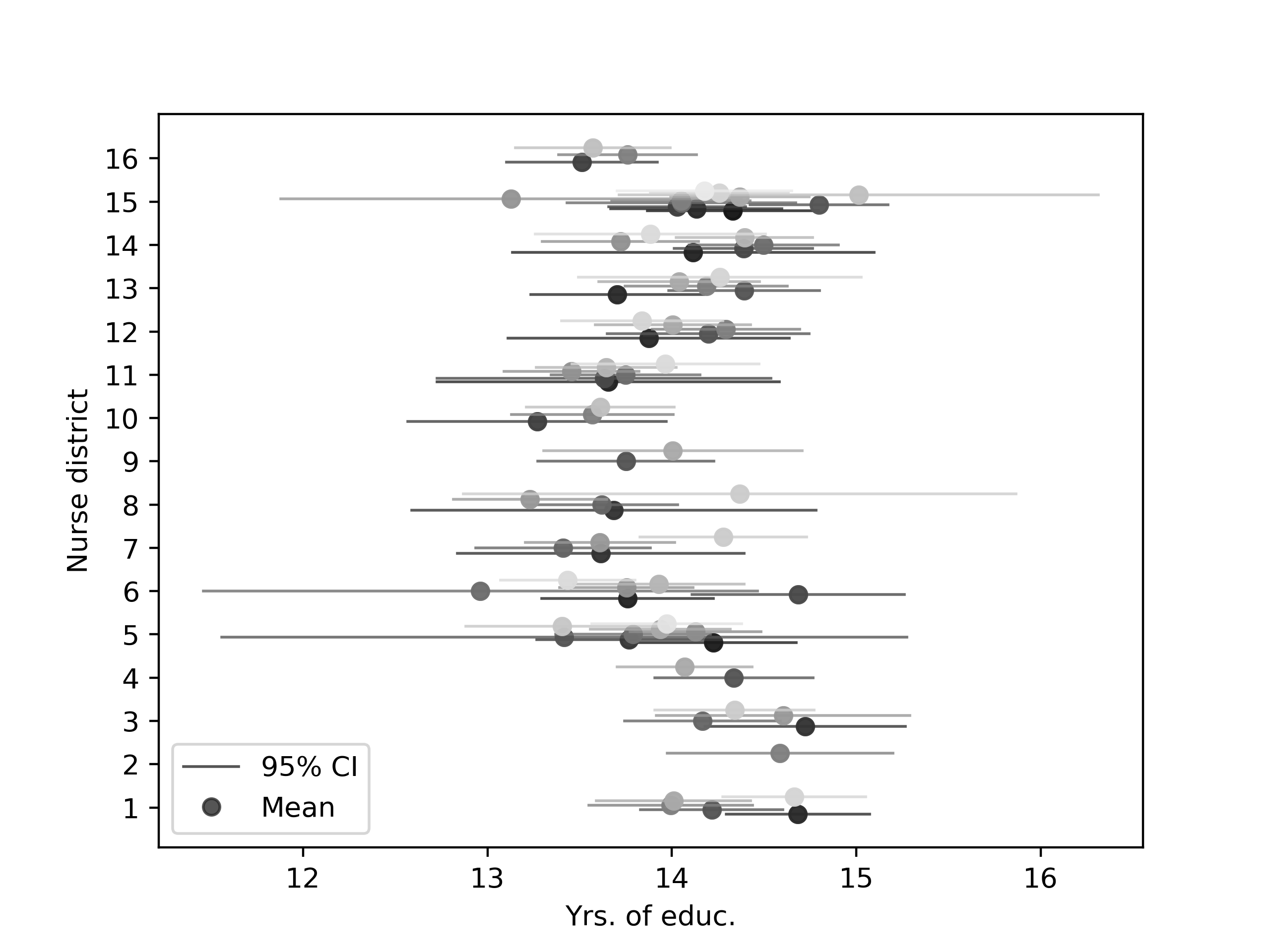}}
    
    \caption{Average Years of Education (with Confidence Intervals) by Nurses for Each Nurse District and Year.}
    \label{fig: ci-plot-by-district-and-year-edulen}
    \begin{minipage}{1\linewidth}
        \vspace{1ex}
        \footnotesize{
        \textit{Notes:}
        The figure shows plots of average years of education with associated 95\% confidence intervals (CIs) of the children by nurse for each district and year, meaning that each point represents an average of the children allocated a specific nurse born in a specific year.
        The different panels refer to different years (1959-1967) and the different rows within each panel refer to different nurse districts (1-16).
		}
	\end{minipage}
\end{figure}

\begin{figure}
    \centering
    \subfloat[1959]{\label{subfig: ci-plot-by-district-and-year-avg_inc_25_50-1959}\includegraphics[width=0.33\linewidth]{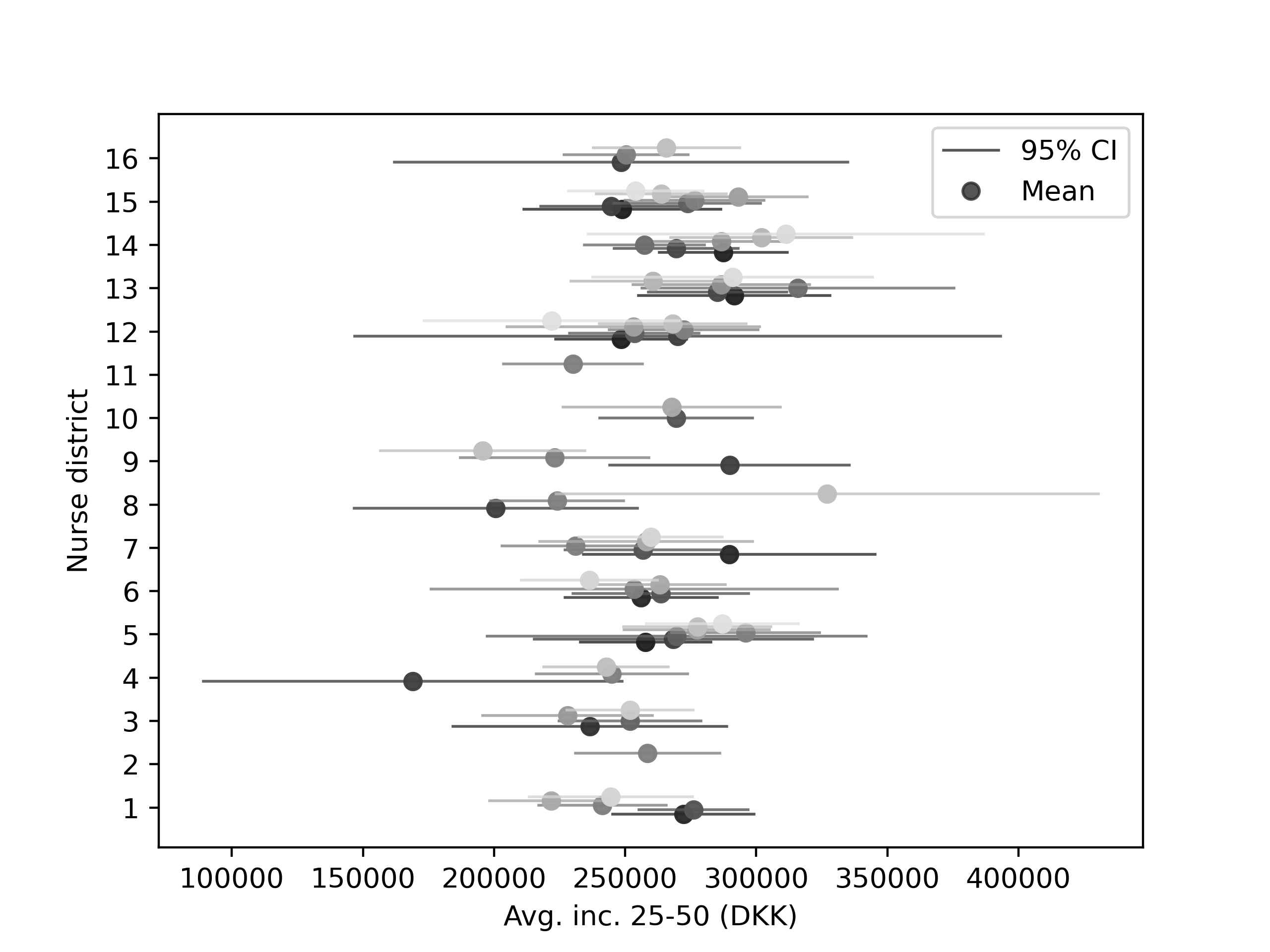}}
    \subfloat[1960]{\label{subfig: ci-plot-by-district-and-year-avg_inc_25_50-1960}\includegraphics[width=0.33\linewidth]{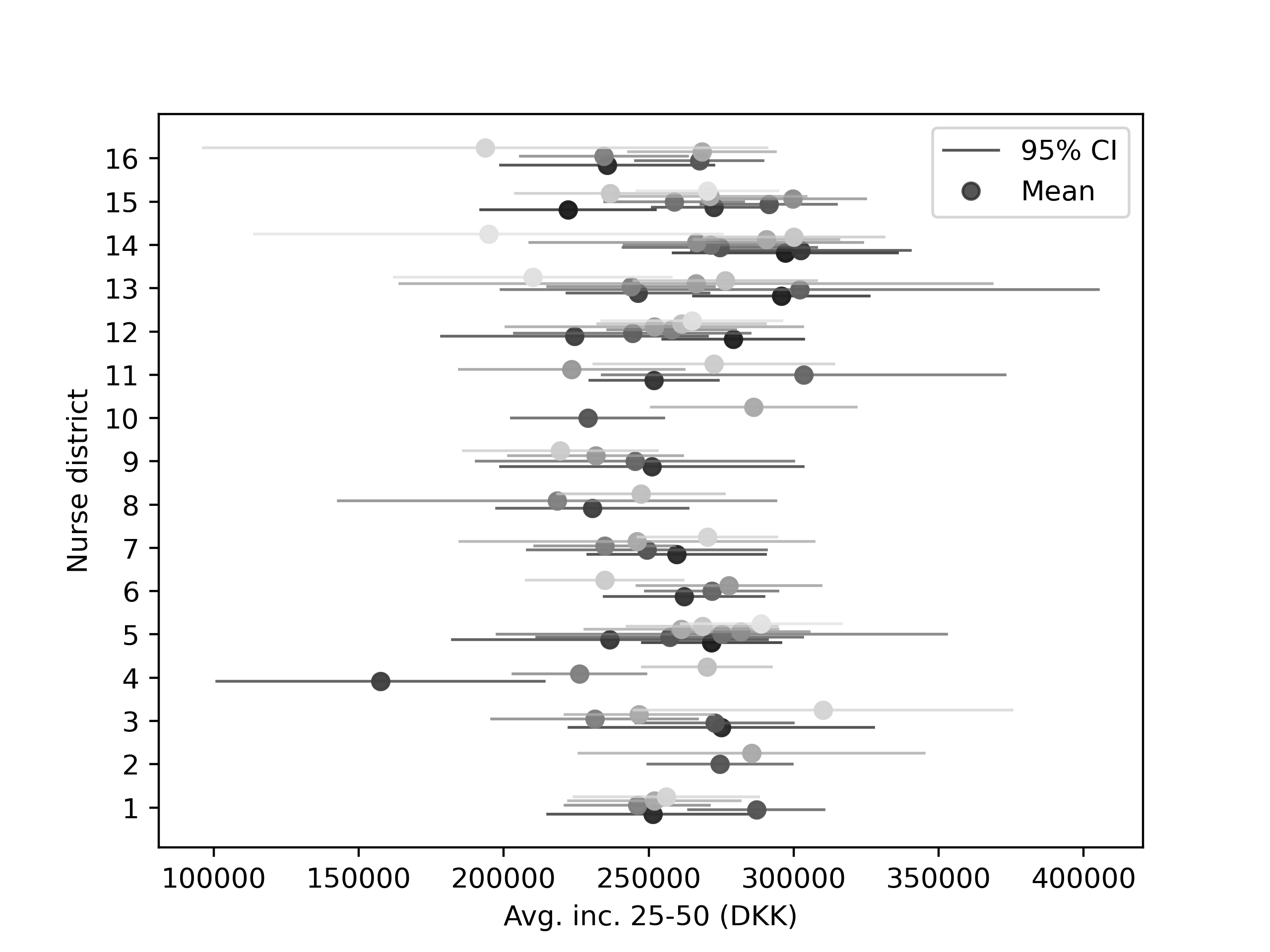}}
    \subfloat[1961]{\label{subfig: ci-plot-by-district-and-year-avg_inc_25_50-1961}\includegraphics[width=0.33\linewidth]{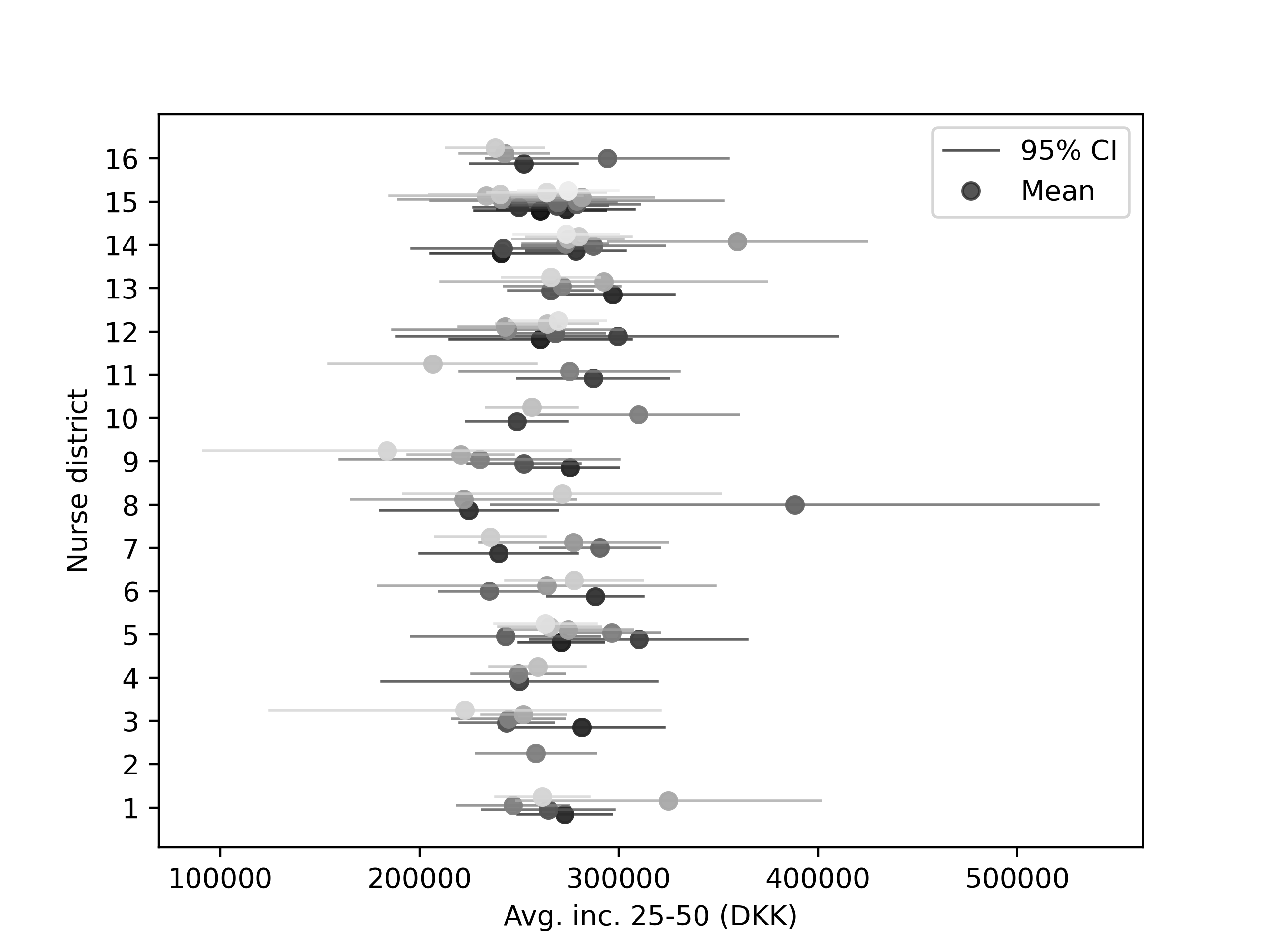}}

    \subfloat[1962]{\label{subfig: ci-plot-by-district-and-year-avg_inc_25_50-1962}\includegraphics[width=0.33\linewidth]{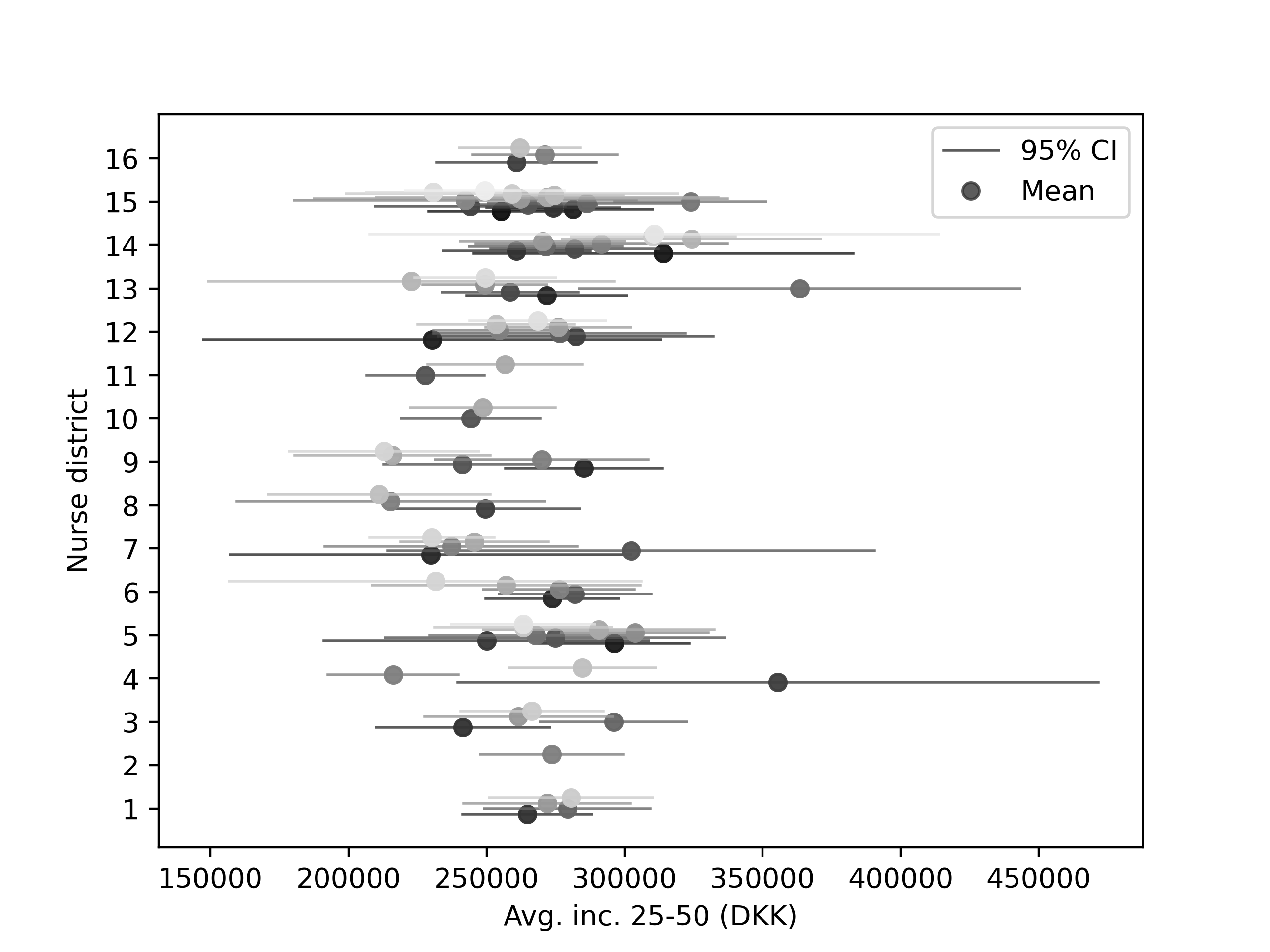}}
    \subfloat[1963]{\label{subfig: ci-plot-by-district-and-year-avg_inc_25_50-1963}\includegraphics[width=0.33\linewidth]{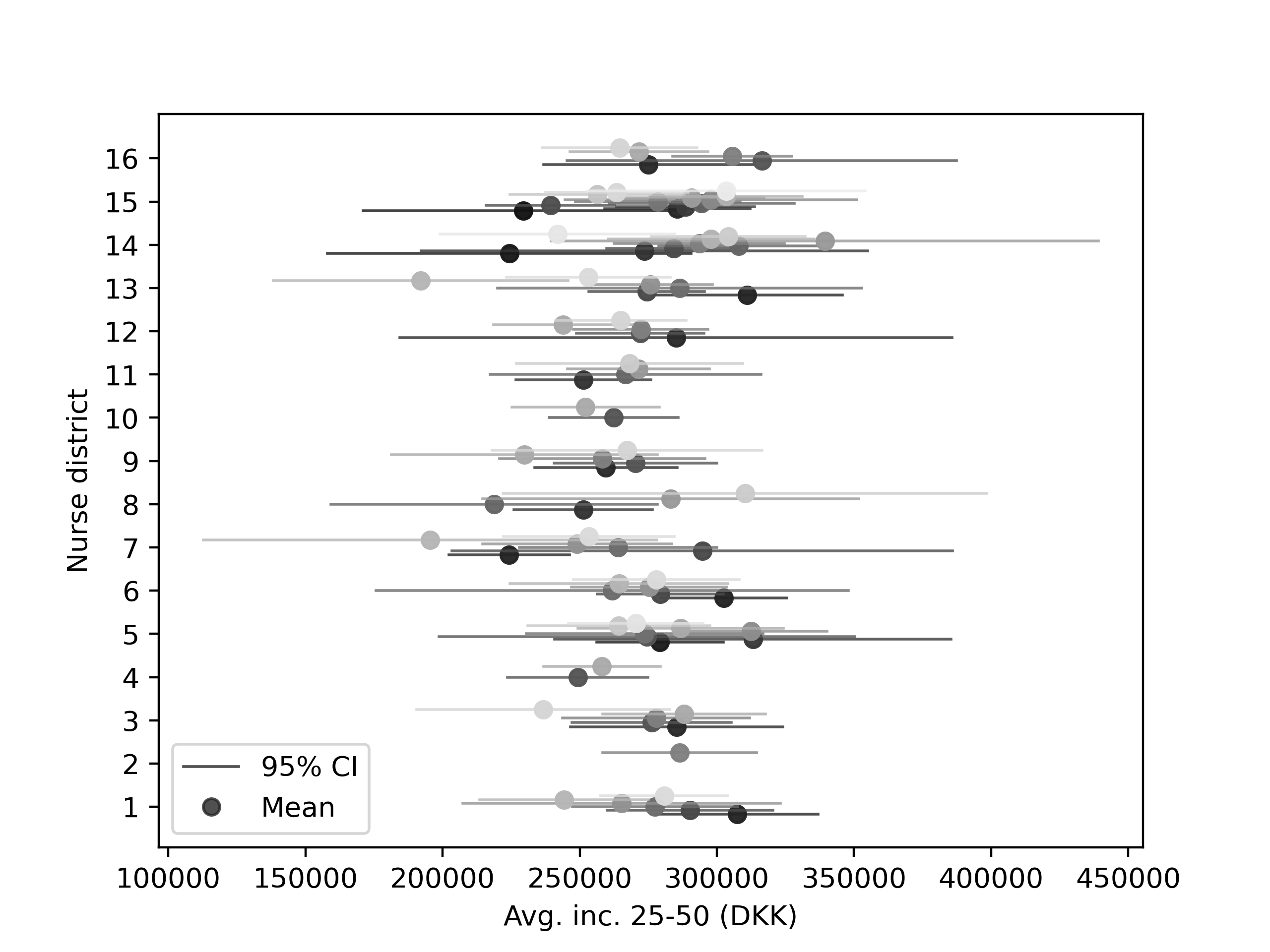}}
    \subfloat[1964]{\label{subfig: ci-plot-by-district-and-year-avg_inc_25_50-1964}\includegraphics[width=0.33\linewidth]{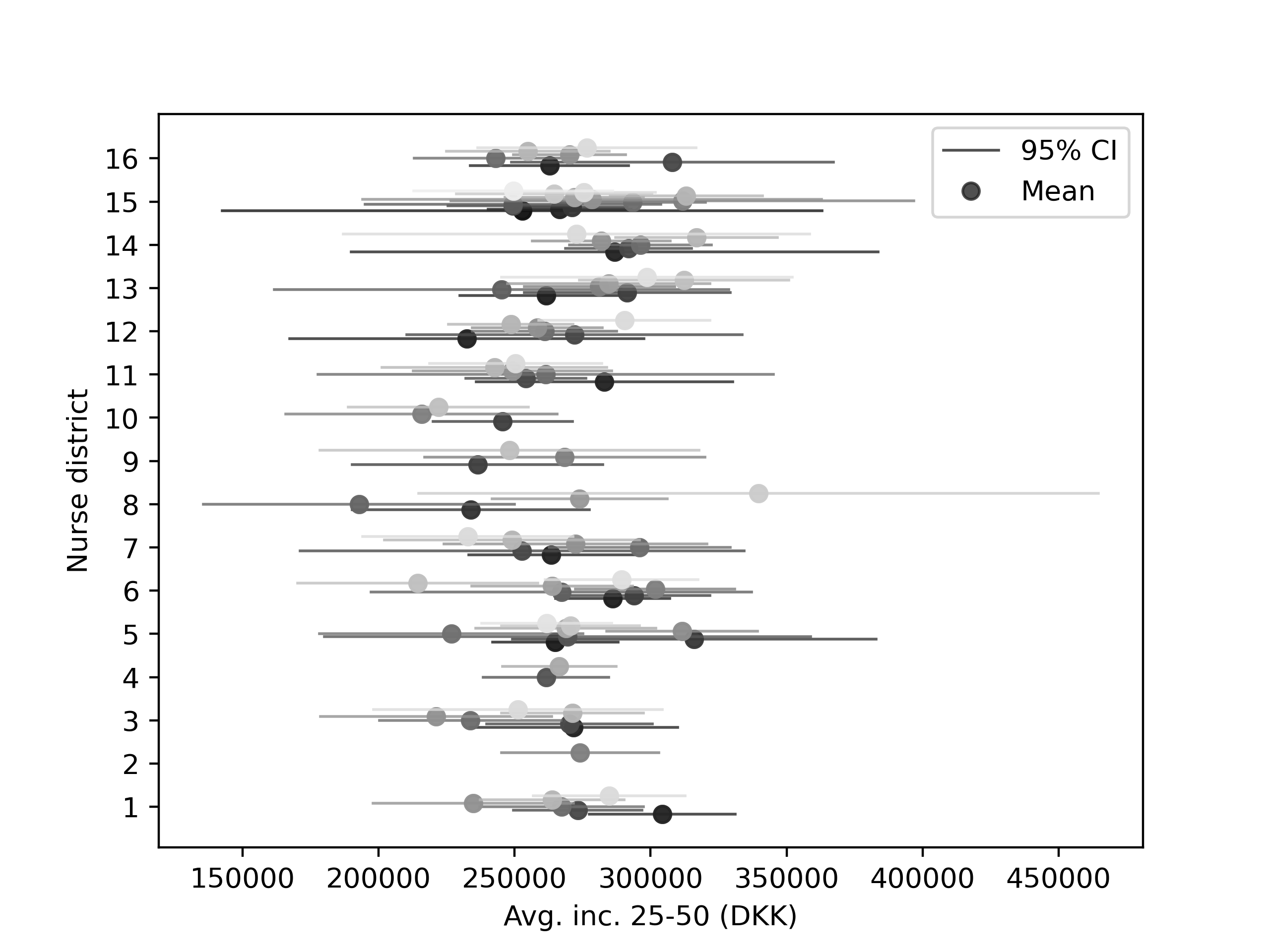}}

    \subfloat[1965]{\label{subfig: ci-plot-by-district-and-year-avg_inc_25_50-1965}\includegraphics[width=0.33\linewidth]{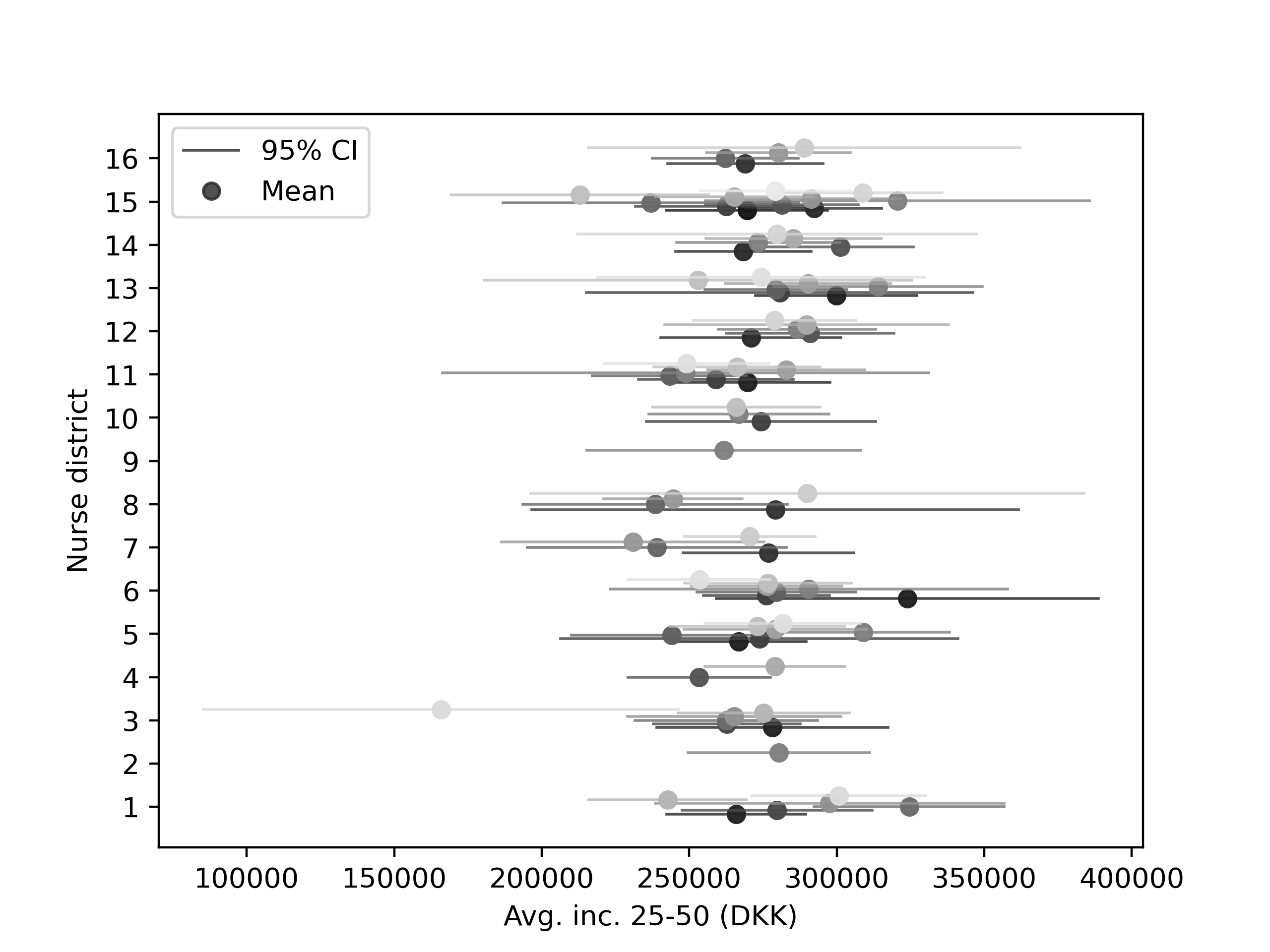}}
    \subfloat[1966]{\label{subfig: ci-plot-by-district-and-year-avg_inc_25_50-1966}\includegraphics[width=0.33\linewidth]{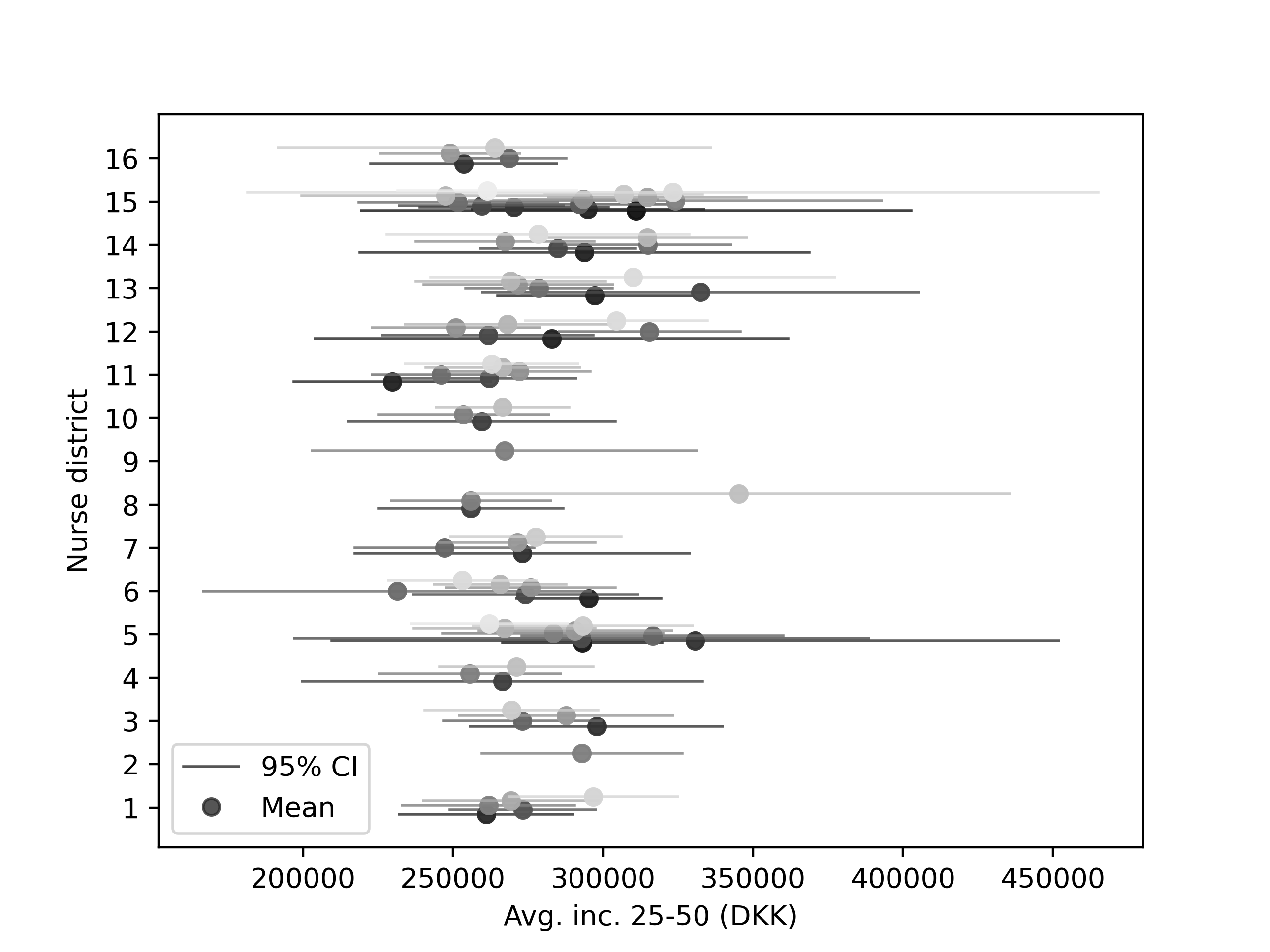}}
    \subfloat[1967]{\label{subfig: ci-plot-by-district-and-year-avg_inc_25_50-1967}\includegraphics[width=0.33\linewidth]{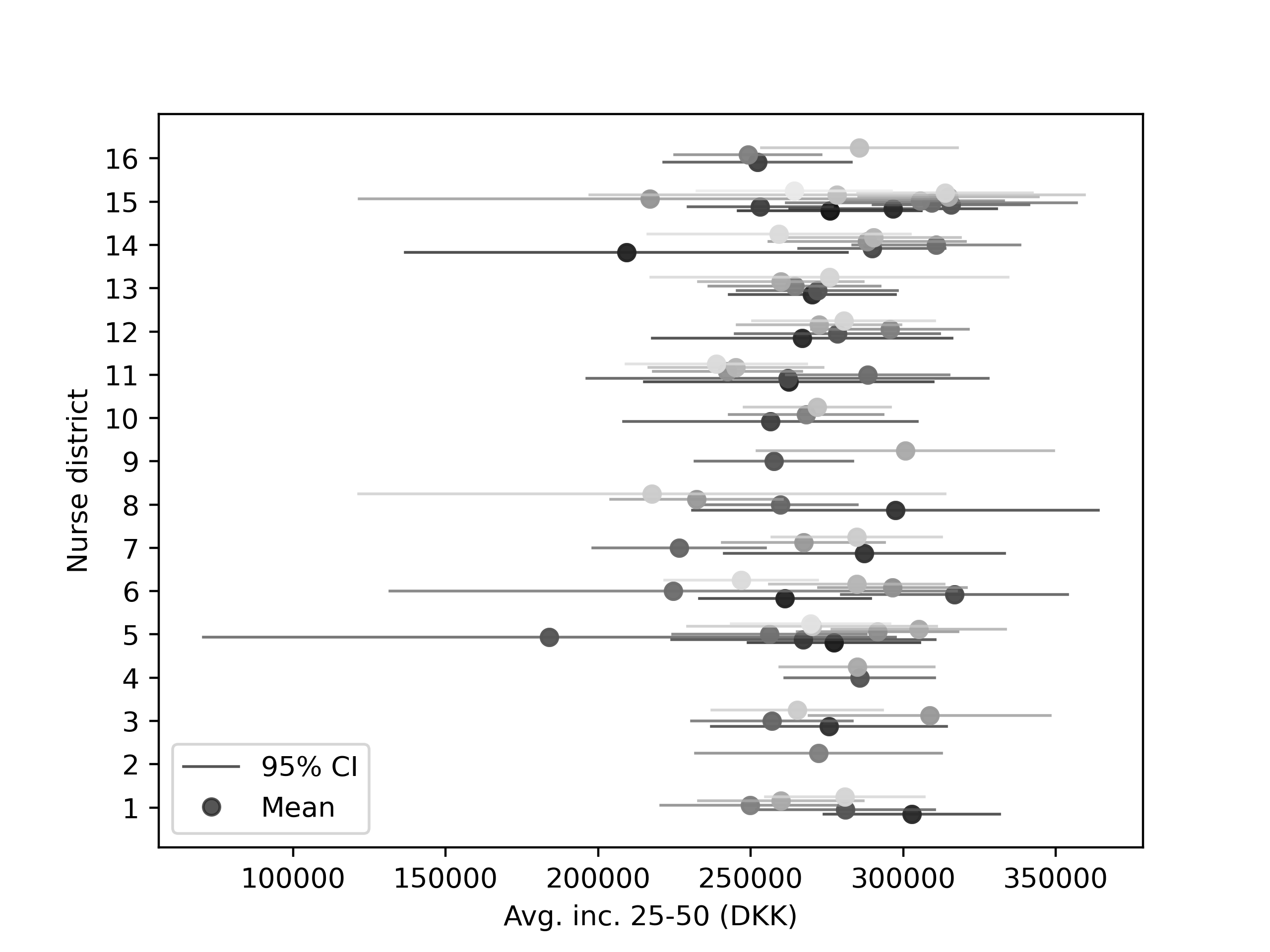}}
    
    \caption{Average Earnings During Ages 25-50 (with Confidence Intervals) by Nurses for Each Nurse District and Year.}
    \label{fig: ci-plot-by-district-and-year-avg_inc_25_50}
    \begin{minipage}{1\linewidth}
        \vspace{1ex}
        \footnotesize{
        \textit{Notes:}
        The figure shows plots of average earnings during ages 25-50 with associated 95\% confidence intervals (CIs) of the children by nurse for each district and year, meaning that each point represents an average of the children allocated a specific nurse born in a specific year.
        The different panels refer to different years (1959-1967) and the different rows within each panel refer to different nurse districts (1-16).
		}
	\end{minipage}
\end{figure}


\begin{figure}
    \centering
    \subfloat[Breastfed at 1 mo.]{\label{subfig: counterfactual-permutation-treatment-effects-breastfeeding_1_mo_pred}\includegraphics[width=0.33\linewidth]{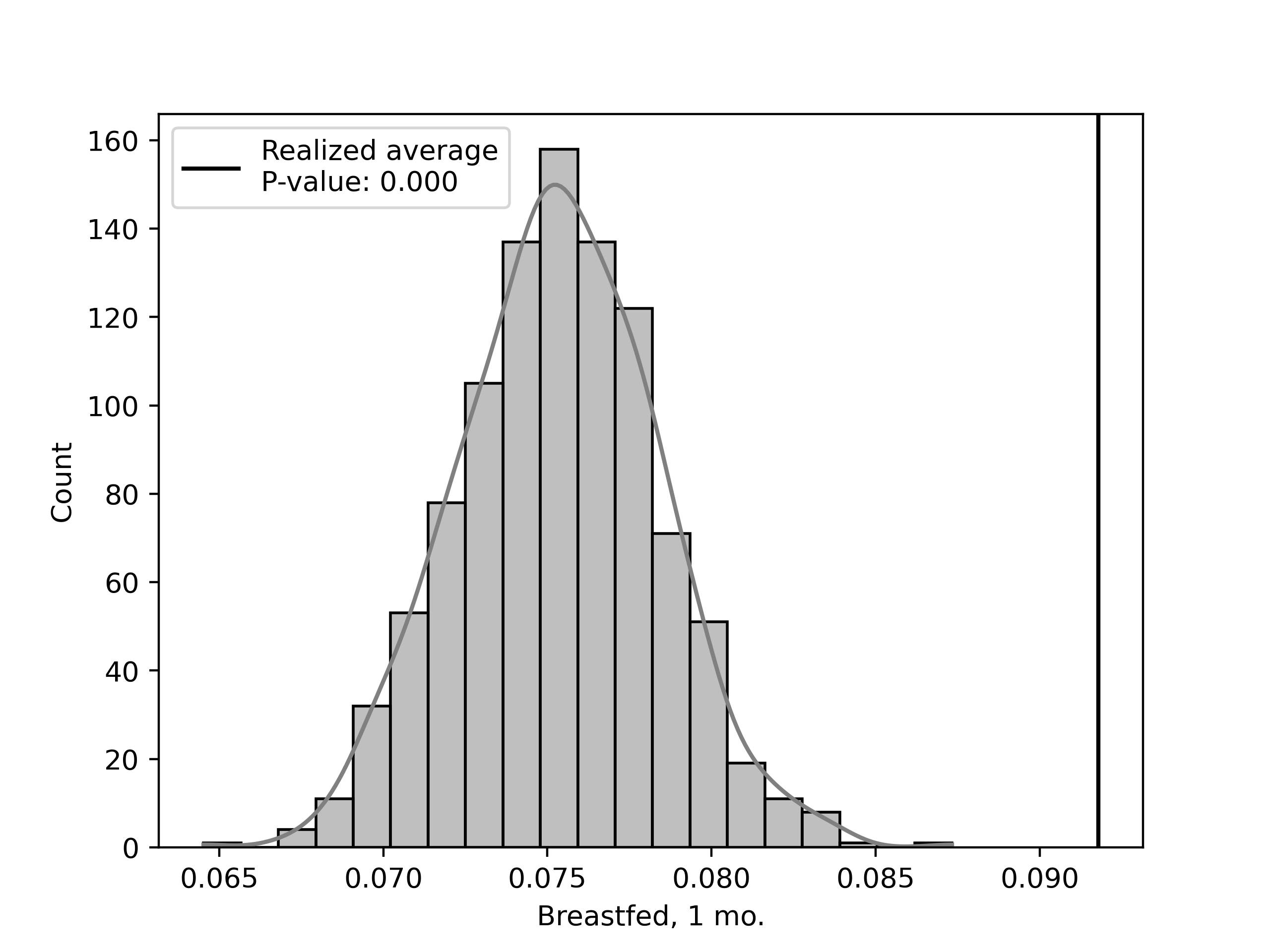}}
    \subfloat[Breastfed at 6 mo.]{\label{subfig: counterfactual-permutation-treatment-effects-breastfeeding_6_mo_pred}\includegraphics[width=0.33\linewidth]{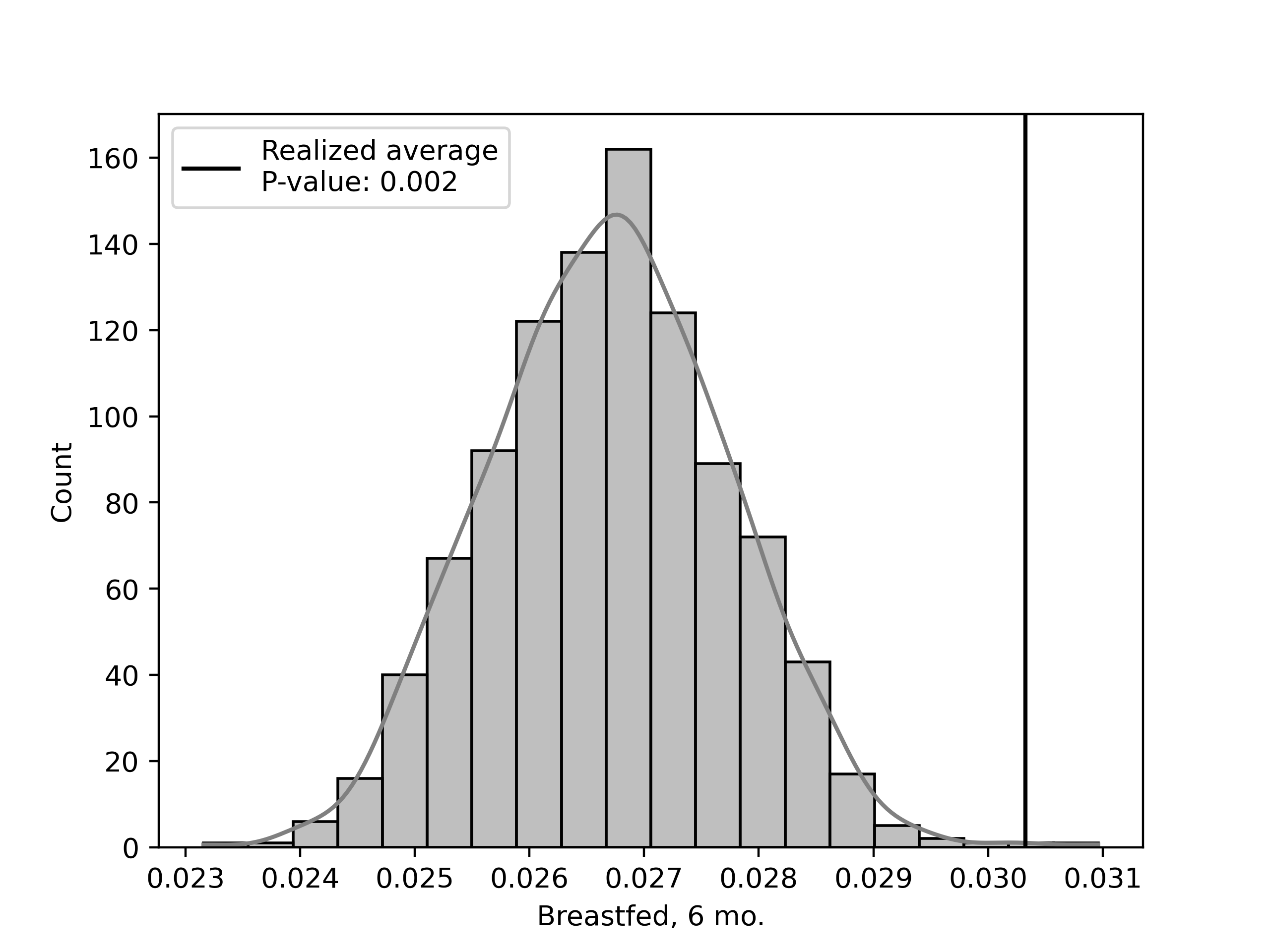}}
    \subfloat[Duration of breastfeeding (mo.)]{\label{subfig: counterfactual-permutation-treatment-effects-bfdurany_pred}\includegraphics[width=0.33\linewidth]{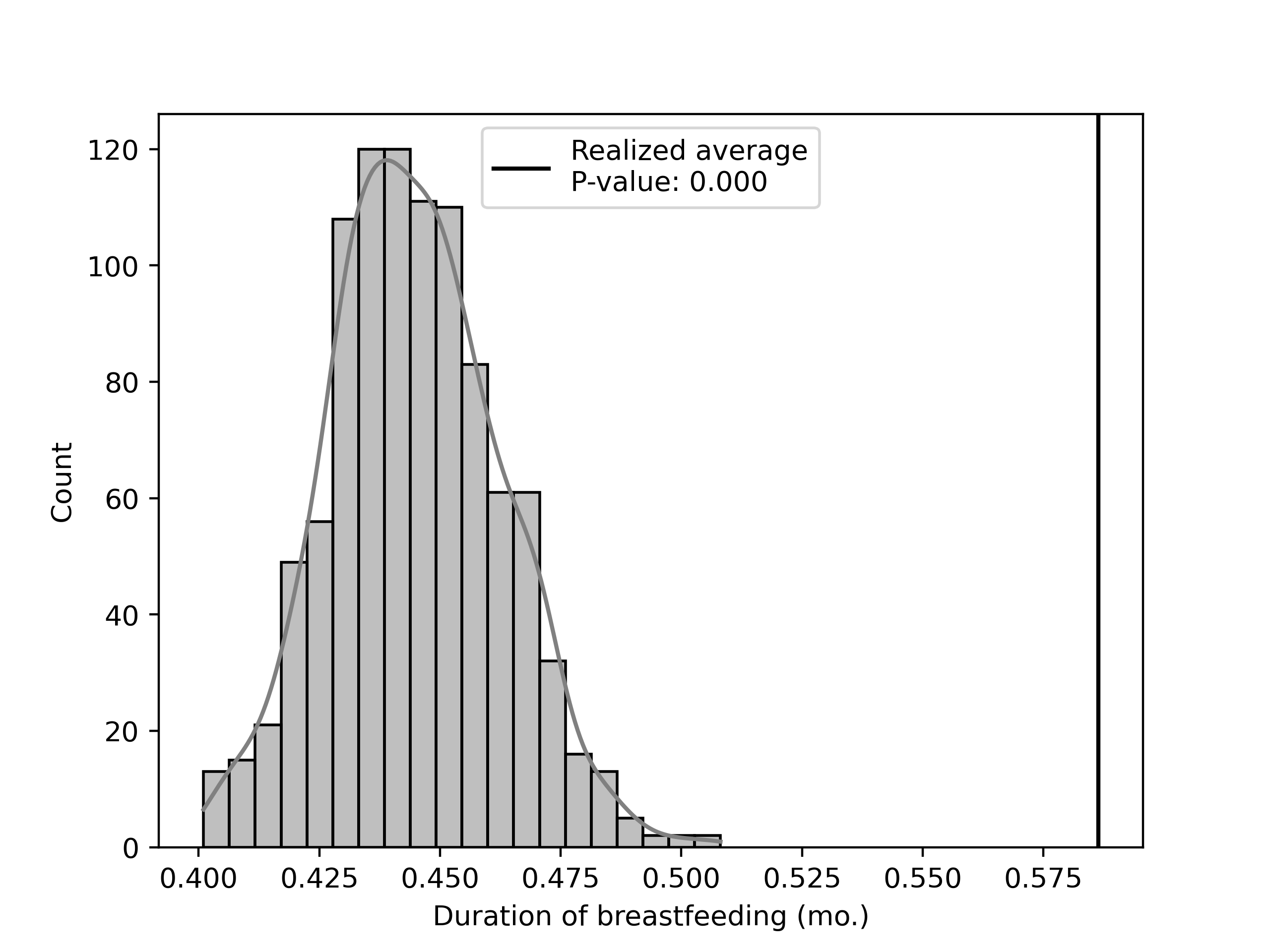}}

    \subfloat[Years of education]{\label{subfig: counterfactual-permutation-treatment-effects-edulen}\includegraphics[width=0.4\linewidth]{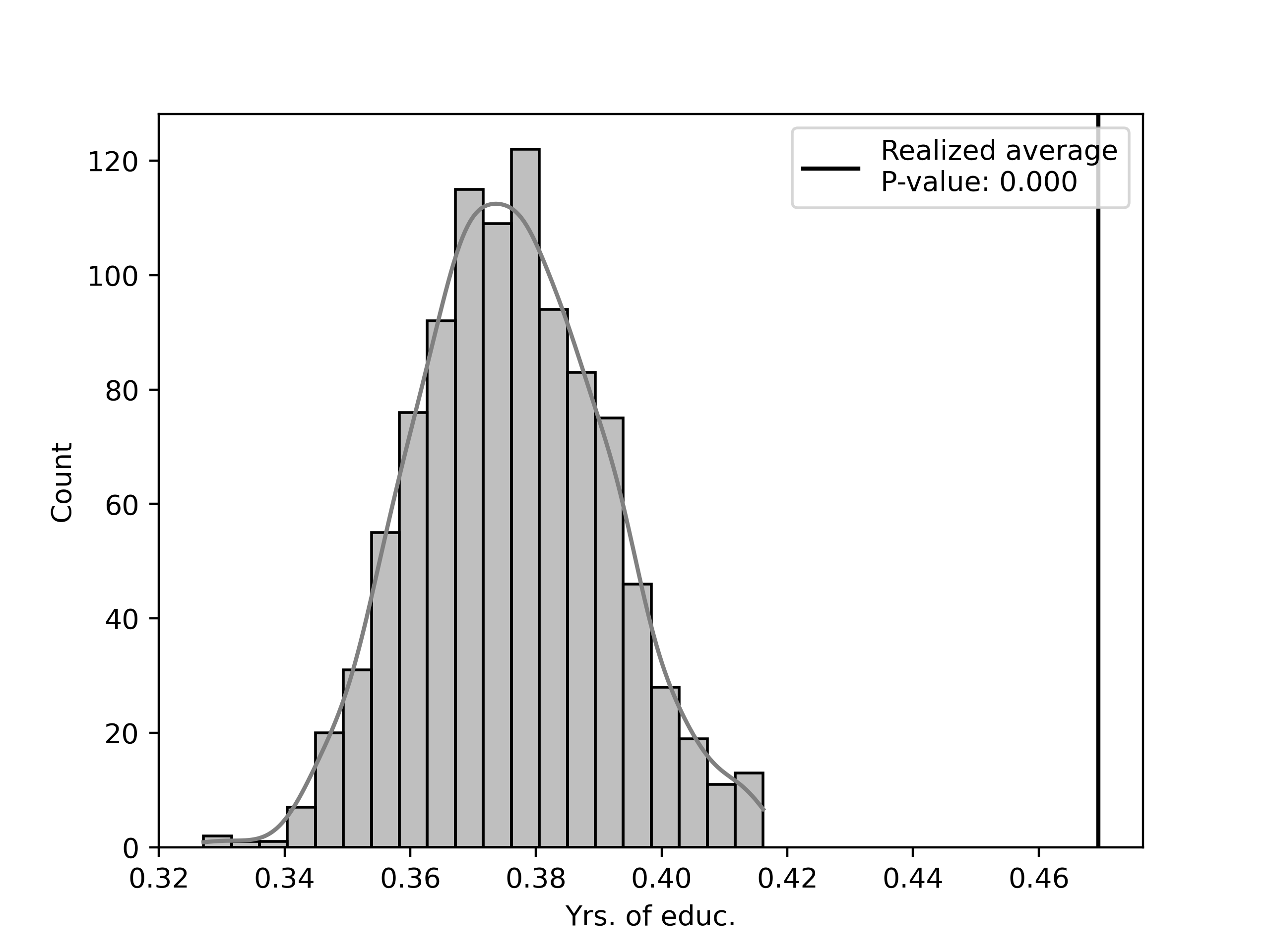}}
    \subfloat[Above mandatory education]{\label{subfig: counterfactual-permutation-treatment-effects-above_mand_edu}\includegraphics[width=0.4\linewidth]{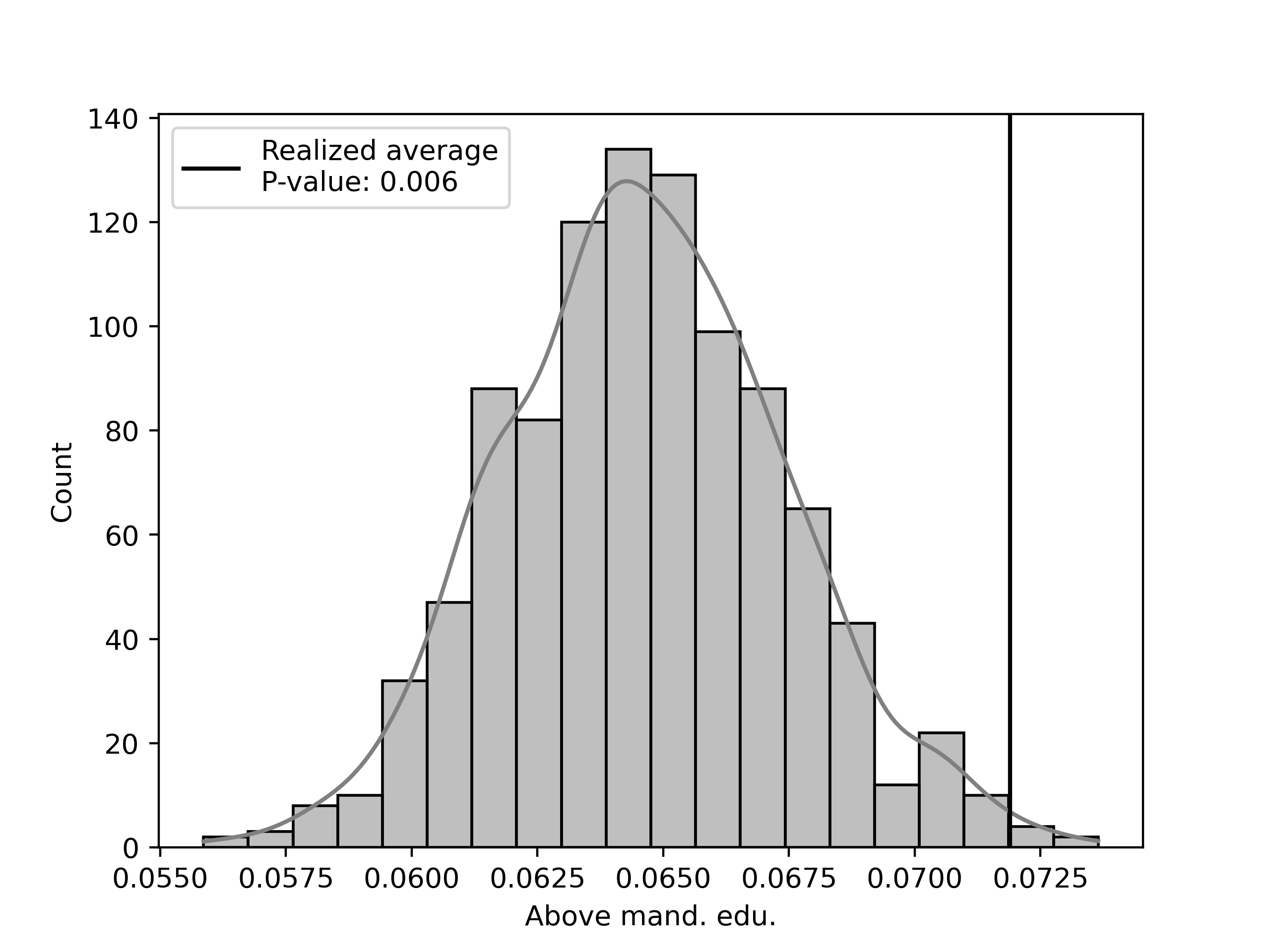}}

    \subfloat[Share of time employed 25-50]{\label{subfig: counterfactual-permutation-treatment-effects-avg_emp_25_50}\includegraphics[width=0.4\linewidth]{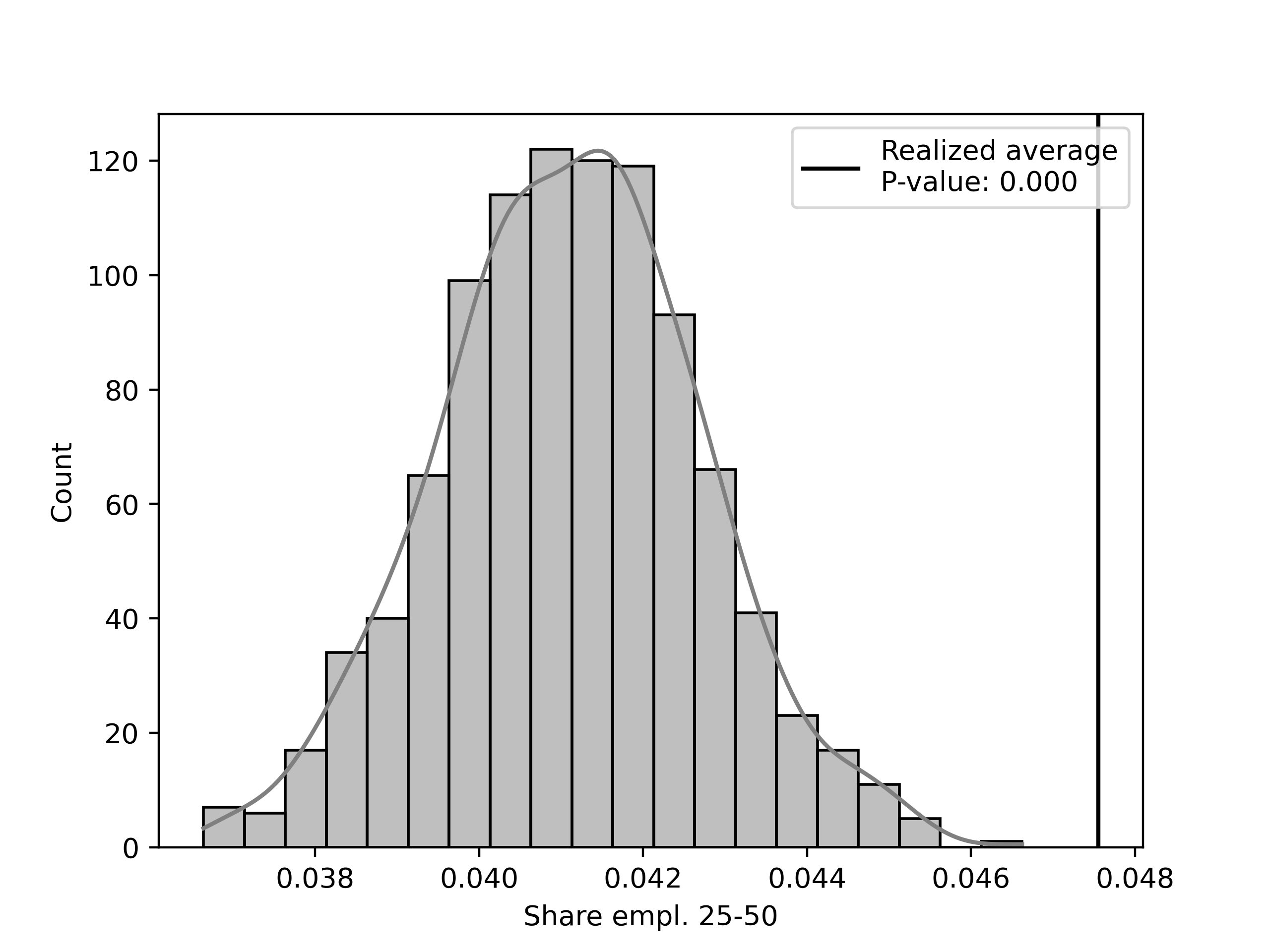}}
    \subfloat[Average income 25-50]{\label{subfig: counterfactual-permutation-treatment-effects-avg_inc_25_50}\includegraphics[width=0.4\linewidth]{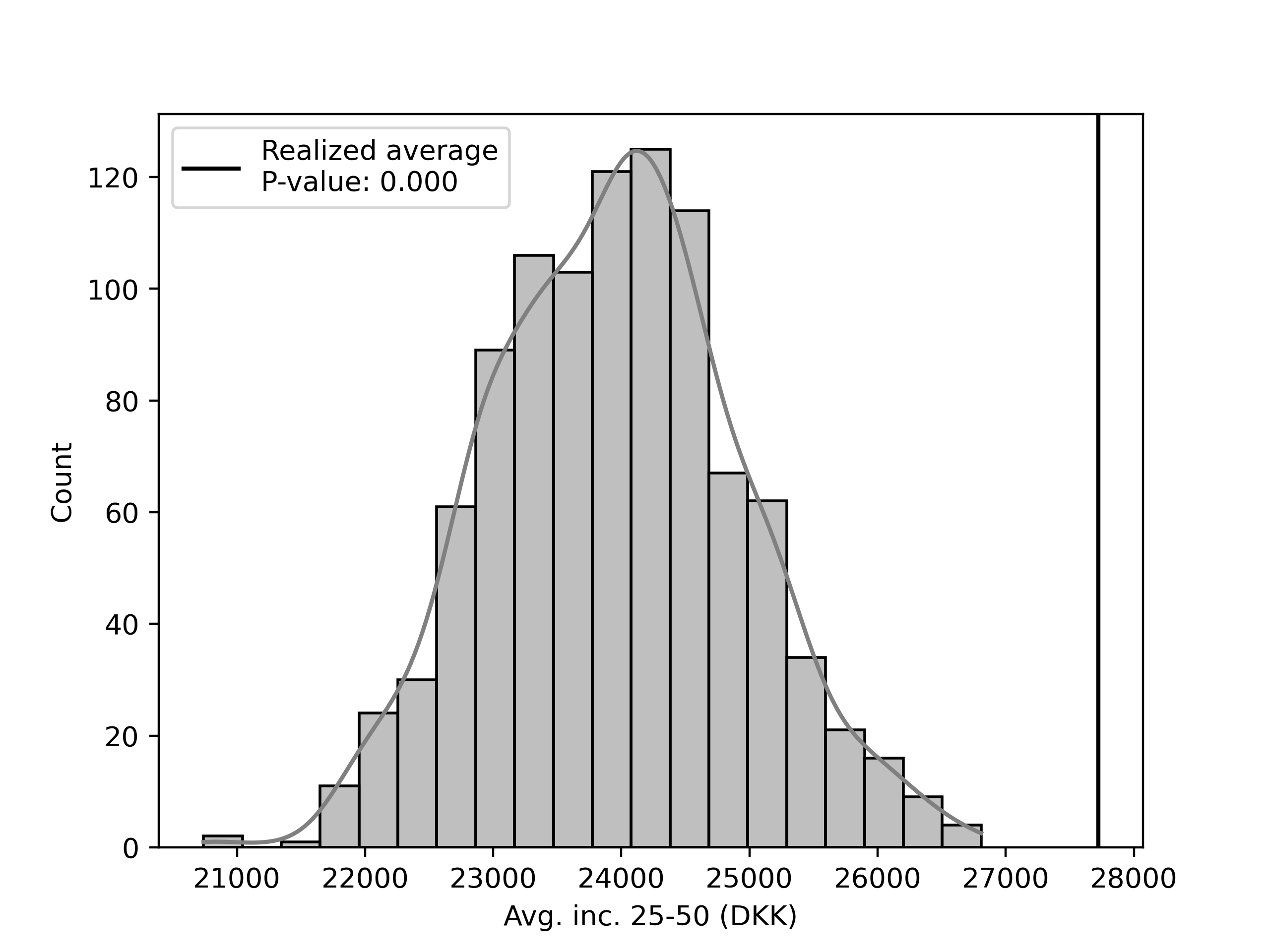}}

    \caption{Counterfactual Average Absolute Differences in Outcomes Between Permuted Children of Different Nurses Within Same District-by-Year.}
    \label{fig: counterfactual-permutation-treatment-effects}
    \begin{minipage}{1\linewidth}
        \vspace{1ex}
        \footnotesize{
        \textit{Notes:}
        The figure shows histograms (with associated density plots) of averages of absolute values of differences in average values of children with different nurses but within same district-by-year group, but where children are permuted in such a way as to be randomly allocated a nurse within the child's district-by-year group, with a total of 1000 permutations per panel.
        The black, solid line in each panel shows the average of the underlying absolute differences in the non-permuted sample, and the p-value reported in each panel is calculated as the share of permutations that led to a more extreme value than the actual average.
        The different panels refer to different outcome variables.
		}
	\end{minipage}
\end{figure}

\begin{figure}
    \centering
    \subfloat[Low BW]{\label{subfig: counterfactual-permutation-treatment-effects-lowbw}\includegraphics[width=0.4\linewidth]{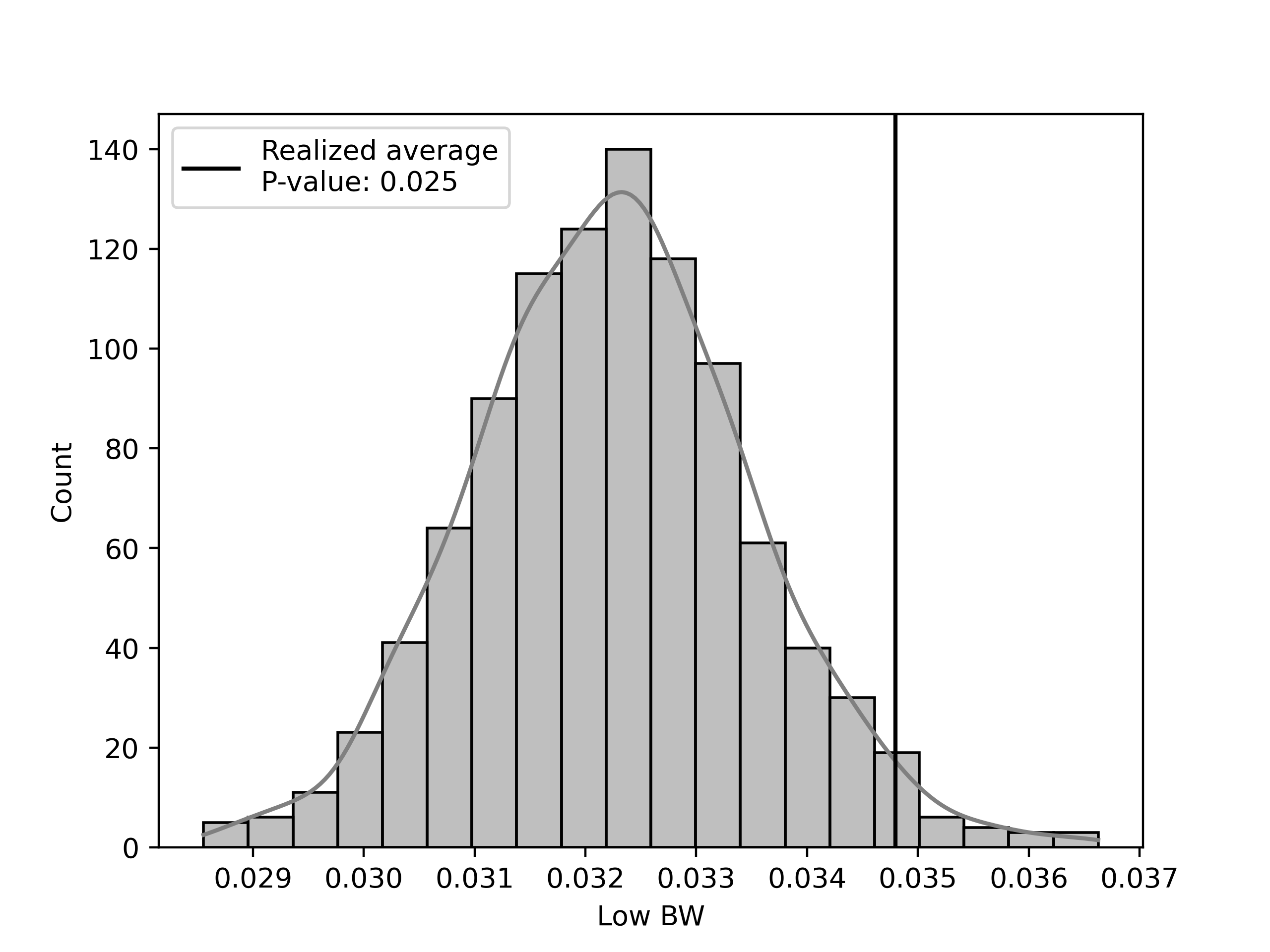}}
    \subfloat[Born prior to due date]{\label{subfig: counterfactual-permutation-treatment-effects-preterm_pred}\includegraphics[width=0.4\linewidth]{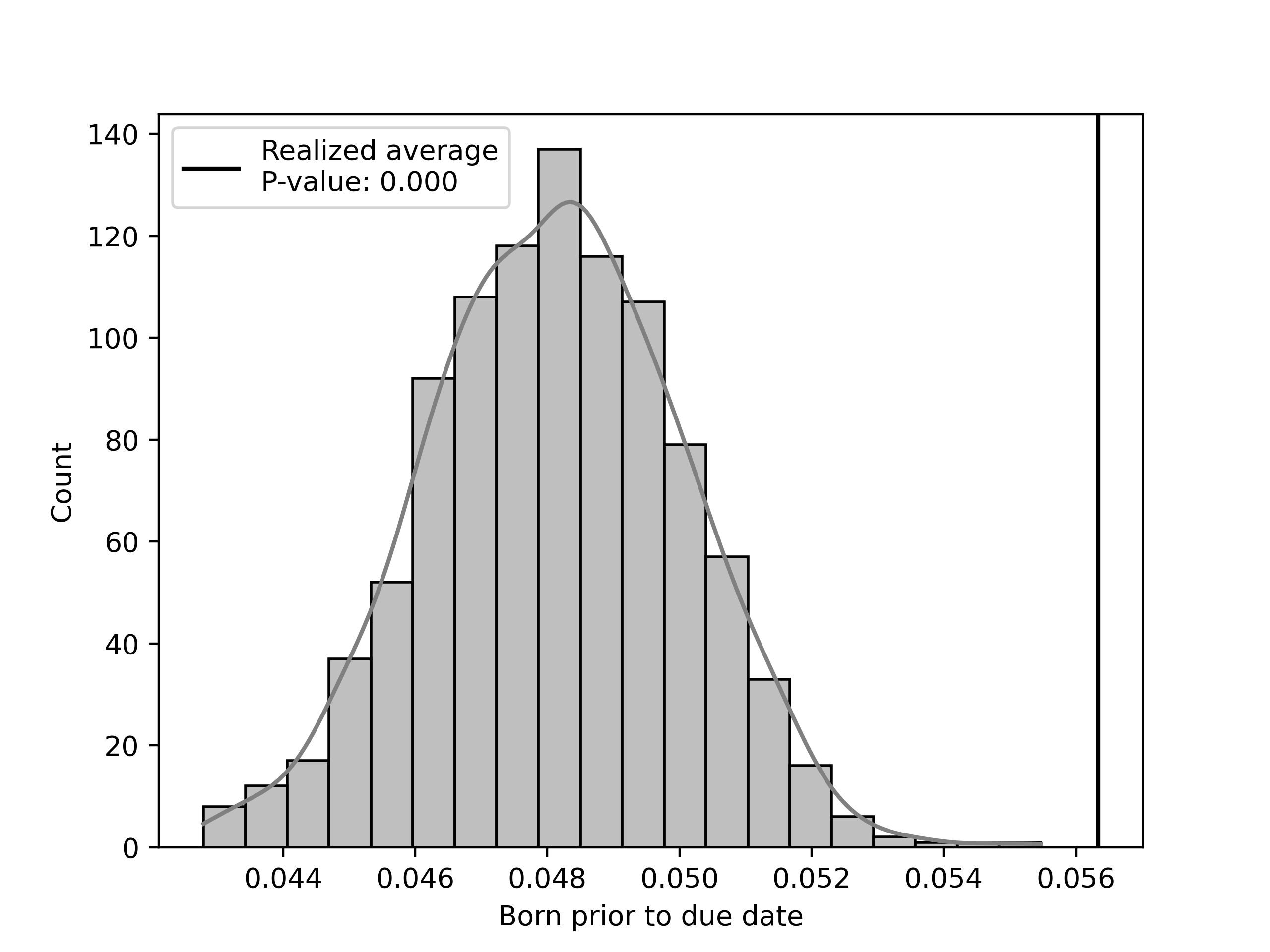}}

    \subfloat[Weeks born prior to due date]{\label{subfig: counterfactual-permutation-treatment-effects-pretermwk_pred}\includegraphics[width=0.4\linewidth]{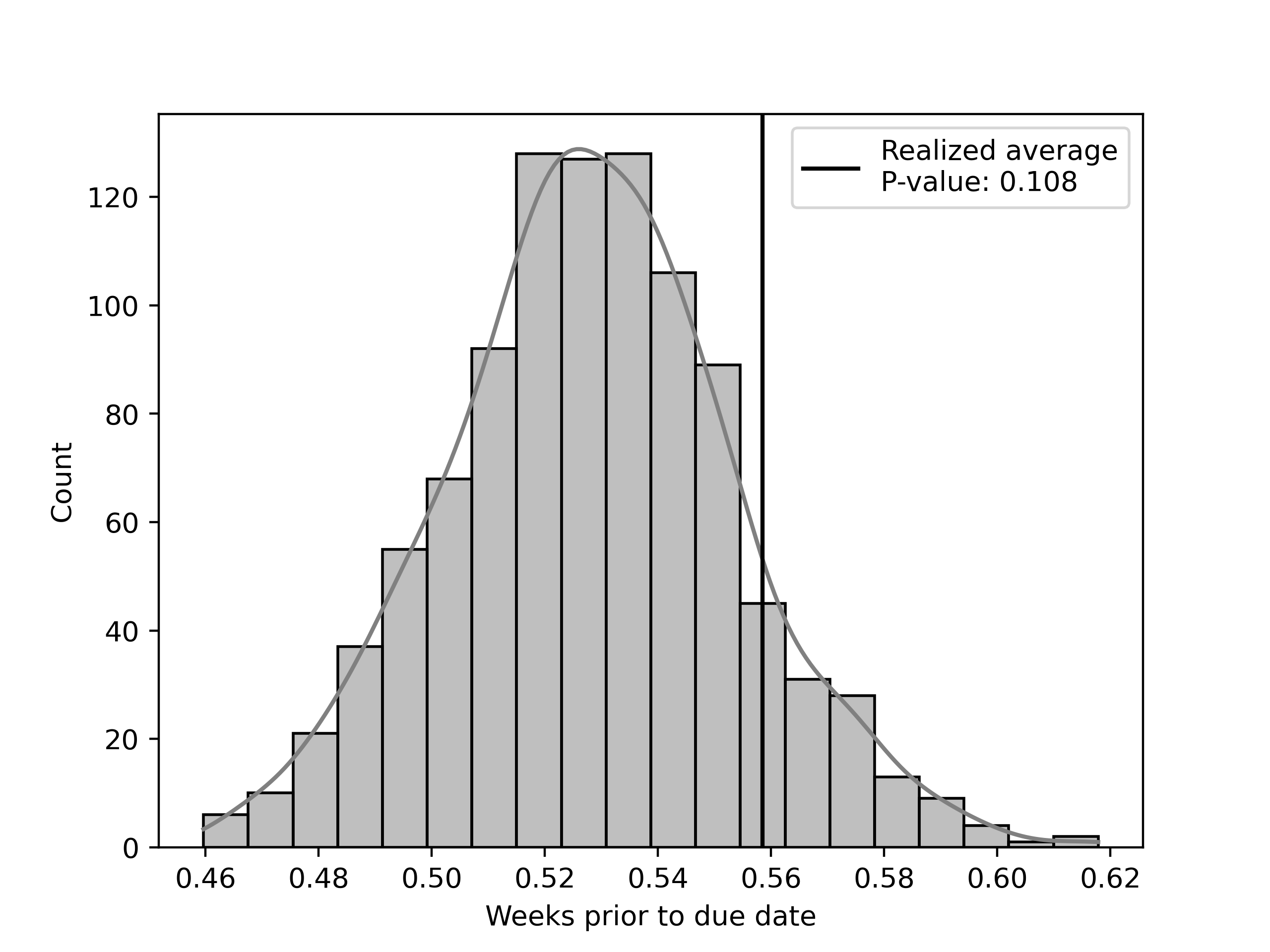}}
    \subfloat[Born 1-3]{\label{subfig: counterfactual-permutation-treatment-effects-born13}\includegraphics[width=0.4\linewidth]{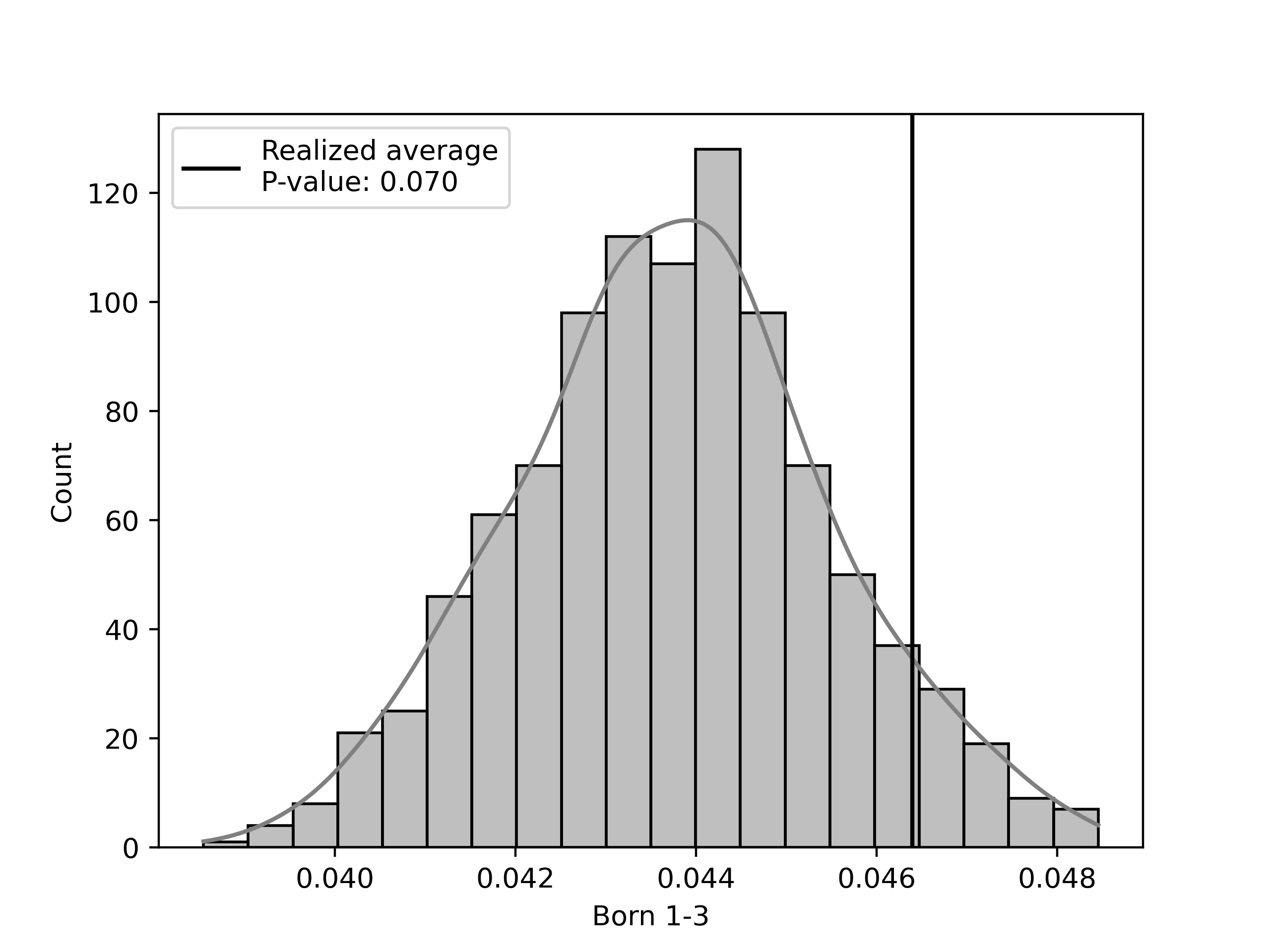}}

    \subfloat[Female]{\label{subfig: counterfactual-permutation-treatment-effects-female}\includegraphics[width=0.4\linewidth]{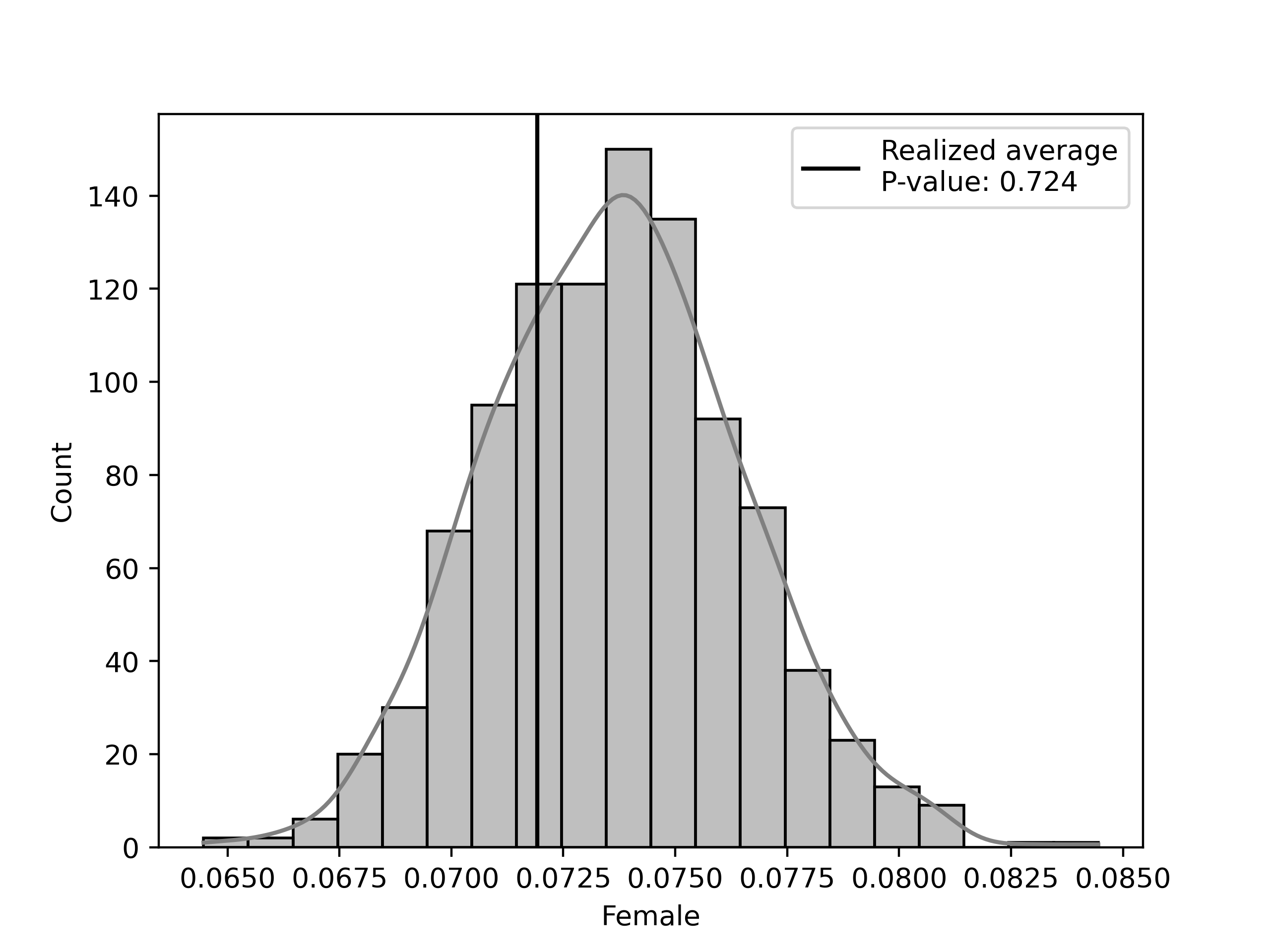}}
    \subfloat[Father Missing]{\label{subfig: counterfactual-permutation-treatment-effects-fmissing}\includegraphics[width=0.4\linewidth]{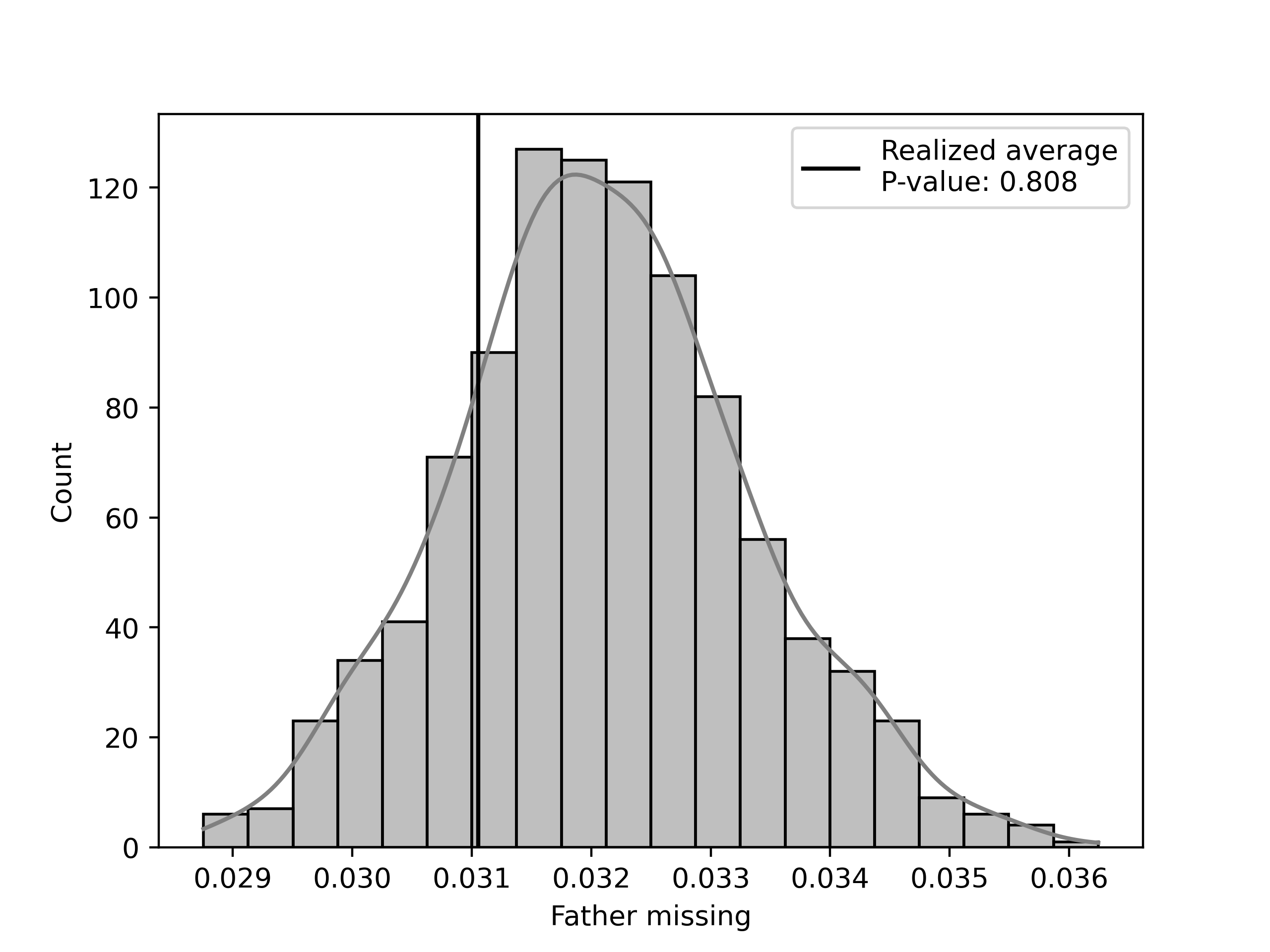}}

    \caption{Counterfactual Average Absolute Differences in Pre-Treatment Variables Between Permuted Children of Different Nurses Within Same District-by-Year.}
    \label{fig: counterfactual-permutation-treatment-effects-pre-treatment-variables}
    \begin{minipage}{1\linewidth}
        \vspace{1ex}
        \footnotesize{
        \textit{Notes:}
        The figure shows histograms (with associated density plots) of averages of absolute values of differences in average values of children with different nurses but within same district-by-year group, but where children are permuted in such a way as to be randomly allocated a nurse within the child's district-by-year group, with a total of 1000 permutations per panel.
        The black, solid line in each panel shows the average of the underlying absolute differences in the non-permuted sample, and the p-value reported in each panel is calculated as the share of permutations that led to a more extreme value than the actual average.
        The different panels refer to different pre-treatment variables.
		}
	\end{minipage}
\end{figure}

\end{document}